\documentclass[12pt]{book}

\usepackage[hmarginratio=3:2]{geometry}

\usepackage{fancyhdr}
\fancypagestyle{plain}{%
	\fancyhf{}%
	\fancyfoot[LE,RO]{\thepage}
	
}

\usepackage[nottoc]{tocbibind}
\usepackage[table]{xcolor} 
\usepackage{tikz}
\usetikzlibrary{matrix,backgrounds,fit,decorations.pathreplacing}
\usetikzlibrary{arrows.meta}
\usetikzlibrary{graphs}
\usetikzlibrary{shadows}
\usepackage[edges]{forest}
\usepackage[utf8]{inputenc} 
\usepackage{todonotes}
\usepackage{blindtext}
\usepackage{amsmath}
\usepackage{amssymb}
\usepackage{pifont}
\usepackage[pagebackref,hidelinks]{hyperref}
\usepackage{xurl}
\usepackage{dsfont}
\usepackage{placeins}
\usepackage{enumerate}
\usepackage{booktabs}
\usepackage{multirow}
\usepackage{tabularx}
\usepackage[font=small,labelfont=bf]{caption}
\usepackage{mathrsfs}
\usepackage{nicefrac}
\usepackage[a-2b,mathxmp]{pdfx}
\usepackage{algorithmicx} 
\usepackage{aligned-overset}
\usepackage[noend]{algpseudocode}
\usepackage[ruled,vlined,linesnumbered,algochapter,dotocloa,noend]{algorithm2e}

\usepackage{mathtools}
\usepackage{comment}
\usepackage{enumitem}
\usepackage{etoolbox}
\usepackage{nameref,cleveref}
\usepackage{graphicx}
\usepackage[T1]{fontenc}
\hypersetup{
	colorlinks,
	linkcolor={red!50!black},
	citecolor={blue!90!black!80},
	urlcolor={blue!80!black}
}

\usepackage{lineno}
\pgfdeclarelayer{background}
\pgfdeclarelayer{foreground}
\pgfsetlayers{background,main,foreground}

\makeatletter

\let\oldchaptermark\chaptermark
\renewcommand{\chaptermark}[1]{\oldchaptermark{\@chaptermark}}% Used stored chapter mark
\let\old@chapter\@chapter
\def\@chapter[#1]#2{%
	\def\@chaptermark{#1}% Store chapter mark
	\old@chapter[#2]{#2}%
}

\let\oldsectionmark\sectionmark
\renewcommand{\sectionmark}[1]{\oldsectionmark{\@sectionmark}}% Used stored section mark
\let\oldsubsectionmark\subsectionmark
\renewcommand{\subsectionmark}[1]{\oldsubsectionmark{\@subsectionmark}}% Used stored subsection mark
\let\old@sect\@sect
\def\@sect#1#2#3#4#5#6[#7]#8{%
	\@namedef{@#1mark}{#7}% Store sectional mark
	\old@sect{#1}{#2}{#3}{#4}{#5}{#6}[#8]{#8}% 
}

\renewcommand*{\backrefalt}[4]{%
	\ifcase #1 %
	No citations.%
	\or
	(Cited on~p.~#2)%
	\else
	(Cited on~pp.~#2)%
	\fi
}
\setitemize{itemsep=0cm}

\usepackage{amsthm}
\theoremstyle{definition}
\newtheorem{definition}{Definition}[chapter]
\theoremstyle{plain}
\newtheorem{theorem}{Theorem}[chapter]
\newtheorem{proposition}[definition]{Proposition}
\newtheorem{observation}[definition]{Observation}

\newtheorem{lemma}[definition]{Lemma}
\newtheorem{corollary}[definition]{Corollary}
\newtheorem{rr}[definition]{Reduction Rule}

\newtheorem{construction}[definition]{Construction}

\usepackage{thmtools}

\newenvironment{propEnum}{%
	\enumerate
}{%
	\endenumerate
}

\newenvironment{claimproof}
  {\noindent \textit{Proof.} }
  {\hfill$\diamond$ \\ \goodbreak}

\newcommand{\kommentar}[1]{}
\newcommand{\Oh}{\ensuremath{\mathcal{O}}}

\newcommand{\proofpara}[1]{\smallskip
	
	\noindent\textit{#1.}}

\DeclareMathOperator{\poly}{poly}
\DeclareMathOperator{\twwithoutN}{{tw}}
\DeclareMathOperator{\vretwithoutN}{{ret}}
\DeclareMathOperator{\eretwithoutN}{{e-ret}}
\DeclareMathOperator{\off}{off}
\DeclareMathOperator{\desc}{desc}
\DeclareMathOperator{\anc}{anc}
\DeclareMathOperator{\DP}{DP}
\DeclareMathOperator{\var}{var}
\DeclareMathOperator{\src}{sources}
\DeclareMathOperator{\wcode}{w-code}
\DeclareMathOperator{\height}{height}

\DeclareMathOperator{\Costs}{Cost}
\newcommand{\numP}{\left\lVert \mathcal{P}\right\rVert}
\newcommand{\Dbar}{{\ensuremath{\overline{D}}}\xspace}
\newcommand{\kbar}{{\ensuremath{\overline{k}}}\xspace}
\newcommand{\w}{{\ensuremath{\lambda}}\xspace}
\newcommand\lb{\linebreak}

\newcommand\Recc[1]{Recurrence~(\ref{#1})}

\newcommand{\predators}[1]{{\ensuremath{N_{>}(#1)}}\xspace}
\newcommand{\prey}[1]{{\ensuremath{N_{<}(#1)}}\xspace}
\newcommand{\sources}{\sourcespersonal{\Food}}
\newcommand{\sourcespersonal}[1]{{\ensuremath{\src(#1)}}\xspace}

\newcommand{\yes}{{\normalfont\texttt{yes}}\xspace}
\newcommand{\no}{{\normalfont\texttt{no}}\xspace}

\newcommand{\Wh}[1]{{\normalfont\texttt W[#1]}\xspace}
\newcommand{\NP}{{\normalfont\texttt{NP}}\xspace}
\renewcommand{\P}{{\normalfont\texttt{P}}\xspace}

\newcommand{\FPT}{{\normalfont\texttt{FPT}}\xspace}
\newcommand{\PFPT}{{\normalfont\texttt{PFPT}}\xspace}
\newcommand{\SETH}{{\normalfont\texttt{SETH}}\xspace}
\newcommand{\ETH}{{\normalfont\texttt{ETH}}\xspace}
\newcommand{\XP}{{\normalfont\texttt{XP}}\xspace}
\newcommand{\PeqNP}{\mbox{\P = \NP}\xspace}
\newcommand{\PneqNP}{\mbox{\P $\neq$ \NP}\xspace}
\newcommand{\NPcoNPpoly}{{\normalfont\texttt{NP~$\not\subseteq$~coNP/poly}}\xspace}

\newcommand{\mcal}{\mathcal}
\newcommand{\mbb}{\mathbb}
\newcommand{\Instance}{{\ensuremath{\mathcal{I}}}\xspace}
\newcommand{\Net}{{\ensuremath{\mathcal{N}}}\xspace}
\newcommand{\Tree}{{\ensuremath{\mathcal{T}}}\xspace}
\newcommand{\Food}{{\ensuremath{\mathcal{F}}}\xspace}
\newcommand{\ret}{{\ensuremath{\vretwithoutN_{\Net}}}\xspace}
\newcommand{\vret}{\ret}
\newcommand{\eret}{{\ensuremath{\eretwithoutN_{\Net}}}\xspace}
\newcommand{\tw}{{\ensuremath{\twwithoutN_{\Food}}}\xspace}
\newcommand{\twN}{{\ensuremath{\twwithoutN_{\Net}}}\xspace}
\newcommand{\PD}{\PDsub\Tree}
\newcommand{\PDsub}[1]{{\ensuremath{\text{PD}_{#1}}}\xspace}

\newcommand{\SstarPrime}[1]{\ensuremath{\mathcal{S}^*_{#1}}\xspace}
\newcommand{\Sstar}{\SstarPrime{t,R,G,B,s}}

\newcommand{\spannbaum}[1]{\spannbaumsub{\Tree}{#1}}
\newcommand{\spannbaumsub}[2]{\ensuremath{#1\langle #2 \rangle}\xspace}

\newcommand{\problemdef}[3]{
	\bigskip
	
	\noindent
	\parbox{0.9\textwidth}{
		\normalsize\textsc{#1} \smallskip \\ \nopagebreak[4]
		\begin{tabularx}{\textwidth}{@{}l@{\hspace{3pt}}X}
			\normalsize\textbf{Input:}    & {\normalsize #2}.\\
			\normalsize\textbf{Question:} & {\normalsize #3}?
		\end{tabularx}
	}
	\bigskip\\
}

\newcommand{\PROB}[1]{{{\normalfont\textsc{#1}}}\xspace}

\newcommand{\VC}{\PROB{Vertex Cover}}
\newcommand{\SC}{\PROB{Set Cover}}
\newcommand{\KP}{\PROB{Knapsack}}
\newcommand{\MCKP}{\PROB{MCKP}}
\newcommand{\MCKPLong}{\PROB{Multiple-Choice Knapsack}}
\newcommand{\SubSum}{\PROB{Subset Sum}}

\newcommand{\msSubSum}{\PROB{Multi-Selectable \SubSum}}
\newcommand{\SubProd}{\PROB{Subset Product}}
\newcommand{\XTClong}{\PROB{Exact Cover by 3-Sets}}
\newcommand{\XTC}{\PROB{X3C}}

\newcommand{\HS}{\PROB{Hitting Set}}
\newcommand{\DS}{\PROB{Dominating Set}}

\newcommand{\ILPF}{\PROB{ILP-Feasibility}}
\newcommand{\rbnb}{\PROB{Red-Blue Non-Blocker}}
\newcommand{\MCNF}{\PROB{MCNF}}
\newcommand{\MCNFlong}{\PROB{Minimum-Cost Network Flow}}
\newcommand{\qSAT}[1]{\PROB{$#1$-SAT}}
\newcommand{\SAT}{\PROB{Satisfiability}}
\newcommand{\WCS}[2]{\PROB{Weighted Circuit Satisfiability over $\mathcal{C}_{#1,#2}$}}

\newcommand{\MPD}{\PROB{Max-PD}}
\newcommand{\MPDlong}{\PROB{Maximize Phylogenetic Diversity}}
\newcommand{\BNAP}{\PROB{Budgeted NAP}}

\newcommand{\PS}{\PROB{Penalty Sum}}
\newcommand{\PenSum}{\PS}
\newcommand{\ucNAP}{\PROB{unit-cost-NAP}}
\newcommand{\HStw}{\PROB{\HS with Tree-Profits}}

\newcommand{\GNAP}{\PROB{GNAP}}
\newcommand{\GNAPLong}{\PROB{Generalized Noah's Ark Problem}}

\newcommand\NAP[4][]{\ensuremath{#1 \stackrel{#2}{\to} #3\:[#4]}-\PROB{NAP}}
\newcommand\NAPLong[4][]{\ensuremath{#1 \stackrel{#2}{\to} #3\:[#4]}-\PROB{Noah's Ark Problem}}

\newcommand{\tPDwslong}{\PROB{Strict Time Sensitive Maximization of Phylogenetic Diversity}}
\newcommand{\tPDslong}{\PROB{Time Sensitive Maximization of Phylogenetic Diversity}}
\newcommand{\tPDws}{\PROB{\mbox{s-Time-PD}}}
\newcommand{\tPDs}{\PROB{\mbox{Time-PD}}}

\newcommand{\ctPDws}{\PROB{\mbox{c-\tPDws}}}
\newcommand{\ctPDs}{\PROB{\mbox{c-\tPDs}}}
\newcommand{\cBartPDslong}[1]{\PROB{extinction $#1$-colored \tPDslong}}
\newcommand{\cBartPDs}[1]{\PROB{\mbox{ex-$#1$-c-\tPDs}}}

\newcommand{\PDDlong}{\PROB{Optimizing PD with Dependencies}}
\newcommand{\sPDDlong}{\PROB{Optimizing PD in Vertex-Weighted Food-Webs}}
\newcommand{\cksPDDlong}{\PROB{$k$-colored \sPDDlong}}
\newcommand{\PDDplong}{\PROB{Optimizing PD with Pattern-Dependencies}}
\newcommand{\cDPDDlong}{\PROB{$D$-colored \PDDlong}}
\newcommand{\ckPDDlong}{\PROB{$2$-colored \PDDlong}}
\newcommand{\PDD}{\PROB{PDD}}
\newcommand{\sPDD}{\PROB{s-PDD}}
\newcommand{\cksPDD}{\PROB{$k$-c-\sPDD}}
\newcommand{\PDDp}{\PROB{\PDD-pattern}}
\newcommand{\cDPDD}{\PROB{$D$-c-\PDD}}
\newcommand{\ckPDD}{\PROB{$2$-c-\PDD}}

\newcommand{\MAPPD}{\PROB{MapPD}}
\newcommand{\MAPPDlong}{\PROB{Max-All-Paths-PD}}
\newcommand{\cMAPPD}{\PROB{colored-\MAPPD}}
\newcommand{\cMAPPDlong}{\PROB{colored-\MAPPDlong}}
\newcommand{\wpSC}{\PROB{wpSC}}
\newcommand{\wpSClong}{\PROB{Item-Weighted Partial Set Cover}}

\newcommand{\MaxNPD}{\PROB{Max-Net-PD}}

\newcommand{\bet}[1]{\ensuremath{\left|#1\right|}}
\newcommand{\sprop}{sur\-vi\-val pro\-ba\-bil\-ity\xspace}
\newcommand{\sprops}{sur\-vi\-val pro\-ba\-bi\-li\-ties\xspace}
\newcommand{\myvec}[1]{\ensuremath{\mathbf{#1}}}
\newcommand{\mo}{-_{\ge 0}}

\DeclareMathOperator{\ex}{ex}
\newcommand{\Pvx}{\ensuremath{P_{v,x}}\xspace}
\newcommand{\mcA}{\ensuremath{\mathcal A}\xspace}
\newcommand{\Hbar}[1]{{\ensuremath{\overline{H_{#1}}}}\xspace}
\newcommand{\Pbar}{{\ensuremath{\overline{P}}}\xspace}
\newcommand{\Cbar}{{\ensuremath{\overline{C}}}\xspace}
\newcommand\proofpart[1]{
	
	\smallskip #1\\ \noindent}

\DeclareMathOperator{\distclust}{dist-clust}
\DeclareMathOperator{\distcoclust}{dist-co-clust}

\newcommand{\apPD}{\apPDsub{\Net}\xspace}
\newcommand{\apPDsub}[1]{\ensuremath{\text{AP-PD}_{#1}}}
\newcommand{\take}{\texttt{take}\xspace}
\newcommand{\leave}{\texttt{leave}\xspace}
\newcommand{\Index}{\texttt{ind}\xspace}

\newcommand{\NetPD}{\NetPDsub{\Net}\xspace}
\newcommand{\NetPDsub}[1]{\ensuremath{\text{Net-PD}_{#1}}}
\newcommand{\N}{\ensuremath{\mathbb{N}}\xspace}
\newcommand{\gam}[2][p]{\ensuremath{\gamma^{#1}_{#2}}}
\newcommand{\floorvar}[2]{\lfloor{#2}\rfloor_{#1}}
\newcommand{\floorH}[1]{\floorvar{H}{#1}}
\newcommand{\ceilvar}[2]{\lceil{#2}\rceil_{#1}}
\newcommand{\ceilH}[1]{\ceilvar{H}{#1}}
\newcommand{\iprop}{in\-her\-i\-tance pro\-por\-tion\xspace}
\newcommand{\iprops}{in\-her\-i\-tance pro\-por\-tions\xspace}

\newcommand{\myrowcols}{\myrowcolsnum{0}}
\newcommand{\myrowcolsnum}[1]{\rowcolors{#1}{gray!10!white}{white}}

\begin{document}

\begin{titlepage}
	\begin{center}
		\vspace{-10cm}
		\includegraphics[width=.5\textheight]{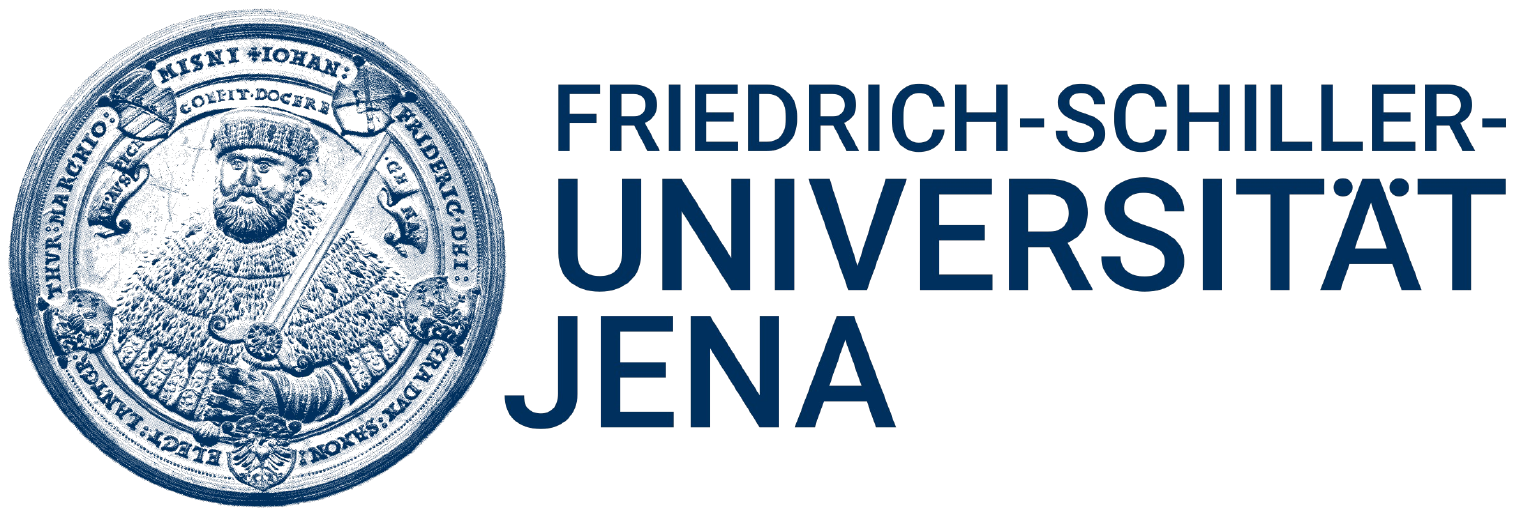}
		\vspace{-3cm}
		
		{\LARGE\textbf{Who Should Have a Place on the Ark?\\[.3cm] Parameterized Algorithms for the Maximization of Phylogenetic Diversity}}\\[1cm]
		
		\vspace{0.5cm}
		
		\textbf{DISSERTATION}\\
		FOR THE DEGREE OF DOCTOR OF NATURAL SCIENCES\\
		DOCTOR RERUM NATURALIUM
		(DR. RER. NAT.)
		
		\vspace{0.5cm}
		
		Submitted to the Council of the Faculty of Mathematics and Computer Science\\
		of the Friedrich Schiller University Jena\\
		on July 29, 2024, by\\
		\textbf{Jannik ``Theodor'' Schestag, M.Sc.},\\
		born May 20, 1995 in Heilbronn, Germany.
	\end{center}
	\vspace{0.5cm}

	\newpage
	\pagestyle{empty}
	
	~
	\vfill
	
	\noindent \textbf{Expert Reviewer} (\textbf{Gutachter}):
	
	\noindent Prof. Dr. Christian Komusiewicz, Friedrich-Schiller-Universität Jena, Germany

	\noindent Prof. Dr.ir. Leo J.J. van Iersel, TU Delft, The Netherlands

	\noindent Dr. Mathias Weller, Université Gustave Eiffel, Paris, France

	\noindent Defended on \textbf{January 17, 2025}
\end{titlepage}
\pagestyle{empty}
\pagenumbering{Roman}
%\setcounter{page}{3}

%\newpage\null\thispagestyle{empty}\newpage

%\include{diss_01_Erklaerung}

%\newpage\null\thispagestyle{empty}\newpage
\pagestyle{plain}

%\include{diss_02_CV}

%\newpage\null\thispagestyle{empty}\newpage

\setcounter{page}{5}
\section*{Acknowledgments}
\addcontentsline{toc}{section}{Acknowledgments}
First and foremost, I humbly thank my redeemer, Jesus Christ, through whom the foundations of the heavens and the earth were created (Colossians 1:16, John 1:3) and who died on the cross as I deserve (Romans 6:23) so that I may live as He deserves (2 Corinthians 5:21).
Jesus not only became my God but also my best friend.
All the strength, concentration, and willpower I needed to accomplish every part of this work were surely given by Him (Philippians 2:13).
All glory, honor, and power be unto Him forever and ever (Revelation 4:11). Amen!

Academically, I express my deepest gratitude to my supervisor, Christian\lb Komusiewicz, for his multitude of great ideas and quick answers to my questions. Additionally, I thank him for the endless hours spent writing applications and for always staying supportive, even in situations where many other professors would have opted~out.

Though officially he is just a co-author, I want to thank Mark Jones as a `minor' supervisor.
His knowledge, humor, and way of working during the cooperative time left a very positive feeling with me and definitively improved my working spirit.

Furthermore, I am grateful to Niels Grüttemeier and Frank Sommer, who initially supervised my Bachelor's and Master's degrees and later became colleagues.
Also, I thank my other colleagues Alexander Bille, Jaroslav Garvardt, Nils Morawietz,\lb Luca Staus, and the Master's student Anton Herrmann, for creating a pleasant working atmosphere, sharing one or the other joke in times, and providing deep talks.\lb
I further want to thank all of my co-authors, listed in lexicographic order:
Matthias Bentert,
Jaroslav Garvardt,
Niels Grüttemeier,
Leo van~Iersel,
Mark Jones,\lb
Christian Komusiewicz,
Simone Linz,
Ber Lorke,
Nils Morawietz,
Malte Renken,
Frank Sommer,
Celine Scornavacca, and
Mathias Weller.

Gratefully, I received a financial support by the Deutscher Akademischer Austauschdienst (DAAD---German Academic Exchange Service), project 57556279,\lb from April to September 2023.
I~thank the TU~Delft and especially Leo van~Iersel,\lb Mark Jones, and their colleagues for hosting me during this period.

Last but not least, for their unwavering support throughout my life, I want to express my sincerest thanks to my mother, Kerstin Schestag, and my sister, Mona Felina Schestag.
Surely, I would have not come to where I am without them.
While it is impossible to list all the important people in my life, I explicitly want to express thanks to
my grandmother Edeltraud Schestag,
Klaus~Pieper,
Holger~\&~Regine~Hellmich,\lb
Ricarda~Schwarz,
the~Slembeck family,
the~Köhler family,
Niklas~Vogt,
Jan~Mischon,
my late father Joachim Schestag,
my late grandfather Ulrich Hellmich,
and everyone who supportively prayed for me!

\newpage\null\thispagestyle{empty}\newpage

\section*{Preface}
\addcontentsline{toc}{section}{Preface}
Within this thesis, I present the results of the research I conducted in the field of parameterized algorithms for the maximization of phylogenetic diversity from May~2022 until June 2024.
During these two years, I have at first been part of the Algorithmics research group led by Christian Komusiewicz at the Philipps-Universität Marburg, Germany.
Because of receiving a six-month scholarship I worked some time in the Discrete Mathematics and Optimization research group led by Karen Aardal at the TU Deft, The Netherlands.
After my pleasant stay in the Netherlands, I again moved to Germany to work with Christian Komusiewicz; this time in the Algorithm Engineering research group at the Friedrich-Schiller-Universität Jena, Germany.

The results presented in this thesis are predominantly published in scientific papers on the platform ArXiv.
Many of them are however not yet published in conference proceedings or journals.
In the following, I will explain which chapters are based on which publications and what participation I had in the research and writing process.
Afterward, I will give a brief list of other scientific publications I coauthored that are not included in this thesis.

It is worth mentioning that the order of chapters in this thesis does not corresponding to when the research occurred.
To minimize confusion: The order of research was \GNAPLong, \MAPPDlong, \tPDs, \MaxNPD, and finally \PDDlong.
This order may be good to keep in min for readers, as algorithmic ideas for \MAPPDlong appear in Chapter~\ref{ctr:Networks} of this thesis but are basis for many algorithms in Chapters~\ref{ctr:TimePD} and~\ref{ctr:FoodWebs}.

Chapter~\ref{ctr:GNAP} is based on the publication ``A Multivariate Complexity Analysis of the Generalized Noah's Ark Problem'' written with Christian Komusiewicz, which can be found on the platform ArXiv~\cite{GNAParxiv}.%, which is currently submitted to \emph{Discrete Applied Mathematics} and can be found on the platform ArXiv~\cite{GNAParxiv}.
A preliminary version of this publication appeared in the \emph{Proceedings of the 19th Cologne-Twente Workshop on Graphs and Combinatorial Optimization ({CTW}~2023)}~\cite{GNAP}.
During my initial literature research on what kind of generalizations of \MPD there are, I first considered studying \GNAPLong (and \PDDlong).
As it was my first project as a Ph.D. student and I was at first only working from home, I was relatively slow in the research and first had to start with basic concepts.
I made the key discovery to find the strong connection to the \MCKPLong.
The overall research was conducted in cooperation with Christian whereby Christian gave the better ideas.
The writing-up process was predominately conducted by me.
Christian provided helpful hints and wrote the introduction.
I gave the talk at CTW~2023.
This work also contains a small observation, Observation~\ref{obs:GNAP-ultrametric-greedy}, which is not published elsewhere.

Chapter~\ref{ctr:TimePD} is based on the publication ``Maximizing Phylogenetic Diversity under Time Pressure: Planning with Extinctions Ahead'' written with Mark Jones, which can be found on the platform ArXiv~\cite{TimePD}.
%A preliminary version of this publication is currently submitted to the \emph{Proceedings of the 19th Cologne-Twente Workshop on Graphs and Combinatorial Optimization ({CTW}~2023)}.
%
The research of \tPDs and \tPDws was initialized by me during my stay in Delft, after Mark and I finished the research on \MAPPDlong.
We were actually stuck in the research process for quite a while because it was difficult to find out how to cope with the scheduling and tree-structures at the same time.
Mark made the final breakthrough and had the rough idea for the \FPT-algorithm for \tPDs with respect to~$D$~(Lemma~\ref{lem:timePD-cs-D}).
I quickly had the idea of how to utilize the idea to show that also \tPDws is \FPT when parameterized by~$D$~(Lemma~\ref{lem:timePD-cws-D}).
Afterward, it took us several months of work to show that \tPDws is \FPT when parameterized by~$\Dbar$~(Theorem~\ref{thm:timePD-DBar}).
The according proof is arguably the most technical and difficult in this thesis.
Mark and I had a fair share in the writing-up process of the results.
I additionally provided the connection to the scheduling field and he wrote up most of the introduction.

Chapter~\ref{ctr:FoodWebs} is based on the publication ``Maximizing Phylogenetic Diversity under Ecological Constraints: A Parameterized Complexity Study'' written with Christian Komusiewicz, which can be found on the platform ArXiv~\cite{PDDarxiv}.
A preliminary version of this publication appeared in the \emph{Proceedings of the 44th IARCS Annual Conference on Foundations of Software Technology and Theoretical Computer Science (FSTTCS~2024)}~\cite{PDD}.
As mentioned above, the idea of the study of \PDDlong came along with my first literature research.
Christian and I conducted the research roughly on an equal level of scientific input.
The writing-up process was predominately conducted by me.
Christian provided helpful hints and wrote the introduction.

Chapter~\ref{ctr:Networks} is predominantly based on the research of \MAPPDlong and only Section~\ref{sec:Net-reduction-MaxNPD} is about \MaxNPD.
However, also the introduction and discussion of the chapter contain material from the \MaxNPD-paper.
The research is contained in the two following publications.
Firstly, ``How Can We Maximize Phylogenetic Diversity? Parameterized Approaches for Networks'' was written with Mark Jones, which can be found on the platform archive~\cite{MAPPDarchive}.
A preliminary version of this publication appeared in the \emph{Proceedings of the 18th International Symposium on Parameterized and Exact Computation (IPEC 2023)}~\cite{MAPPD}.
Secondly, ``Maximizing Network Phylogenetic Diversity'' was written with Leo van Iersel, Mark Jones, Celine Scornavacca, and Mathias Weller, which can be found on the platform ArXiv~\cite{MaxNPD}.
%A preliminary version of this publication is currently submitted to the \emph{44th IARCS Annual Conference on Foundations of Software Technology and Theoretical Computer Science (FSTTCS~2024)}.
%
Christian Komusiewicz suggested me also to study phylogenetic diversity in phylogenetic networks.
We then came into contact with Mark Jones and Leo van Iersel and asked them whether I could join them in Delft for a time.
After their confirmation, I applied for a scholarship from DAAD with a great deal of help from Christian.
Mark and I were both very active in the research of \MAPPDlong so that Leo did not join the paper in the end.
Mark and I conducted the research roughly on an equal level of scientific input and also the writing-up roughly on equal sides.
I gave the talk at IPEC~2023.
After the conference version, I proposed a color-coding algorithm with respect to~$\Dbar$ (Theorem~\ref{thm:Net-Dbar}) which improves to the previous algorithm in running time and working on general graphs and not only on binary graphs.
Further, Mark and I together developed a kernelization algorithm with respect to the reticulation number.
Arising from the research of \MAPPDlong, Mark proposed to study \MaxNPD in an established research group where I joined.
In the meetings, Mark and I quickly showed that there is a reduction from \GNAP to \MaxNPD on level-1-networks~(Theorem~\ref{thm:Net-level-1}).
We then collectively also proved the \NP-hardness of \PS (Section~\ref{sec:PenSum}) and an \FPT-algorithm for the number of reticulations.
I wrote up what is Section~\ref{sec:Net-reduction-MaxNPD} in this thesis, Mark what is Section~\ref{sec:PenSum} in this thesis, Mathias the \FPT-algorithm, and Celine and Leo the intro, prelims, and discussion of the paper.

Chapters~\ref{ctr:intro},~\ref{ctr:prelims}, and~\ref{ctr:conclusion}, these are the introduction, the preliminaries, and the conclusion, are based on parts of all five papers.
The examination of \MCKPLong, presented in Section~\ref{sec:MCKP}, is part of the \GNAP-paper~\cite{GNAP}.
The examination of \PS, presented in Section~\ref{sec:PenSum}, is part of the \MaxNPD-paper~\cite{MaxNPD}, except for Proposition~\ref{prop:Pre-PS-pseudo} which was not published elsewhere.

\newpage
During my time as a Ph.D. student, I also participated in the research of the following publications.
These are ordered by when the research project started.

\begin{itemize}
	\item ``On Critical Node Problems with Vulnerable Vertices'', with Niels\linebreak Grüttemeier, Christian Komusiewicz, and Frank Sommer.
	This paper is based on my Master's thesis.
	Journal: \emph{Journal of Graph Algorithms and Applications}~\cite{Vulnerable}.
	Conference: \emph{33rd International Workshop on Combinatorial Algorithms (IWOCA 2022)}~\cite{VulnerableIWOCA}.
	
	\item ``On the Complexity of Parameterized Local Search for the Maximum Parsimony Problem'', with Christian Komusiewicz, Simone Linz, and Nils\lb Morawietz.
	Conference: \emph{34th Annual Symposium on Combinatorial Pattern Matching (CPM 2023)}~\cite{Parsimony}.
	
	\item ``Finding Degree-Constrained Acyclic Orientations'', with Jaroslav Garvardt, Malte Renken, and Mathias Weller.
	Conference: \emph{18th International Symposium on Parameterized and Exact Computation (IPEC 2023)}~\cite{Ortientation}.
	
	\item ``On the Complexity of Finding a Sparse Connected Spanning Subgraph in a Non-Uniform Failure Model'', with Matthias Bentert, and Frank Sommer.
	Conference: \emph{18th International Symposium on Parameterized and Exact Computation (IPEC 2023)}~\cite{unsafeSpanning}.
	ArXiv: ~\cite{unsafeSpanningArXiv}.
	
	\item ``Protective and Nonprotective Subset Sum~Games: A Parameterized Complexity Analysis'', with Jaroslav Garvardt, Christian Komusiewicz, and Ber Lorke.
	Conference: \emph{8th International Conference on Algorithmic Decision Theory (ADT 2024)}~\cite{SubsetSumGame}.
\end{itemize}

\pagestyle{plain}
\pagenumbering{arabic}
\setcounter{page}{11}

%\newpage\null\thispagestyle{empty}\newpage
%
%\newpage\null\thispagestyle{empty}\newpage

\section*{Abstract}
\addcontentsline{toc}{section}{Abstract}

Phylogenetic Diversity (PD) is a well-regarded measure of the overall biodiversity of a set of present-day species (taxa) that indicates its ecological significance.
In the \MPDlong (\MPD) problem one is asked to find a small set of taxa in a phylogenetic tree for which this measure is maximized.
\mbox{\MPD} is particularly relevant in conservation planning, where limited resources necessitate prioritizing certain taxa to minimize biodiversity loss.
Although \mbox{\MPD} can be solved in polynomial time~[Steel, Systematic Biology, 2005; Pardi and Goldman, PLoS Genetics, 2005], its generalizations---which aim to model biological processes and other aspects in conservation planning with greater accuracy---often exhibit \NP-hardness, making them computationally challenging.
This thesis explores a selection of these generalized problems within the framework of parameterized complexity.

In \GNAPLong (\GNAP), each taxon only survives at a certain \sprop, which can be increased by investing more money in the taxon.
We show that \GNAP is \Wh{1}-hard with respect to the number of taxa but is \XP with respect to the number of different costs and different \sprops.
Additionally, we show that \ucNAP, a special case of \GNAP, is \NP-hard.

In \tPDslong\lb (\tPDs), different extinction times of taxa are considered after which they can no longer be saved.
For \tPDs, we present color-coding algorithms that prove that \tPDs is fixed-parameter tractable (\FPT) with respect to the threshold of diversity and the acceptable loss of diversity.

In \PDDlong (\PDD), each saved taxon must be a source in the ecological system or a predator of another saved species.
These dependencies are given in a food-web.
We show that \PDD is \FPT when parameterized with the size of the solution plus the height of the phylogenetic tree.
Further, we consider parameters characterizing the food-web
and prove that \PDD is \FPT with respect to the treewidth if the phylogenetic tree is a star, disproving an open conjecture.

Phylogenetic Diversity is traditionally defined on phylogenetic trees, but phylogenetic networks have gained popularity as they enable the modeling of so-called reticulation events.
\MAPPDlong (\MAPPD) and \MaxNPD are two problems in maximizing phylogenetic diversity on phylogenetic networks.
\MAPPD is \Wh{2}-hard with respect to the solution size but
we show that \MAPPD is \FPT with respect to the number of reticulations and the treewidth, two parameters that show how tree-like the phylogenetic network is.
\MaxNPD however remains \NP-hard even on level-1-networks.

\newpage

\section*{Zusammenfassung (Translation of the Abstract)}
\addcontentsline{toc}{section}{Zusammenfassung}

Die phylogenetische Diversität (PD) ist ein anerkanntes Maß für die gesamte Bio-\lb diversität, das die ökologische Bedeutung einer Gruppe heutiger Arten (Taxa) anzeigt.
Im Problem \MPDlong (Maximierung der phylogenetischen Diversität, \MPD) wird gefordert, eine kleine Menge an Taxa in einem phylogenetischem Baum zu finden, für welche dieses Maß maximiert wird.
Dieses Problem ist besonders relevant in der Naturschutzplanung, bei der begrenzte Ressourcen die Priorisierung bestimmter Taxa zur Minimierung des Biodiversitätsverlusts erfordern.
Obwohl \MPD in polynomialer Zeit gelöst werden kann~[Steel, Systematic Biology, 2005; Pardi und Goldman, PLoS Genetics, 2005], weisen seine Verallgemeinerungen - welche biologische Prozesse und andere Aspekte des Artenschutzes genauer modellieren sollen - oft \NP-Härte auf, was sie rechnerisch anspruchsvoller macht.
Diese Arbeit untersucht eine Auswahl dieser generalisierten Probleme im Rahmen der parametrisierten Komplexität.

In \GNAPLong (Generalisiertes Noahs-Arche-Pro\-blem, \GNAP) überlebt jedes Taxon nur mit einer bestimmten Überlebenswahrscheinlichkeit, die durch mehr Investitionen in das Taxon erhöht werden kann.
Wir zeigen, dass \GNAP in Bezug auf die Anzahl der Taxa \Wh{1}-schwer, jedoch \XP in Bezug auf die Anzahl der verschiedenen Kosten und Überlebenswahrscheinlichkeiten ist.
Zusätzlich zeigen wir, dass das \ucNAP (Einheitspreis-NAP), ein Spezialfall von \GNAP, \NP-schwer ist.

In \tPDslong (zeitkritische Maximierung der phylogenetischen Diversität, \tPDs) werden unterschied\-liche Aussterbezeiten der Taxa berücksichtigt, nach denen diese nicht mehr gerettet werden können.
Für \tPDs präsentieren wir Farbmarkierungsalgorithmen, die beweisen, dass \tPDs in Bezug auf die geforderte Diversität und den akzeptablen Verlust an Diversität festparameter handhabbar (\FPT) ist.

In \PDDlong (Optimieren phylogenetischer Diversität bei Abhängigkeiten, \PDD) muss jedes gerettete Taxon eine Quelle im ökologischen System oder ein Raubtier einer anderen geretteten Art sein.
Diese Abhängigkeiten sind in einem Nährungsnetz gegeben.
Wir zeigen, dass \PDD{} \FPT ist, wenn es in Bezug auf die Größe der Lösung plus die Höhe des phylogenetischen Baumes parametrisiert wird.
Darüber hinaus betrachten wir Parameter, die das Nährungsnetz charakterisieren, und beweisen, dass \PDD in Bezug auf die Baumweite \FPT ist, wenn der phylogenetische Baum ein Stern ist, womit wir eine offene Vermutung widerlegen.

Phylogenetische Diversität wird traditionell auf phylogenetischen Bäumen de\-fi\-niert, aber phylogenetische Netzwerke haben an Popularität gewonnen, da sie auch die Modellierung sogenannter retikulierender Ereignisse ermöglichen.
\MAPPDlong (Maximierung aller Pfade [zur Bestimmung] der phylogenetischen Diversität, \MAPPD) und \MaxNPD (Maximierung der phylogenetischen Netzwerk-Diversität) sind zwei Probleme zur Maximierung der phylogenetischen Diversität in phylogenetischen Netzwerken.
\MAPPD ist in Bezug auf die Lösungsgröße \Wh{2}-schwer, aber wir zeigen, dass \MAPPD in Bezug auf die Anzahl der Retikulationen und die Baumweite \FPT ist, zwei Parameter, welche zeigen, wie Baum-ähnlich das phylogenetische Netzwerk ist.
\MaxNPD bleibt jedoch auch auf Level-1-Netzwerken \NP-schwer.

%\newpage\null\thispagestyle{empty}\newpage

\setcounter{tocdepth}{1}%Depthcounter: 1 = Only Parts, Chapters, and Sections
\tableofcontents

\chapter{Introduction}
\label{ctr:intro}

%\todosi{
%Mass extinction exists. Something on how horrible it is. Man made climate change enhances it.
%}

Human activities~\cite{crist} in the current Western economic system, driven by significant overconsumption and the externalization of costs to society and the environment~\cite{marques,lynch}, have significantly contributed to the world's ecological system approaching a sixth mass extinction~\cite{WarningToHumanityTwo} within the last decades and thereby created arguably one of mankind's biggest challenges for this century.
While activists and scientists alike are waiting for a serious political response to this crisis~\cite{WarningToHumanity}, the looming threat posed by climate change is predicted to further exacerbate the decline in species.
The 2023 report of the Intergovernmental Panel on Climate Change (IPCC)~\cite[Page~73]{ipcc2023}
%\todos{In final version make sure that the cite is IPCC23 and not IPC23.}
warns that as the average global temperature increases, thousands of plant and animal species around the world face increasing risk of extinction.
With a ``climate breakdown''~\cite{climatebreakdown} on the horizon, media outlets worldwide report with drastic words on an increasing number of species that are already extinct or are inevitably about to become extinct in the near future~\cite{BBC,CNN,Reuters}.

%\todosi{
%What species should we protect.
%}

The reason for quantifying the loss of biodiversity in public media by naming the number of extinct species is relatively evident.
The sheer number of extinct species is a very easy-to-understand metric.
Additionally, it subliminally carries a political message: Not a single further species should face extinction.
However, a bit more realism reveals the tragedy that certainly not enough resources---of monetary nature, political willpower, land-wise, time-wise, and further---are available to save each and every single species alive today.
Already in 1990, May~\cite{May} argued in a seminal note that the subsequent question of which set of species should be saved, ``raises larger questions about the ways in which relative values are assigned to different creatures''.

Indeed, when it seems impossible to save the entire ``tree of life'', some species have to be selected to focus on.
May further describes this selection as ``making choices for the ineluctably limited number of places on the ark''~\cite{May},
where he compares to the biblical narrative of Noah who built an ark to save his family and pairs of animals from a great flood.
In the de facto approach, a lot of effort in protection is spent on species which are perceived to be ``cute'' or ``dangerous''~\cite{cuteness} or which have the ``right color''~\cite{CutenessColor}.

\section{Phylogenetic Diversity}

%\todosi{
%What is biodiversity.
%}

While humans can bond to species~\cite{wilson1984} and therefore feel the need to protect a species that they consider aesthetically pleasing,
such a factor most certainly does not state how important a species is in the ecological system.
However, requiring the information which species have such an importance is crucial, as ``it is foolish to allow destruction of nature without knowing what it is worth''~\cite{crozier}.
In the years after May's note, a lot of effort in the scientific world has been put into the question of how to measure the relevance of a set of taxa or species.
The key for this search was the usage of \emph{phylogenetic trees}~\cite{humphries},
which are represented by graphs with vertices and edges.
Each leaf represents a living taxon (species) and internal vertices represent common ancestors of the leaves that can be reached or biological categories---technically known as operational taxonomic units.
Usually, the set of taxa is denoted with~$X$ and one then refers to phylogenetic $X$-trees.\footnote{Note that phylogenetic trees may be rooted or unrooted. In this thesis, we only consider the rooted variant.}

As an inaugural idea, it was proposed to count the number of internal vertices that the spanning tree of a given set of taxa in the phylogenetic tree has~\cite{vane,nixon} to measure the diversity of this set.
This measure, called \emph{cladistic diversity}, was soon improved by adding weights on the edges of the phylogenetic tree.
Independently, Faith~\cite{FAITH1992} and Weitzmann~\cite{Weitzman1992} proposed \emph{phylogenetic diversity}, where Weitzmann's definition even generalizes the usage of this metric to arbitrary items such as books in a library.
The phylogenetic diversity of a set~$A$ of taxa, denoted by~$\PD(A)$, is the total weight of all edges that are on a path from the root to a leaf that corresponds to a taxon in~$A$.
In Figure~\ref{fig:CD-and-PD} an example of cladistic and phylogenetic diversity is given.

\begin{figure}[t]
	\centering
	
	\begin{forest}
		forked edges,
		/tikz/every pin edge/.append style={Latex-, shorten , gray},
		/tikz/every pin/.append style={darkgray, font=\sffamily},
		/tikz/every label/.append style={gray, font=\sffamily},
		before typesetting nodes={
			delay={
				where content={}{coordinate}{},
			},
			where n children=0{tier=terminus, label/.process={Ow{content}{right:#1}}, content=}{},
		},
		for tree={
			grow'=0,
			s sep'+=10pt,
			l sep'+=20pt,
		},
		l sep'+=50pt,
		[, !1.edge label={node [pos=.85, every label, below] {150}}
		[, !1.edge label={node [pos=.75, every label, below] {154}}
		[
		[
		[(a) Acropora cf rotumana]
		[(b) Acropora valencinessi]
		]
		[, !1.edge label={node [pos=.25, every label, below] {66}}
		[(c) Acropora abtrotanoides]
		]
		]
		[(d) Acropora derawanensis]
		]
		[, !1.edge label={node [pos=.3, every label, below] {113}}
		[(e) Acropora batunai]
		]
		]
	\end{forest}
	\caption{A weighted phylogenetic tree of a selection of Acropora corals. For edges without numbers, the evolutionary distances are not specified~\cite{faith12}.\\
		The cladistic diversity of the sets of corals $\{a, c, e\}$ and $\{a, e\}$ is 4, each.
		The phylogenetic diversity of the sets of corals $\{a, c, e\}$ is 483 and of $\{a, c\}$ is 370, assuming that edges without numbers have a weight of~0.}
	\label{fig:CD-and-PD}
\end{figure}%

Intuitively, with phylogenetic diversity one measures the expected range of biological features of a given set of taxa.
We, however, need to be a bit cautious as it is not always correct to estimate the number of features with the phylogenetic diversity~\cite{wicke2021formal}.
Nevertheless, maximizing phylogenetic diversity has become the standard, albeit imperfect, measure for the biological diversity of a set of taxa~\cite{GCW+15,LondonEDGE,MPC+18}.
This measure forms the basis of the Fair Proportion Index and the Shapley Value~\cite{haake2008shapley,Hartmann2013TheEO,redding2006incorporating}, which are used to evaluate the individual contribution of individual taxa to overall biodiversity.
Over the years, phylogenetic diversity nevertheless became the most well-regarded metric in measuring the biodiversity and therefore the value of a set of species (see~\cite{crozier}),
which for example is used by the IUCN's  Phylogenetic Diversity Task Force \mbox{(\url{https://www.pdtf.org/})}
and the Zoological Society of London's EDGE of Existence program~\cite{LondonEDGE}.
Despite all the mentioned positive features of phylogenetic diversity, to not present a one-sided picture, we also want to mention that there are some objections raised against the usage of phylogenetic diversity~\cite{winter2013}
and many preservation projects also do not consider phylogenetic diversity~\cite{mace2008,santamaria2012}.

%\todosi{
%	Explain \MPD
%}

Based on the phylogenetic diversity index, Daniel Faith proposed a maximization problem, \MPD, in which one has to find a relatively small set of taxa that achieves maximal phylogenetic diversity~\cite{FAITH1992}.
The set needs to be small because limited resources impose monetary, environmental, or other restrictions to the preservation project.
More formally, \MPD, formulated as a decision problem, is defined as follows.

\problemdef{\MPDlong (\MPD)}{
	A phylogenetic~$X$-tree~\Tree and integers~$k$ and~$D$}{
	Is there a set of taxa~$S \subseteq X$ of size at most~$k$ such that~$\PD(S) \ge D$}

Faith already stated that \MPD can be computed with a greedy algorithm and therefore is computationally easy~\cite{FAITH1992}.
Formal proofs were given by Steel~\cite{steel} and Pardi and Goldman~\cite{Pardi2005}, independently.
Consequently, \MPD can be solved within seconds, even on large instances~\cite{PDinSeconds}.

\section{Modeling Real-World Biological Processes}

%\todosi{
%Further problems are needed to cope complexity in reality.
%}

As much as the lightweight definition of \MPD is an advantage in understanding and quickly solving the problem,
just as much of a drawback it is in modeling real-world biological processes.
Among these disadvantages stand that it is not realistic that all species can be saved for the same price, that all species have the same remaining time before an extinction, or that the survival of each species is decided without regarding the interaction of species.
To cope with these issues, several further problems were introduced which, in contrast to \MPD, are \NP-hard~\cite{pardi07,hartmann,TimePD,moulton}.
Therefore, it is unlikely that these problems can be solved with an algorithm that only consumes time polynomial to the size of the input.
In this thesis, we examine a selection of \NP-hard problems in which one aims to maximize phylogenetic diversity.
We study these problems from a classical and parameterized point of view and want to help understanding what makes these problems tractable and provide algorithms that break this intractability to a certain degree.
While more mathematical definitions follow in the next chapter, the reader at this point only needs to understand that if a problem is \XP or \FPT, then an algorithm exists which solves the problem in a desirable time.
For \Wh{1}-hard problems, the existence of an \FPT-algorithm is unlikely.
A deeper view of the theoretical framework is given in the next chapter.
In the following, we briefly present the problems considered in this work and summarize some results.

%\todosi{
%Explain \GNAP.
%}
\paragraph{Generalized Noah's Ark Problem.}
One of the first steps was to allow that the protection of taxa may have different costs~\cite{pardi07}.
This approach introduced the problem~\BNAP, also referred to as~\NAP[0]{c_i}{1}{2}, which is shown to be \NP-hard by a reduction from \KP~\cite{pardi07}.

Subsequent approaches also allowed to model uncertainty as follows.
For example, performing an action to protect some species does not guarantee the survival of that species but only raises the survival probability~\cite{weitzman}.
In this model, one now aims to maximize the \emph{expected} phylogenetic diversity.
That is, the weight of an edge is only added with the probability that at least one of the offspring in the subtree survives.
Finally, one may also consider the even more realistic case when for each species, one may choose from a set of different actions or even from combinations of different actions.
Each choice is then associated with a cost and with a resulting survival probability.
This model was studied by Hartmann and Steel~\cite{hartmann}, Pardi~\cite{PardiThesis} and Billionnet~\cite{billionnet13,billionnet17} as~\GNAPLong{} (\GNAP).
An example of an instance of \GNAP is given in Figure~\ref{fig:example-gnap}.
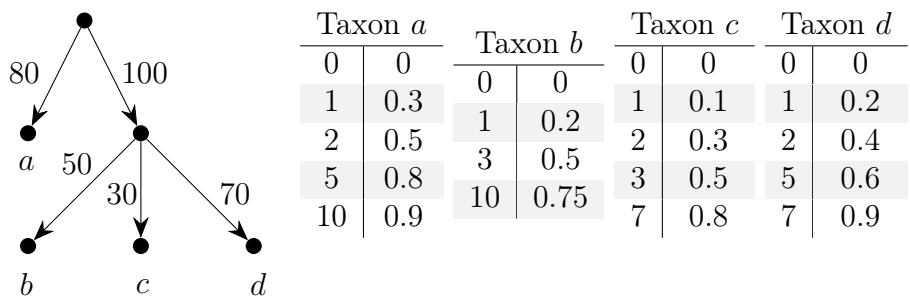
\begin{figure}[t]
	\centering
	\begin{tabular}{m{3.5cm}m{11cm}}
		\begin{tikzpicture}[edge from parent/.style={draw,-{Stealth[length=8pt]}},every node/.style={circle,fill=black,draw=black,inner sep=2pt}]
			\tikzstyle{txt}=[circle,fill=white,draw=white,inner sep=0pt]
			
			\node[txt] at (-.77,-1.9) {$a$};
			\node[txt] at (-.77,-3.5) {$b$};
			\node[txt] at (.77,-3.54) {$c$};
			\node[txt] at (2.31,-3.5) {$d$};
			
			\node[txt] at (-.8,-.7) {$80$};
			\node[txt] at (.8,-.7) {$100$};
			\node[txt] at (-.1,-1.9) {$50$};
			\node[txt] at (.5,-2.3) {$30$};
			\node[txt] at (2,-2.3) {$70$};
			
			\node {}
			child {node {}}
			child {node {}
				child {node {}}
				child {node {}}
				child {node {}}
			};
		\end{tikzpicture}
		& \raisebox{23mm}{%
			\myrowcolsnum3
			\begin{tabular}{c|c}
				\multicolumn{2}{c}{Taxon~$a$}\\
				\hline
				0 & 0\\
				1 & 0.3\\
				2 & 0.5\\
				5 & 0.8\\
				10 & 0.9
			\end{tabular}
			\myrowcolsnum3
			\begin{tabular}{c|c}
				\multicolumn{2}{c}{Taxon~$b$}\\
				\hline
				0 & 0\\
				1 & 0.2\\
				3 & 0.5\\
				10 & 0.75
			\end{tabular}
			\myrowcolsnum3
			\begin{tabular}{c|c}
				\multicolumn{2}{c}{Taxon~$c$}\\
				\hline
				0 & 0\\
				1 & 0.1\\
				2 & 0.3\\
				3 & 0.5\\
				7 & 0.8
			\end{tabular}
			\myrowcolsnum3
			\begin{tabular}{c|c}
				\multicolumn{2}{c}{Taxon~$d$}\\
				\hline
				0 & 0\\
				1 & 0.2\\
				2 & 0.4\\
				5 & 0.6\\
				7 & 0.9
			\end{tabular}
		}
	\end{tabular}
	\caption{An example of an instance of \GNAP with a phylogenetic tree on the left and the lists of projects to the right.
		In each table, in the left column the cost of the project is shown and in the right column the associated \sprop.
		Spending~2 on Taxon~$a$, 1~on Taxon~$b$, 0~on Taxon~$c$, and~5 on Taxon~$d$ would cost~8 and give an expected diversity of~$80 \cdot 0.5 + 50 \cdot 0.2 + 30 \cdot 0 + 70 \cdot 0.6 + 100 \cdot (1 - 0.8 \cdot 0.4) = 40 + 10 + 42 + 68 = 160$.}
	\label{fig:example-gnap}
\end{figure}%

It feels natural that putting more effort into a certain species increases their \sprop.
While the biblical story tells us that for God it was possible to restore species with a single male and female after the great flood, such a project seems unlikely to succeed for men without divine skills.
So, in \GNAP one can consider the projects to refer to the cost and the according \sprop that are associated with the protection of a certain amount of individuals of the same kind.
However, projects can also be a lot more general.

We consider \GNAP and also put a focus on a special case of \GNAP in which every taxon only has a single project that has a cost of~1 and a positive \sprop.
We show that this special case, which we call \ucNAP, is \NP-hard even if the phylogenetic tree of the input has a height of~2 or is ultrametric and has a height of~3.
Furthermore, we show that \GNAP is \XP with respect to the number of unique costs or the number of unique \sprops,
but when parameterized with the number of taxa, \GNAP is \Wh{1}-hard and therefore unlikely to admit an \FPT algorithm.

%\todosi{
%Explain \tPDs
%}
\paragraph{Consideration of Extinction Times.}
Efforts in another direction were made to consider that species could have differing times, after which they will die out if they have not been protected.
In this extension of \BNAP, denoted \tPDslong(\tPDs), each taxon has an associated \emph{rescue length} (the amount of time it takes to save them) and also an \emph{extinction time} (the time after which the taxon can not be saved anymore).
Additionally, we know when preservation teams are able to work on saving taxa.
Thus, to ensure that a set of taxa can be saved with the available resources, it is not enough to guarantee that their total cost is below a certain threshold.
One also needs to ensure that there is a schedule for the involved intervention teams under which each taxon is saved before its moment of extinction.
The first steps to addressing this issue were made by defining and analyzing \tPDs and \tPDws \cite{TimePD}, which is a part of this thesis.
\tPDws and \tPDs are related versions of the problem which differ only in the specification of whether operating teams may cooperate on saving a specific taxon or not.

\tPDs and \tPDws have much in common with machine scheduling problems, insofar as we may think of the taxa as corresponding to jobs with a certain due date and the teams corresponding to machines.
One may think of \tPDs and \tPDws as machine scheduling problems, in which the objective to be maximized is the phylogenetic diversity of the set of completed tasks---which are the saved taxa.

We show that both problems, \tPDs and \tPDws, are \FPT when parameterized with the threshold of diversity~$D$.
Further, we present an \FPT-algorithm for \tPDs with respect to the acceptable loss of phylogenetic diversity.
Such an algorithm is unlikely to exist for \tPDws, as this problem is \NP-hard even if every taxon needs to be saved.

%\todosi{
%Explain \PDD
%}
\paragraph{Considering Ecological Dependencies.}
Moulton et~al.~\cite{moulton} were the first to consider a formal problem in which the survival of some taxa may also depend on the survival of other taxa and modeled the \PDDlong~(\PDD).
Here, the input additionally contains a directed acyclic graph (DAG)\footnote{A DAG is a directed graph without loops in which there is no path from vertex~$v$ to vertex~$u$ if there is a path from~$u$ to~$v$.}~\Food with vertex set~$X$ where an arc~$uv$ is present if the existence of taxa~$u$ provides all the necessary foundations for the existence of taxon~$v$.
In other words,~\Food models ecological dependencies between taxa.
Now, a taxon~$v$ may survive only if (i)~it does not depend on other taxa at all, that is, it has no incoming arcs, or (ii)~at least one taxon~$u$ survives such that~\Food contains the arc~$uv$.
The most widespread interpretation of such ecological dependency networks are food-webs where the arc~$uv$ means that taxon~$v$ feeds on taxon~$u$.\footnote{We remark that previous works~\cite{moulton,faller} consider a reversed interpretation of the arcs. We define the order in such a way that a source of the network also corresponds to a source of the ecosystem.}
A subset of taxa~$X$ where every vertex fulfills~(i) or~(ii) is called \emph{viable}.
The task in \PDD{} is to select a \emph{viable} set of~$k$ taxa that achieves maximal phylogenetic diversity.
An example of a food-web is given in Figure~\ref{fig:example-PDD}.
\begin{figure}[t]
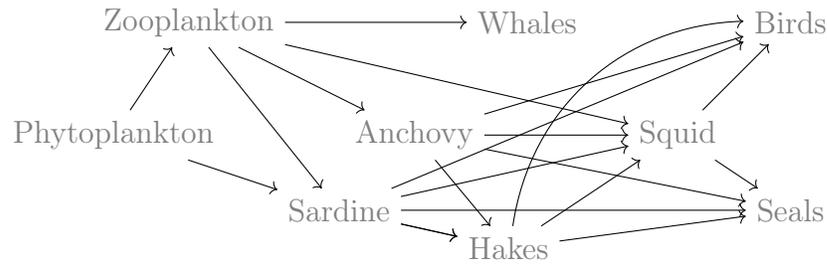

	\centering
	\usetikzlibrary{graphs}
	\tikz {
		\node[gray] (Phy) at (0,0) {Phytoplankton};
		\node[gray] (Zoo) at (1,1.5) {Zooplankton};
		\node[gray] (Sd) at (3,-1) {Sardine};
		\node[gray] (W) at (5.5,1.5) {Whales};
		\node[gray] (Sl) at (4,0) {Anchovy};
		\node[gray] (H) at (5.25,-1.5) {Hakes};
		\node[gray] (T) at (7.5,0) {Squid};
		\node[gray] (R) at (9,-1) {Seals};
		\node[gray] (V) at (9,1.5) {Birds};
		
		\graph { 	
			(Phy) -> (Zoo) -> (W),
			(Phy) -> (Sd) -> (H) ->[bend left=40] (V),
			(Zoo) -> (Sd) -> (T) -> (V),
			(Zoo) -> (Sl) -> (R),
			(Zoo) -> (T) -> (R),
			(Sd) -> (H) -> (R),
			(Sd) -> (R), (Sd) -> (V),
			(Sl) -> (H), (Sl) -> (T), (Sl) -> (V),
			(H) -> (T) };
	}
	\caption{Here, a model of the food-web of the Benguela ecosystem is depicted~\cite{planque}.
		Phytoplankton is the only source.
		If whales are to be saved, then zooplankton and phytoplankton also need to be saved.}
	\label{fig:example-PDD}
\end{figure}%

We consider \PDD and the special case \sPDD, where the phylogenetic tree is a star.
We show that \PDD is \FPT when parameterized by the size of the solution~$k$ plus the height of the phylogenetic tree and therefore with respect to the desired diversity~$D$.
We further examine \PDD with respect to parameters analyzing the structure of the food-web.
Herein, we show that \sPDD is \FPT when parameterized by the distance to cluster or the treewidth of the food-web.
With the distance to cluster it is measures how few vertices could be removed to result in a cluster graph and with the treewidth, one can state how similar a graph is to a tree.
The latter disproves a conjecture of Faller et al.~\cite[Conjecture~4.2]{faller} stating that \sPDD is \NP-hard even when the food-web is a tree, unless \PeqNP.

\section{Phylogenetic Networks}
%\todosi{
%Explain networks
%}

Phylogenetic trees are the familiar model to represent the inheritance of a set of taxa.
Even though such a representation may models the biological ground truth acceptably well for many applications, it is not flawless.
In the interaction between taxa it is possible that reticulation events occur.
These events include hybridization, horizontal gene transfer and other forms of genetic recombination~\cite{huson2006}.
With hybridization, evolutionists describe the crossing of two groups of species that have developed different features over time~\cite{stebbins1959,stull2023}.
Bacteria and fungi use the mechanism of horizontal gene transfer to pass DNA from a donor cell to a receiver cell in order to make themselves more resistant~\cite{nimmakayala2019}.

These reticulation events can not be represented in a phylogenetic tree.
Biologists have therefore generalized phylogenetic trees and have introduced the concept of \emph{phylogenetic networks}.
A phylogenetic network is a directed acyclic graph~(DAG)
in which there is exactly one vertex with an in-degree of~0 and each vertex with an in-degree bigger than~1---which are called \emph{reticulations}---has an out-degree of~1.
As the name suggests, in the reticulations, reticulation events can be represented.
Similarly as in phylogenetic trees, vertices with an out-degree of~0 represent (a subset of) present day taxa and the other vertices represent assumed extinct ancestors.
Weights on the edges of a phylogenetic network mark a distance between the two endpoints and can stand for a number of features or an estimated evolutionary distance.
We note that this definition of phylogenetic networks is sometimes referred to as phylogenetic \emph{reticulate} networks or \emph{explicit} phylogenetic networks and that other types of phylogenetic networks exist.
An overview of different types of phylogenetic networks can be found in~\cite{huson2006,phylogeneticNetworks}.
An example of a phylogenetic reticulate network is given in Figure~\ref{fig:example-network}.
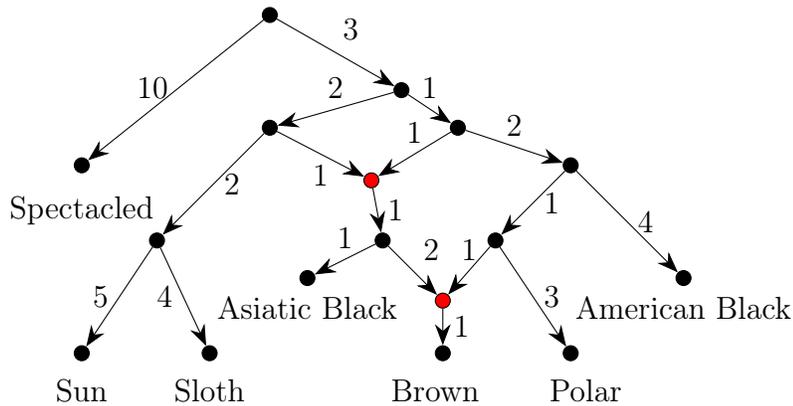
\begin{figure}[t]
	\centering
	\begin{tikzpicture}
		\tikzstyle{txt}=[circle,fill=white,draw=white,inner sep=0pt]
		\tikzstyle{nde}=[circle,fill=black,draw=black,inner sep=2pt]
		
		\node[txt] at (2,2.4) {Spectacled};
		\node[txt] at (2,0) {Sun};
		\node[txt] at (3.7,0) {Sloth};
		\node[txt] at (5,1.1) {Asiatic Black};
		\node[txt] at (10,1.1) {American Black};
		\node[txt] at (8.7,0) {Polar};
		\node[txt] at (6.7,0) {Brown};

		\node[nde] (r) at (4.5,5) {};
		\node[nde] (spec) at (2,3) {};
		\node[nde] (sun) at (2,.5) {};
		\node[nde] (sloth) at (3.7,.5) {};
		\node[nde] (asi) at (5,1.5) {};
		\node[nde] (am) at (10,1.5) {};
		\node[nde] (pol) at (8.5,.5) {};
		\node[nde] (brown) at (6.8,.5) {};

		\node[nde] (v1) at (6.25,4) {};
		
		\node[nde] (v2) at (4.5,3.5) {};
		\node[nde] (v21) at (3,2) {};
		
		\node[nde] (v3) at (7,3.5) {};
		
		\node[nde,fill=red] (v4) at (5.85,2.8) {};
		\node[nde] (v41) at (6,2) {};
		\node[nde,fill=red] (v411) at (6.8,1.2) {};
		
		\node[nde] (v32) at (8.5,3) {};
		\node[nde] (v321) at (7.5,2) {};

		\draw[-{Stealth[length=8pt]}] (4.5,5) -- node[left]{10} (spec);
		\draw[-{Stealth[length=8pt]}] (4.5,5) -- node[above right]{3} (v1);
		
		\draw[-{Stealth[length=8pt]}] (v1) -- node[above]{2} (v2);
		\draw[-{Stealth[length=8pt]}] (v1) -- node[above]{1} (v3);
		\draw[-{Stealth[length=8pt]}] (v3) -- node[above]{2} (v32);
		
		\draw[-{Stealth[length=8pt]}] (v3) -- node[above]{1} (v4);
		
		\draw[-{Stealth[length=8pt]}] (v2) -- node[below]{1} (v4);
		\draw[-{Stealth[length=8pt]}] (v4) -- node[right]{1} (v41);
		\draw[-{Stealth[length=8pt]}] (v41) -- node[above]{1} (asi);
		
		\draw[-{Stealth[length=8pt]}] (v2) -- node[right]{2} (v21);
		\draw[-{Stealth[length=8pt]}] (v21) -- node[left]{5} (sun);
		\draw[-{Stealth[length=8pt]}] (v21) -- node[left]{4} (sloth);
		
		\draw[-{Stealth[length=8pt]}] (v32) -- node[right]{4} (am);
		\draw[-{Stealth[length=8pt]}] (v32) -- node[right]{1} (v321);
		\draw[-{Stealth[length=8pt]}] (v321) -- node[right]{3} (pol);
		
		\draw[-{Stealth[length=8pt]}] (v321) -- node[above]{1} (v411);
		\draw[-{Stealth[length=8pt]}] (v41) -- node[above right]{2} (v411);
		\draw[-{Stealth[length=8pt]}] (v411) -- node[right]{1} (brown);
	\end{tikzpicture}
	\caption{This figure depicts a likely heritage of several bears in a weighted phylogenetic network~\cite{kumar}.
		The two reticulations are depicted in red.}
	\label{fig:example-network}
\end{figure}%

%\todosi{
%Explain \MAPPD and \MaxNPD
%}

Phylogenetic networks generalize phylogenetic trees.
Therefore, the question of what measure generalizes phylogenetic diversity on phylogenetic networks naturally arises.
First concepts have been introduced for so-called phylogenetic split systems~\cite{volkmann2014,chernomor2016}
and later also for explicit phylogenetic reticulate networks~\cite{WickeFischer2018,bordewichNetworks}.

Arguably the most easy-to-understand measure of phylogenetic diversity of a set of taxa~$A$ in a phylogenetic network~\Net is \emph{all-paths phylogenetic diversity}, denoted with~$\apPD(A)$,
which was first defined as ``phylogenetic subnet diversity'' in~\cite{WickeFischer2018}.
In all-paths phylogenetic diversity, one simply adds the weight of all the edges~$uv$ for which there is a path from~$v$ to a taxon in~$A$.

In all-paths phylogenetic diversity it is therefore assumed that taxa which result from reticulation events have all the features of the parents.
A more realistic assumption is that all parents contribute some features at some \iprop.
As a consequence, to compute \emph{network phylogenetic diversity}, denoted with~$\NetPD(A)$, of a set of taxa~$A$ in a phylogenetic network~$\Net$ we need to know the \iprop~$p(e)$ of every edge~$e$ incoming at some reticulation~$v$.
Then, if a taxon below~$v$ is selected, the features above~$e$ only survive with a probability~$p(e)$---if not covered by a better path.
Observe that if all the \iprops are~1, then the network phylogenetic diversity takes the same value as the all-paths phylogenetic diversity.

Corresponding to these two measures, decision problems~\MAPPDlong (\MAPPD) and~\MaxNPD were introduced~\cite{bordewichNetworks} in which we are given a phylogenetic network, two integers~$k$ and~$D$, and \iprops in the case of the latter problem.
It is asked whether a subset~$S$ of taxa exists such that~$S$ has a size of at most~$k$ and~$\apPD(S) \ge D$ or~$\NetPD(S) \ge D$, respectively.

Based on a known hardness reduction for \MAPPD~\cite{bordewichNetworks}, we establish the \Wh{2}-hardness with respect to the size~$k$ of a solution.
We then show that \MAPPD is, however, \FPT when parameterized by the threshold of diversity~$D$ or by the acceptable loss of diversity~$\Dbar$.
We further show that \MAPPD is \FPT with respect to the number of reticulations and with respect to the number of edges incoming in reticulations and admits a kernelization of polynomial size.

As \MaxNPD generalizes \MAPPD, the hardness-results for \MAPPD also hold for \MaxNPD.
We moreover show that \MaxNPD remains \NP-hard in instances in which the phylogenetic network has a level of~1.
The level measures the tree-likeness of a phylogenetic network.

\section{Structure of this Thesis}
In the following chapter, we formally define phylogenetic diversity in the considered variants and the problems regarded in this thesis.
Furthermore, we give an overview of further relevant definitions such as terms of computational complexity.
At the end of Chapter~\ref{ctr:prelims}, we show parameterized complexity results for \MCKPLong and \PenSum, two problems that are unrelated to phylogenetic diversity but are relevant for results in the following chapters.

In Chapter~\ref{ctr:GNAP}, we consider \GNAPLong with its special cases;
in Chapter~\ref{ctr:TimePD}, we consider \tPDs and \tPDws;
and in Chapter~\ref{ctr:FoodWebs}, we consider \PDDlong with its special cases.
At the beginning of each of these chapters, we give an overview of the known results and then we examine the respective problem within the framework of parameterized complexity.

In Chapter~\ref{ctr:Networks}, with \MAPPDlong and \MaxNPD, we examine two problems that are based on measures of phylogenetic diversity in phylogenetic networks.
As in the chapters before, we first present known results.
Afterward, we show parameterized complexity results for \MAPPDlong before we show that \MaxNPD is \NP-hard even on networks of level~1.

Finally, in Chapter~\ref{ctr:conclusion} we discuss open problems and future research ideas.

\chapter{Preliminaries}
\label{ctr:prelims}

In this chapter, we give the fundamental definitions that we use throughout this thesis.
We start with some mathematical notation, then formally define phylogenetic diversity and \MPDlong.
Afterward, we define notions of classic and parameterized complexity.
Finally, in Sections~\ref{sec:MCKP} and~\ref{sec:PenSum}, we consider two problems, \MCKPLong and~\PS, to which we refer later in the thesis.
These problems are not about the maximization of phylogenetic diversity but are very useful for results later in the thesis.

\section{Mathematical Notation}
For an integer~$n$, let~$[n]$ denote the set~$\{1,\dots,n\}$ and let~$[n]_0$ denote the set~$\{0\}\cup [n]$.
Define $\mathbb{R}_{[0,1]} := \{x \in \mathbb{R} \mid 0 \leq x \leq 1\}$ and $\mathbb{R}_{(0,1)} := \mathbb{R}_{[0,1]} \setminus \{0,1\}$.
A \textit{partition of a set~$N$} is a family of pairwise disjoint sets~$\{ N_1, \dots, N_m \}$ such that~$\bigcup_{i=1}^m N_i=N$.

We extend functions $f: A\to B$, where~$B$ is a family of sets, to handle subsets~$A'\subseteq A$ of the domain by defining $f(A') := \bigcup_{a\in A'} f(a)$.
Functions~$f:A\to \mathbb{N}$ are extended to subsets~$A'\subseteq A$ of the domain by defining $f_\Sigma(A') := \sum_{a\in A'} f(a)$.

Throughout the thesis, we will use both natural and binary logarithms.
We will write $\ln x$ to denote the natural logarithm of $x$ (to the base $e$), and $\log_i x$ to denote the logarithm of $x$ to the base $i\in \mathbb{N}$.
If~$i$ is not further specified, then we use~$i=2$.

The \emph{encoding length} of a number is the number~$z$ of bits necessary to encode~$z$ in binary.
Let~${\myvec{0}}$ denote the multidimensional-dimensional zero.
For a vector~$\myvec{p}$ of numbers, we use the following operations.
We let~$\myvec{p}_{(j) +z}$ denote the vector~$\myvec{p}$ in which in position~$i$, value~$z$ is added.
We write~$\myvec{p}\le \myvec{q}$ if~$\myvec{p}$ and~$\myvec{q}$ have the same dimension~$d$ and~$p_i\le q_i$ for every~$i\in[d]$.

\subsection{Graph-Theoretic Notation}
\label{sec:graph-theory}
We define the graph-theoretic notion used throughout the work.
Standard monographs provide a deeper view~\cite{West00,Diestel12}.

For a set~$V$ and a set~$E \subseteq \{ uv \mid u,v \in V \}$,
a \textit{(directed) graph}~$G$ is a tuple~$(V,E)$, where~$V$ is called the \textit{set of vertices} of~$G$ and~$E$ the \textit{set of edges} of~$G$, respectively.
We write $uv$ or~$(u,v)$ for a \emph{(directed) edge} from~$u$ to~$v$, and $\{u,v\}$ for an \emph{undirected edge} between $u$ and $v$.
For any graph~$G$, we write~$V(G)$ and~$E(G)$, respectively, to denote the set of vertices and edges of~$G$.
If~$E$ is a family of undirected edges, then~$G = (V,E)$ is called an~\emph{undirected graph}.
For a directed graph~$G=(V,E)$, the directed graph~$G' := (V,E')$ with~$E' := \{ \{u,v\} \mid uv \in E \}$ is called the~\emph{underlying undirected graph}.

An edge~$e = uv$ or~$e = \{u,v\}$ is \textit{incident} with~$u$ and~$v$.
Further, the edge $uv$ is \textit{incoming} at $u$ and \textit{outgoing} from $v$. 
We say two vertices are \emph{adjacent} or \emph{neighbors} if they are incident with the same edge, and similarly we say two edges are \emph{adjacent} if they are incident with the same vertex.

The \textit{degree of a vertex~$v$} is the number of edges that are incident with~$v$.
Similarly, the \textit{in-degree} (and \textit{out-degree}, respectively) of $v$ is the number of incoming (outgoing) edges of $v$.

For a graph~$G$ and a set of vertices~$V'\subseteq V(G)$, \emph{the subgraph of~$G$ induced by~$V'$} is~$G[V']:=(V',E')$, where~$E' := \{ e = uv \in E(G) \mid u,v \in V'\}$.
Moreover, with~$G - V':=G[V\setminus V']$ with~$V' \subseteq V$, we denote the graph obtained from~$G$ by removing~$V'$ and its incident edges.
For an edge set~$E' \subseteq E$, we define~$G - E'$ as the graph~$(V,E\setminus E')$.

We say that there is \textit{a path of length~$p$ from $u$ to~$w$} if $v=w$ and~$p=0$ or there is an edge~$uv$ or~$\{u,v\}$ such that there is a path of length~$p-1$ from $v$ to $w$.
We sometimes omit the length of the path.
Two vertices~$u$ and~$v$ are \emph{connected} if there is a path from~$u$ to~$v$ or from~$v$ to~$u$.
An undirected graph~$G$ is \emph{connected} if the vertices in~$V(G)$ are pairwise connected.
For an undirected graph~$G$ and a non-empty set of vertices~$V' \subseteq V(G)$, the subgraph of~$G$ induced by~$V'$ is a \emph{connected component} if~$G[V']$ is connected, and vertices~$u$ and~$v$ are not connected for each~$v\in V(G) \setminus V'$ and each~$u\in V'$.
An undirected graph~$G$ is \emph{biconnected} if~$G - \{v\}$ is connected for every vertex~$v\in V(G)$.

A \emph{topological order} induced by a given directed graph~$G = (V,E)$ is a bijective mapping~$f: V \to [|V|]$ such that~$f(u) \le f(v)$ for every edge~$uv \in E$.

The \emph{complement graph} of an undirected graph is obtained by replacing edges with non-edges and vice versa.

\paragraph{Other Graph Classes.}
An undirected graph~$G$ is \emph{cyclic} if there are\lb vertices~$u,v,w\in V(G)$ such that neighbors~$u$ and~$w$ are neighbors of~$v$ and~$u$ and~$w$ are connected in~$G - \{v\}$.
An undirected graph is not cyclic is \emph{acyclic}.
A self-loop is an edge~$uu$ for some vertex~$u$.
A \emph{directed, acyclic graph} (DAG) is a directed graph that does not contain self-loops and we can conclude~$u=v$ if there is a path from~$u$ to~$v$ and from~$v$ to~$u$.

An undirected graph~$G = (V,E)$ is \emph{bipartite} if there is a partition~$\{V_1, V_2\}$ of the vertex set~$V$ such that~$e \cap V_1$ and~$e \cap V_2$ are both non-empty for each edge~$e\in E$.
In other words, each edge is between~$V_1$ and~$V_2$.
The vertex sets~$V_1$ and~$V_2$ are a \emph{vertex bipartition}.

An undirected graph is a \emph{cluster graph} if the existence of edges~$\{u,v\}$ and~$\{v,w\}$ implies the existence of the edge~$\{u,w\}$.
Equivalently, a graph is a cluster graph if every connected component is a clique.
An undirected graph is a \emph{co-cluster graph} if its complement graph is a cluster graph.
In other words, a graph is a co-cluster graph if its vertex set can be partitioned into independent sets such that each pair of vertices from different independent sets is adjacent.

For a graph class~$\Pi$ defined exclusively on undirected graphs,
we say that a directed graph is of class~$\Pi$ if the underlying undirected graph is of class~$\Pi$.

\paragraph{Trees.}
A \textit{rooted tree~$T$ with root~$\rho$} is a directed, acyclic graph with~$\rho\in V(T)$ where each vertex of~$T$ can be reached from~$r$ via exactly one path.
A vertex~$v$ of a tree~$T$ is a \textit{leaf} when the out-degree of~$v$ is~0.
We refer to the non-leaf vertices of a tree as the \emph{internal vertices}.
In a rooted tree, the~\textit{height of a vertex~$v$} is the length of the (only) path from the root~$\rho$ to~$v$ for each vertex~$v$.
The~\textit{height~$\height_T$ of a rooted tree~$T$} is the maximal height of one of the vertices of~$T$.
A~\textit{star} is a tree with a height of~1.

For an edge~$uv$ of a rooted tree, we call~$u$~\textit{the parent of~$v$} and~$v$~\textit{a child of~$u$}.
For two edges $e$ and $\hat e$ incident with the same vertex $v$, we say $\hat e$ is the \emph{parent-edge} of $e$ if $\hat e = uv$ and $e = vw$. If $e = vu$ and $\hat e = vw$, we say $\hat e$ is a \emph{sibling-edge} of $e$.
For a vertex~$v$ with parent~$u$, the \textit{subtree $T_v$ rooted at $v$} is the connected component containing $v$ in~$T - \{u\}$.
In the special case that $\rho$ is the root of $T$, we define $T_\rho := T$.
For a vertex $v$ with children~$w_1,\dots,w_t$, the \textit{$i$-partial subtree $T_{v,i}$ rooted at~$v$} for~$i\in [t]$ is the connected component containing~$v$ in~$T_v - \{ w_{i+1}, \dots, w_t \}$.

For an undirected graph~$G$ and a set of vertices~$V' \subseteq V(G)$, a \emph{spanning tree of~$V'$} is a tree~$T$ with~$V' \subseteq V(T) \subseteq V(G)$ and~$E(T) \subseteq E(G)$.
The \emph{size of a spanning tree} is the number of vertices.
A \emph{minimum spanning tree of~$V'$} is a spanning tree of~$V'$ of minimum size.
For a tree~$T = (V,E)$ and a vertex set~$V'\subseteq V$, the \emph{minimum spanning tree of~$V'$} is denoted by~$\spannbaumsub{T}{ V' }$.

\paragraph{Graph Parameters.}
For an undirected graph $G = (V,E)$ and a graph class~$\Pi$ the \emph{distance to~$\Pi$} of~$G$ is the smallest number~$d$ such that there exists a set $Y\subseteq V$ of~$d$ such that~$G - Y$ is of class~$\Pi$.
The set~$Y$ in this definition is called a~\emph{modulator} to~$\Pi$.

The \emph{max leaf number} of an undirected graph~$G = (V,E)$ is the maximum number of leaves a spanning tree of~$V$ has.

\begin{definition}[Tree Decomposition]
	\label{def:tw}
	A \emph{tree decomposition} of a directed graph~$G$ is a rooted tree~$T$ and a mapping which assigns each node $t \in T$ a subset $Q_t \subseteq V(G)$ such that
	\begin{propEnum}
		\item every vertex of~$G$ is contained in some bag,
		\item for every edge~$e = \{u,v\}$, there is some bag containing~$u$ and~$v$, and
		\item the nodes of~$T$ whose bags contain any particular vertex~$v$ are connected.
	\end{propEnum}
	
	A tree decomposition is \emph{nice} \cite{CyganTreeDecomp} if
	the bags of the root and all leaves of~$T$ are empty,
	and every non-leaf node $t$ has one of the following types
	\begin{propEnum}
		\item[~] \emph{Introduce vertex node}: $t$ has only one child~$t'$ and $Q_{t'} = Q_t \setminus \{v\}$ for some vertex~$v$ which is said to \emph{be introduced} at~$t$;
		\item[~] \emph{Introduce edge node}: $t$ has only one child~$t'$ and $Q_{t'} = Q_t$ and $t$ is labeled with an edge $e \subseteq Q_t$ which is said to \emph{be introduced} at~$t$;
		\item[~] \emph{Forget node}: $t$ has only one child $t'$ and $Q_{t'} = Q_t \cup \{v\}$ for some vertex~$v$;
		\item[~] \emph{Join node}: $t$ has exactly two children $t_1$, and~$t_2$ and $Q_t = Q_{t_1} = Q_{t_2}$.
	\end{propEnum}
	Every edge of~$G$ must be introduced exactly once, and we may assume this happens as high in~$T$ as possible,
	i.e. we can introduce edges right before some of their endpoints is forgotten.
	For a node $t\in T$, we denote by $G_t = (V_t, E_t)$ the subgraph
	of~$G$ that contains exactly those vertices~$V_t$ and edges~$E_t$ and edges that are introduced at~$t$ or any descendant of~$t$.
\end{definition}

If we do not define introduce edge nodes, then all edges are introduced once both vertices are introduced.
The \emph{width} of a tree decomposition is the size of the biggest bag minus 1.
The \emph{treewidth} of an undirected graph~$G$ is the minimum width of a tree decomposition of~$G$.

For a graph parameter~$\kappa$ defined exclusively on undirected graphs,
we define~$\kappa$ on directed graphs with the size of~$\kappa$ in the underlying undirected graph.

\section{Phylogenetics}
For a given set $X$, a \emph{phylogenetic~$X$-tree~$\Tree=(V,E,\w)$} (in short,~$X$-tree or phylogenetic tree) is a tree~$T=(V,E)$ with an \emph{edge-weight} function~$\w: E\to \mathbb{N}_{>0}$ and a bijective labeling of the leaves with elements from~$X$ where all non-leaves in \Tree have out-degree at least~2.
The set~$X$ is a set of~\emph{taxa}~(species).
Because of the bijective labeling, we interchangeably use the words taxon and leaf.
We write $\max_\w$ to denote the biggest edge weight in~$\Tree$.
An~$X$-tree~$\Tree$ is \textit{ultrametric} if there is an integer~$p$ such that the sum of the weights of the edges of the path from the root~$\rho$ to~$x_i$ equals~$p$ for every leaf~$x_i$.

If there is a directed path from $u$ to $v$ in $\Tree$ (including when~$u=v$ or when~$v$ is a child of $u$), we say that $u$ is an \emph{ancestor} of $v$ and $v$ is a \emph{descendant} of $u$.
If in addition~$u \neq v$, we say $u$ is a \emph{strict ancestor} of $v$ and $v$ a \emph{strict descendant} of $u$.
The sets of ancestors and descendants of $v$ are denoted $\anc(v)$ and $\desc(v)$, respectively.
The set of descendants of~$v$ which are in~$X$ are \emph{offspring} $\off(v)$ of a vertex $v$.
Moreover, we denote $\off(e) = \off(v)$ for an edge $e = uv \in E$.

\paragraph{Phylogenetic Diversity.}
Given a subset of taxa $A \subseteq X$,
let $E_{\Tree}(A)$ denote the set of edges in $e \in E$ with $\off(e)\cap A \neq \emptyset$.
The \emph{phylogenetic diversity} $\PD(A)$ of $A$ is defined by 

\begin{equation}
	\label{eqn:PDdef}
	\PD(A) := \sum_{e \in E_{\Tree}(A)}\w(e).
\end{equation}
That is, $\PD(A)$ is the total weight of all edges $uv$ in $\Tree$ so that there is a path from~$v$ to a vertex in~$A$. 
The phylogenetic diversity model assumes that features of interest appear along the edges of the tree with frequency proportional to the weight of that edge and that any feature belonging to one species is inherited by all its descendants.
Thus, $\PDsub{\Tree}(Z)$ corresponds to the expected number of distinct features appearing in all species in $Z$.

See Figure~\ref{fig:example-PD} for an example of the definition of phylogenetic diversity.

\begin{figure}[t]
	\centering
	\begin{tikzpicture}[scale=1,every node/.style={scale=0.9}]
		\tikzstyle{txt}=[circle,fill=white,draw=white,inner sep=0pt]
		\tikzstyle{nde}=[circle,fill=black,draw=black,inner sep=2.5pt]
		
		\node[nde] (v0) at (1,10) {};
		\node[nde] (v1) at (0,9) {};
		\node[nde] (v2) at (2,9) {};
		\node[nde] (v3) at (1,8) {};
		\node[nde] (v4) at (3,8) {};
		\node[nde] (v5) at (0.6,7) {};
		\node[nde] (v6) at (1.4,7) {};
		\node[nde] (v7) at (2.6,7) {};
		\node[nde] (v8) at (3.4,7) {};
		
		\node[txt,xshift=5mm,yshift=-6mm] (c1) [above=of v1] {$1$};
		\node[txt,xshift=-5mm,yshift=-6mm] (c3) [above=of v2] {$1$};
		\node[txt,xshift=5mm,yshift=-6mm] (c1) [above=of v3] {$5$};
		\node[txt,xshift=-5mm,yshift=-6mm] (c3) [above=of v4] {$10$};
		\node[txt,xshift=1mm,yshift=-6mm] (c1) [above=of v5] {$2$};
		\node[txt,xshift=-1mm,yshift=-6mm] (c1) [above=of v6] {$4$};
		\node[txt,xshift=1mm,yshift=-6mm] (c1) [above=of v7] {$1$};
		\node[txt,xshift=-1mm,yshift=-6mm] (c1) [above=of v8] {$2$};

		\node[txt,yshift=10mm] (x1) [below=of v1] {$x_1$};
		\node[txt,yshift=10mm] (x2) [below=of v5] {$x_2$};
		\node[txt,yshift=10mm] (x3) [below=of v6] {$x_3$};
		\node[txt,yshift=10mm] (x4) [below=of v7] {$x_4$};
		\node[txt,yshift=10mm] (x5) [below=of v8] {$x_5$};

		\draw[thick,arrows = {-Stealth[length=8pt]}] (v0) to (v1);
		\draw[line width=0.03cm,blue,dashed,arrows = {-Stealth[length=8pt]}] (v0) to (v2);
		\draw[line width=0.03cm,blue,dashed,arrows = {-Stealth[length=8pt]}] (v2) to (v3);
		\draw[line width=0.03cm,blue,dashed,arrows = {-Stealth[length=8pt]}] (v2) to (v4);
		\draw[line width=0.03cm,blue,dashed,arrows = {-Stealth[length=8pt]}] (v3) to (v5);
		\draw[line width=0.03cm,blue,dashed,arrows = {-Stealth[length=8pt]}] (v3) to (v6);
		\draw[thick,arrows = {-Stealth[length=8pt]}] (v4) to (v7);
		\draw[line width=0.03cm,blue,dashed,arrows = {-Stealth[length=8pt]}] (v4) to (v8);
	\end{tikzpicture}
	
	\caption{An example of a phylogenetic~$X$-tree $\Tree$ with taxa $X = \{x_1,x_2,x_3,x_4,x_5\}$.
		The set $A = \{x_2, x_3, x_5\}$ has a phylogenetic diversity of~$\PD(A) = 1 + 5 + 10 + 2 +4 + 2 = 24$.
		The set of edges $E_{\Tree}(A)$ is blue and dashed.
	}
	\label{fig:example-PD}
\end{figure}
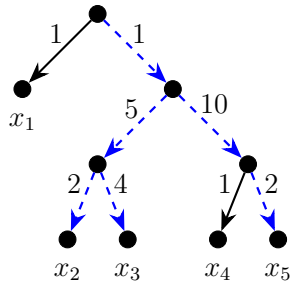%

In \Cref{sec:GNAP-PD-def} and \Cref{sec:Net-PD-def}, we present the generalizations of phylogenetic diversity to a) handle that species do not survive indefinite but at a given probability and b) phylogenetic networks, respectively.

\paragraph{Phylogenetic Networks.}
A \textit{phylogenetic $X$-network~$\Net=(V,E,\w)$} (in\lb short, $X$-network or network) is a directed acyclic graph with an \textit{edge-weight}\lb function~$\w: E\to \mathbb{N}_{>0}$ and a single vertex with an in-degree of~$0$ (the \textit{root~$\rho$}), in which the vertices with an out-degree of~$0$ (the \textit{leaves}) have an in-degree of~$1$ and are bijectively labeled with elements from a set~$X$, and such that all vertices either have an in-degree of at most~$1$ or an out-degree at most~$1$.
The vertices with in-degree at least~$2$ and out-degree~$1$ are called \textit{reticulations}; the other non-leaf vertices are called \textit{tree vertices}.
Every edge incoming at a reticulation is a \emph{reticulation edge}.
A network is \emph{binary} if the root has an out-degree of 2 and every other non-leaf vertex has a degree of~3.

The \emph{number of reticulations} of a network~$\Net$ is denoted with~$\ret$.
The \emph{number of reticulation-edges}~$\eret$ of a network~$\Net$ is the number of edges that need to be removed such that~\Net is a tree.
We note that this is not state what the number of the reticulation edges is.
Some authors refer to the number of reticulation-edges as the reticulation number of~$\Net$.
If $\Net$ is binary, then the number of reticulations is exactly the number of reticulation-edges.
The \emph{level} of a network~$\Net$ is the maximum reticulation number of a subgraph~$\Net[V']$ for some~$V' \subseteq V(\Net)$ where we require that the underlying undirected graph of~$\Net[V']$ is biconnected.

\section{Classic and Parameterized Complexity}
In this section, we provide the necessary background in complexity theory, both for classical concepts of \NP-hardness and \NP-completeness, and for parameterized complexity.
Readers are referred to monographs for classic computational complexity~\cite{GareyJohnson79,Papadimitriou07,AB09},
and to monographs for parameterized complexity~\cite{FG06,N06,downeybook,cygan} for more detailed introductions.

\subsection{Classic Computational Complexity}
Within the field of classic computational complexity, we deal with \emph{computational problems} and the question of how many computational resources are required to find an answer to the raised question.
In this thesis, we only consider \emph{decision problems}~$\Pi$ (problems for short).
The closely related \emph{optimization problems} can usually be solved with only a small overhead, if an algorithm is given that solves the decision variant.
We consider the following formal definitions.

\begin{definition}[Languages, Decision Problems and Instances]
	\label{def:lang}
	~
	\begin{propEnum}
		\item A \emph{language}~$L \subseteq \Sigma^*$ is a set of (finitely long) strings over a finite \emph{alphabet}~$\Sigma$.
		\item The decision problem associated with a language~$L$ is to determine if~$x\in L$ for any given \emph{instance}~$x\in \Sigma^*$.
		\item An instance~$x\in \Sigma^*$ is a \emph{\yes-instance} if~$x\in L$ and a \emph{\no-instance} if~$x\notin L$.
	\end{propEnum}
\end{definition}

An \emph{algorithm} is a defined order of computational steps to solve a computational problem.
An algorithm~$\mathcal{A}$ solves a decision problem~$\Pi$ if, after a finite number of computation steps,~$\mathcal{A}$ correctly determines whether an instance is a \yes-instance or a \no-instance.
If such an algorithm exists, then~$\Pi$ is \emph{decidable}.
In this thesis, we only consider decidable decision problems.

The complexity measures of algorithms are the computational resources that are required to let the algorithm run on any instance, usually in correspondence with the \emph{encoding length}~$|x|$ of the instance~$x$.
The most frequently used complexity measures are \emph{time} and~\emph{space}.
With time, we refer to the number of steps the computation needs, and with space we refer to the required memory.
Throughout this thesis, we only consider the running time of algorithms.
The complexity of a decision problem~$\Pi$ is given by the fastest, usually unknown algorithm solving~$\Pi$.

In our running time analyses, we assume a unit-cost RAM model where arithmetic addition and multiplication for numbers of any length have a constant running time.
This model is unrealistically strong but avoids that we have to add factors to the running time which are not all too interesting from a theoretical point of view.
We will, in order to have more clarity, recall this fact in the running time analyses of algorithms in which we operate with numbers that are potentially bigger than~$\Oh(\log(|\Instance|))$, where~$|\Instance|$ is the size of the entire instance.

There are several \emph{complexity classes} for decision problems, of which~\P and \NP are the most prominent.

\newpage
\begin{definition}[The classes \P and \NP]
	\label{def:P+NP}
	~
	\begin{propEnum}
		\item A decision problem is in \P if a \emph{deterministic} Turing machine exists which,\lb for each given instance~$x\in \Sigma^*$, determines whether~$x$ is a \yes-instance\lb within~$\poly(|x|)$ time.
		\item A decision problem is in \NP if a \emph{non-deterministic} Turing machine exists which, for each given instance~$x\in \Sigma^*$, accepts~$x$ within~$\poly(|x|)$ time.
	\end{propEnum}
\end{definition}

A full definition of non-determinism is beyond the scope of this work.
Informally, in a non-deterministic Turing machine some configurations may have several successor configurations.
A Turing machine \emph{accepts} an input~$x$ if it is possible to reach an accepting configuration from the starting configuration.

Every problem in~\P is also in~\NP, but it is assumed that certain problems are in~\NP but not in~\P, and, therefore,~\PneqNP.
In other words, it is believed that there are problems in \NP that are not deterministically solvable in polynomial time---the \NP-complete problems.
To mathematically define these, we introduce the concept of reductions, to show that a problem is at least as hard as another.
For two decision problems~$\Pi_1,\Pi_2 \subseteq \Sigma^*$,
a \emph{reduction} from~$\Pi_1$ to~$\Pi_2$ is an algorithm---a computable function~$f: \Sigma^* \to \Sigma^*$---which takes instances of~$\Pi_1$ and returns instances of~$\Pi_2$, and
we require that~$x\in \Sigma^*$ is a \yes-instance of $\Pi_1$ if and only if~$f(x)\in \Sigma^*$ is a \yes-instance of $\Pi_2$.
The instances~$x$ and~$f(x)$ are then called~\emph{equivalent}.
A reduction~$f$ is a \emph{polynomial-time reduction} if~$f(x)$ is computed within~$\poly(|x|)$~time.
A decision problem~$\Pi$ is \emph{\NP-hard} if there is a polynomial-time reduction from~$\Psi$ to~$\Pi$ for every decision problem $\Psi$ in \NP.
A decision problem~$\Pi$ is~\NP-complete if~$\Pi$ is~\NP-hard and in~\NP.
If there was an \NP-complete problem which is also in~\P, then all problems in \NP could be solved deterministically in polynomial time.
As discussed earlier, it is not believed that such a decision problem exists.
In Section~\ref{sec:probs}, a small list of \NP-complete problems is given.

A function~$f: \mathbb{N} \to \mathbb{N}$ is in~$\Oh(g(n))$ for another function~$g: \mathbb{N} \to \mathbb{N}$, if there are integers~$c$ and~$m$ such that~$f(n) \le c \cdot g(n)$ for each integer~$n\ge m$.
Similarly, function~$f: \mathbb{N} \to \mathbb{N}$ is in~$\Oh^*(g(n))$ for another function~$g: \mathbb{N} \to \mathbb{N}$, if there is an integer~$m$ such that~$f(n) \le g(n) \cdot \poly(n)$ for each integer~$n\ge m$.
In other words, if~$f(n) \in \Oh(g(n))$ then~$g$ grows at least as fast as~$f$.
We use in the $\Oh$ and\lb the~$\Oh^*$-notation to describe the running times of algorithms and in the~$\Oh^*$-notation we omit factors polynomial in the input size.

\subsection{Parameterized Complexity} 
\label{sec:para-complexity}
In this section, we give a brief overview of formal definitions of parameterized complexity.

\paragraph*{Fixed-Parameter Tractability.}
Here, we first define what a parameterized problem is before we define the complexity classes \FPT and \XP.
The formal definition for parameterized languages and parameterized decision problems are similar to the definition of languages and decision problems given in Definition~\ref{def:lang}.

\begin{definition}[Parameterized Languages and Parameterized Problems]
	\label{def:paraLan}
	~
	\begin{propEnum}
		\item A \emph{parameterized language}~$L \subseteq \Sigma^* \times \mathbb{N}_0$ is a set of tuples~$(x,k)$, where $x$ is a (finitely long) string over a finite alphabet~$\Sigma$ and~$k$ is an integer.
		\item We call~$x$ the \emph{input},~$k$ the \emph{parameter}, and~$(x,k)$ the \emph{instance}.
		\item The \emph{parameterized decision problem} associated with a parameterized language~$L$ is to determine if~$(x,k)\in L$ for any given \emph{instance}~$(x,k)\in \Sigma^* \times \mathbb{N}_0$.
		\item An instance~$(x,k) \in \Sigma^* \times \mathbb{N}_0$ is a \emph{\yes-instance} if~$(x,k)\in L$ and a \emph{\no-instance} if~$(x,k)\notin L$.
	\end{propEnum}
\end{definition}

The \emph{encoding length} of an instance~$(x,k)$ is denoted by~$|(x,k)|$.
We usually simply use the term problem, when the context makes it clear that it is a parameterized decision problem.
Now, with the definition of parameterized decision problems at hand, we define the two most relevant classes of parameterized decision problems, \FPT and \XP.

\begin{definition}[Fixed-Parameter Tractability]
	\label{def:FPT}
	~
	\begin{propEnum}
	\item\label{def:FPT-algo}An algorithm~$\mathcal{A}$ for a parameterized decision problem~$\Pi$ is \emph{fixed-parameter} if for every instance~$(x,k) \in \Sigma^* \times \mathbb{N}_0$, within~$f(k) \cdot \poly(|x|)$~(deterministic) time the algorithm~$\mathcal{A}$ correctly determines whether~$(x,k)$ is a \yes- or a \no-instance of~$\Pi$.
	Here, the function~$f$ can be any computable function that only depends on~$k$.
	We write in short that~$\mathcal{A}$ is an~\FPT-algorithm.
	\item\label{def:FPT-prop}A parameterized decision problem is \emph{fixed-parameter tractable} (\FPT) if it can be solved by a fixed-parameter algorithm.
	\item\label{def:FPT-class}\FPT is the complexity class of all parameterized decision problems that are fixed-parameter tractable.
	\end{propEnum}
\end{definition}

An algorithm has a \emph{pseudo-polynomial running time} if it has a running time polynomial in the input size when all numbers are encoded in unary, but the running time is not polynomial in the input size when all numbers are encoded in binary.
A parameterized decision problem is \emph{polynomial fixed-parameter tractable} (\PFPT) if the function~$f$ in Definition~\ref{def:FPT}(\ref{def:FPT-algo}) is a polynomial.

\begin{definition}[Slice-wise Polynomial (\XP)]
	\label{def:XP}
	~
	\begin{propEnum}
		\item An algorithm~$\mathcal{A}$ for a parameterized decision problem~$\Pi$ is \emph{slice-wise polynomial} if for every instance~$(x,k) \in \Sigma^* \times \mathbb{N}_0$, within~$\poly(|(x,k)|) ^ {g(k)}$~(deterministic) time the algorithm~$\mathcal{A}$ correctly determines whether~$(x,k)$ is a \yes- or a \no-instance of~$\Pi$.
		Here, the function~$g$ can be any computable function that only depends on~$k$.
		We write in short that~$\mathcal{A}$ is an~\XP-algorithm.
		\item A parameterized decision problem is slice-wise polynomial (\XP) if it can be solved by an \XP-algorithm.
		\item \XP is the complexity class of all parameterized decision problems that are slice-wise polynomial.
	\end{propEnum}
\end{definition}

We observe that the definition of \FPT is stricter than the definition of \XP.
Consequently, any \FPT-algorithm is an \XP-algorithm, but the converse is not given.
We conclude that~$\FPT \subseteq \XP$.

We observe further that problems which are \XP can be solved in polynomial time if we require the parameter to be a constant.
In order to show that a parameterized decision problem~$\Pi$ is not \XP, we therefore show that~$\Pi$ is \NP-hard even for constant values of the parameter.
Next, we define how we exclude that a parameterized decision problem is \FPT under some complexity-theoretic assumptions.

\paragraph*{The W-Hierarchy.}
Of course, we do not assume that all parameterized decision problems are \FPT.
Proving that a parameterized decision problem~$\Pi$ is \NP-hard for a constant size of the parameter shows that~$\Pi$ is not in \XP (assuming \PneqNP).
Because~$\FPT \subseteq \XP$ we can also conclude that~$\Pi$ is not in \FPT in that case.
But then automatically the question arises of how to provide evidence that a problem in \XP is not in \FPT.

To answer this question, Downey and Fellows defined the \emph{W-Hierarchy}~\cite{DF95,downey,downeybook}.
%For each positive integer~$i$, a complexity class~\Wh{$i$} is defined.
%
This hierarchy is defined with parameterized reductions which generalize reductions.

\begin{definition}[Parameterized Reductions]
	\label{def:paraReductions}
	For two parameterized decision problems~$\Pi_1$ and~$\Pi_2$ and two computable functions~$f,g: \mathbb{N} \to \mathbb{N}$,
	a \emph{parameterized reduction} from~$\Pi_1$ to~$\Pi_2$ is an algorithm~$\mathcal{A}$ which takes instances~$(x,k)\in \Sigma^* \times \mathbb{N}$ of~$\Pi_1$ and returns instances~$(x',k')\in \Sigma^* \times \mathbb{N}$ of~$\Pi_2$ such that
	\begin{propEnum}
		\item $(x,k)$ is a \yes-instance of $\Pi_1$ if and only if~$(x',k')$ is a \yes-instance of $\Pi_2$, and
		\item the computation of~$\mathcal{A}$ takes~$f(k) \cdot \poly(|x|)$~time, and
		\item $k' \le g(k)$.
	\end{propEnum}
\end{definition}

Essential for the W-hierarchy are the following problems.
In \WCS{t}{d}, we are given a logic circuit with a depth of~$d$ and a weft of~$t$.
It is asked whether there is an input that satisfies the circuit.
The depth of a circuit is the maximum number of nodes on a path from an input variable to the output and
the weft of a circuit is the maximum number of nodes with in-degree greater than~2 on a path from an input variable to the output.
It is evident that \WCS{t}{d} generalizes the famous \SAT problem.
The parameter we consider in the definition of the W-hierarchy is the weft of the circuit.

\begin{definition}[The W-Hierarchy]
	\label{def:WHierarchy}
	~
	\begin{propEnum}
		\item The complexity class~\Wh{$i$} for an integer~$i$ consists of the parameterized decision problems~$\Pi$ for which there is a parameterized reduction from~$\Pi$ to~\WCS{i}{d} for some~$d \in \mathbb{N}$.
		\item A parameterized decision problem~$\Pi$ is~\emph{\Wh{$i$}-hard} if there is a parameterized reduction from~$\Psi$ to~$\Pi$ for every parameterized decision problem~$\Psi$ in~\Wh{$i$}.
		\item A parameterized decision problem~$\Pi$ is~\emph{\Wh{$i$}-complete} if~$\Pi$ is~\Wh{$i$}-hard and in~\Wh{$i$}.
	\end{propEnum}
\end{definition}

It is know that~$\FPT \subseteq \Wh{1} \subseteq \Wh{2} \subseteq \dots \subseteq \XP$.
Further, it is widely believed that~$\FPT \subsetneq \Wh{1} \subsetneq \Wh{2} \subsetneq \dots \subsetneq \XP$,
but like \PneqNP, this claim is not yet proven.

\paragraph*{Kernelization.}
For parameterized decision problems, the concept of kernelization provides a means of measuring how much the preprocessing decreases the size of the instance.

\begin{definition}[Kernelization Algorithms]
	\label{def:kernelization}
	A \emph{kernelization algorithm} (also\lb kernelization, kernel or problem kernel) for a parameterized decision problem~$\Pi$\lb is an algorithm~$\mathcal{A}$ which takes instances~$(x,k)\in \Sigma^* \times \mathbb{N}$ of~$\Pi$ and returns\lb instances~$(x',k')\in \Sigma^* \times \mathbb{N}$ of~$\Pi$ such that
	\begin{propEnum}
		\item the computation of~$\mathcal{A}$ takes time polynomial in~$|(x,k)|$, and
		\item $(x,k)$ is a \yes-instance of $\Pi$ if and only if~$(x',k')$ is a \yes-instance of $\Pi$, and
		\item $|x'| + k' \le g(k)$ for some computable function~$g: \mathbb{N} \to \mathbb{N}$.
	\end{propEnum}
\end{definition}

Kernelization algorithms are usually explained in several steps.
These steps are called reduction rules.
A \emph{(data) reduction rule} for a parameterized decision problem~$\Pi$ is an algorithm that given an instance~$(x,k)$ of~$\Pi$ returns an instance~$(x',k')$ of~$\Pi$.
A reduction rule is \emph{correct} if~$(x,k)$ is a \yes-instance of~$\Pi$ if and only if~$(x',k')$ is a \yes-instance of~$\Pi$.
We say that a reduction rule has been \emph{exhaustively applied} on an instance if an application does not change the instance.

Kernelizations have a close connection to the class \FPT, as we see in this theorem.
\begin{theorem}[\cite{cygan,downeybook}]
	A parameterized decision problem~$\Pi$ admits a problem kernel if and only if~$\Pi$ is \FPT.
\end{theorem}

A \emph{polynomial kernelization (algorithm)} is a kernelization where the function~$g$ in Definition~\ref{def:kernelization} is a polynomial.
To disprove that a parameterized decision problem admits a polynomial kernelization, we use one of the following two concepts.
All these concepts are based on the assumption that~\NPcoNPpoly.
%This assumption is not as strong as \PneqNP.
%
Even though it is not proven yet, it is widely believed that \NPcoNPpoly.
In particular, if \NP~$\subseteq$~\texttt{coNP/poly}, then the polynomial hierarchy would collapse at the third level~\cite{yap1983}.

\begin{definition}[Polynomial Parameter Transformation]
	For two parameterized decision problems~$\Pi_1$ and~$\Pi_2$,
	a \emph{polynomial parameter transformation (PPT)} is an algorithm~$\mathcal{A}$ which maps instances~$(x,k)\in \Sigma^* \times \mathbb{N}$ of~$\Pi_1$ to instances~$(x',k')\in \Sigma^* \times \mathbb{N}$ of~$\Pi_2$ such that
	\begin{propEnum}
		\item the computation of~$\mathcal{A}$ takes time polynomial in~$|(x,k)|$, and
		\item $(x,k)$ is a \yes-instance of $\Pi$ if and only if~$(x',k')$ is a \yes-instance of $\Pi$, and
		\item $k' \le p(k)$ for some polynomial~$p: \mathbb{N} \to \mathbb{N}$.
	\end{propEnum}
\end{definition}

If there is a PPT from~$\Pi_1$ to~$\Pi_2$ for decision problems~$\Pi_1$ and~$\Pi_2$ with~$\Pi_1\in \NP$ and~$\Pi_2$ is \NP-hard and admits a kernel of polynomial size, then also~$\Pi_1$ admits a kernel of polynomial size.
We conclude that if there is a PPT from~$\Pi_1$ to~$\Pi_2$ and~$\Pi_1$ does not admit a polynomial kernelization, assuming \NPcoNPpoly, then also~$\Pi_2$ does not admit a polynomial kernelization, assuming \NPcoNPpoly.

Next, we define cross-compositions, another technique for excluding polynomial kernelizations.
For this definition, we require polynomial equivalence relations.

\begin{definition}[Polynomial Equivalence Relations]
	A \emph{polynomial equivalence relation} is an equivalence relation~$R$ on~$\Sigma^*$ for which the following holds.
	\begin{propEnum}
		\item There is an algorithm that for given strings~$x,y \in \Sigma^*$ can check if~$x \sim_R y$ in~$\poly(|x|+|y|)$~time, and
		\item for any finite set~$S\subseteq \Sigma^*$ there are at most~$\poly(\max_{x\in S}|x|)$ equivalence classes with regard to the relation~$R$.
	\end{propEnum}
\end{definition}

\begin{definition}[Cross-Compositions, \cite{BJK14,cygan}]
	Let a decision problems~$\Pi_1 \subseteq \Sigma^*$, 
	a parameterized decision problems~$\Pi_2 \subseteq \Sigma^* \times \mathbb{N}$, and
	a polynomial equivalence relation~$R$,
	and an integer~$t$ be given.
	A \emph{cross-composition} from~$\Pi_1$ into~$\Pi_2$ 
	is an algorithm~$\mathcal{A}$ which takes~$2^t$ instances~$x_1,\dots,x_{2^t}$ of~$\Pi_1$ that are pairwise equivalent with respect to~$R$ and returns an instance~$(x',k')$ of~$\Pi_2$ such that
	\begin{propEnum}
		\item the computation of~$\mathcal{A}$ takes time polynomial in~$\sum_{i=1}^{2^t} |x_i|$, and
		\item $x_i$ is a \yes-instance of $\Pi_1$ for some~$i \in [2^t]$ if and only if~$(x',k')$ is a \yes-instance of $\Pi_2$, and
		\item $k' \le p(t + \max_{i=1}^{2^t}|x_i|)$ for some polynomial~$p: \mathbb{N} \to \mathbb{N}$.
	\end{propEnum}
\end{definition}

\paragraph*{Exponential Time Hypothesis (ETH).}
The \emph{Exponential Time Hypothesis} (\ETH) and the \emph{Strong Exponential Time Hypothesis} (\SETH) are important conjectures in the field of parameterized complexity and have been proposed by Impagliazzo, Paturi, and Zane~\cite{IPZ01}.
Relevant to the definition of \ETH and \SETH is \qSAT{q}, a special case of \SAT.
In \qSAT{q}, we are given a formula~$\psi$ with~$n$ variables that is a conjunction of disjunctions of at most~$q$ literals.
It is asked whether there is an assignment of the variables that fulfills~$\psi$.
The problem \qSAT{q} is \NP-hard for each~$q \ge 3$~\cite{karp}.

Now, for any~$q$ let~$C_q$ be the set of numbers such that~\qSAT{q} can be solved in~$\Oh(2^{c \cdot n})$ time.
As~\qSAT{q} is a special case of~\qSAT{q+1}, we conclude~$C_q \subseteq C_{q+1}$.
Let~$\delta_q$ be the infimum of~$C_q$.
The formal definitions of \ETH and \SETH are as follows.
\begin{definition}[\ETH and \SETH]
	~
	\begin{propEnum}
		\item The Exponential Time Hypothesis (\ETH) states~$\delta_3 > 0$.
		\item The Strong Exponential Time Hypothesis (\SETH) states~$\lim\limits_{q\to \infty} \delta_q = 1$.
	\end{propEnum}
\end{definition}

It is widely believed that \ETH is correct while there are some doubts for the correctness of \SETH.
However, we again do not have a prove of the correctness or incorrectness of either \ETH or \SETH, yet.

%\paragraph*{Dynamic Programming.}
%A \emph{dynamic programming table}~$\DP$ is a data structure characterized by a certain number of dimensions and \emph{entries}.
%A \emph{dynamic programming algorithm} utilizes this table to compute and store values in its entries.
%The calculation for each entry typically depends on the values stored in previously computed entries.
%
%We sometimes define entries to store~$\infty$ or~$-\infty$, but it is sufficient to assume that a big positive or a big negative value is stored.

\paragraph*{Color Coding.}
Here, we briefly want to define some mathematical objects that are relevant for the technique of color coding which we use several times throughout this thesis.
For an in-depth treatment of color coding, we refer the reader to~\cite[Sec.~5.2~and~5.6]{cygan} and~\cite{alon}.
The technique of color-coding is traditionally used for randomized algorithms.
Over the years, concepts for derandomization have been developed.
Derandomization is necessary to meet the formal definition of (deterministic) \FPT-algorithms.
However, we want to mention that randomized algorithms with a very low error probability usually have faster running times.

\begin{definition}[Perfect Hash Families]
	\label{def:perfectHashFamily}
	For integers~$n$ and~$k$,
	an~\emph{$(n,k)$-perfect hash family $\mathcal{H}$} is a family of functions $f: [n] \to [k]$ such that for every subset~$Z$ of~$[n]$ of size~$k$, some $f \in \mathcal{H}$ exists that is injective when restricted to~$Z$.
\end{definition}

We will also resort to another data structure relevant for color coding.

\begin{definition}[Universal Sets]
	For integers~$n$ and~$k$,
	an \emph{$(n,k)$-universal set} is a family ${\cal U}$ of subsets of $[n]$ such that for any~$S\subseteq [n]$ of size $k$, $\{A \cap S \mid A \in {\cal U}\}$ contains all $2^k$ subsets of $S$.
\end{definition}

It has been proven how big~$(n,k)$-perfect hash families and $(n,k)$-universal~sets can be and how fast they can be computed.

\begin{theorem}[\cite{Naor1995SplittersAN}]
	For any integers $n,k \geq 1$,
	an~$(n,k)$-perfect hash family which contains~$e^k k^{\Oh(\log k)} \cdot \log n$ functions can be constructed in time $e^k k^{\Oh(\log k)} \cdot n \log n$.
\end{theorem}

\begin{theorem}[\cite{Naor1995SplittersAN}]
	\label{thm:Net-universalSet}
	For any integers $n,k \geq 1$,
	an~$(n,k)$-universal set which contains~$2^k k^{\Oh(\log k)} \cdot \log n$ functions can be constructed in time $2^kk^{\Oh(\log k)} \cdot n\log n$.
\end{theorem}

\section{A List of Frequently Used Problems}
\label{sec:probs}
In this section, we define some problems and results that we frequently refer to throughout this thesis.
This list is far from complete and we also recall problem definitions when we use them.

\problemdef{\KP}
{A set of items~$N=\{a_1,\dots,a_n\}$,
	two functions~$c,d: N\to \mbb N$,
	and two integers~$B$ and~$D$}
{Is there a set~$S\subseteq N$ such that
	$c_\Sigma(S)\le B$ and $d_\Sigma(S)\ge D$}

\problemdef{$k$-\SubSum}
{A set of items~$N=\{a_1,\dots,a_n\}$,
	a function~$c: N\to \mbb N$,
	and two integers~$k$ and~$G$}
{Is there a set~$S\subseteq N$ such that
	$|S| = k$ and $c_\Sigma(S) = G$}
We mostly use an analogous definition in which we say that we are given a multiset of integers instead of~$N$ and the function~$c$.
\KP is a generalization of $k$-\SubSum.
Both problems are \NP-hard~\cite{karp} and \Wh{1}-hard with respect to the size of the solution~\cite{downey}.

\problemdef{\VC}
{An undirected graph~$G = (V,E)$ and an integer~$k$}
{Is there a vertex set~$C \subseteq V$ of size at most~$k$ such that
	$e \cap C \ne \emptyset$ for each edge~$e\in E$}
In other words, every edge has at least one endpoint in~$C$.
Such a vertex set~$C$ is called a \emph{vertex cover} of~$G$.
\VC is \NP-hard even if each vertex in~$G$ has a degree of exactly~3~\cite{mohar}.

\problemdef{\SC}
{A universe $\mathcal{U}$, a family $\mathcal{F}$ of subsets over $\mathcal{U}$, and an integer $k$}
{Are there sets $F_1,\dots,F_k\in \mathcal F$ such that
	$\mathcal{U} := \bigcup_{i=1}^k F_i$}
In other words, every item of $\mathcal{U}$ occurs at least in one of the sets~$F_1,\dots,F_k\in \mathcal F$.
\SC is \NP-hard~\cite{karp} and \Wh{2}-complete when parameterized by~$k$~\cite{downeybook}.
Assuming \NPcoNPpoly, \SC does not admit a polynomial kernel when parameterized by the size of the universe $|\mathcal{U}|$~\cite{dom}.

\problemdef{\rbnb}
{An undirected bipartite graph $G$ with vertex bipartition $V(G)=V_r \cup V_b$ and an integer $k$}
{Is there a set $S\subseteq V_r$ of size at least $k$ such that
	each vertex $v$ of $V_b$ has a neighbor in $V_r \setminus S$}
\rbnb is \Wh{1}-hard when parameterized by~$k$~\cite{downey}.

\problemdef{\ILPF}
{A matrix~$\myvec{A} \in \mathbb{R}^{m \times n}$ and a vector~$\myvec{b} \in \mathbb{R}^{m}$}
{Does a vector~$\myvec{x} \in \mathbb{Z}^{n}$ exist such that~$\myvec{A}\myvec{x} \le \myvec{b}$}
It is possible to redefine equations~$\myvec{a}\myvec{x} = b$ for some vectors~$\myvec{a}$ and~$\myvec{x}$, and a number~$b$ to be two inequalities~$\myvec{a}\myvec{x} \le b$ and~$\myvec{-a}\myvec{x} \le -b$.
In general, \ILPF is \NP-hard~\cite{karp}.
However, it is known that instances of \ILPF with~$n$ variables and input length~$s$ can be solved using~$s \cdot n^{2.5n+o(n)}$ arithmetic operations~\cite{frank,lenstra}.

\section{Multiple-Choice Knapsack}
\label{sec:MCKP}
In this section, we consider \MCKPLong (\MCKP), a variant of \KP, in which the set of items is divided into classes. From every class, exactly one item can be chosen.
\MCKP is algorithmically closely related with \GNAPLong and therefore particularly interesting for the examination we want to conduct in the next chapter.
While \MCKP has been studied from a classical and approximation point of view~\cite{kellerer}, a parameterized point of view has not been considered yet, to the best of our knowledge.
We want to close this gap in this section.
The problem is formally defined as follows.

\problemdef{Multiple-Choice Knapsack Problem (MCKP)}
{A set of items~$N=\{a_1,\dots,a_n\}$,
	a partition~$\{N_1,\dots,N_m\}$ of~$N$,
	two functions~$c,d: N\to \mbb N$,
	and two integers~$B$, and~$D$}
{Is there a set~$S\subseteq N$ such that
	$c_\Sigma(S)\le B$, $d_\Sigma(S)\ge D$,
	and~$|S\cap N_i|=1$ for each~$i\in [m]$}
Recall that we write~$c_\Sigma(A) := \sum_{a_i\in A} c(a_i)$ and~$d_\Sigma(A) := \sum_{a_i\in A} d(a_i)$ for a set~$A\subseteq N$.
We call~$c(a_i)$ the {\it cost of~$a_i$} and~$d(a_i)$ the {\it value of~$a_i$}.
Further, for a set~$A\subseteq N$ we define~$c(A) := \{ c(a) \mid a\in A \}$ and~$d(A) := \{ d(a) \mid a\in A \}$.
A set~$S$ holding the criteria of the question is called a~\textit{solution} for the instance~$\mathcal I$ of \MCKP.
The sets~$N_i$ for~$i\in [m]$ are called~\emph{classes}.

We examine \MCKP with respect to the following parameters.
The input directly gives the \textit{number of classes~$m$}, the \textit{budget~$B$}, and the desired~\textit{value~$D$}.
Closely related to~$B$ is the \textit{maximum cost for an item~$C := \max_{a_j \in N} c(a_j)$}.
By~$\var_c$, we denote the \textit{number of different costs}, that is,~$\var_c := |\{ c(a_j) ~:~ a_j \in N \}|$. We define the \textit{number of different values}~$\var_d$ analogously.
The size of the biggest class is denoted by~$L$.
If a class~$N_i$ contains two items~$a_p$ and~$a_q$ with the same cost and~$d(a_p)\le d(a_q)$, the item~$a_p$ can be removed from the instance.
Thus, we may assume that no class contains two items with the same cost and so~$L\le \var_c$. Analogously, we may assume that no class contains two same-valued items and consequently~$L\le\var_w$.

Since projects whose cost exceeds the budget can be removed from the input, we may assume~$C\le B$. Further, we assume that~$B\le C\cdot m$, as otherwise, we can return \yes if the total value of the most valuable items per class exceeds~$D$, and \no otherwise.

\begin{table}[t]
	\centering
	\caption{Complexity results for \MCKPLong. The two question marks indicate unknown results.}
	\footnotesize
	\label{tab:results-mckp}
	\myrowcols
	\begin{tabular}{lcccl}
		\hline
		Parameter & \PFPT & \FPT & \XP & Remarks\\
		\hline
		Number of classes $m$ & \ding{55} & \ding{55} & \checkmark & Thm.~\ref{thm:Pre-MCKP-m-W1}\\
		Budget $B$ & \checkmark & \checkmark & \checkmark & $\Oh(B \cdot |N|)$~\cite{pisinger}\\
		Maximum cost $C$ & \checkmark & \checkmark & \checkmark & $\Oh(C \cdot |N|\cdot m)$; Obs.~\ref{obs:Pre-C-MCKP}\\
		Threshold $D$ & \checkmark & \checkmark & \checkmark & $\Oh(D \cdot |N|)~$\cite{bansal}\\
		Largest class $L$ & \ding{55} & \ding{55} & \ding{55} & \NP-hard for $L=2$ \cite{kellerer}\\
		Number of costs $\var_c$ & \ding{55} & ? & \checkmark & $\Oh(m^{\var_c-1}\cdot |N|)$; Prop.~\ref{prop:Pre-MCKP-XP-varc}\\
		Number of values $\var_d$ & \ding{55} & ? & \checkmark & $\Oh(m^{\var_d-1}\cdot |N|)$; Prop.~\ref{prop:Pre-MCKP-XP-vard}\\
		$\var_c+\var_d$ & \ding{55} & \checkmark & \checkmark & Thm.~\ref{thm:Pre-MCKP-ILPF}\\
		\hline
	\end{tabular}
	
\end{table}

Table~\ref{tab:results-mckp} presents an overview of known and new complexity results for \MCKP.
Observe that because $\var_c$ and~$\var_d$ are bound in the size of the instance and \MCKP is \NP-hard~\cite{kellerer}, \MCKP can not be \PFPT with respect to $\var_c+\var_d$, unless~\PeqNP.

\subsection{Algorithms for Multiple-Choice Knapsack}
\label{subsec:algos-MCKP}
First, we provide some algorithms that solve~\MCKP. 
It is known that~\MCKP can be solved in~$\Oh(B \cdot |N|)$ time~\cite{pisinger}, or in~$\Oh(D \cdot |N|)$ time~\cite{bansal}.
As we may assume that~$C\cdot m\ge B$, we may also observe the following.
\begin{observation}
	\label{obs:Pre-C-MCKP}
	\MCKP can be solved in~$\Oh(C \cdot \bet N \cdot m)$ time.
\end{observation}

In  the following we want to study \MCKP with respect to the number of different costs and different values.
\KP is \FPT with respect to the number of different costs,~$\var_c$~\cite{etscheid}.
This result is shown by a reduction to \ILPF instances with~$f(\var_c)$ variables.
This approach can not be adopted easily, as it has to be checked whether a solution contains exactly one item per class.
In Proposition~\ref{prop:Pre-MCKP-XP-varc} and~\ref{prop:Pre-MCKP-XP-vard} we show that \MCKP is XP with respect to the number of different costs and different values, respectively.
Then, in Theorem~\ref{thm:Pre-MCKP-ILPF} we show that \MCKP is \FPT with respect to the parameter~$\var_c+\var_d$.
In the following, let~$\mcal I=(N,\{N_1,\dots,N_m\},c,d,B,D)$ be an instance of~\MCKP, and let~$\{c_1,\dots,c_{\var_c}\}:=c(N)$ and~$\{d_1,\dots,d_{\var_d}\}:=d(N)$ denote the set of different costs and the set of the different values in~\Instance, respectively.
Without loss of generality, assume~$c_i<c_{i+1}$ for each~$i\in[\var_c-1]$ and likewise we can assume~$d_j<d_{j+1}$ for each~$j\in [\var_d-1]$. In other words,~$c_i$ is the~$i$th cheapest cost in~$c(N)$ and~$d_j$ is the~$j$th smallest value in~$d(N)$.
Recall also that we assume that there is at most one item with cost~$c_p$ and at most one item with value~$d_q$ in~$N_i$, for every~$i\in[m]$,~$p\in[\var_c]$, and~$q\in[\var_d]$.

\begin{proposition}
	\label{prop:Pre-MCKP-XP-varc}
	\MCKP can be solved in~$\Oh(m^{\var_c-1}\cdot |N|)$ time, where~$\var_c$ is the number of different costs.
\end{proposition}
\begin{proof}
	\proofpara{Table definition}
	We describe a dynamic programming algorithm with a table~$\DP$ that has~$\var_c$ dimensions.
	We want to store the largest value of a set~$S$ that contains exactly one item of each set of~$N_1,\dots,N_i$ and contains exactly~$p_j$ items of cost~$c_j$ for each~$j\in [\var_c-1]$ in entry~$\DP[i,p_1,\dots,p_{\var_c-1}]$.
	Consequently,~$S$ contains exactly~$p_{\var_c}^{(i)} := i-\sum_{j=1}^{\var_c-1} p_j$ items of cost~$c_{\var_c}$.

	We denote by~$\myvec{p}$ in the following~$(p_1,\dots,p_{\var_c-1})$.
	
	\proofpara{Algorithm}
	As a base case, we consider~$i=1$.
	For entries~$\DP[1,\myvec{p}]$, only subsets of~$N_1$ with a single number are considered.
	Thus, for every~$a\in N_1$ with a cost of~$c(a) = c_j < c_{\var_c}$, store~$\DP[1,{\myvec{0}}_{(j) + 1}]=d(a)$.
	If~$N_1$ contains an item~$a$ with a cost of~$c(a)=c_{\var_c}$, then store~$\DP[1,{\myvec{0}}]=d(a)$
	and otherwise store~$\DP[1,{\myvec{0}}]=-\infty$.
	For all other~$\myvec{p}$, store~$\DP[1,\myvec{p}]=-\infty$.
	
	Fix an~$i\in [m]$.
	Once the entries~$\DP[i,\myvec p]$ for all~$\myvec p$ have been computed, we use the following recurrence to compute further values
	\begin{eqnarray}
		\label{eqn:Pre-MCKP-varc}
		\DP[i+1,\myvec{p}] &=&
		\max_{a \in N_{i+1}}
		\left\{
		\begin{array}{ll}
			\DP[i,\myvec{p}_{(j) -1}] + d(a)
			& \text{if } c(a) = c_j < c_{\var_c} \text{ and } p_j \ge 1\\
			\DP[i,\myvec{p}] + d(a)
			& \text{if } c(a) = c_{\var_c}
		\end{array}
		\right.
		.
	\end{eqnarray}
	
	Return \yes if~$\DP[m,\myvec{p}]\ge D$ for some~$\myvec p$ with~$p_{\var_c}^{(m)} \cdot c_{\var_c} + \sum_{i=1}^{\var_c-1} p_i \cdot c_i \le B$ and return \no, otherwise.
	
	\proofpara{Correctness}
	For given integers~$i\in [m]$ and~$\myvec{p} \in [i]_0^{\var_c-1}$, we define~$\mathcal S^{(i)}_{\myvec{p}}$ to be the family of~$i$-sized sets~$S\subseteq N$ that contain exactly one item of each of~$N_1,\dots,N_i$ and where~$p_\ell$ is the number of items in~$S$ with cost~$c_\ell$ for each~$\ell \in [\var_c-1]$.
	
	For fixed a~$\myvec{p}\in [i]_0^{\var_c-1}$, we prove that~$\DP[i,\myvec{p}]$ stores the largest value of a set~$S\in \mathcal S^{(i)}_{\myvec{p}}$, by an induction.
	This implies that the algorithm is correct.
	The base cases are correct.
	Now, as an induction hypothesis assume that the claim is correct for a fixed~$i\in [m-1]$.
	We first prove that if~$\DP[i+1,\myvec{p}]=q$, then there exists a set~$S\in \mathcal S^{(i+1)}_{\myvec{p}}$ with~$d_\Sigma(S)=q$.
	Afterward, we prove that~$\DP[i+1,\myvec{p}]\ge d_\Sigma(S)$ for every set~$S\in \mathcal S^{(i+1)}_{\myvec{p}}$.

	Now, let~$\DP[i+1,\myvec{p}]=q$.
	Let~$a\in N_{i+1}$ be an item with a cost of~$c(a) = c_j$ be an item of~$N_{i+1}$ that maximizes the right side of Equation~(\ref{eqn:Pre-MCKP-varc}) for~$\DP[i+1,\myvec{p}]$.
	Assume first that~$c_j < c_{\var_c}$, and thus~$\DP[i+1,\myvec{p}]=q=\DP[i,\myvec{p}_{(j) -1}] + d(a)$.
	By the induction hypothesis, there is a set~$S\in S^{(i)}_{\myvec{p}_{(j) -1}}$ such that~$\DP[i+1,\myvec{p}_{(j) -1}]=d_\Sigma(S)=q-d(a)$.
	Observe that~$S':=S\cup\{a\}\in S^{(i+1)}_{\myvec{p}}$. The value of~$S'$ is~$d_\Sigma(S')=d_\Sigma(S)+d(a)=q$.
	The other case with~$c(a) = c_{\var_c}$ is shown analogously.
	
	Conversely,
	let~$S\in \mathcal S^{(i+1)}_{\myvec{p}}$ be a set of items and let~$a\in S\cap N_{i+1}$ be an item.
	Assume first that~$c(a) = c_j < c_{\var_c}$.
	Observe that~$S' := S\setminus\{a\}\in S^{(i)}_{\myvec{p}_{(j) -1}}$.
	Consequently,
	\begin{eqnarray}
		\label{eqn:Pre-MCKP-varc-IH1}
		\DP[i+1,\myvec{p}]
		&\ge& \DP[i,\myvec{p}_{(j) -1}] + d(a)\\
		\label{eqn:Pre-MCKP-varc-IH2}
		&=& \max\{ d_\Sigma(S) \mid S \in S^{(i)}_{\myvec{p}_{(j) -1}}\} + d(a)\\
		\nonumber
		&\ge& d_\Sigma(S') + d(a) = d_\Sigma(S).
	\end{eqnarray}
	Herebin, Inequality~(\ref{eqn:Pre-MCKP-varc-IH1}) is the definition of the recurrence in Equation~(\ref{eqn:Pre-MCKP-varc}), and Equation~(\ref{eqn:Pre-MCKP-varc-IH2}) follows by the induction hypothesis.
	The other case with~$c(a) = c_{\var_c}$ is shown analogously.
	
	\proofpara{Running time}
	First, we show how many options of vectors~$\myvec{p}$ there are and then how many equations have to be computed for one of these options.
	
	For~$p_1,\dots,p_{\var_c-1}$ with~$\sum_{j=1}^{\var_c-1} p_j \ge m$ and each~$i\in [m]$, the entry~$\DP[i,\myvec{p}]$ stores~$-\infty$. Consequently, we consider a vector~$\myvec{p}$ with~$p_j=m$ only if~$p_\ell=0$ for each~$\ell \ne j$.
	Thus, we are only interested in~$\myvec{p}\in [m-1]_0^{\var_c-1}$ or~$\myvec{p}={\myvec{0}}_{(j) +m}$ for each~$j\in[\var_c-1]$.
	So, there are~$m^{\var_c-1} + m \in \Oh(m^{\var_c-1})$ options of~$\myvec{p}$.
	
	For a fixed~$\myvec{p}$, each item~$a\in N_i$ is considered exactly once in the computation of~$\DP[i,\myvec{p}]$. Thus, overall~$\Oh(m^{\var_c-1} \cdot |N|)$ time is needed to compute the table~$F$.
	Additionally, we need~$\Oh(\var_c \cdot |N|)$ time to check, whether~$\DP[m,p_1,\dots,p_{\var_c-1}]\ge D$ and~$p_{\var_c}^{(m)} \cdot c_{\var_c} + \sum_{i=1}^{\var_c-1} p_i \cdot c_i \le B$ for any~$\myvec{p}$. As we may assume~$m^{\var_c-1} > \var_c$, the running time of the entire algorithm is~$\Oh(m^{\var_c-1} \cdot |N|)$.
\end{proof}

One can define a dynamic programming algorithm which is very similar to the one in Proposition~\ref{prop:Pre-MCKP-XP-varc} in which the table has~$\var_d$ dimensions.
Herein, instead of storing the maximum value of a set of items with a given set of costs, we store the minimum cost a set of items with a given set of values can have.
\begin{proposition}
	\label{prop:Pre-MCKP-XP-vard}
	\MCKP{} can be solved in~$\Oh(m^{\var_d-1}\cdot |N|)$, where~$\var_d$ is the number of different values.
\end{proposition}

By Propositions~\ref{prop:Pre-MCKP-XP-varc} and~\ref{prop:Pre-MCKP-XP-vard}, \MCKP is~\XP with respect to~$\var_c$ and~$\var_d$, respectively.
In the following, we show that~\MCKP is \FPT with respect to the combined parameter~$\var_c+\var_d$.
To prove this, we reduce an instance of \MCKP to an instance of~\ILPF, in which the number of variables is in~$2^{\var_c+\var_d} \cdot \var_c$.
\begin{theorem}
	\label{thm:Pre-MCKP-ILPF}
	For an instance of~\MCKP one can define an equivalent instance of~\ILPF with~$\Oh(2^{\var_c+\var_d} \cdot \var_c)$ variables.
	Thus,~\MCKP is \FPT with respect to~$\var_c+\var_d$.
\end{theorem}
%It is known that~\KP is \FPT with respect to~$\var_c$~\cite{etscheid}. We were not able to use the technique for this result to prove that also \MCKP is \FPT with respect to $\var_c$.
\begin{proof}
	\proofpara{Description}
	We may assume that any class~$N_i$ does not contain two items of the same cost or the same value.
	We conclude that~$c(a) \ne c(b)$,~$d(a) \ne d(b)$, and if~$c(a)<c(b)$ then~$d(a)<d(b)$ for items~$a,b\in N_i$.
	Thus, each class~$N_i$ is described by the set of costs~$c(N_i)$ of the items in~$N_i$ and the set of values~$d(N_i)$ of the items in~$N_i$.
	
	In the following, we call~$T=(C,Q)$ a \textit{type}, for sets~$C\subseteq c(N), Q\subseteq d(N)$ with~$|C| = |Q|$. Let~$\Tree$ be the family of types.
	We say that \textit{class~$N_i$ is of\lb type~$T=(c(N_i),d(N_i))$.}
	For each~$T\in\Tree$, let~$m_T$ be the number of classes of type~$T$.
	Clearly,~$\sum_{T\in\Tree} m_T = m$.
	
	Observe, for every class~$N_j$ of type~$T=(C,Q)$, and each item~$a\in N_j$ with a cost of~$c(a)\in C$ the value~$d(a)\in Q$ can be determined.
	More precisely, if~$c(a)$ is the~$\ell$th cheapest cost in~$C$, then the value of~$a$ is the~$\ell$th smallest value in~$Q$.
	For every type~$T=(C,Q)$ and each~$i\in [\var_c]$, we define a constant~$d_{T,i} := -\sum_{i=1}^m \max d(N_i)$ if~$c_i \not\in C$. Otherwise, let~$d_{T,i}$ be the~$\ell$th smallest value in~$Q$, if~$c_i$ is the~$\ell$th smallest cost in~$C$.
	
	We define an instance of~\ILPF that is equivalent to the instance~\Instance of~\MCKP.
	The variable~$x_{T,i}$ expresses the number of items with cost~$c_i$ that are chosen in a class of type~$T$.
	
	\begin{align}
		\label{eqn:Pre-MCKP-ILPF-B}
		\sum_{T \in \mathcal{T}_C} \sum_{i=1}^{\var_c} x_{T,i} \cdot c_i \le & \; B\\
		\label{eqn:Pre-MCKP-ILPF-D}
		\sum_{T \in \mathcal{T}_C} \sum_{i=1}^{\var_c} x_{T,i} \cdot d_{T,i} \ge & \; D\\
		\label{eqn:Pre-MCKP-ILPF-Nr}
		\sum_{i=1}^{\var_c} x_{T,i} = & \; m_{T} & \quad \forall T\in \mathcal{T}\\
		\label{eqn:Pre-MCKP-ILPF-initial}
		x_{T,i} \ge & \; 0 & \quad \forall T\in \mathcal{T}, i\in [\var_c]
	\end{align}
	
	\proofpara{Correctness}
	Observe that if~$c_i\not\in C$, then Inequality~(\ref{eqn:Pre-MCKP-ILPF-D}) would not be fulfilled if~$x_{T,i}>0$  because we defined~$d_{T,i}$ to be~$-\sum_{i=1}^m \max d(N_i)$. Consequently,~$x_{T,i}=0$ if~$c_i\not\in C$ for each type~$T=(C,Q)\in\Tree$ and~$i\in[\var_c]$.
	Inequality~(\ref{eqn:Pre-MCKP-ILPF-B}) can only be correct if the total cost is at most~$B$.
	Inequality~(\ref{eqn:Pre-MCKP-ILPF-D}) can only be correct if the total value is at least~$D$.
	Equation~(\ref{eqn:Pre-MCKP-ILPF-Nr}) can only be correct if exactly~$m_T$ elements are picked from the classes of type~$T$, for each~$T\in \Tree$.
	It remains to show that the instance of the~\ILPF has~$\Oh(2^{\var_c+\var_d}\cdot \var_c)$ variables.
	Because~$\Tree\subseteq 2^{c(N)} \times 2^{d(N)}$, the size of~$\Tree$ is~$\Oh(2^{\var_c+\var_d})$.
	Consequently, there are~$\Oh(2^{\var_c+\var_d}\cdot \var_c)$ different options for the variables~$x_{T,i}$.
\end{proof}
Observe that with the same technique, an instance of~\ILPF\lb with~$2^{\var_c+\var_d}\cdot \var_d$ variables can be described.

\subsection{Hardness With Respect to the Number of Classes}
Kellerer et al. gave a reduction from \KP{} to~\MCKP in which each item in the instance of~\KP{} is added to a unique class with a new item that has no costs and no value~\cite{kellerer}.
\begin{observation}[\cite{kellerer}]
	\label{obs:Pre-MCKP-NP-L=2}
	\MCKP is \NP-hard even if every class contains two items.
\end{observation}

The reduction above constructs an instance with many classes.
In the following, we prove that~\MCKP is \Wh{1}-hard with respect to the number of classes~$m$, even if~$B=D$ and~$c(a)=d(a)$ for each~$a\in N$. This special case of~\MCKP is called~\textsc{Multiple-Choice \SubSum}~\cite{kellerer}.
\begin{theorem}
	\label{thm:Pre-MCKP-m-W1}
	\MCKP is \XP and \Wh1-hard with respect to the number of classes~$m$.
\end{theorem}

To show the hardness, we reduce from \msSubSum, a version of \SubSum in which every integer can be chosen arbitrarily often, parameterized by~$k$.
More formally, in~\msSubSum, a multiset~$Z=\{z_1,\dots,z_n\}$ of integers and two integers~$Q$ and~$k$ are given, and it is asked, whether there is a \textit{multi}-set~$S\subseteq Z$ of size~$k$ such that~$\sum_{s\in S}=Q$.

In parameterized complexity, the problem is often defined for set inputs.
The original \Wh1-hardness proof for \SubSum relies on a reduction from the \textsc{Perfect Code} problem~\cite[Lemma~4.4]{DF95}.
It is easy to observe that this reduction also works if every constructed integer is added~$k$ times to~$Z$.
We will not repeat the details of the reduction but give a brief intuition.
In the reduction, the target number~$Q$ has a value of one at each digit.
Now, adding~$k$ copies of a number to~$Z$ in the construction maintains correctness because including any number twice in the solution produces carries in the summation which destroys the property that every digit has a value of 1.
\begin{proposition}[\cite{DF95}]
	\label{prop:MultiSubSum}
	\SubSum{} is \Wh{1}-hard with respect to~$k$, even when every integer in~$Z$ has multiplicity at least~$k$.
\end{proposition}

\begin{proof}[Proof of \Cref{thm:Pre-MCKP-m-W1}]
	\proofpara{Slicewice-polynomial}
	We first discuss the \XP-algorithm and then focus on the \Wh1-hardness.
	For every class~$N_i$, iterate over the items~$a_{j_i}\in N_i$ such that there are~$m$ nested loops.
	We check whether~$\sum_{i=1}^m c(a_{j_i}) \le B$\lb and~$\sum_{i=1}^m d(a_{j_i}) \ge D$ in~$\Oh(m)$~time, such that a solution of~\MCKP is computed in~$\Oh(L^m \cdot m)$~time.
	
	\proofpara{Reduction}
	We reduce from~\SubSum with multi-set inputs~$Z$ where every integer~$z$ occurs~$k$ times. 	Let~$\Instance=(Z,Q,k)$ be such an instance of~\SubSum
	and assume, without loss of generality, that~$Z=\{z_{i,j}\mid i\in [n], j\in [k]\}$ such that~$z_{i,1}=z_{i,2}=\ldots =z_{i,k}$ for all~$i\in [n]$. 
	
	We define~$N_j:=\{a_{i,j}\mid i\in [n]\}$ for each~$j\in [k]$ and set~$c(a_{i,j})=d(a_{i,j})=z_{i,j}$ for each element $a_{i,j}$. Then~$\Instance'$ is~$(N,\{N_1,\ldots, N_k\},Q,Q)$ where~$N$ is the union of the sets~$N_j$,~$j\in [k]$. Observe that~$m=k$. 
	
	\proofpara{Correctness}
	We show that~\Instance is a yes-instance if and only if~\Instance' is a yes-instance.
	
	Let~$S=\{z_{i_1},\ldots,z_{i_k}\}$ be a solution for~$\Instance$. Observe that, by construction,~for each~$j\in k$, the set~$N_j$ contains one element~$a_{\ell_j}$ such that~$c(a_{\ell_j})=d(a_{\ell_j})=z_{i_j}$. Then, the set~$S':=\{a_{\ell_1}, a_{\ell_2},\ldots, a_{\ell_k} \}$ is a solution for~$\Instance'$ since
	$$\sum_{j\in [k]} c(a_{\ell_j}) = \sum_{j\in [k]} d(a_{\ell_j}) = \sum_{j\in [k]} z_{i_j}=Q.$$
	
	Conversely, let~$S':=\{a_{\ell_1}, a_{\ell_2},\ldots, a_{\ell_k} \}$ be a solution for~$\Instance'$. By construction, for each~$j\in [k]$, there is a number~$z_{i_j}\in Z$ such that~$c(a_{\ell_j})=z_{i_j}$.  Then, since
	$$Q=\sum_{j\in [k]} c(a_{\ell_j}) = \sum_{j\in [k]}  z_{i_j},$$
	$S:=\{z_{i_j}\mid j\in [k] \}$ is a solution for~$\Instance$.    
\end{proof}

\section{Penalty Sum}
\label{sec:PenSum}
In this section, we examine \PS, a problem that has been introduced in~\cite{GNAP}.
In the next chapter we will present a hardness reduction which builds on results from this chapter.
\PS is defined as follows.

\problemdef{\PS}
{A set of tuples $\{t_i = (a_i,b_i) \mid i \in [m], a_i \in \mathbb{Q}_+\cup \{0\}$, $b_i \in \mathbb(0,1)\}$,
	two integers~$k$ and~$Q$, and
	a number $D \in \mathbb{Q}_+$}
{Is there a set $S\subseteq [m]$
	such that $|S| = k$ and
	$\sum_{i\in S}a_i - Q\cdot \prod_{i\in S}b_i\geq D$}

We first show that \PS can be solved in polynomial running time if the numbers in the input are given in unary.
Afterward, we show that \PS nevertheless is \NP-hard in general.

\begin{proposition}
	\label{prop:Pre-PS-pseudo}
	\PS can be solved in~$\Oh((Q + \lceil D \rceil) \cdot m \cdot k)$~time.
\end{proposition}
\begin{proof}
	Let~$\Instance = (T, k, Q, D)$ be an instance of \PS, where~$T$ is a set of tuples~$t_i = (a_i,b_i)$ of size~$m$.
	We say that a set~$S\subseteq [m]$ is~\emph{$(c,A)$-feasible} for integers~$c$ and~$A$ if~$S$ has a size of~$c$ and~$\sum_{j\in S} a_j = A$.
	
	\proofpara{Table definition}
	We describe a dynamic programming algorithm with a table~$\DP$.
	We want that entry~$\DP[i,c,A]$ stores the smallest value~$\prod_{j\in S} b_j$ for a~$(c,A)$-feasible set~$S \subseteq [i]$ for each~$i\in [m]$, each~$c\in [k]$, and each~$A \in [Q + \lceil D \rceil]$.
	We want that~$\DP[i,c,A]$ stores~$\infty$ if no such set~$S$ exists.
	
	\proofpara{Algorithm}
	As a base case, for any~$i\in [m]$, store~$\DP[i,0,0] = 1$ and for~$c,A>0$ store~$\DP[i,c,0] = \infty$ and~$\DP[i,0,A] = \infty$.
	Further, store~$\DP[1,c,A] = \infty$ if~$c>1$ or~$A \not\in \{0,a_1\}$.
	We consider $\DP[i,c,A] = \infty$ when we call invalid values such as~$A<0$.
	
	We compute the table for increasing values of~$i$.
	For a fixed~$i$, let the values~$\DP[i,c,A]$ be computed for each~$c$ and~$A$.
	To compute further values we use the recurrence
	\begin{eqnarray}
		\label{eqn:Pre-PS}
		\DP[i+1,c,A] &=& \min \{ \; \DP[i,c,A] ; \; \DP[i,c-1,A-a_{i+1}] \cdot b_{i+1} \; \}.
	\end{eqnarray}
	
	If there is an~$A \in [Q + \lceil D \rceil]$ such that~$\DP[m,k,A] = B$ and~$A - BQ \ge D$
	return~\yes.
	Otherwise, return \no.

	\proofpara{Correctness}
	The base cases are correct.
	We show first that \Recc{eqn:Pre-PS} is correct.
	Afterward, we prove that this algorithm returns the right value.

	As an induction hypothesis, assume that~$\DP[i,c,A]$ stores the desired value for a fixed~$i \in [m]$.
	We show first that if~$\DP[i+1,c,A] = B$, then there is a~$(c,A)$-feasible set~$S \subseteq [i+1]$.
	Then, we show that~$\DP[i+1,c,A] \le \prod_{j\in S} b_j$ for every~$(c,A)$-feasible set~$S\subseteq [i+1]$.

	If~$\DP[i+1,c,A] = B$, then by \Recc{eqn:Pre-PS} $\DP[i,c-1,A-a_{i+1}] \cdot b_{i+1} = B$ or~$\DP[i,c,A] = B$.
	In the latter case, the induction hypothesis directly proves that there is a~$(c,A)$-feasible set~$S\subseteq [i]$.
	In the first case, a~$(c-1,A-a_{i+1})$-feasible set~$S\subseteq [i]$ exists.
	Consequently,~$a_{i+1} + \sum_{j\in S} a_j = A$ such that~$S \cup \{i+1\}$ is~$(c,A)$-feasible.

	Now, let~$S\subseteq [i+1]$ be a~$(c,A)$-feasible set.
	If~$i+1 \not\in S$, then we conclude that~$\DP[i+1,c,A] \le \DP[i,c,A] \le \prod_{j\in S} b_j$, where the first inequality follows from \Recc{eqn:Pre-PS} and the second by the induction hypothesis.
	Otherwise, if~$i+1 \in S$, then~$S' := S\setminus \{i+1\}$ is a~$(c-1,A-a_{i+1})$-feasible set and with the arguments from before~$\DP[i+1,c,A] \le \DP[i,c-1,A-a_{i+1}] \cdot b_{i+1} \le b_{i+1} \cdot \prod_{j\in S'} b_j = \prod_{j\in S} b_j$.

	Finally, assume that there is an~$A \in [Q + \lceil D \rceil]$ such that~$\DP[m,k,A] = B$ and~$A - BQ \ge D$.
	Consequently, a~$(k,A)$-feasible set~$S \subseteq [m]$ with~$\prod_{j\in S} b_j = B$ exists
	and so~$S$ is a solution for~\Instance.
	Conversely, let~$S$ be a solution for~\Instance.
	We define~$A := \sum_{j\in S} a_j$, then~$S$ is a~$(c,A)$-feasible set.
	So,~$B = \DP[m,k,A] \le \prod_{j\in S} b_j$ and so~$A - BQ \ge A - Q \prod_{j\in S} b_j \ge D$.

	\proofpara{Running time}
	The table contains~$m \cdot k \cdot (Q + \lceil D \rceil)$ entries.
	Each of these stores a number between~0 and~1 (if not~$\infty$), that is a multiplication of at most~$k$ numbers~$b_j$.
	Consequently, the encoding length of each table entry is at most~$k$ times the longest encoding length of a~$b_j$.
	Then, \Recc{eqn:Pre-PS} can be computed in constant time in our RAM-model.
	We want to declare, however, that the table entries may store numbers with an encoding length of up to~$|k| \cdot \max_b$, where~$\max_b$ is the maximum encoding length of a number~$b$.
\end{proof}

\subsection{Hardness of Subset Product}
\label{sec:reduction-PenSum}
Now, we prove the \NP-hardness of \PS.
To this end, we first show the \NP-hardness of the following variant of~\SubProd.

\problemdef{\SubProd}
{A multiset of positive integers $\{v_1,v_2,\dots, v_m\}$ and integers~$M$ and~$k$}
{Is there a set $S\subseteq [m]$
	such that $|S| = k$ and
	$\prod_{i \in S}v_i = M$}

We note the definition of \SubProd is slightly different here from the formulation that appears for example in the book of Garey and Johnson \cite{GareyJohnson79}.
In particular, we assume that the size $k$ of the set~$S$ is given and that all integers are positive.
This makes the subsequent \NP-hardness reductions slightly simpler.

The \NP-hardness of \SubProd is not a new result. It was stated by~\cite{GareyJohnson79} without a full proof (the authors indicate that the problem is \NP-hard by reduction from \XTClong (\XTC), citing "Yao, private communication") and a full proof appears in the book of Moret~\cite{Moret97}.
We reprove it here for our slightly adapted variant.

In \XTC, the input is a universe~$\mathcal{U}$ and a family~$\mathcal{C}$ of subsets of~$\mathcal{U}$ which have a size of three each.
It is asked whether there is a subset~$\mathcal{C}'$ of~$\mathcal{C}$ such that~$\mathcal{U} = \bigcup_{C \in \mathcal{C}'} C$ and the sets in~$\mathcal{C}'$ are pairwise disjoint.
In other words, for each item~$u$ of~$\mathcal{U}$ there is exactly one set in~$\mathcal{C}'$ containing~$u$.

\begin{lemma}[\cite{Moret97}]
	\label{lem:Pre-XTCtoSubProb}
	\SubProd is \NP-hard.
\end{lemma}
\begin{proof}
	\proofpara{Reduction}
	Let $(\mathcal{U} := \{u_1,\dots, u_{3n}\},\mathcal{C} := \{C_1,\dots, C_m\})$ be an instance of \XTC.
	Let $p_1,\dots, p_{3n}$ be the first $3n$ prime numbers, so that we may associate each~$u_j \in \mathcal{U}$ with a unique prime number $p_j$.
	For each set $C_i = \{u_a, u_b, u_c\}$, define $v_i := p_a\cdot p_b \cdot p_c$, that is, $v_i$ is the product of the three primes associated with the elements of $C_i$.
	Now, let $M$ be the product of the prime numbers $p_1$ to $p_{3n}$.
	Finally, set $k := n$.
	This completes the construction of the instance $(\{v_1,\dots v_m\}, M, k)$ of \SubProd.
	
	\proofpara{Correctness}
	Now, observe that if $\prod_{i \in S}v_i = M$ for some $S \subseteq [m]$, then by the uniqueness of prime factorization, every prime number $p_1,\dots, p_m$ must appear exactly once across the prime factorizations of all numbers in $\{v_i \mid i\in S\}$. It follows by the construction that the collection of subsets $\mathcal{C}' := \{C_i \mid i \in S\}$ contains each element of~$\mathcal{U}$ exactly once.
	Thus, if $(\{v_1,\dots v_m\}, M, k)$ is a \yes-instance of \SubProd then $(\mathcal{U},\mathcal{C})$ is a \yes-instance of \XTC.
	Conversely, if $(\mathcal{U},\mathcal{C})$ is a \yes-instance of~\XTC with solution $\mathcal{C'}$, then we can define $S:= \{i \in [m] \mid C_i \in \mathcal{C}'\}$. Since every element of $\mathcal{U}$ appears in exactly one $C_i \in \mathcal{C}'$ and $|C_i| = 3$ for all $i \in [m]$, we have that~$|\mathcal{C}'| = |\mathcal{U}|/3 = n = k$, and $\prod_{i \in S}v_i = p_1\cdot\dots\cdot p_{3n} = M$. Thus, $(\{v_1,\dots v_m\}, M, k)$ is a \yes-instance of \SubProd.
	
	It remains to show that the construction of $(\{v_1,\dots v_m\}, M, k)$ from  $(\mathcal{U},\mathcal{C})$ takes polynomial time.
	In particular, we need to show that each of the primes $p_1,\dots p_{3n}$ (and thus the product $M$) can be constructed in polynomial time. This can be shown using two results from number theory: 
	$p_j <  j(\ln j + \ln \ln j)$ for $j\geq 6$, \cite{Rosser41, Dusart99} %https://en.wikipedia.org/wiki/Prime_number_theorem#cite_ref-rosser_34-1
	and the set of all prime numbers in $[Z]$ can be computed in time $\Oh(Z/ \ln \ln Z)$~\cite{AtkinBernstein04}. 
	Combining these, we have that the first $3n$ prime numbers can be generated in time~$\Oh(n \ln n / \ln \ln n)$.
	
	Given the prime numbers $p_1,\dots, p_{3n}$, it is clear that the numbers $\{v_i \mid i\in [m]\}$ can also be computed in polynomial time. The number $M$, being the product of $3n$ numbers each less than $3n(\ln 3n + \ln \ln 3n)$, can also be computed in time polynomial in $n$ (though $M$ itself is not polynomial in $n$).
	It follows that $(\{v_1,\dots v_m\}, M, k)$ can be constructed in polynomial time.
\end{proof}

\subsection{Hardness of Penalty Sum}
In the following, we first describe a simple reduction from an instance of \SubProd to an equivalent `instance' of \PenSum, but one in which the numbers involved are irrational (and as such, cannot be produced in polynomial time). 
We then show how this transformation can be turned into a polynomial-time reduction by replacing the irrational numbers with suitably chosen rationals.

The reduction from \SubProd to \PenSum can be informally described as follows:
For an instance $(\{v_1,\dots v_m\}, M, k')$ of \SubProd and a big integer $A$, we let $a_i$ be (a rational close to) $A - \ln v_i$ and  let $b_i := 1/v_i$, for each~$i \in [n]$. Let $Q := M$, let $k := k'$, and let $D$ be (a rational close to) $ kA - \ln Q - 1$.

Observe that we cannot set $a_i := A - \ln v_i$ or $D := kA - \ln Q - 1$ exactly,
because in general these numbers are irrational and cannot be calculated exactly in finite time, or even stored in finite space.
We temporarily forget about the need for rational numbers, and consider how the function $\sum_{i\in S}a_i - Q\cdot \prod_{i\in S}b_i$ behaves when we drop the `a rational close to' qualifiers from the descriptions above.
In particular, we show that the function reaches its theoretical maximum exactly when $S$ is a solution to the \SubProd instance.

\subsubsection{Reduction with irrational numbers}
\begin{construction}
	\label{cons:irrational reduct}
	Let $(\{v_1,\dots v_m\}, M, k)$ be an instance of \SubProd.
	We define the following (not necessarily rational) numbers.
	\begin{itemize}
		\item Define $A := \lceil\max_{i \in [m]} (\ln v_i)\rceil+1$;
		\item Define $a_i^* := A - \ln v_i$ for each $i\in [m]$;
		\item Define $b_i := 1/v_i$ for each $i \in [m]$;
		\item Define $Q := M$;
		\item Define $D^*: = kA - \ln Q - 1$.
	\end{itemize}
	Finally, output instance $(\{(a^*_i,b_i) \mid i \in [m]\}, k, Q, D^*)$ of \PenSum.
\end{construction}
\medskip

We note the purpose of $A$ is to ensure that $a_i^* > 0$ for each $i \in [m]$,
as is required by the formulation of \PenSum.
Now, define $f^*$ to be a function mapping a subset~$S$ of~$[m]$ to~$\mathbb{R}$ by
\begin{equation}
	f^*(S) := \sum_{i\in S}a^*_i - Q\cdot \prod_{i\in S}b_i.
\end{equation}

\begin{lemma}
	\label{lem:Pre-PenSumIrrationalMax}
	For every set $S \subseteq [m]$ of size~$k$ the following hold.
	\begin{enumerate}
		\item $f^*(S) \leq D^*$, and
		\item $f^*(S) = D^*$ if and only if $\prod_{i\in S} v_i = Q$.
	\end{enumerate}
\end{lemma}
\begin{proof}
	Observe that given~$S$ with $|S| = k$, the value $f^*(S)$ can be written as 
	\begin{equation*}
		f^*(S) = kA - \sum_{i \in S}\ln v_i - Q/\prod_{i\in S}v_i
		= kA - \ln \left(\prod_{i\in S}v_i\right) - Q/\prod_{i\in S}v_i\
	\end{equation*}
	Letting $x_S := \prod_{i\in S}v_i$, we therefore have $f^*(S) = kA - \ln x_S - Qx_S^{-1}$.
	Define a function $g^*:\mathbb{R}_{>0} \to \mathbb{R}$ by $g^*(x) := kA - \ln x - Qx^{-1}$
	and observe $f^*(S) = g^*(x_S)$ for any $S \in \binom{[m]}{k}$. 
	Recall that a function $f$ has a critical point at $x'$ if $\frac{\text{d} f}{\text{d} x}(x') = 0$.
	Since~$\frac{\text{d} g^*}{\text{d} x} = -x^{-1} + Qx^{-2}$, there are critical points~$x'$ in ${x'}^{-1} = Q{x'}^{-2}$, which is the case in~$x' = Q$.
	Moreover, for $Q > x > 0$, we have $Qx^{-1} > 1$, implying
	\begin{equation*}
		\frac{\text{d} g^*}{\text{d} x} = -x^{-1} + Qx^{-2}  >  -x^{-1} + x^{-1} = 0.
	\end{equation*}
	Further, for $x > Q > 0$, we have $Qx^{-1} < 1$, implying
	\begin{equation*}
		\frac{\text{d} g^*}{\text{d} x} = -x^{-1} + Qx^{-2}  <  -x^{-1} + x^{-1} = 0.
	\end{equation*}
	It follows that $g^*(x)$ is strictly increasing in the range $0 < x < Q$ and
	strictly decreasing in the range~$x > Q$.
	Thus, $g^*$ has a unique maximum in the range~$x > 0$, which is in~$x = Q$.
	In particular, for all $S \subseteq [m]$, we have	
	\begin{equation}
		f^*(S) = g^*(x_S) \leq g^*(Q) = kA - \ln Q - 1 = D^*.
	\end{equation}
	With equality if and only if $x_S = \prod_{i\in S}v_i = Q$.		
\end{proof}

The above result implies that $(\{(a_i^*,b_i) \mid i \in [m]\}, k, Q, D^*)$ is a \yes-instance of `\PenSum' if and only if  $(\{v_1,\dots v_m\}, M, k')$ is a \yes-instance of \SubProd, when allowing irrational numbers.

We are now ready to fully describe the polynomial-time reduction from \SubProd to \PenSum, showing how we can adapt the ideas above to work for rational $a_i$ and $D$.

\subsubsection{Reduction with rational numbers}
Let  $(\{v_1,\dots v_m\}, M, k')$ be an instance of \SubProd.
So far we have seen that if we define~$a_i^*$,~$b_i$,~$Q$,~$k$, and~$D^*$ as beforehand,
then,
by Lemma~\ref{lem:Pre-PenSumIrrationalMax},
for any~$S \subseteq [m]$,
$f^*(S) = \sum_{i\in S} a^*_i - Q\cdot \prod_{i\in S} b_i \geq D^*$ if and only if $\prod_{i \in S} v_i = Q = M$.

Our task now is to show how to replace $a_i^*$ and $D^*$ with rationals $a_i$ and $D$, in such a way that this property---$\sum_{i\in S} a_i - Q\cdot \prod_{i\in S} b_i \geq D$ if and only if $\prod_{i \in S} v_i = M$---still holds, so that an instance of \PS can be constructed in polynomial time.
The key idea is to find rational numbers that can be encoded in polynomially many bits, but are close enough to their respective irrationals such that the difference between $f^*(S)$ and $f(S)$ (and between $D^*$ and $D$) is guaranteed to be small.
To this end, let us fix a positive integer $H$ to be defined later, and we will require all numbers~$a_i$,~$b_i$, and~$D$ to be a multiple of $2^{-H}$. This ensures that the denominator part of any of these rationals can be encoded using $\Oh(H)$ bits.

Given a number $x \in \mathbb{R}$ and a positive integer~$H$,
we define $\floorH{x} : = r_x/2^H$, where~$r_x$ is the largest integer such that~$r_x/2^H \leq x$.
Observe for example that because~$25/2^3 < \pi < 26/2^3$ we have~$r_\pi = 25$ and~$\floorvar{3}{\pi} = 3.125 = 25/2^3$.
One may think of~$\floorH{x}$ as the number derived from the binary representation of $x$ by deleting all digits more than $H$ positions after the binary point. Thus, as the binary expansion of $\pi$ begins~\texttt{11.00100 10000 11111}$\dots$, the binary expression of $\floorvar{3}{\pi}$ is~\texttt{11.001}.
Similarly, define~$\ceilH{x}: = s_x/2^H$, where $s_x$ is the smallest integer such that $x \leq s_x/2^H$.
Finally, define~$\delta:= 1/{2^H}$.

\begin{observation}
	\label{obs:Pre-roundBounds}
	Let $x \in \mathbb{R}$. Then,
	$x- \delta   <   \floorH{x}  \leq x  \leq \ceilH{x} < x + \delta$.
\end{observation}

We can now describe the reduction from \SubProd to \PenSum.

\begin{construction}
	\label{cons:rational reduct}
	Let $(\{v_1,\dots v_m\}, M, k)$ be an instance of \SubProd.
	\begin{itemize}
		\item Define $A := \lceil\max_{i \in [m]} (\ln v_i)\rceil+1$;
		\item Define $a_i := \ceilH{a_i^*} = \ceilH{A - \ln v_i}$ for each $i\in [m]$;
		\item Define $b_i := 1/v_i$ for each $i \in [m]$;
		\item Define $Q := M$;
		\item Define $D := \floorH{D^*}  = \floorH{kA - \ln Q - 1}$.
	\end{itemize}
	Finally, output instance $\Instance:=(\{(a_i,b_i) \mid i \in [m]\}, k, Q, D)$ of \PenSum.
\end{construction}
\medskip

In the following, we show that the two instances are equivalent.
Now, define $f$ to be a function mapping a subset~$S$ of~$[m]$ to~$\mathbb{R}$ by
\begin{align*}
	f(S) := \sum_{i\in S}a_i - Q\cdot \prod_{i\in S}b_i.
\end{align*}

Note that $f$ is the same as the function $f^*$ defined previously but with each $a_i^*$ replaced by~$a_i$.
Then, $\Instance$ is a \yes-instance of \PenSum
if and only if
there is some $S \subseteq [m]$ such that $f(S) \geq D$.
The next lemma shows the close relation between~$f^*$ and~$f$, and between $D^*$ and $D$.
This will be used in both directions to show the equivalence between \yes-instances of \SubProd and \PenSum.

\begin{lemma}
	\label{lem:Pre-fDbounds}
	$f^*(S) \leq f(S) < f^*(S) + k\delta$
	for each $S \subseteq [m]$;
	and
	$D^* - \delta < D \leq D^*$.
\end{lemma}
\begin{proof}
	By Observation~\ref{obs:Pre-roundBounds},~$\ceilH{x} - x < \delta$ and~$x - \floorH{x} < \delta$ for each~$x\in \mathbb{R}$.
	
	Observe that $f(S) - f^*(S) = \sum_{i \in S}(a_i - a_i^*)$ and $|S|=k$.
	We conclude\lb that~$0 \leq a_i - a_i^* = \ceilH{a_i^*} - a_i^* < \delta$ for all $i \in [m]$.
	Consequently, $0 \leq f(S) - f^*(S) < k \delta$ and~$D^* - D = D^* - \floorH{D^*} < \delta$.
\end{proof}

\begin{corollary}\label{cor:Pre-SubProdyesToPenSumyes}
	For every $S \subseteq [m]$ with $\prod_{i \in S}v_i = Q$
	it follows $f(S) \geq D$.
\end{corollary}
\begin{proof}
	$D \stackrel{\text{Lem.~\ref{lem:Pre-fDbounds}}}{\leq} D^* \stackrel{\text{Lem.~\ref{lem:Pre-PenSumIrrationalMax}~(2)}}{=} f^*(S)
	\stackrel{\text{Lem.~\ref{lem:Pre-fDbounds}}}{\leq} f(S)$.
	% \todo{MW: why is $f^*(S)=D^*$ if $S$ can be any size-$k$ set of indices? MJ: we also require  $\prod_{i \in S}v_i = Q$.}
\end{proof}

We now have that $(\{v_1,\dots v_m\}, M, k')$ being a \yes-instance of \SubProd implies~$\Instance$ being a \yes-instance of \PenSum.
To show the converse, we show that if~$\prod_{i \in S}v_i = Q' \neq Q$ for every~$S \subseteq [m]$, then~$f(S) < D$.
Since $f(S) < f^*(S) + k\delta$ and $D^* - \delta < D$, it is sufficient to show that $f^*(S) + k\delta \leq D^* - \delta$.
This is equivalent to~$(k+1)\delta \leq D^* - f^*(S)$.
To this end, we first establish a lower bound on $D^* - f^*(S')$ in terms of $Q$, using the following technical lemma.

\begin{lemma}
	\label{lem:Pre-logDifferenceBound}
	For every pair of positive integers~$Q$ and~$Q'$ with $Q \geq 2$ and $Q\neq Q'$
	it follows that $\ln Q' - \ln Q + Q/Q' - 1 > Q^{-4}$.
\end{lemma}
We explicitly note that we use the natural logarithm.
For other logarithms, this lemma is not true.
Consider for example integers $Q=2$ and $Q'=1$ in the logarithm on basis~2.
We have $\log_2(1) - \log_2(2) + 2/1 -1 = 0 - 1 + 2 - 1 = 0 < 2^{-4}$.
\begin{proof}%[Proof of \Cref{lem:Pre-logDifferenceBound}]
	We first show that it is enough to consider the cases $Q' = Q+1$ and $Q' = Q-1$.
	Fix an integer $Q \in \mathbb{N}_+$ with $Q \geq 2$.
	Consider the function $h_Q:\mathbb{R}_{>0}\rightarrow \mathbb{R}$ given by
	$$h_Q(x) = \ln x - \ln Q + Q/x - 1.$$
	
	So, we aim to show that $h_Q(Q') \geq Q^{-4}$.
	Similar to the proof of Lemma~\ref{lem:Pre-PenSumIrrationalMax}, we can observe that the derivation
	\begin{eqnarray*}
		\frac{\text{d} h_Q}{\text{d} x} = x^{-1} - Qx^{-2} = \frac{1}{x}\left(1- \frac{Q}{x}\right).
	\end{eqnarray*}
	is less than $0$ if $x < Q$; exactly $0$ if $x = Q$; and greater than $0$ if $x>Q$.
	It follows that in the range $x>0$, the function $h_Q$ has a unique minimum at $x = Q$, and is decreasing in the range $x< Q$ and increasing in the range $x>Q$.
	Thus, in particular~$h_Q(Q') \geq h_Q(Q-1)$ if $Q' \leq Q-1$ and $h_Q(Q') \geq h_Q(Q+1)$ if $Q' \geq Q+1$.
	Since either $Q' \leq Q-1$ or $Q' \geq Q+1$ for any integer $Q'\neq Q$, it remains to show that 
	$h_Q(Q-1) > Q^{-4}$ and  $h_Q(Q+1) > Q^{-4}$.
	
	To show $h_Q(Q-1) > Q^{-4}$ for any $Q \in \mathbb{N}_{\ge 2}$,
	we define the function $\lambda:\mathbb{R}_{>0}\rightarrow\mathbb{R}$ given by 
	\begin{eqnarray*}
		\lambda(Q) & = & h_Q(Q-1) - Q^{-4}\\
		&=& \ln(Q-1) - \ln Q + Q/(Q-1) - 1 - Q^{-4}\\
		&=& \ln(Q-1) - \ln Q + 1/(Q-1) - Q^{-4}.
	\end{eqnarray*}
	We then observe
	\begin{eqnarray*}
		\frac{\text{d} \lambda}{\text{d} Q} &=& (Q-1)^{-1} - Q^{-1} + (Q-1)^{-2} + 4Q^{-5}\\
		&>& (Q-1)^{-2} + 4Q^{-5} > 0.
	\end{eqnarray*}
	
	Therefore, $\lambda$ is a (strictly) increasing function.
	We conclude $\lambda(Q) > 0$, because of~$\lambda(2) = 0 - \ln 2 + 1 - 1/16 \approx 0.244 > 0$, for all $Q\geq 2$, and thus $h_Q(Q-1) > Q^{-4}$.

	To show that $h_Q(Q+1) > Q^{-4}$ for all $Q \in \mathbb{N}_{\geq 2}$,
	observe first that if $Q = 2$, $h_Q(Q+1) = \ln(3) - \ln(2) + 2/3 - 1 \approx 0.0721 > 0.0625 = 2^{-4}$ and so the claim is true.
	For any $Q \geq 3$, observe that $h_Q(Q+1) = \ln(Q+1) - \ln Q + \frac{Q}{Q+1} - 1 = \ln(\frac{Q+1}{Q}) - \frac{1}{Q+1}$.
	We use the Mercator series for the natural logarithm.
	\begin{eqnarray*}
		\ln\left(\frac{Q+1}{Q}\right)
		= \ln\left(1+\frac{1}{Q}\right)  
		= \sum_{k = 1}^\infty \frac{(-1)^{k+1}}{kQ^k}
		= \frac{1}{Q} - \frac{1}{2Q^2} + \frac{1}{3Q^3} - \frac{1}{4Q^4} + \dots
	\end{eqnarray*}
	
	Since $\frac{1}{kQ^k} - \frac{1}{(k+1)Q^{k+1}} >0$ for all $k > 0$, we conclude
	\begin{eqnarray*}
		\ln\left(\frac{Q+1}{Q}\right)
		> \frac{1}{Q} - \frac{1}{2Q^2}
		= \frac{2Q - 1}{2Q^2}.
	\end{eqnarray*}
	
	Then,
	\begin{eqnarray*}
		\ln\left(\frac{Q+1}{Q}\right) - \frac{1}{Q+1}
		& > & \frac{2Q - 1}{2Q^2} - \frac{1}{Q+1} \\
		& = & \frac{(2Q-1)(Q+1) - 2Q^2}{2Q^2(Q+1)} \\
		& = & \frac{2Q^2 + Q - 1 - 2Q^2}{2Q^2(Q+1)} \\
		& = & \frac{Q-1}{2Q^2(Q+1)} \\
		% & \geq & \frac{2}{2Q^2(Q+1)} \\
		& \geq & \frac{1}{Q^2(Q+1)} \\
		& > & \frac{1}{Q^4}
	\end{eqnarray*}
	where the last two inequalities use $Q \geq 3$.
\end{proof}

\begin{corollary}
	\label{cor:Pre-deltaGap}
	If $\prod_{i \in S}v_i = Q' \neq Q$ for some $Q\ge 2$ and $S \subseteq [m]$,\lb
	then $D^* - f^*(S) > Q^{-4}$.
\end{corollary}
\begin{proof}
	%	We may assume without loss of generality that  $Q \geq 2$.
	Recall $f^*(S) = kA - \ln(\prod_{i \in S}v_i) - Q/(\prod_{i \in S}v_i) = kA - \ln Q' - Q/Q'$ and that~$D^* = kA - \ln Q - 1$.
	Then, $D^* - f^*(S) = \ln Q' - \ln Q + Q/Q' - 1$.
	With Lemma~\ref{lem:Pre-logDifferenceBound} we conclude $D^* - f^*(S) > Q^{-4}$.
\end{proof}

Given the above, we can now fix a suitable value for $H$. Given that we want to conclude~$D^* - f^*(S) \geq (k+1)\delta = \frac{(k+1)}{2^H}$ from $\prod_{i \in S} v_i \neq Q$, and assuming without loss of generality that $k < Q$, it is sufficient to set $H := 5\lceil\log_2 Q\rceil$.

\begin{corollary}
	\label{cor:Pre-deltaGap2}
	If $H = 5 \lceil\log_2 Q\rceil$ and $\delta = (1/2^H)$,
	then the instance of \PS constructed in Construction~\ref{cons:rational reduct} holds $Q^{-4} \geq (k+1)\delta$.
\end{corollary}
\begin{proof}
	We assume $k < Q$. 
	% (Indeed, if $k \geq Q$ and there is some $S \in \binom{[m]}{k}$ with $\prod_{v_i} = Q$ then $v_i = 1$ for some $i \in S$.)
	Then, $(k+1)\delta \leq Q/2^H \leq Q/Q^5 = Q^{-4}$.
\end{proof}
We now have all the necessary pieces to reduce from \SubProd to \PenSum and therefore show the intractability of \PenSum.

\begin{theorem}
	\label{thm:Pre-NP-PenSum}
	\PenSum is \NP-hard.
\end{theorem}
\begin{proof}
	Given an instance $\Instance = (\{v_1,\dots v_m\}, M, k')$ of \SubProd,\lb define~$Q:= M$, $H := 5 \lceil\log_2 Q\rceil$, and $\delta:= (1/2^H)$.
	Construct $A$,~$a_i$,~$b_i$,~$k$, and~$D$ as in Construction~\ref{cons:rational reduct}.
	That is:
	\begin{itemize}
		\item Define $A := \lceil\max_{i \in [m]} (\ln v_i)\rceil+1$;
		\item Define $a_i := \ceilH{a_i^*} = \ceilH{A - \ln v_i}$ for each $i\in [m]$;
		\item Define $b_i := 1/v_i$ for each $i \in [m]$;
		\item Define $Q := M$;
		\item Define $D := \floorH{D^*}  = \floorH{kA - \ln Q - 1}$.
	\end{itemize}
	Let $\Instance' = (\{(a_i, b_i) \mid i \in [m]\}, k, Q, D)$ be the resulting instance of \PenSum.
	
	We first show that $\Instance'$ is a \yes-instance of \PenSum if and only if $\Instance$ is a \yes-instance of \SubProd.
	Suppose first that $\Instance$ is a \yes-instance of \SubProd.
	Then, $\prod_{i \in S} v_i = M = Q$ for some $S \subseteq [m]$.
	Consequently, by Corollary~\ref{cor:Pre-SubProdyesToPenSumyes}, $f(S) \geq D$ and so instance $\Instance'$ is a \yes-instance of~\PenSum.
	
	Conversely, suppose that $\Instance'$ is a \yes-instance of \PenSum.
	Then, there is some $S \subseteq [m]$ such that $f(S) \geq D$.
	By Lemma~\ref{lem:Pre-fDbounds} and Corollary~\ref{cor:Pre-deltaGap2}, we conclude~$f^*(S) > f(S) - k\delta \geq D - k\delta > D^* - (k+1)\delta \geq D^* - Q^{-4}$.
	Consequently, $D^* - f^*(S) \leq Q^{-4}$, which by Corollary~\ref{cor:Pre-deltaGap} implies that $\prod_{i \in S} v_i = Q$.
	Therefore, \Instance~is a \yes-instance of \SubProd.
	
	It remains to show that the reduction takes polynomial time.
	For this, it is sufficient to show that the rationals $A$, $k$, $D$, $a_i$, and~$b_i$ can all be computed in polynomial time for each~$i \in [m]$.
	Observe that $A = \lceil\max_{i \in [m]} (\ln v_i)\rceil$ is the unique integer such that $e^A > \max_{i \in [m]} v_i > e^{A-1}$.
	We have
	$1\leq A \leq \lceil\max_{i \in [m]} \log_2 v_i\rceil$ since~$\ln v_i < \log_2 v_i$.
	So we can compute $A$ in polynomial time by checking all integers in this range.
	
	For each $i \in [m]$, $a_i = \ceilH{A - \ln v_i} = r_i/2^H$, where $r_i$ is the minimum integer such that $A - \ln v_i \leq r_i/2^H$.
	Thus, we can compute $r_i$ by checking $e^{A - r_i/2^H}\leq v_i$ with $r_i = 2^H\cdot (A - \ceilH{\ln v_i})$,
	setting $r_i$ to its successor if the inequality is not satisfied.
	Thus, we can construct $a_i$ in polynomial time, and $a_i$ can be represented with~$\Oh(\log_2 r + H)$ bits.
	The construction of $D$ can be handled in a similar way.
	
	For each $i \in [m]$, rational $b_i = 1/v_i$ can be represented with $\Oh(\log_2 v_i)$ bits
	(recall that we represent $1/v_i$ with binary representations of the integers~$1$ and $v_i$).
	It takes~$\Oh(\log_2 v_i)$ time to construct~$b_i$. 
	The integers~$Q$ and $k$ are taken directly from the instance \Instance of \SubProd.
\end{proof}

\chapter{The Generalized Noah's Ark Problem}
\label{ctr:GNAP}

\section{Introduction}
The definition of \MPDlong was given with the intention of helping to systematically address the challenge of  preservation of biological diversity~\cite{FAITH1992,crozier}.
We, however, have to acknowledge the fact that selecting a species for protection in an ecological intervention will not necessarily help to protect them.
Real-world interaction always have to consider uncertainties such as diseases, famines, floods, or other natural catastrophes.
It is therefore useful to consider the protection of a given taxon only at a specific \sprop~\cite{weitzman}.
Further, we may assume that investing more into a given species will increase their \sprop by a certain amount.

In such a model, we therefore have to consider the \emph{expected} phylogenetic diversity.
In the case of \GNAPLong (\GNAP), for each species one may choose from a set of different actions.
Each choice is then associated with a cost and with a resulting survival probability.
With a preprocessing step, we can even consider combinations of different actions.

Introducing cost differences for species protection makes the problem of maximizing phylogenetic diversity \NP-hard~\cite{pardi07} and thus all of the even richer models are \NP-hard as well.
However, in special cases even then a solution can still be computed greedily~\cite{hartmann,hartmann07}.

Billionnet provided some practical algorithmic ideas and results for \GNAP with integer linear programming~\cite{billionnet13,billionnet17}.
Pardi showed that \GNAP can be solved in pseudo-polynomial running time in the budged plus the number of possible solutions~\cite{PardiThesis}.

Apart from these algorithms and the \NP-hardness, there is no work that systematically studies which structural properties of the input make \GNAP tractable.
In this chapter, we try to fill this gap.
On the way, we also observe close relations to \MCKPLong and to \PS, which we observed in the previous chapter.

\paragraph{Structure of the chapter.}
In the next section, we  give a formal definition of \GNAPLong and we  give an overview over results.
Further, we  provide first observations.
In Section~\ref{sec:GNAP-GNAP}, we provide the algorithmic ideas for solving \GNAP in its most generalized form.
We, additionally, prove that \GNAP is \Wh{1}-hard when parameterized by the number of taxa.
In Section~\ref{sec:GNAP-two-projects}, we consider the special case of \GNAP in which for each taxon we can chose between only two options.

\section{Preliminaries}
In this section, we consider definitions used in this chapter, an overview over the results and some preliminary observations.

\subsection{Projects and Phylogenetic Diversity}
\label{sec:GNAP-PD-def}
A \textit{project}~$p_{i,j}$ is a tuple~$(c_{i,j},w_{i,j})\in \mbb N_0\times \mbb Q\cap \mathbb{R}_{[0,1]}$, where~$c_{i,j}$ is the \textit{cost} and~$w_{i,j}$ is the \textit{\sprop} of~$p_{i,j}$.
For a given phylogenetic~$X$-tree~$\Tree$ and a taxon~$x_i\in X$, a \textit{project list}~$P_i$ is an~$\ell_i$-tuple of projects~$(p_{i,1},\dots,p_{i,\ell_i})$.
As a project with a higher cost will only be considered when the \sprop is higher, we assume the costs and the \sprops{} to be ordered.
That is,~$c_{i,j}<c_{i,j+1}$ and~$w_{i,j}<w_{i,j+1}$ for every project list~$P_i$ and each~$j \in [\ell_i-1]$.
An \textit{$m$-collection of projects}~$\mcal P$ is a set of~$m$~project lists~$\{P_1,\dots,P_m\}$.
For a project set~$S$, the~\textit{total cost~$\Costs(S)$ of~$S$} is~$\sum_{p_{i,j}\in S} c_{i,j}$.

For a given phylogenetic~$X$-tree~$\Tree$, the~\textit{phylogenetic diversity~$PD_\Tree(S)$ of a set of projects~$S=\{p_{1,j_1}, \dots, p_{|X|,j_{|X|}}\}$} is given by

\begin{equation}
PD_\Tree(S) := \sum_{uv\in E} \w(uv) \cdot \left(1 - \prod_{x_i\in \off(v)} (1 - w_{i,j_i}) \right).
\end{equation}
The formula~$\left(1 - \prod_{x_i\in \off(v)} (1 - w_{i,j_i}) \right)$ describes the likelihood that at least one offspring of~$v$ survives. Thus, the phylogenetic diversity is the sum of the expected values of the edges of~$\Tree$ when applying~$S$.
If the \sprops of all projects in~$S$ is either~0 or~1, then this definitions coincides with the standard definition of phylogenetic diversity given in Equation~(\ref{eqn:PDdef}).

In this chapter, we use the convention that~$n := |V(\Tree)|$. Observe~$n \in \Oh(|X|)$.

\subsection{Problem Definitions}
%\todos[inline]{Diskussion: Ich würde gerne bei~$\var_w$ bleiben und Wahrscheinlichkeiten mit~$w_i$ ausdrücken, da~$p$ bereits für projects reserviert ist. Billionnet hat auch für die Wahrscheinlichkeiten~$w$ benutzt~\cite{billionnet13}.}
We now define this chapter's main problem, \GNAPLong, and the special case where each species has two projects,~\NAPLong[a_i]{c_i}{b_i}{2}.
\problemdef{\GNAPLong (\GNAP)}
{A phylogenetic~$X$-tree~$\Tree=(V,E,\w)$,
	an~$\bet{X}$-collection of projects~$\mathcal{P}$,
	an integer~$B\in \mbb N_0$, and a number~$D\in \mbb Q_{\ge 0}$}
{Is there a set of projects~$S=\{p_{1,j_1}, \dots, p_{|X|,j_{|X|}}\}$, one from each project list of~$\mathcal{P}$, 
	such that~$PD_\Tree(S)\ge D$ and~$\Costs(S)\le B$}
A project set~$S$ is called~\textit{a solution for instance~$\mathcal I=(\Tree,\mathcal P,B,D)$} if~$S$ satisfies the conditions in the question of the problem definition.

\problemdef{\NAPLong[a_i]{c_i}{b_i}{2} (\NAP[a_i]{c_i}{b_i}{2})}
{A phylogenetic~$X$-tree~$\Tree=(V,E,\w)$,
	a~$\bet{X}$-collection of projects~$\mathcal{P}$ in which the project list~$P_i$ contains exactly two projects~$(0,a_i)$ and~$(c_{i},b_i)$ for each~$i\in [\bet{X}]$,
	an integer~$B\in \mbb N_0$, and a number~$D\in \mbb Q_{\ge 0}$}
{Is there a set of projects~$S=\{p_{1,j_1}, \dots, p_{|X|,j_{|X|}}\}$, one from each project list of~$\mathcal{P}$,
	such that~$PD_\Tree(S)\ge D$, and~$\Costs(S)\le B$}
In other words, in an instance~\NAP[a_i]{c_i}{b_i}{2} we can decide for each taxon~$x_i$ whether we want to spend~$c_i$ to increase the \sprop of~$x_i$ from~$a_i$ to~$b_i$.
If $P_i$ contains two projects~$(\bar c,a_i)$ and~$(\hat c,b_i)$ with $\bar c > 0$ for an~$i\in [\bet{X}]$, we can reduce the cost of the two projects and the budget by $\bar c$ each to obtain an equivalent instance.
Therefore, we will assume that the project with the lower \sprop has a cost of~0 and we refer to the cost of the other project as $c_i$.
%Thus,~\NAP[a_i]{c_i}{b_i}{2} is a special case of~\GNAP.

As a special case of this problem, we consider~\NAP[0]{c}{b_i}{2} where every project with a positive \sprop has a cost of~$c$.
Observe that for each~$c\in\mathbb{N}$, an instance~$\mathcal{I}=(\Tree,\mathcal{P},B,D)$ of~\NAP[0]{c}{b_i}{2} can be reduced to an equivalent instance~$\mathcal{I}'=(\Tree,\mathcal{P}',B',D)$ of~\NAP[0]{1}{b_i}{2} by replacing each project~$(c,b_i)$ with~$(1,b_i)$, and setting~$B'=\left\lfloor B/c \right\rfloor$. Thus,~\NAP[0]{c}{b_i}{2} can be considered as the special case of~\GNAP with unit costs for projects.
In the following we refer to this problem as~\ucNAP.

\subsection{Parameters, and Results Overview}
We study \GNAP and the special cases of \GNAP with respect to several parameters which we describe in the following; For an overview of the results see Table~\ref{tab:results-gnap}.
If not stated differently, we assume in the following that~$i \in [\bet X]$ and~$j\in [|P_i|]$.
The input of \GNAP directly gives the following natural parameters: The \textit{number of taxa~$|X|$}, the \textit{budget~$B$}, and the required \textit{diversity~$D$}. Closely related to~$B$ is~$C := \max_{i,j} c_{i,j}$, the \textit{maximum cost of a project}. We may assume that no projects have a cost that exceeds the budget, as we can delete them from the input and so~$C\le B$. We may further assume that~$B\le C\cdot \bet X$, as otherwise we can compute in polynomial time whether the diversity of the most valuable projects of the taxa exceeds~$D$ and return \yes, if it does and \no, otherwise.

Further, we consider the \textit{maximum number of projects per taxon~$L:=\max_{i} \bet{P_i}$}. By definition,~$L=2$ in~\NAP[a_i]{c_i}{b_i}{2} and in~\GNAP we have~$L\le C+1$.
We denote the \textit{number of projects} by~$\numP := \sum_{i} |P_i|$. Clearly,~$\bet X\le \numP$, $L\le \numP$,\lb and~$\numP \le \bet X \cdot L$.
We denote with~$\var_c:=|\{ c_{i,j} : (c_{i,j},w_{i,j})\in P_i, P_i\in \mathcal{P} \}|$, the \textit{number of different costs}.
We define the \textit{number of different \sprops}~$\var_w$, analogously.
The consideration of this type of parameterization, called the \textit{number of numbers} parameterization was initiated by Fellows et al.~\cite{fellows12};
It is motivated by the idea that in many real-life instances the number of numbers may be small. In fact, for non-negative numbers, it is never larger than the maximum values which are used in pseudo-polynomial time algorithms.
Also, we consider the~\textit{maximum encoding length for \sprops $\wcode := \max_{i,j}(\text{binary length of } w_{i,j})$} and the~\textit{maximum edge weight~$\max_\w := \max_{e\in E} \w(e)$}. Observe that because the maximal \sprop of a taxon could be smaller than 1, one can not assume that~$\max_\w\le D$, in this chapter.

%\kommentar{\paragraph{Results}}{}
\begin{table}[t]
	\centering
	\footnotesize
	\caption{Complexity results for \GNAPLong.
		Recall that the special cases mentioned are \NAP[0]{c_i}{1}{2} which is the special case where the \sprops are only 0 or 1, and \ucNAP is the special case where each project has unit costs.
		Entries with the sign ``---'' mark parameters that are (partially) constant in the specific problem definition and thus are not interesting.}
	\label{tab:results-gnap}
	\myrowcols
	\resizebox{\columnwidth}{!}{%
		\begin{tabular}{l ll ll}
			\hline
			Parameter & \multicolumn{2}{c}{\GNAP} & \multicolumn{2}{c}{\GNAP{} with $\height_{\Tree}=1$} \\
			\hline
			$|X|$ & \Wh1-hard, \XP & Thm.~\ref{thm:GNAP-X-W1hard}, Prop.~\ref{prop:GNAP-GNAP-X-XP} & \Wh1-hard, \XP & Thm.~\ref{thm:GNAP-X-W1hard}, Prop.~\ref{prop:GNAP-GNAP-X-XP}\\
			$B$ & \XP; \FPT is \textit{open} & Obs.~\ref{obs:GNAP-XP-B} & \PFPT $\Oh(B \cdot \numP)$ & Prop.~\ref{prop:GNAP-height=1->mckp}\\
			$C$ & \NP-h for $C=1$ & Thm.~\ref{thm:GNAP-ucNAP-hardness} & \PFPT $\Oh(C \cdot \numP \cdot |X|)$ & Prop.~\ref{prop:GNAP-height=1->mckp}\\
			$D$ & \NP-h for $D=1$ & Obs.~\ref{obs:GNAP-D=height=1,varw=2} & \NP-h for $D=1$ & Obs.~\ref{obs:GNAP-D=height=1,varw=2}\\
			\hline
			$\max_\w$ & \NP-h for $\max_\w=1$ & Thm.~\ref{thm:GNAP-X-W1hard} & \NP-h for $\max_\w=1$ & Thm.~\ref{thm:GNAP-X-W1hard}\\
			$\var_c$ & \NP-h for $\var_c=2$ & Thm.~\ref{thm:GNAP-ucNAP-hardness} & \XP $\Oh(|X|^{\var_c-1} \cdot \numP)$ & Prop.~\ref{prop:GNAP-height=1->mckp}\\
			$\var_w$ & \NP-h for $\var_w = 2$ & Obs.~\ref{obs:GNAP-KP-to-height=1} & \NP-h for $\var_w = 2$ & Obs.~\ref{obs:GNAP-KP-to-height=1}\\
			$D+\wcode$ & \textit{open} &  & \FPT $\Oh(D \cdot 2^{\wcode} \cdot \numP)$ & Prop.~\ref{prop:GNAP-height=1->mckp}\\
			$B+\var_w$ & \XP $\Oh(B \cdot |X|^{2\cdot \var_w + 1}$ & Thm.~\ref{thm:GNAP-B+varw} & \PFPT & Prop.~\ref{prop:GNAP-height=1->mckp}\\
			$D+\var_w$ & \NP-h for $D=1, \var_w=2$ & Obs.~\ref{obs:GNAP-D=height=1,varw=2} & \NP-h for $D=1, \var_w=2$ & Obs.~\ref{obs:GNAP-D=height=1,varw=2}\\
			$\var_c+\var_w$ & \XP $\Oh(|X|^{2\cdot (\var_c + \var_w) +1}$ & Thm.~\ref{thm:GNAP-varc+varw} & \FPT & Thm.~\ref{thm:GNAP-height=1:varc+varw}\\
			\hline
			\hline
			Parameter & \multicolumn{2}{c}{\NAP[0]{c_i}{1}{2}} & \multicolumn{2}{c}{\ucNAP} \\
			\hline
			$|X|$ & \FPT & Obs.~\ref{obs:GNAP-FPT-X-2-NAP} & \FPT & Obs.~\ref{obs:GNAP-FPT-X-2-NAP}\\
			$B$ & \PFPT $\Oh(B^2 \cdot n)$ & \cite{pardi07} & \XP & Obs.~\ref{obs:GNAP-XP-B} \\
			$C$ & \PFPT $\Oh(C^2 \cdot n^3)$ & Cor.~\ref{cor:GNAP-C-01-NAP} & --- & \\
			$D$ & \PFPT $\Oh(D^2 \cdot n)$ & Prop.~\ref{prop:GNAP-wcode=1-D} & \textit{open} & \\
			\hline
			$\max_\w$ & \PFPT $\Oh((\max_\w)^2 \cdot n^3)$ & Cor.~\ref{cor:GNAP-wcode=1,val-lambda} & \textit{open} & \\
			$\var_c$ & \XP & Cor.~\ref{cor:GNAP-wcode=1-varc} & --- & \\
			$\var_w$ & --- &  & \XP & Cor.~\ref{cor:GNAP-C=1-XP-varw}\\
			\hline
		\end{tabular}
	}
\end{table}

\subsection{Observations for GNAP}
We first present some basic observations that provide some first complexity classifications.
In the problem with exactly two projects per taxa,~\NAP[a_i]{c_i}{b_i}{2}, one can iterate over all subsets~$X'$ of taxa and check if it is a possible solution pay~$c_i$ to increase the \sprop for each~$x_i\in X'$.
To this end, we check if~$\sum_{x_i\in X'} c_i \le B$ and compute if the phylogenetic diversity is at least~$D$, when the \sprop of every~$x_i\in X'$ is~$b_i$ and~$a_i$ otherwise. Thus,~\NAP[a_i]{c_i}{b_i}{2} is fixed-parameter tractable with respect to the number of taxa.
\begin{observation}
	\label{obs:GNAP-FPT-X-2-NAP}
	\NAP[a_i]{c_i}{b_i}{2} can be solved in~$2^{\bet X} \cdot \poly(|\Instance|)$~time.
\end{observation}

A~$\GNAP$ solution contains at most~$B$ projects with positive costs.
Hence, a solution can be found by iterating over all~$B$-sized subsets~$X'$ of taxa and checking every combination of projects for~$X'$. Like before, we have to check that the budget is not exceeded and the phylogenetic diversity of the selected projects is at least~$D$.
This brute-force algorithm shows that~\GNAP is \XP with respect to the budget.
\begin{observation}
	\label{obs:GNAP-XP-B}
	\GNAP can be solved in~$(|X| \cdot L)^B \cdot \poly(|\Instance|)$~time.
\end{observation}

In \KP, one is given a set of items~$N$, a cost-function~$c: N \to \mathbb N$, a value-function~$d: N \to \mathbb N$, and two integers~$B$ and~$D$ and asks whether there is an item set~$N'$ such that~$c_\Sigma(N') \le B$ and~$d_\Sigma(N') \ge D$.
We describe briefly a known reduction from~\KP to~\NAP[0]{c_i}{1}{2}~\cite{pardi07}.
Let~$\Instance=(N,c,d,B,D)$ be an instance of~\KP.
Define an~$N$-tree~$\Tree:=(V,E,\w)$ with vertices~$V:=\{\rho\} \cup N$ and edges~$E:=\{ \rho x_i \mid x_i\in N \}$ of weight~$\w(\rho x_i) := d(x_i)$.
For each taxon~$x_i$ we define a project list~$P_i$ that contains two projects~$(0,0)$ and~$(c(x_i),1)$.
Then, the instance~$(\Tree,\mathcal{P},B':=B,D':=D)$ is a \yes-instance of~\NAP[0]{c_i}{1}{2} if and only if~$(N,c,d,B,D)$ is a \yes-instance of~\KP.
Because \KP is \NP-hard, also \NAP[0]{c_i}{1}{2} is \NP-hard.
\begin{observation}[\cite{pardi07}]
	\label{obs:GNAP-KP-to-height=1}
	\NAP[0]{c_i}{1}{2} is \NP-hard, even if the tree~$\Tree$ is a star.
\end{observation}
Because~\NAP[0]{c_i}{1}{2} is a special case of~\GNAP in which~$L=2,~\wcode=1$, and~$\var_w=2$, we conclude that~\GNAP is \NP-hard, even if~$\height_{\Tree}=\wcode=1$ and~$L=\var_w=2$.
In this reduction, one could also set~$D':=1$ and set the \sprop of every project with a positive cost to~$1/D$.
\begin{observation}
	\label{obs:GNAP-D=height=1,varw=2}
	\NAP[0]{c_i}{b}{2} is \NP-hard, even if~$b\in (0,1]$ is a constant,~$D=1$, and the given phylogenetic~$X$-tree~$\Tree$ is a star.
\end{observation}

\section{The Generalized Noah's Ark Problem}
\label{sec:GNAP-GNAP}
In this section, with \GNAP we consider the most general form of the problem.
\subsection{Algorithms for the Generalized Noah's Ark Problem}
First, we observe that for a constant number of taxa, we can solve \GNAP in polynomial time by considering all the possible project choices for each taxon. 
\begin{proposition}
	\label{prop:GNAP-GNAP-X-XP}
	\GNAP is \XP with respect to~$|X|$.
\end{proposition}
\begin{proof}
	For every taxon~$x_i \in X$, iterate over the projects~$p_{i,j_i}$ of~$P_i$ such that there are~$|X|$~nested loops to compute a set~$S:=\{p_{i,j_i} \mid i\in[|X|]\}$.
	Return \yes, if~$\Costs(S) \le B$ and~$PD_{\Tree}(S) \ge D$.
	Otherwise, return \no, after the iteration.
	
	This algorithm is clearly correct.
	We can check whether~$\Costs(S) \le B$\lb and $PD_{\Tree}(S) \ge D$ in~$\Oh(n^2)$~time.
	Therefore, in~$\Oh(L^{|X|} \cdot n^2)$~time a solution is computed.
\end{proof}

In Theorem~\ref{thm:GNAP-varc+varw}, we show that \GNAP can be solved in polynomial time when the number of different project costs and the number of different survival probabilities is constant.
In the following, let~$\Instance=(\Tree,\mathcal P,B,D)$ be an instance of~\GNAP, and let~$\mathcal C:=\{c_1,\dots,c_{\var_c}\}$ and~$\mathcal W:=\{w_1,\dots,w_{\var_w}\}$ denote the sets of different costs and different \sprops in~\Instance, respectively.
Without loss of generality, assume~$c_i<c_{i+1}$ for each~$i\in[\var_c-1]$ and assume~$w_j<w_{j+1}$ for each~$j\in [\var_w-1]$, likewise. In other words,~$c_i$ is the~$i$th cheapest cost in~$\mathcal C$ and~$w_j$ is the~$j$th smallest \sprop in~$\mathcal W$.
Recall that we assume that there is at most one item with cost~$c_p$ and at most one item with \sprop~$w_q$ in every project list~$P_i$, for each~$p\in[\var_c]$ and~$q\in[\var_w]$.
For the rest of the section, by~$\myvec{a}$ and~$\myvec{b}$ we denote vectors~$(a_1,\dots,a_{\var_c-1})$ and~$(b_1,\dots,b_{\var_w-1})$, respectively.

\begin{theorem}
	\label{thm:GNAP-varc+varw}
	\GNAP can be solved in~$\Oh\left(|X|^{2(\var_c+\var_w-1)}\cdot (\var_c+\var_w)\right)$ time.
\end{theorem}
\begin{proof}
	\proofpara{Table Defintion}
	We describe a dynamic programming algorithm with two tables~$\DP$~and~$\DP'$ that have a dimension for all the~$\var_c$ different costs, except for~$c_{\var_c}$ and all the~$\var_w$ different \sprops, except for~$w_{\var_w}$.
	Recall that~$T_v$ is the subtree rooted at $v$ and the offspring $\off(v)$ of~$v$ are the leaves in $T_v$, and
	that the $i$-partial subtree $T_{v,i}$ rooted at $v$ is the subtree of~$T_v$ containing only the first children~$w_1,\dots,w_i$ of $v$ for some~$i\in [t]$ for a vertex $v$ with children~$w_1,\dots,w_t$ of~$v$.
	For a vertex~$v\in V$ and given vectors~\myvec{a} and~\myvec{b}, we define~$\mathcal S^{(v)}_{\myvec{a},\myvec{b}}$ to be the family of sets of projects~$S$ such that
	\begin{itemize}
		\item $S$ contains exactly one project of~$P_i$ for each~$x_i\in \off(v)$,
		\item $S$ contains exactly~$a_k$ projects of cost~$c_k$ for each~$k \in [\var_c-1]$, and
		\item $S$ contains exactly~$b_\ell$ projects of \sprop~$w_\ell$ for each~$\ell \in [\var_w-1]$.
	\end{itemize}
	For a vertex~$v\in V$ with children~$u_1,\dots,u_t$, given vectors~\myvec{a}, and~\myvec{b} and a given integer~$i\in[t]$ we define~$\mathcal S^{(v,i)}_{\myvec{a},\myvec{b}}$ analogously, just that exactly one project of $P_j$ is chosen for each~$x_j\in \off(u_1)\cup\dots\cup\off(u_i)$.
	
	It follows that we can compute how many projects with cost~$c_{\var_c}$ and \sprop~$w_{\var_w}$ a set~$S\in \mathcal S^{(v)}_{\myvec{a},\myvec{b}}$ contains.
	That are~$a_{\var_c}^{(v)} := |\off(v)|-\sum_{j=1}^{\var_c-1} a_j$ projects with a cost of~$c_{\var_c}$ and~$b_{\var_w}^{(v)} := |\off(v)|-\sum_{j=1}^{\var_w-1} b_j$ projects with a \sprop of~$w_{\var_w}$.
	We want entries~$\DP[v,\myvec{a},\myvec{b}]$ to store the largest expected phylogenetic diversity~$\PDsub{T_v}(S)$ of the tree $T_v$, for a set~$S\in \mathcal S^{(v)}_{\myvec{a},\myvec{b}}$.
	We analogously want~$\DP'[v,i,\myvec{a},\myvec{b}]$ to store the largest expected phylogenetic diversity~$\PDsub{T_{v,i}}(S)$ of the tree $T_{v,i}$, for a set~$S\in\mathcal S^{(v,i)}_{\myvec{a},\myvec{b}}$.
	We further define the total \sprop to~be
	\begin{eqnarray}
		\label{eqn:sprops}
		w(b_{\var_w},\myvec{b}) := 1-
		(1-w_{\var_w})^{b_{\var_w}}
		\cdot
		\prod_{i=1}^{\var_w-1}
		(1-w_i)^{b_i},
	\end{eqnarray}
	when~$b_{\var_w}$ and~$\myvec{b}$ describe the number of chosen single~\sprops.
	
	\proofpara{Algorithm}
	As a base case, fix a taxon~$x_i$ with project list~$P_i$.
	As we want to select exactly one project of $P_i$, the project is clearly defined by $\myvec a$ and $\myvec b$.
	So, we store~$\DP[x_i,\myvec{a},\myvec{b}] = 0$, if~$P_i$ contains a project~$p = (c_k,w_\ell)$ such that
	\begin{itemize}
		\item ($k<\var_c$ and~$\myvec{a}={\myvec{0}}_{(k)+1}$ or~$k=\var_c$ and~$\myvec{a}={\myvec{0}}$), and
		\item ($\ell<\var_w$ and~$\myvec{b}={\myvec{0}}_{(\ell)+1}$ or~$\ell=\var_w$ and~$\myvec{b}={\myvec{0}}$).
	\end{itemize}
	Otherwise, store~$\DP[x_i,\myvec{a},\myvec{b}] = -\infty$.
	
	Let~$v$ be an internal vertex with children~$u_1,\dots,u_t$.
	We define
	\begin{eqnarray}
		\label{eqn:varc+varw-G-1}
		\DP'[v,1,\myvec{a},\myvec{b}]
		&=&
		\DP[u_1,\myvec{a},\myvec{b}]
		+ \w(v u_1) \cdot
		w\left(b_{\var_w}^{(u_{1})},\myvec{b}\right).
	\end{eqnarray}
	To compute further values of~$\DP'$, we use the recurrence
	\begin{eqnarray}
		\label{eqn:varc+varw-G-i+1}
		&&\DP'[v,i+1,\myvec{a},\myvec{b}]\\
		\nonumber
		&=&
		\max_{
			\begin{array}{l}
				\myvec{a'}, \myvec{b'}
			\end{array}}
		\DP'[v,i,\myvec{a}-\myvec{a'},\myvec{b}-\myvec{b'}]
		+ \DP[u_{i+1},\myvec{a'},\myvec{b'}]
		+ \w(v u_{i+1}) \cdot
		w\left(b_{\var_w}^{(u_{i+1})},\myvec{b'}\right)
		.
	\end{eqnarray}
	Herein, $\myvec{a'}$ and~$\myvec{b'}$ are selected to satisfy $\myvec{0}\le \myvec{a'}\le \myvec{a}$ and $\myvec{0}\le \myvec{b'}\le \myvec{b}$.
		
	Finally, we define~$\DP[v,\myvec{a},\myvec{b}] = \DP'[v,t,\myvec{a},\myvec{b}]$.
	
	Return \yes if there are~$\myvec{a}$ and~$\myvec{b}$ such that~$\sum_{i=1}^{\var_c-1} a_i \le |X|$, and~$\sum_{i=1}^{\var_w-1} b_i \le |X|$, and~$a_{\var_c}^{(r)} \cdot c_{\var_c} + \sum_{i=1}^{\var_c-1} a_i \cdot c_i \le B$, and~$\DP[\rho,\myvec{a},\myvec{b}]\ge D$ where~$\rho$ is the root of~$\Tree$.
	Otherwise, if no such~$\myvec{a}$ and~$\myvec{b}$ exist, return \no.
	
	\proofpara{Correctness}
	For any vertex~$v$, vectors~\myvec{a},~\myvec{b}, and an integer~$i$, we prove that~$\DP[v,\myvec{a},\myvec{b}]$ and~$\DP'[v,i,\myvec{a},\myvec{b}]$ store~$\max PD_{\Tree_v}(S)$ for~$S\in \mathcal S^{(v)}_{\myvec{a},\myvec{b}}$ and~$\max PD_{\Tree_{v,i}}(S)$ for~$S\in \mathcal S^{(v,i)}_{\myvec{a},\myvec{b}}$, respectively. This implies that the algorithm is correct.
	For a taxon~$x_i$, the tree~$T_{x_i}$ does not contain edges and so there is no diversity. We can only check if~\myvec{a} and~\myvec{b} correspond to a feasible project. So, the table~$F$ stores the correct value in the base cases.
	For an internal vertex~$v$ with children~$u_1,\dots,u_t$, and~$i\in [t-1]$,
	observe that~$PD_{\Tree_{v,1}}(S) = PD_{\Tree_{u_1}}(S) + \w(v u_1) \cdot w(b_{\var_w}^{(u_{i+1})},\myvec{b})$ for~$S\in \mathcal S^{(v,1)}_{\myvec{a},\myvec{b}}$, where~$w(b_{\var_w}^{(u_{i+1})},\myvec{b})$ is the \sprop at~$u_1$.
	Thus entry~$\DP'[v,1,\myvec{a},\myvec{b}]$ stores the correct value.
	Further, the value in entry~$\DP[v,\myvec{a},\myvec{b}]$ stores the correct value, when~$\DP'[v,t,\myvec{a},\myvec{b}]$ stores the correct value, because~$\mathcal S^{(v)}_{\myvec{a},\myvec{b}}=\mathcal S^{(v,t)}_{\myvec{a},\myvec{b}}$.
	It remains to show that the correct value is stored in~$\DP'[v,i+1,\myvec a,\myvec b]$.
	
	Now, assume as an induction hypothesis that in~$\DP[u_j,\myvec a,\myvec b]$ and~$\DP'[v,i,\myvec a,\myvec b]$ the correct value is stored, for an internal vertex~$v$ with children~$u_1,\dots,u_t$ and~$i\in [t-1]$.
	We first prove that if~$\DP'[v,i+1,\myvec{a},\myvec{b}]=d$, then there exists a set~$S\in \mathcal S^{(v,i+1)}_{\myvec{a},\myvec{b}}$ with~$PD_{\Tree_{v}}(S)=d$.
	Afterward, we prove that~$\DP'[v,i+1,\myvec{a},\myvec{b}]\ge PD_{\Tree_{v,i+1}}(S)$ for every set~$S\in \mathcal S^{(v,i+1)}_{\myvec{a},\myvec{b}}$.
	
	Let~$\DP'[v,i+1,\myvec{a},\myvec{b}]=d$.
	Let~\myvec{a'} and~\myvec{b'} be the vectors that maximize the right side of Equation~(\ref{eqn:varc+varw-G-i+1}) for~$\DP'[v,i+1,\myvec{a},\myvec{b}]$.
	By the induction hypothesis, there is a set~$S_G\in S^{(v,i)}_{\myvec{a}-\myvec{a'},\myvec{b}-\myvec{b'}}$ such that~$\DP'[v,i,\myvec{a}-\myvec{a'},\myvec{b}-\myvec{b'}]=PD_{\Tree_{v,i}}(S_G)$ and there is a set~$S_F\in S^{(u_{i+1})}_{\myvec{a'},\myvec{b'}}$ such that~$\DP[u_{i+1},i,\myvec{a'},\myvec{b'}]=PD_{\Tree_{u_{i+1}}}(S_F)$.
	Define~$S := S_G \cup S_F$.
	Then,
	\begin{eqnarray}
		PD_{\Tree_{v,i+1}}(S) &=& PD_{\Tree_{v,i}}(S_G) + PD_{\Tree_{v}}(S_F)\\
		\label{eqn:varc+varw-RA2}
		&=& PD_{\Tree_{v,i}}(S_G) + PD_{\Tree_{u_{i+1}}}(S_F) + \w(v u_{i+1}) \cdot w(b_{\var_w}^{(u_{i+1})},\myvec{b}).
	\end{eqnarray}
	This equals the right side of \Recc{eqn:varc+varw-G-i+1} and we conclude~$PD_{\Tree_v}(S)=d$.
	
	Conversely,
	let~$S\in \mathcal \mathcal S^{(v,i+1)}_{\myvec{a},\myvec{b}}$.
	Let~$S_F$ be the subset of projects of~$S$ that are from a project list of an offspring of~$u_{i+1}$ and define~$S_G = S\setminus S_F$.
	Let~$a_k$ be the number of projects in~$S_F$ with a cost of~$c_k$ and let~$b_\ell$ be the number of projects in~$S_F$ with a \sprop of~$b_\ell$. Define~$\myvec{a'}=(a_1,\dots,a_{\var_c-1})$ and~$\myvec{b'}=(b_1,\dots,b_{\var_w-1})$.
	Then,
	\begin{eqnarray}
		\label{eqn:varc+varw-LA1}
		&& \DP'[v,i+1,\myvec a,\myvec b]\\
		&\ge & \DP'[v,i,\myvec a-\myvec{a'},\myvec b-\myvec{b'}] + \DP[u_{i+1},\myvec{a'},\myvec{b'}] + \w(v u_{i+1}) \cdot w\left({b'}_{\var_w}^{(u_{i+1})},\myvec{b'}\right)\\
		\label{eqn:varc+varw-LA2}
		&=& PD_{\Tree_{v,i}}(S_G) + PD_{\Tree_{u_{i+1}}}(S_F) + \w(v u_{i+1}) \cdot w\left({b'}_{\var_w}^{(u_{i+1})},\myvec{b'}\right)\\
		\label{eqn:varc+varw-LA3}
		&=& PD_{\Tree_{v,i+1}}(S).
	\end{eqnarray}
	Inequality~(\ref{eqn:varc+varw-LA1}) follows from~\Recc{eqn:varc+varw-G-i+1}.
	By the definition of~$S_F$ and~$S_G$, Equation~(\ref{eqn:varc+varw-LA2}) is correct.
	Finally, Equation~(\ref{eqn:varc+varw-LA3}) follows from Equation~(\ref{eqn:varc+varw-RA2}).

	\proofpara{Running time}
	First, we prove how many options for vectors~$\myvec{a}$ and \myvec{b} there are.
	Because~$\sum_{i=1}^{\var_c-1} a_i \le |X|$, we conclude that if~$a_i=|X|$, then~$a_j=0$ for~$i\ne j$.
	Otherwise, for~$a_i \in [|X|-1]_0$ there are~$\Oh(|X|^{\var_c-1})$ options for \myvec a of not containing~$|X|$, such that altogether there are~$\Oh(|X|^{\var_c-1}+|X|)=\Oh(|X|^{\var_c-1})$ options for a suitable~\myvec a.
	Likewise, there are~$\Oh(|X|^{\var_w-1})$ options for a suitable~\myvec b.
	
	Clearly, the base cases can be computed in~$\Oh(\numP)$~time, each.
	Let~$v$ be an internal vertex and fix~$\myvec{a}$ and~$\myvec{b}$.
	For a vertex~$w\in V$, we can compute in~$\Oh(n)$~time the set~$\off(w)$.
	It follows that~$w(b_{\var_w}^{(u_{i})},\myvec{b})$ can be computed in~$\Oh(n+\var_w)$~time such that~\Recc{eqn:varc+varw-G-1} can be computed in~$\Oh(n+\var_w)$~time for fixed~$v$,~$i$,~$\myvec{a}$, and~$\myvec{b}$.
	For fixed~$\myvec{a}$ and \myvec{b}, there are~$\Oh\left(|X|^{\var_c+\var_w-2}\right)$ options to choose~$\myvec{a'}$ and \myvec{b'}.
	Therefore,~\Recc{eqn:varc+varw-G-i+1} can be evaluated in~$\Oh\left(|X|^{\var_c+\var_w-2}\cdot (n+\var_w)\right)$~time.
	
	\Recc{eqn:varc+varw-G-1} has to be computed once for every internal vertex.
	\Recc{eqn:varc+varw-G-i+1} has to be computed once for every vertex except the root.
	Altogether, in~$\Oh\left(|X|^{2(\var_c+\var_w-2)}\cdot (n+\var_w) + |X| \cdot \numP\right)$~time all entries of the tables~$\DP$ and~$\DP'$ can be computed.
	Additionally,~$\Oh(|X|^{\var_c+\var_w-2}\cdot (\var_c+\var_w))$~time is needed to check whether there are vectors~$\myvec{a}$ and~$\myvec{b}$ such that
	\begin{itemize}
		\item $\DP[r,\myvec{a},\myvec{b}]\ge D$,
		\item $\sum_{i=1}^{\var_c-1} a_i \le |X|$,
		\item $\sum_{i=1}^{\var_w-1} b_i \le |X|$, and
		\item $a_{\var_c}^{(r)} \cdot c_{\var_c} + \sum_{i=1}^{\var_c-1} a_i \cdot c_i \le B$.
	\end{itemize}
	Because~$\Oh(n)=\Oh(|X|)$ and~$\Oh(\numP)\le \Oh(|X| \cdot \var_w)$, the overall running time is~$\Oh\left(|X|^{2(\var_c+\var_w-1))}\cdot (\var_c+\var_w)\right)$ in our RAM-model.
	We want to declare, however, that the table entries may store numbers with an encoding length of up to~$|X| \cdot \wcode + \log(D)$, which is not linear in the input size.
\end{proof}

The algorithm of Theorem~\ref{thm:GNAP-varc+varw} can easily be adjusted for~\NAP[0]{c_i}{1}{2}, in which the only \sprops are 0 and 1.
In this problem, we additionally can compute~$w(b_{\var_w}^{(u_{1})},\myvec{b})$ faster, as it can only be 0 or 1.
\begin{corollary}
	\label{cor:GNAP-wcode=1-varc}
	\NAP[0]{c_i}{1}{2} can be solved in~$\Oh\left(|X|^{2(\var_c+1)}\cdot \var_c\right)$ time.
\end{corollary}

As each project with a cost higher than~$B$ can be deleted, we may assume that there are no such projects.
This implies that~$\var_c\le C+1\le B+1$.
Thus, Theorem~\ref{thm:GNAP-varc+varw} also implies that~\GNAP{} is \XP with respect to~$C+\var_w$ and~$B+\var_w$ with astronomical running times of $\Oh\left(|X|^{2(C+\var_w-1)}\cdot (C+\var_w)\right)$ and $\Oh\left(|X|^{2(B+\var_w-1)}\cdot (B+\var_w)\right)$, respectively.
However, we can adjust the algorithm so that~$B$ is not in the exponent of the running time.
Instead of declaring how many projects of cost $c_i$ for $i\in [\var_c]$ are selected, we declare the budget that can be spent.

\begin{theorem}
	\label{thm:GNAP-B+varw}
	\GNAP can be solved in~$\Oh\left(B^2\cdot |X|^{2(\var_w-1)}\cdot \var_w\right)$ time.
\end{theorem}
\begin{proof}
	\proofpara{Table definition}
	We describe a dynamic programming algorithm with two tables~$\DP$ and~$\DP'$ that have a dimension for all the~$\var_w$ different \sprops, except for~$\var_w-1$.
	For a vertex~$v\in V$, a given vector~\myvec{b} and $k\in [B]_0$, we define~$\mathcal S^{(v)}_{k,\myvec{b}}$ to be the family of sets of projects~$S$ such that
	\begin{itemize}
		\item $S$ contains exactly one project of~$P_i$ for each~$x_i\in \off(v)$,
		\item $\Costs(S)\le k$, and
		\item $S$ contains exactly~$b_\ell$ projects of \sprop~$w_\ell$ for each~$\ell \in [\var_w-1]$.
	\end{itemize}
	For a vertex~$v\in V$ with children~$u_1,\dots,u_t$, a given vector~\myvec{b} and integers~$k\in [B]_0$ and~$i\in[t]$ we define~$\mathcal S^{(v,i)}_{k,\myvec{b}}$ analogously, just that exactly one project of $P_j$ is chosen for each~$x_j\in \off(u_1)\cup\dots\cup\off(u_i)$.

	As in the previous proof, there are~$b_{\var_w}^{(v)} := |\off(v)|-\sum_{j=1}^{\var_w-1} b_j$ projects with \sprop~$w_{\var_w}$.
	We want entries~$\DP[v,k,\myvec{b}]$ to store the largest expected phylogenetic diversity~$\PDsub{T_v}(S)$ of the tree $T_v$, for a set~$S\in \mathcal S^{(v)}_{k,\myvec{b}}$ and,
	respectively,
	$\DP'[v,i,k,\myvec{b}]$ to store the largest expected phylogenetic diversity~$\PDsub{T_{v,i}}(S)$ of the tree~$T_{v,i}$, for a set~$S\in\mathcal S^{(v,i)}_{k,\myvec{b}}$.
	We further define the total \sprop analogously as in Equation~(\ref{eqn:sprops}).

	\proofpara{Algorithm}
	As a base case, fix a taxon~$x_i$ with project list~$P_i$.
	As we want to select exactly one project of $P_i$, the project is clearly defined by $k$ and $\myvec b$.
	So, we store~$\DP[x_i,k,\myvec{b}] = 0$, if~$P_i$ contains a project~$p = (c_t,w_\ell)$ such that $c_t\le k$, and
	\begin{itemize}
		\item ($\ell<\var_w$ and~$\myvec{b}={\myvec{0}}_{(\ell)+1}$ or~$\ell=\var_c$ and~$\myvec{b}={\myvec{0}}$).
	\end{itemize}
	Otherwise, store~$\DP[x_i,k,\myvec{b}] = -\infty$.

	Let~$v$ be an internal vertex with children~$u_1,\dots,u_t$.
	We define
	\begin{eqnarray}
		\label{eqn:B+varw-G-1}
		\DP'[v,1,k,\myvec{b}]
		&=&
		\DP[u_1,k,\myvec{b}]
		+ \w(v u_1) \cdot
		w\left(b_{\var_w}^{(u_{1})},\myvec{b}\right).
	\end{eqnarray}
	To compute further values of~$G$, we can use the recurrence
	\begin{eqnarray}
	\label{eqn:B+varw-G-i+1}
	&&\DP'[v,i+1,k,\myvec{b}]\\
	\nonumber
	&=&
	\max_{
		\begin{array}{l}
			k', \myvec{b'}
	\end{array}}
	\DP'[v,i,k-k',\myvec{b}-\myvec{b'}]
	+ \DP[u_{i+1},k',\myvec{b'}]
	+ \w(v u_{i+1}) \cdot
	w\left(b_{\var_w}^{(u_{i+1})},\myvec{b'}\right)
	.
	\end{eqnarray}
	Herein, $k$ and $\myvec{b'}$ are selected to satisfy $k' \in [k]_0$ and $\myvec{0}\le \myvec{b'}\le \myvec{b}$.

	Finally, we define~$\DP[v,\myvec{a},\myvec{b}] = \DP'[v,t,\myvec{a},\myvec{b}]$.
	
	Return \yes if there is a vector~$\myvec{b}$ such that~$\sum_{i=1}^{\var_w-1} b_i \le |X|$, and~$\DP[\rho,B,\myvec{b}]\ge D$ where~$\rho$ is the root of~$\Tree$.
	Otherwise, if no such vector~$\myvec{b}$ exists, return \no.

	\proofpara{Correctness and running time}
	The correctness and the running time can be proven analogously to the correctness and running time of the algorithm in Theorem~\ref{thm:GNAP-varc+varw}.
\end{proof}

Generally, we have to assume that~$B$ is exponential in the input size. However, there are special cases of~\GNAP, in which this is not the case, such as the case with unit costs for projects.
Since~$\var_w\le 2^{\wcode}$, we conclude the following from Theorem~\ref{thm:GNAP-B+varw}.
\begin{corollary}
	\label{cor:GNAP-C=1-XP-varw}
	\GNAP is \XP with respect to~$\var_w$ and~$\wcode$, if~$B$ is bounded polynomially in the input size.
\end{corollary}

\subsection{Generalized Noah's Ark Problem on Stars}
We now consider the special case of \GNAP where the given phylogenetic tree is a star.
We first show that this special case---and therefore~\GNAP in general---is \Wh1-hard with respect to the number of taxa,~$|X|$. This implies that the algorithm in Proposition~\ref{prop:GNAP-GNAP-X-XP} is tight in some sense.
Afterward, we prove that most of the \FPT and \XP algorithms that, in Section~\ref{sec:MCKP}, we presented for \MCKP can also be adopted for this special case of~\GNAP.

\subsubsection{Hardness}
\begin{theorem}
	\label{thm:GNAP-X-W1hard}
	\GNAP is \Wh1-hard with respect to~$|X|+\Delta$, even on ultrametric phylogenetic trees with~$\max_\w=\height_{\Tree}=1$, and~$D=1$, where $\Delta$ is the largest degree of a vertex in the phylogenetic tree.
\end{theorem}
\begin{proof}
	\proofpara{Reduction}
	We reduce from \MCKP, which by Theorem~\ref{thm:Pre-MCKP-m-W1} is \Wh1-hard with respect to the number of classes~$m$.
	Let~$\mathcal I=(N,\{N_1,\dots,N_m\},c,d,B,D)$ be an instance of \MCKP.
	We define an instance~$\mathcal I'=(\Tree,\mathcal P,B':=B,D':=1)$ of \GNAP in which the phylogenetic~$X$-tree~$\Tree=(V,E,\w)$ is a star with root~$\rho$ and the vertex set is~$V:=\{\rho\} \cup X$, with~$X:=\{x_1,\dots,x_m\}$. Set~$\w(e):=1$ for every~$e\in E$.
	For every class~$N_i=\{ a_{i,1}, \dots, a_{i,\ell_i} \}$, define projects~$p_{i,j} := (c_{i,j} := c(a_{i,j}), w_{i,j} := d(a_{i,j})/D )$ of a project list~$P_i$ for taxon~$x_i$. The~$|X|$-collection of projects~$\mathcal P$ contains all these project lists~$P_i$.
	
	\proofpara{Correctness}
	Because we may assume~$0\le d(a)\le D$ for all~$a\in N$, the \sprops~$w_{i,j}$ are in~$\mathbb{R}_{[0,1]}$ for all~$i\in [m]$ and $j\in [|N_i|]$.
	The tree has $m$ taxa and a maximum degree of~$m$. The reduction is clearly computable in polynomial time, so it only remains to show the equivalence.
	
	Let~$S$ be a solution for instance~\Instance and assume~$S\cap N_i=\{a_{i,j_i}\}$ without loss of generality.
	We show that~$S'=\{p_{i,j_i} \mid i\in [m]\}$ is a solution for~$\mathcal I'$:
	The cost of the set~$S'$ is~$\sum_{i=1}^m c_{i,j_i} = \sum_{i=1}^m c(a_{i,j_i}) \le B$ and further
	
	\begin{eqnarray*}
		PD_\Tree(S')
		&=& \sum_{(v,x_i)\in E} \w(v x_i) \cdot w_{i,j_i}\\
		&=& \sum_{(v,x_i)\in E} 1 \cdot d(a_{i,j})/D\\
		&=& \frac 1D \cdot \sum_{i=1}^m d(a_{i,j}) \ge 1 = D'.
	\end{eqnarray*}
	
	Conversely, let~$S=\{p_{1,i_1},\dots,p_{m,j_{m}}\}$ be a solution for instance~$\Instance'$.
	We show that~$S'=\{a_{1,i_1},\dots,a_{m,j_{m}}\}$ is a solution for \Instance.
	Clearly,~$S'$ contains exactly one item per class.
	The cost of the set~$S'$ is~$c_\Sigma(S') = \sum_{i=1}^m c(a_{i,j_i}) = \sum_{i=1}^m c_{i,j_i} \le B$.
	The value of~$S'$ is~$d_\Sigma(S') = \sum_{i=1}^m d(a_{i,j_i}) = \sum_{i=1}^m w_{i,j_i} \cdot D = PD_\Tree(S)\cdot D \ge D$.
\end{proof}

By Observation~\ref{obs:Pre-MCKP-NP-L=2}, \MCKP is \NP-hard, even if every class contains at most two items (of which one has no cost and no value). Because the above reduction is computed in polynomial time, we conclude the following.
\begin{corollary}
	\label{cor:GNAP-D=height=1,L=2,ultrametric}
	\NAP[0]{c_i}{b_i}{2} is \NP-hard, even on ultrametric phylogenetic trees with~$\height_{\Tree}=\max_\w=1$, and~$D=1$.
\end{corollary}

The~$X$-tree that has been constructed in the reduction in the proof of Theorem~\ref{thm:GNAP-X-W1hard}, is a star and therefore has a relatively high degree.
In the following, we show that~\GNAP is also \Wh{1}-hard with respect to~$|X|$ on binary phylogenetic trees.
\begin{corollary}
	\label{cor:GNAP-X+Delta}
	\GNAP is \Wh1-hard with respect to~$|X|+D+\height_\Tree$ even on binary phylogenetic trees with~$\max_\w=1$.
\end{corollary}
\begin{proof}
	\proofpara{Reduction}
	We reduce from~\GNAP, which by Theorem~\ref{thm:GNAP-X-W1hard} is \Wh{1}-hard with respect to~$|X|$, even if~$\max_\w=\height_{\Tree}=D=1$.
	Let~$\Instance=(\Tree,\mathcal P,B,D)$ be an instance of~\GNAP with $D=1$.
	Define a phylogenetic tree~$\Tree':=(V,E)$ as follows.
	Let~$V$ be the union of~$X$ and a set of vertices~$\{v_1, \dots, v_{|X|}, x^*\}$.
	Let the edges be~$E:= \{v_i x_i, v_i v_{i+1} \mid i\in [|X|-1]\} \cup \{ v_n x_n, v_n x^* \}$ and let every edge have a weight of~1.
	Define a project-list~$P_{x^*} = (p_{*,0} := (0,0),p_{*,1} := (1,1))$ for~$x^*$.
	Finally, let~$\Instance' := (\Tree',\mathcal P\cup P_{x^*},B':=B+1,D':=|X|+1)$ be an instance of~\GNAP.
	
	\proofpara{Correctness}
	The reduction can be computed in polynomial time.
	We show that~$S$ is a solution for instance~\Instance if and only if~$S':=S\cup\{ p_{*,1} \}$ is a solution for instance~$\Instance'$.
	Clearly,~$\Costs(S') = \sum_{p_{i,j} \in S'} c_{i,j} = c_{*,1} + \sum_{p_{i,j} \in S} c_{i,j} = \Costs(S)+1$.
	Because~$w_{*,1}=1$, we conclude that the \sprop at each vertex~$v_i$ is exactly~1. Thus, the value of~$PD_{\Tree'}(S')$ is
	\begin{eqnarray*}
	\sum_{i=1}^{|X|-1} \w(v_i v_{i+1}) \cdot 1
	+
	\w(v_{|X|} x^*) \cdot 1
	+
	\sum_{p_{i,j}\in S} \w(v_i x_i) \cdot w_{i,j} =
	|X| + PD_{\Tree}(S).
	\end{eqnarray*}
	Hence, the set~$S$ satisfies~$\Costs(S') \le B+1$ and~$PD_{\Tree'}(S')\ge |X|+1$ if and only if~$\Costs(S) \le B$ and~$PD_{\Tree}(S)\ge 1$.
	Therefore,~$S$ is a solution for~$\Instance$.
	We can assume that a solution of~$\Instance'$ contains~$p_{*,1}$ because otherwise, we can exchange one project with~$p_{*,1}$ to obtain a better solution.
\end{proof}

\subsubsection{Algorithmic results}
In Section~\ref{sec:MCKP}, we presented algorithms solving~\MCKP.
Many of these algorithms can be adopted for instances of \GNAP in which the phylogenetic tree~$\Tree$ is a star.

\newpage
\begin{proposition}
	\label{prop:GNAP-height=1->mckp}
	\GNAP can be solved
	\begin{propEnum}
		\item in~$\Oh(D\cdot 2^{\wcode} \cdot \numP + |\mathcal{I}|)$ time,
		\item in~$\Oh(B \cdot \numP + |\mathcal{I}|)$ time,
		\item in~$\Oh(C \cdot \numP \cdot |X| + |\mathcal{I}|)$ time, or
		\item in~$\Oh(|X|^{\var_c-1} \cdot \numP + |\mathcal{I}|)$ time,
	\end{propEnum}
	if the given phylogenetic tree is a star.
	Herein,~$\numP=\sum_{i=1}^{|X|} |P_i|$ is the number of projects and~$|\mathcal{I}|$ is the size of the input.
\end{proposition}
\begin{proof}
	To see the correctness of the statement, we reduce an instance of~\GNAP with a phylogenetic star tree to an instance~\MCKP and then use algorithms presented in Section~\ref{subsec:algos-MCKP}.
	
	\proofpara{Reduction}
	Let~$\mathcal{I}=(\Tree,\w,\mathcal{P},B,D)$ be an instance of~\GNAP with~$\height_{\Tree}=1$.
	We define an instance~$\mathcal{I}'=(N,\{N_1,\dots,N_{|X|}\},c,d,B',D')$ of~\MCKP.
	Without loss of generality, each \sprop is in the form~$w_i=w_i'/2^{\wcode}$ with~$w_i'\in [2^{\wcode}]_0$.\lb
	For every taxon~$x_i$ with project list~$P_i$, we define a class~$N_i$.
	We add an item~$a_{i,j}$\lb with cost~$c(a_{i,j}):=c_{i,j}$ and value~$d(a_{i,j}):=w_{i,j}'\cdot \w(\rho x_i)$ to~$N_i$ for every\lb project~$p_{i,j}=(c_{i,j},w_{i,j})\in P_i$.
	We set~$B':=B$ and~$D':=D\cdot 2^{\wcode}$.
	
	\proofpara{Correctness}
	Let~$S=\{p_{1,j_1},\dots,p_{|X|,j_{|X|}}\}$ be a solution for the instance~\Instance of~\GNAP.
	Define the set~$S'=\{a_{1,j_1},\dots,a_{|X|,j_{|X|}}\}$. Clearly,~$c_\Sigma(S')=\Costs(S)\le B$.
	Further,
	\begin{eqnarray*}
		\sum_{i=1}^{|X|} d(a_{i,j_i}) &=& \sum_{i=1}^{|X|} w_{i,j_i}'\cdot \w(\rho x_i)\\
		&=& 2^{\wcode}\cdot \sum_{i=1}^{|X|} w_{i,j_i}\cdot \w(\rho x_i)\\
		&=& 2^{\wcode}\cdot PD_\Tree(S) \ge 2^{\wcode}\cdot D = D'.
	\end{eqnarray*}
	Thus,~$S'$ is a solution for~$\Instance'$.
	Analogously, one can show that if~$S'$ is a solution for instance~$\Instance'$ of~\MCKP, then~$S$ is a solution for instance~\Instance of \GNAP.
	
	\proofpara{Running time}
	The instance~$\Instance'$ of \MCKP is computed in~$\Oh(|\Instance|)$ time.
	We observe that in~$\Instance'$ the size of~$N$ equals the number of projects~$\numP$, the number of classes~$m$ is the number of taxa~$|X|$, and the budget~$B$ remains unchanged.
	Because all costs are simply copied, the maximal cost~$C$ and the number of different costs~$\var_c$ remain the same.
	Because the \sprops are multiplied with an edge weight, it follows that~$\var_d \in \Oh(\var_w \cdot \max_\w)$.
	By definition,~$D' = D\cdot 2^{\wcode}$.
	
	Thus, after computing instance $\Instance'$, one can use any algorithm for solving \MCKP of which we saw some in Section~\ref{subsec:algos-MCKP} to compute an optimal solution in the stated time.
\end{proof}

By Proposition~\ref{prop:GNAP-height=1->mckp} and Theorem~\ref{thm:Pre-MCKP-ILPF}, we conclude that \GNAP is \FPT with respect to~$\var_c+\var_w+\max_\w$ when restricted to instances in which the phylogenetic tree is a star.
In the following, we present a better algorithm---a reduction from an instance of \GNAP in which the phylogenetic tree is a star to an instance of~\ILPF, in which the number of variables is bound in~$\Oh(2^{\var_c+\var_d} \cdot \var_c)$.
This reduction uses a technique that was used to show that~\KP is \FPT with respect to~$\var_c$~\cite{etscheid}, which was not necessary for proving Theorem~\ref{thm:Pre-MCKP-ILPF}.
\begin{theorem}
	\label{thm:GNAP-height=1:varc+varw}
	There is a reduction from instances of~\GNAP in which the phylogenetic tree is a star to instances of~\ILPF with~$\Oh(2^{\var_c+\var_d}\cdot \var_c)$ variables.
	Thus, \GNAP is \FPT with respect to~$\var_c+\var_w$ when restricting to instances in which the phylogenetic tree is a star.
\end{theorem}
\begin{proof}
	\proofpara{Description}
	Let~$\mathcal{I}=(\Tree,\w,\mathcal{P},B,D)$ be an instance of~\GNAP in which~$\Tree$ is a star and has root~$\rho$.
	
	We may assume that a project list~$P_i$ does not contain two projects of the same cost or the same value.
	In the following, we call~$T=(C,W)$ a \textit{type}, for sets~$C\subseteq \mathcal C$ and~$W\subseteq \mathcal W$ with~$|C| = |W|$, where~$\mathcal C$ and~$\mathcal W$ are the sets of different costs and~\sprops, respectively. Let~$\mathcal F$ be the family of all types.
	We say that the \textit{project list~$P_i$ is of type~$T=(C,W)$} if~$C$ and~$W$ are the set of costs and \sprops of~$P_i$.
	For each~$T\in\mathcal F$, we define~$m_T$ to be the number of classes of type~$T$.
	
	Observe, for each type~$T=(C,W)$, project list~$P_i$ of type~$T$, and a project~$p\in P_i$, we can determine the \sprop of~$p$ when we know the cost~$c$ of~$p$.
	More precisely, if~$c$ is the~$\ell$th cheapest cost in~$C$, then the \sprop of~$p$ is the~$\ell$th smallest \sprop in~$W$.
	For a type~$T=(C,W)$ and~$i\in [\var_c]$, we define the constant~$w_{T,i}$ to be~$-n\cdot \max_\w$ if~$c_i \not\in C$. Otherwise, let~$w_{T,i}\in \mathbb{R}_{[0,1]}$ be the~$\ell$th smallest \sprop in~$W$, if~$c_i$ is the~$\ell$th smallest cost in~$C$.
	
	For two taxa~$x_i$ and~$x_j$ with project lists~$P_i$ and~$P_j$ of the same type~$T$, it is possible that~$\w(\rho x_i) \ne \w(\rho x_j)$.
	Hence, it can make a difference if a project is selected for the taxon~$x_i$ instead of~$x_j$.
	For a type~$T$, let~$x_{T,1},\dots,x_{T,m_T}$ be the taxa, such that the project lists~$P_{T,1},\dots,P_{T,m_T}$ are of type~$T$ and~$\w(\rho x_{T,i}) \ge \w(\rho x_{T,i+1})$ for each~$i\in [m_T -1]$.
	For each type $T$, we define a function~$f_T: [m_T]_0 \to \mathbb{N}$ by~$f_T(0):=0$ and~$f_T(\ell)$ stores total value of the first~$\ell$ edges. More precisely, that is~$f_T(\ell) := \sum_{i=1}^{\ell} \w(\rho x_i)$.
	
	The following describes an instance of~\ILPF.
	\begin{align}
		\label{eqn:ILPF-B}
		\sum_{T \in \mathcal{F}} \sum_{i=1}^{\var_c} y_{T,i} \cdot c_i \le & \; B\\
		\label{eqn:ILPF-D}
		\sum_{T \in \mathcal{F}} \sum_{i=1}^{\var_c} w_{T,i} \cdot g_{T,i} \ge & \; D\\
		\label{eqn:ILPF-g}
		f_{T}\left(\sum_{\ell=i}^{\var_c} y_{T,\ell}\right) - f_{T}\left(\sum_{\ell=i+1}^{\var_c} y_{T,\ell}\right) = & \; g_{T,i} & \quad \forall T\in\mathcal{F}, i\in [\var_c]\\
		\label{eqn:ILPF-Nr}
		\sum_{i=1}^{\var_c} y_{T,i} = & \; m_{T} & \quad \forall T\in \mathcal{F}\\
		\label{eqn:ILPF-variables}
		y_{T,i}, g_{T,i} \ge & \; 0 & \quad \forall T\in\mathcal{F}, i\in [\var_c]
	\end{align}
	The variable~$y_{T,i}$ expresses the number of projects with cost~$c_i$ that are chosen in a project list of type~$T$.
	We want to assign the most valuable edges that are incident with taxa that have a project list of type~$T$ to the taxa in which the highest \sprop is chosen.
	To receive an overview, in~$g_{T,j}$ we store the total value of the~$y_{T,j}$ most valuable edges that are incident with a taxon that has a project list of type~$T$ for~$j\in [\var_c]$.
	
	For each type~$T$, the function~$f_T$ is not necessarily linear. However, for each~$f_T$ there are affine linear functions~$p_T^{(1)},\dots,p_T^{(m_T)}$ such that~$f_T(i) = \min_\ell p_T^{(\ell)}(i)$ for each~$i\in [m_T]$~\cite{etscheid}.
	
	\proofpara{Correctness}
	Observe that if~$c_i\not\in C$, then because we defined~$d_{T,i}$ to be~$-n\cdot\max_\w$, Inequality~(\ref{eqn:ILPF-D}) would not be satisfied if~$g_{T,i} > g_{T,i+1}$ and consequently $y_{T,i}=0$ if~$c_i\not\in C$ for each type~$T=(C,W)\in\mathcal F$ and~$i\in[\var_c]$.
	Inequality~(\ref{eqn:ILPF-B}) can only be satisfied if the total cost is at most~$B$.
	The correctness of Inequality~(\ref{eqn:ILPF-D}):
	The variable~$g_{T,i}$ stores is the total weight of the edges towards the~$y_{T,i}$ taxa with projects of project lists of type~$T$, in which a project of cost~$c_i$ is selected.
	All these projects a have \sprop of~$w_{T,i}$, and thus the phylogenetic diversity of these projects is~$w_{T,i} \cdot (g_{T,i} - g_{T,i+1})$.
	Consequently, Inequality~(\ref{eqn:ILPF-D}) can only be satisfied if the total phylogenetic diversity is at least~$D$.
	Equation~(\ref{eqn:ILPF-g}) ensures that the value of~$g_{T,i}$ is chosen correctly.
	Equation~(\ref{eqn:ILPF-Nr}) can only be correct if exactly~$m_T$ projects are picked from the project lists of type~$T$, for each~$T\in \mathcal F$.
	It remains to show that the instance of the~\ILPF has~$\Oh(2^{\var_c+\var_d}\cdot \var_c)$ variables.
	Because~$\mathcal F\subseteq 2^{\mathcal C} \times 2^{\mathcal W}$, the size of~$\mathcal F$ is~$\Oh(2^{\var_c+\var_d})$.
	We follow that there are~$\Oh(2^{\var_c+\var_d}\cdot \var_c)$ different options for the variables~$y_{T,i}$ and~$g_{T,i}$.
\end{proof}

\section{Restriction to Two Projects per Taxon}\label{sec:GNAP-two-projects}
We finally study two special cases of~\NAP[a_i]{c_i}{b_i}{2}---the special case of~\GNAP where every project list contains exactly two projects.
In this section, we write~$c(x_i)$ for the cost of the project with \sprop~1 in~$P_i$.

\subsection{Sure Survival or Extinction For Each Project}
First, we consider~\NAP[0]{c_i}{1}{2}, the special case where each taxon~$x_i$  survives definitely if~$c_i$ is paid and  becomes extinct, otherwise.
This special case was introduced by Pardi and Goldman~\cite{pardi07} under the name~\textsc{Budgeted NAP}.
They also presented a pseudopolynomial-time algorithm that computes a solution for an instance in~$\Oh(B^2 \cdot n)$~time.
Because we may assume that~$B\le C\cdot \bet X$, we conclude the following.
\begin{corollary}
	\label{cor:GNAP-C-01-NAP}
	\NAP[0]{c_i}{1}{2} can be solved in~$\Oh(C^2 \cdot n^3)$ time.
\end{corollary}

We observe that~\NAP[0]{c_i}{1}{2} is \FPT with respect to~$D$, with an adaption of the above-mentioned algorithm of Pardi and Goldman~\cite{pardi07} for the parameter~$B$.
\begin{proposition}
	\label{prop:GNAP-wcode=1-D}
	\NAP[0]{c_i}{1}{2} can be solved in~$\Oh(D^2 \cdot n)$ time.
\end{proposition}
\begin{proof}
	\proofpara{Table definition}
	For a set~$A$ of vertices, a vertex~$v$, and integers~$b$ and~$d$, we call a set of projects~$S$ an~\textit{$(A,v,d,b)$-respecting set}, if~$PD_{\Tree_v}(S)\ge d$, and~$S$ contains exactly one project of the projects lists of the offspring of~$A$, and~$S$ contains at least~$b$ projects with \sprop~1.
	
	We describe a dynamic programming algorithm with two tables~$\DP$ and~$\DP'$.
	We want entry~$\DP[v,d,b]$ for a vertex~$v\in V$, an integer~$d \in [D]_0$, and a boolean~$b\in \{0,1\}$ to store the minimal cost of a~$(\{v\},v,d,b)$-respecting set.
	If~$v$ is an internal vertex with children~$u_1,\dots,u_t$ and~$i\in [t]$,
	then we want entry~$\DP'[v,i,d,b]$ to store the minimal cost of an~$(\{u_1,\dots,u_i\},v,d,b)$-respecting~set.
	
	We use non-negative subtraction,~$\mo$, which for integers~$a$ and~$b$ is~$a \mo b = a - b$ if~$a \ge b$ and~$a \mo b = 0$, otherwise.
	
	\proofpara{Algorithm}
	As a base case, for a leaf~$x_i$ store~$\DP[x_i,0,1] = c(x_i)$ and for each~$d > 0$ store~$\DP[x_i,d,1] = \infty$.
	For any vertex~$v$ and each~$d\in [D]_0$ store~$\DP[v,d,0] = 0$ and~$\DP[v,i,d,0] = 0$.
	
	For an internal vertex~$v$, we define~$\DP'[v,1,d,1] = \DP[u_1,d \mo \w(v u_1),1]$.
	
	Now, let~$v$ be an internal vertex with children~$u_1,\dots,u_t$
	and we assume that for a fixed~$i\in [t]$ the values of~$\DP'[v,i,d,b]$ and~$\DP[u_{i+1},d,b]$ are known for each~$d\in [D]_0$ and~$b\in \{0,1\}$.
	To compute the value of~$\DP'[v,i+1,d]$, we use the recurrence
	\begin{eqnarray}
		\label{eqn:wcode=1-D}
		&&\DP'[v,i+1,d,1]\\
		& = & \min \{ \DP'[v,i,d,1]; \min_{d'\in [d_{i+1}]} \{ \DP'[v,i,d_{i+1}-d',1] + \DP[u_{i+1},d',1] \} \},
		 \nonumber
	\end{eqnarray}
	where~$d_{i+1} := d \mo \w(v u_{i+1})$.
	We store~$\DP[v,d,1]=\DP'[v,t,d,1]$, eventually.
	
	We return \yes if~$\DP[\rho,D] \le B$ for the root~$\rho$ of~$\Tree$.
	Otherwise, we return \no.
	
	\proofpara{Correctness}
	The base case, as well as the computation of~$\DP'[v,1,d,1]$, and the computation of~$\DP[v,d,1]$ for an internal vertex~$v$ and integers~$d\in [D]_0$ are correct by definition.
	It remains to show that~$\DP'[v,i+1,d,1]$ stores the correct value if~$\DP'[v,i,d,1]$ and~$\DP[u_{i+1},d,1]$ store the correct value.
	We first show that if~$S$ is an~$(\{u_1,\dots,u_{i+1}\},v,d,1)$-respecting set, then~$\DP'[v,i+1,d,1]\le \Costs(S)$.\lb
	Afterwards, we show that if~$c$ is stored in~$\DP'[v,i+1,d,1]$, then there is\lb an~$(\{u_1,\dots,u_{i+1}\},v,d,1)$-respecting set~$S$ with~$\Costs(S)=c$.
	
	Let~$S$ be an~$(\{u_1,\dots,u_{i+1}\},v,d,1)$-respecting set.
	Let~$S_1$ be the set of\lb projects that are in one of the project lists of the offspring of~$u_{i+1}$.
	We\lb define~$S_2 := S\setminus S_1$.
	Then $\DP'[v,i+1,d,1] = \DP'[v,i,d,1]$, if~$S_2$ only contains\lb projects with a \sprop of~0.
	Otherwise, let~$d' := PD_{\Tree_v}(S_1)$.\lb
	Then, $PD_{\Tree_v}(S_2) = PD_{\Tree_v}(S) - PD_{\Tree_v}(S_1) \ge d - d'$.
	We conclude that~$S_2$ is\lb an~$(\{u_1,\dots,u_{i}\},v,d-d',1)$-respecting set, and~$S_1$ is an~$(\{u_{i+1}\},v,d',1)$-respecting set.
	Consequently,
	\begin{eqnarray}
		\label{eqn:wcode=1-LA-1}
		\DP'[v,i+1,d,1] &\le& \DP'[v,i,d-d',1] + \DP[u_{i+1},d' \mo \w(v u_{i+1}),1]\\
		\label{eqn:wcode=1-LA-2}
		&\le& \Costs(S_G) + \Costs(S_F) = \Costs(S).
	\end{eqnarray}
	Here, Inequality~(\ref{eqn:wcode=1-LA-1}) follows from the recurrence in Recurrence~(\ref{eqn:wcode=1-D}) and Inequality~(\ref{eqn:wcode=1-LA-2}) follows from the induction hypothesis.
	
	Conversely,
	we assume that~$\DP'[v,i+1,d,1]$ stores~$c$.
	Unless~$\DP'[v,i+1,d,1]$ takes the value of~$\DP'[v,i,d,1]$ there is an integer~$d'\in [d_{i+1}]_0$, such\lb that~$\DP'[v,i+1,d,1] = \DP'[v,_{i+1},d_{i+1}-d',1] + \DP[u_{i+1},d',1]$.
	By the induction hypothesis, there is an~$(\{u_1,\dots,u_{i}\},v,d-d',1)$-respecting set~$S_2$ and an~$(\{u_{i+1}\},v,d',1)$-respecting set~$S_1$ such that~$\DP[u_{i+1},d',1]=\Costs(S_1)$ and~$\DP'[v,i,d-d',1]$\lb stores~$\Costs(S_2)$.
	We conclude that~$S:=S_1\cup S_2$ is an~$(\{u_1,\dots,u_{i+1}\},v,d,1)$-respecting set and~$c = \Costs(S_1)+\Costs(S_2)=\Costs(S)$.

	\proofpara{Running time}
	Table~$\DP$ has~$\Oh(D \cdot n)$ entries and each entry can be computed in constant time.
	Also,~$\DP'$ has~$\Oh(D \cdot n)$ entries.
	Entry~$\DP'[v,i+1,d]$ is computed by checking at most~$D+1$ options for~$d'$.
	Altogether, a solution can be found in~$\Oh(D^2 \cdot n)$~time.
\end{proof}

We may further use this pseudopolynomial-time algorithm to obtain an algorithm for the maximum edge weight~$\max_\w$:
For an instance~\Instance of~\NAP[0]{c_i}{1}{2} with~$\sum_{e\in E} \w(e) < D$, we can return \no since the desired diversity can never be reached.
Otherwise, we may assume~$D \le \sum_{e\in E} \w(e)\le \max_\w \cdot (n-1)$.
This gives the following running time-bound. 
\begin{corollary}
	\label{cor:GNAP-wcode=1,val-lambda}
	\NAP[0]{c_i}{1}{2} can be solved in~$\Oh((\max_\w)^2 \cdot n^3)$ time.
\end{corollary}

\subsection{Unit Costs For Each Project}
\label{subsec:unitc}
Next, we consider~\ucNAP---the special case of \GNAP in which every project with a positive \sprop has the same cost.
Observe, every instance~$\Instance=(\Tree,\mathcal{P},B,D)$ of~\NAP[0]{c}{b_i}{2} for each~$c\in\mathbb{N}$ can be reduced to an equivalent instance~$\mathcal{I}'=(\Tree,\mathcal{P}',B',D)$ of~\NAP[0]{1}{b_i}{2} by setting the costs of every project with a positive~\sprop to 1, and~$B'=\lfloor B/c \rfloor$.
Thus, the problem~\NAP[0]{1}{b_i}{2} can be considered as~\ucNAP.

In the remainder of this chapter, we use the term \emph{solution} to denote only those projects with a cost of~1 which have been chosen.
Further, with~$w(x_i)$ we denote the \sprop of the project in~$P_i$ which a costs of~1.

We show that even in very restricted instances, \ucNAP is \NP-hard by a reduction from \PS.
Recall that in \PS we are given a set of tuples~$T=\{ t_i=(a_i,b_i) \mid i\in [n], a_i\in \mathcal{Q}_{\ge 0},b_i\in \mathbb{R}_{(0,1)} \}$,
two integers~$k$, and~$Q$, and a number~$D\in \mathbb{Q}_+$.
It is asked whether there is a set of~$k$ tuples~$S\subseteq T$, such\lb that $\sum_{t_i\in S} a_i - Q\cdot \prod_{t_i\in S} b_i \ge D$.
\PS is \NP-hard by Theorem~\ref{thm:Pre-NP-PenSum}.

\begin{theorem}
	\label{thm:GNAP-ucNAP-hardness}
	\ucNAP is \NP-hard even on instances with a phylogenetic tree with a height of~2 and a degree of~1 in the root.
\end{theorem}
\begin{proof}
	\proofpara{Reduction}
	Let~$\Instance = (T,k,Q,D)$ be an instance of~\PS.
	Let~$bin(a)$ and~$bin(1-b)$ be the maximum binary encoding length of~$a_i$ and~$1-b_i$, respectively, and define~$t$ to be~$bin(a)+bin(1-b)$.
	We define an instance~$\Instance' = (\Tree,\mathcal{P},B,D')$ of~\ucNAP as follows.
	Let~$\Tree$ contain the vertices~$V := \{\rho,v,x_1,\dots,x_{|T|}\}$ and let the underlying undirected graph of~$\Tree$ be a star with center~$v$.
	Therefore,~$v$ is the only child of the root~$\rho$ and~$x_1,\dots,x_{|T|}$ are the leaves.
	We define~$\w(\rho v)=2^t Q$ and~$\w(v x_i) = 2^t (\nicefrac{a_i}{1-b_i})$ for each~$t_i\in T$.
	For each tuple~$t_i$, we define a project list~$P_i := ((0,0),(1,1-b_i))$.
	Then,~$\mathcal P$ is defined to be the set of these project lists.
	Finally, we set~$B:=k$, and~$D':=2^t (D+Q)$.
	The reduction can clearly be  computed in polynomial time.
	
	\proofpara{Correctness}
	We show that instance~\Instance is a \yes-instance of~\PS if and only if instance~$\Instance'$ is a \yes-instance of~\ucNAP.
	
	First, let~$S$ be a solution of an instance~\Instance of~\PS.
	We define a\lb set~$S' := \{ (1,w(x_i)) \mid t_i \in S \}$.
	Then,
	\begin{eqnarray*}
		&& \sum_{x_i\in S'} \w(v x_i) \cdot w(x_i) + \w(\rho v)\cdot \left(1-\prod_{x_i\in S'} (1- w(x_i))\right)\\
		&=& \sum_{t_i\in S'} 2^t\cdot(\nicefrac{a_i}{1-b_i})\cdot (1-b_i) + 2^t Q\cdot \left(1-\prod_{t_i\in S'} (1-(1-b_i))\right)\\
		&=& \sum_{t_i\in S} 2^t a_i - 2^t Q\cdot \prod_{t_i\in S} b_i + 2^t Q \ge 2^t \cdot (D + Q) = D'.
	\end{eqnarray*}
	Then, $S'$ is a solution for the instance~$\Instance'$, because~$|S'|=|S|\le B$.
	
	Consequently,
	let~$S'$ be a solution for the instance~$\Instance'$ of~\ucNAP.
	We conclude~$\sum_{x_i\in S'} \w(v x_i) \cdot w(x_i) \ge D' - \w(\rho v)\cdot \left(1-\prod_{x_i\in S'} (1- w(x_i))\right)$.
	Define~$S\subseteq T$ to contain a tuple~$t_i$ if and only if~$S'$ contains a project of the taxon~$x_i$.
	
	Further, observe that~$a_i = 2^{-t}\cdot \w(v x_i)\cdot (1-b_i)$.
	Then,
	\begin{eqnarray*}
		&&\sum_{t_i\in S} a_i - Q\cdot \prod_{t_i\in S} b_i\\
		&=& \sum_{x_i\in S'} 2^{-t}\cdot \w(v x_i)\cdot (1-b_i) - 2^{-t}\cdot \w(\rho v)\cdot \prod_{x_i\in S'} (1- w(x_i))\\
		&=& 2^{-t}\cdot \left(\sum_{x_i\in S'} \w(v x_i)\cdot w(x_i) - \w(\rho v)\cdot \prod_{x_i\in S'} (1- w(x_i))\right)\\
		&\ge& 2^{-t}\cdot \left(D' - \w(\rho v)\cdot \left(1-\prod_{x_i\in S'} (1- w(x_i))\right) - \w(\rho v)\cdot \prod_{x_i\in S'} (1- w(x_i))\right)\\
		&=& 2^{-t}\cdot \left(D' - \w(\rho v)\right) = 2^{-t} \cdot (2^t(D+Q) - 2^t Q) = D.
	\end{eqnarray*}
	Because~$|S'|=|S|\le B$, we conclude that~$S$ is a solution of instance \Instance.
\end{proof}

Recall that in an ultrametric tree, the weighted distance from the root to a vertex is the same for all vertices.
Observe that in an instance of \ucNAP in which the phylogenetic tree is ultrametric and has a height of at most~2, and the root has only one child one can select the set of taxa that have the highest \sprop.
Therefore, we can solve general instances of \ucNAP with an ultrametric phylogenetic tree that has a height of at most~2 with a dynamic programming algorithm which assigns each child of the root the number of taxa that are saved in their offspring.
We conclude the following.
\begin{observation}
	\label{obs:GNAP-ultrametric-greedy}
	\ucNAP can be solved in polynomial time on instances in which the phylogenetic tree is ultrametric and has a height of at most~2.
\end{observation}

In the following theorem, we show that, however,~\ucNAP is \NP-hard even when restricted to ultrametric phylogenetic trees with a height of~3.
\begin{theorem}
	\label{thm:GNAP-C=1,height=3,ultrametric}
	\ucNAP is \NP-hard even if the given phylogenetic tree is ultrametric and has a height of at most~3.
\end{theorem}
\begin{proof}
	By Theorem~\ref{thm:GNAP-ucNAP-hardness}, it suffices to reduce from \ucNAP with the restriction that the root has only one child and the height of the tree is~2.
	
	\proofpara{Reduction}
	Let~$\Instance = (\Tree,\mathcal{P},B,D)$ be an instance of~\ucNAP in which the root~$\rho$ of~$\Tree$ has only one child~$v$ and the height of~\Tree is~2.
	Without loss of generality, assume~$\w(v x_i)\ge \w(v x_{i+1})$ for each~$i\in [|X|-1]$ and there is a fixed~$s\in [|X|]$ with~$w(x_s)\ge w(x_j)$ for each~$j\in [|X|]$.
	Observe, that by the reduction in Theorem~\ref{thm:GNAP-ucNAP-hardness} we may assume~$w(x_j) \ne 1$ for each~$x_j\in X$.
	
	Consider Figure~\ref{fig:reduction} for an illustration of this reduction.
	
	We define an instance~$\Instance' := (\Tree',\mathcal{P},B',D')$ of~\ucNAP.
	Let~$X_1\subseteq X$ be the set of vertices~$x_i$ with~$\w(v x_1)=\w(v x_i)$.
	If~$X_1=X$, then~$\Instance$ is already ultrametric and simply output~$\Instance$.
	Otherwise, define~$X_2 := X\setminus X_1$.
	Fix an integer~$t\in\mathbb{N}$ that is large enough such that~$1-2^{-t}>w_{s,1}$ and~$2^t\cdot \w(v x_{|X|})>\w(v x_1)$.
	Define a tree~$\Tree'=(V',E',\w')$, in which~$V'$ contains the vertices~$V$ and add two vertices~$u_i$ and~$x_i^*$ for every~$x_i\in X_2$.
	Let~$X^*$ be the set of leaves~$x_i^*$.
	The set of edges is defined by~$E'=\{ \rho v \}\cup \{ v x_i \mid x_i\in X_1\} \cup \{ v u_i, u_i x_{i}, u_i x_i^* \mid x_i\in X_2 \}$. Observe that the leaf set of~$\Tree$ is~$X\cup X^*$.
	The weights of the edges are defined as:
	\begin{itemize}
		\item We set $\w'(\rho v) = \w(\rho v) \cdot 2^{t\cdot |X_2|} \cdot (2^t - 1)$.
		\item For~$x_i\in X_1$ we set~$\w'(v x_i) = (2^t - 1) \cdot \w(v x_1)$.
		\item For~$x_i\in X_2$ we set~$\w'(v u_i) = 2^t \cdot (\w(v x_1) - \w(v x_i))$\\
		and~$\w'(u_i x_{i}) = \w'(u_i x_i^*) = (2^t\cdot \w(v x_i)) - \w(v x_1)$.
	\end{itemize}
	
	Define project lists~$P_{i}^*:=((0,0),(1,1-2^{-t}))$ for each taxa~$x_i^*\in X^*$.
	Then, define~$\mathcal P' := \mathcal P \cup \{ P_i^* \mid x_i^*\in X^* \}$.
	Finally, we set~$B':=B+|X^*|$ and
	$$D'=(2^t-1)\cdot \left(D + |X_2|\cdot (2^t-1)\cdot \w(v x_1)+\left(2^{t|X_1|}-1\right)\cdot \w(w v)\right).$$
	
	\proofpara{Correctness}
	Observe that~$t$ can be chosen to be in~$\Oh(\wcode + \max_\w)$. Hence, the reduction can be computed in polynomial time.
	Moreover, the new phylogenetic tree has a height of~3.
	Before showing that \Instance is a \yes-instance if and only if $\Instance'$ is a \yes-instance, we first prove that~$\Tree'$ is ultrametric.
	%	
	%	In the rest of the proof we denote with~$\w_{w}$ the value~$\w(e)$, where~$e$ is the only edge that is directed towards~$w$.
	
	To show that~$\Tree'$ is ultrametric, for each~$x_i\in X$ and~$x_i^*\in X^*$ we show that the paths from~$v$ to~$x_i$ and from~$v$ to~$x_i^*$ have the same length, as the edge from~$v$ to~$x_1$. This is sufficient because every path from the root to a taxon visits~$v$.
	By definition, the claim is correct for every~$x_i\in X_1$.
	For an~$x_i\in X_2$, the path from~$v$ to~$x_i$ is
	\begin{eqnarray*}
		\w'(v u_i) + \w'(u_i x_i)
		&=& 2^t \cdot (\w(v x_1) - \w(v x_i)) + 2^t\cdot \w(v x_i) - \w(v x_1)\\
		&=& (2^t-1) \cdot \w(v x_1) = \w'(v x_1).
	\end{eqnarray*}
	By definition, this is also the length of the path from~$v$ to~$x_i^*\in X_i^*$.
	We conclude that~$\Tree'$ is an ultrametric tree.
	We now show that~\Instance is a \yes-instance of~\ucNAP if and only if~$\Instance'$ is a \yes-instance of~\ucNAP.

	Let~$S\subseteq X$ be a set of taxa.
	Define~$S' := S \cup X_i^*$.
	We show that~$S$ is a solution for~\Instance if and only if~$S'$ is a solution of~$\Instance'$.
	First,~$|S'|=|S|+|X^*|$ and hence~$|S|\le B$ if and only if~$|S'|\le B'=B+|X^*|$.
	We compute the phylogenetic diversity of~$S'$ in~$\Tree'$. Here, we first consider the diversity from the subtree that consists of~$v$, $u_i$, $x_i$, and~$x_i^*$ for both options, whether~$S'$ contains~$x_i$ or not.
	Afterward, we consider the additional value from the edge~$\rho v$.
	
	If~$S'$ contains~$x_i^*$ but not~$x_i$, the contribution is
	\begin{eqnarray*}
		&& \sum_{x_i\in X_2 \setminus S} (\w'(v u_i) + \w'(u_i x_i))\cdot w'(x_i^*)\\
		&=& \sum_{x_i\in X_2 \setminus S} (2^t \cdot (\w(v x_1) - \w(v x_i)) + 2^t\cdot \w(v x_i) - \w(v x_1))\cdot (1-2^{-t})\\
		&=& \sum_{x_i\in X_2 \setminus S} \w(v x_1)\cdot (2^t-1) \cdot (1-2^{-t}).
	\end{eqnarray*}

	For those vertices where~$S'$ contains~$x_i$ and~$x_i^*$, the contributed phylogenetic diversity is
	\begin{eqnarray*}
		&& \sum_{x_i\in X_2 \cap S} \w'(v u_i)\cdot (1-(1-w'(x_i^*))\cdot (1-w'(x_i)))\\
		&& +\sum_{x_i\in X_2 \cap S} \w'(u_i x_i)\cdot w'(x_i) + \sum_{x_i\in X_2 \cap S} \w'(u_i x_i^*)\cdot w'(x_i^*)\\
		&=& \sum_{x_i\in X_2 \cap S} 2^t \cdot (\w(v x_1) - \w(v x_i))\cdot (1-2^{-t}\cdot (1-w(x_i)))\\
		&& +\sum_{x_i\in X_2 \cap S} (2^t\cdot \w(v x_i) - \w(v x_1))\cdot (w(x_i)+1-2^{-t})\\
		&=& \sum_{x_i\in X_2 \cap S} (\w(v x_1) - \w(v x_i))\cdot (2^t - (1-w(x_i)))\\
		&& +\sum_{x_i\in X_2 \cap S} (2^t\cdot \w(v x_i) - \w(v x_1))\cdot (w(x_i)+1-2^{-t})\\
		&=& \sum_{x_i\in X_2 \cap S} \w(v x_1)\cdot (2^t - (1-w(x_i))-w(x_i)-1+2^{-t}))\\
		&& + \w(v x_i)\cdot (-2^t + (1-w(x_i))+2^t (w(x_i)+1-2^{-t}))\\
		&=& \sum_{x_i\in X_2 \cap S} \w(v x_1)\cdot (2^t-2+2^{-t}) + \w(v x_i)\cdot ((2^t-1) \cdot w(x_i))\\
		&=& (2^t-1)\cdot \left(\sum_{x_i\in X_2 \cap S} \w(v x_1)\cdot (1-2^{-t}) + \w(v x_i)\cdot w(x_i)\right).
	\end{eqnarray*}
	
	Finally, for the edge~$\rho v$ the contribution is
	\begin{eqnarray*}
		&& \w'(\rho v)\cdot \left( 1- \prod_{x_i^*\in X^*} (1-w'(x_i^*))\cdot \prod_{x_i\in S} (1-w'(x_i)) \right)\\
		&=& \w(\rho v) \cdot 2^{t\cdot |X_2|} \cdot (2^t - 1)\cdot \left( 1- 2^{-t\cdot|X_2|}\cdot \prod_{x_i\in S} (1-w(x_i))\right)\\
		&=& \w(\rho v) \cdot 2^{t\cdot |X_2|} \cdot (2^t - 1) - \w(\rho v) \cdot (2^t - 1)\cdot \prod_{x_i\in S} (1-w(x_i)).
	\end{eqnarray*}
	
	Altogether, we conclude
	\begin{eqnarray*}
		&& PD_{\Tree'}(S')\\
		&=& \w'(\rho v)\cdot \prod_{x_i^*\in X^*} (1-w'(x_i^*))\cdot \prod_{x_i\in S} (1-w'(x_i))\\
		&& +\sum_{x_i\in X_1 \cap S} \w'(v x_1)\cdot w'(x_i)\\
		&& +\sum_{x_i\in X_2 \cap S} \w'(v u_i)\cdot (1-(1-w'(x_i^*))\cdot (1-w'(x_i)))\\
		&& +\sum_{x_i\in X_2 \cap S} \w'(u_i x_i)\cdot w'(x_i) + \sum_{x_i\in X_2 \cap S} \w'(v x_1)\cdot (1-2^t)\\
		&& +\sum_{x_i\in X_2 \setminus S} (\w'(v u_i) + \w'(u_i x_i))\cdot w'(x_i^*)\\
		&=& \w(\rho v) \cdot 2^{t\cdot |X_2|} \cdot (2^t - 1) - \w(\rho v) \cdot (2^t - 1)\cdot \prod_{x_i\in S} (1-w(x_i))\\
		&& +\sum_{x_i\in X_1 \cap S} (2^t - 1) \cdot \w(v x_1)\cdot w(x_i)\\
		&& +(2^t-1)\cdot \left(\sum_{x_i\in X_2 \cap S} \w(v x_1)\cdot (1-2^{-t}) + \w(v x_i)\cdot w(x_i)\right)\\
		&& +\sum_{x_i\in X_2 \setminus S} \w(v x_1)\cdot (2^t-1) \cdot (1-2^{-t})\\
		&=& (2^t-1)\cdot \left[ \w(\rho v) \cdot 2^{t\cdot |X_2|} - \w(\rho v) \cdot \prod_{x_i\in S} (1-w(x_i))\right.\\
		&& \left.+\sum_{x_i\in S} \w(v x_i)\cdot w(x_i) + \sum_{x_i\in X_2} \w(v x_1)\cdot (1-2^{-t})\right]\\
		&=& (2^t-1)\cdot \left[ \w(\rho v) \cdot \left(1- \prod_{x_i\in S} (1-w(x_i))\right) + \sum_{x_i\in S} \w(v x_i)\cdot w(x_i) \right.\\
		&& \left.+\w(\rho v) \cdot (2^{t\cdot |X_2|}-1) + \sum_{x_i\in X_2} \w(v x_1)\cdot (1-2^{-t})\right]\\
		&=& (2^t-1)\cdot \left[ PD_\Tree(S) + (2^{t\cdot |X_2|}-1)\cdot \w(\rho v) + |X_2|\cdot (1-2^{-t})\cdot \w(v x_1)\right].
	\end{eqnarray*}
	It follows from this equation that~$PD_{\Tree'}(S')\le D'$ if and only if~$PD_{\Tree}(S)\le D$.
	
	\begin{figure}[t]
		\centering
		\begin{tikzpicture}[xscale=1.1,
			sibling distance=5em,
			every node/.style = {
				shape=circle,
				fill=white,
				rounded corners,
				minimum size=9,
				inner sep=0,
				draw,}
			]
			\tikzstyle{txt}=[circle,fill=white,draw=white,inner sep=0pt]
			
			\draw (-2,-1) -- (2,-1);
			
			\node (r) at (0,0) {};
			\node (v) at (0,-1) {};
			\node (x1) at (-2,-7) {};
			\node (x2) at (-1,-7) {};
			\node (x3) at (0,-5) {};
			\node (x4) at (1,-3) {};
			\node (x5) at (2,-2) {};
			
			\node[txt] at (-0.25,0.25) {$r$};
			\node[txt] at (-0.25,-0.75) {$v$};
			\node[txt] at (-2,-7.45) {$x_1$};
			\node[txt] at (-1,-7.45) {$x_2$};
			\node[txt] at (0,-5.45) {$x_3$};
			\node[txt] at (1,-3.45) {$x_4$};
			\node[txt] at (2,-2.45) {$x_5$};
			
			\node[txt] at (0.25,-0.5) {$1$};
			\node[txt] at (-2.25,-4) {$6$};
			\node[txt] at (-1.25,-4) {$6$};
			\node[txt] at (-0.25,-3) {$4$};
			\node[txt] at (0.75,-2) {$2$};
			\node[txt] at (1.75,-1.5) {$1$};
			
			\draw[-stealth] (r) -- (v);
			\draw[-stealth] (-2,-1) -- (x1);
			\draw[-stealth] (-1,-1) -- (x2);
			\draw[-stealth] (v) -- (x3);
			\draw[-stealth] (1,-1) -- (x4);
			\draw[-stealth] (2,-1) -- (x5);
		\end{tikzpicture}
		\hfill
		\begin{tikzpicture}[xscale=1.6,
			sibling distance=5em,
			every node/.style = {
				shape=circle,
				fill=white,
				rounded corners,
				minimum size=9,
				inner sep=0,
				draw}
			]
			\tikzstyle{txt}=[circle,fill=white,draw=white,inner sep=0pt]
			
			%lines across u3,u4,u5
			\draw (-0.25,-3.286) -- (0.25,-3.286);
			\draw (0.75,-5.571) -- (1.25,-5.571);
			\draw (1.75,-6.714) -- (2.25,-6.714);
			
			\node[txt] at (1.2,-0.5) {$2^9\cdot 7=3584$};
			\draw (-2,-1) -- (2,-1);
			
			\node[txt] at (1.5,-6.8) {$2$};
			
			\node[txt] at (-0.25,0.25) {$r$};
			\node[txt] at (-0.25,-0.75) {$v$};
			\node[txt] at (-2,-7.45) {$x_1$};
			\node[txt] at (-1,-7.45) {$x_2$};
			\node[txt] at (-0.25,-7.45) {$x_3$};
			\node[txt] at (0.75,-7.45) {$x_4$};
			\node[txt] at (1.75,-7.45) {$x_5$};
			\node[txt] at (0.25,-7.45) {$x_3^*$};
			\node[txt] at (1.25,-7.45) {$x_4^*$};
			\node[txt] at (2.25,-7.45) {$x_5^*$};
			\node[txt] at (-0.25,-3.036) {$u_3$};
			\node[txt] at (0.75,-5.321) {$u_4$};
			\node[txt] at (1.75,-6.464) {$u_5$};
			
			\node (r) at (0,0) {};
			\node (v) at (0,-1) {};
			\node (u3) at (0,-3.286) {};
			\node (u4) at (1,-5.571) {};
			\node (u5) at (2,-6.714) {};
			\node (x1) at (-2,-7) {};
			\node (x2) at (-1,-7) {};
			\node (x3) at (-0.25,-7) {};
			\node (x4) at (0.75,-7) {};
			\node (x5) at (1.75,-7) {};
			\node (x3*) at (0.25,-7) {};
			\node (x4*) at (1.25,-7) {};
			\node (x5*) at (2.25,-7) {};
			
			\node[txt] at (-2.25,-4) {$42$};
			\node[txt] at (-1.25,-4) {$42$};
			\node[txt] at (-0.25,-2.1) {$16$};
			\node[txt] at (0.75,-3.3) {$32$};
			\node[txt] at (1.75,-3.9) {$40$};
			
			\draw[dashed,->] (r) -- (v);
			
			\draw[-stealth] (-2,-1) -- (x1);
			\draw[-stealth] (-1,-1) -- (x2);
			
			\draw[-stealth] (v) -- (u3);
			\draw[-stealth] (1,-1) -- (u4);
			\draw[-stealth] (2,-1) -- (u5);
			
			\draw[-stealth] (-0.25,-3.286) -- (x3);
			\draw[-stealth] (0.25,-3.286) -- (x3*);
			\draw[-stealth] (0.75,-5.571) -- (x4);
			\draw[-stealth] (1.25,-5.571) -- (x4*);
			\draw[-stealth] (1.75,-6.714) -- (x5);
			\draw[-stealth] (2.25,-6.714) -- (x5*);
			\node[txt] at (0,-5.1) {$26$};
			\node[txt] at (1,-6.3) {$10$};
		\end{tikzpicture}
		\caption{This figure shows an example of the reduction presented in Theorem~\ref{thm:GNAP-C=1,height=3,ultrametric}, where on the left side the tree of an example-instance~\Instance and on the right side the tree of instance~$\Instance'$ is depicted. Here, the~\sprops are omitted and we used~$t=3$.}
		\label{fig:reduction}
	\end{figure}
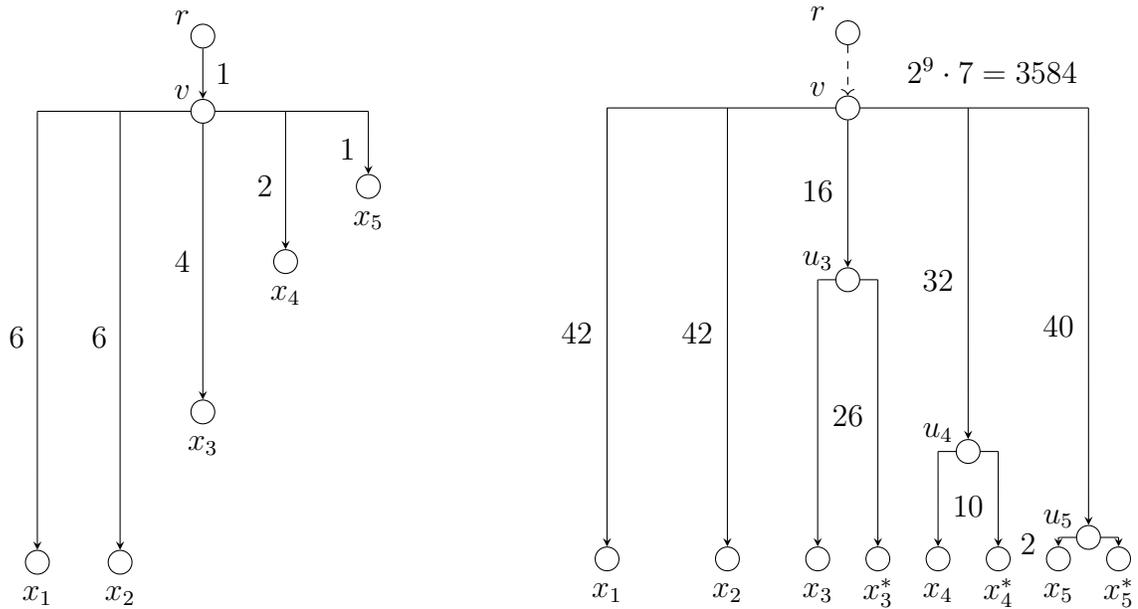
	
	Conversely,
	we show that~$\Instance'$ has a solution~$S'$ with~$X^*\subseteq S'$. This then implies that~$S'\setminus X^*$ is a solution for~\Instance by the argument we saw beforehand. 
	
	Let~$S'$ be a solution for~$\Instance'$ that contains a maximum number of elements of~$X^*$ among all solutions.
	If~$X^*\subseteq S'$, then we are done.
	Assume otherwise and choose some~$x_i^* \not\in S'$. We show that there is a solution containing~$x_i^*$ and all elements of~$S'\cap X^*$, contradicting the choice of~$S'$.
	If~$x_i$ is in~$S'$, then we consider the\lb set~$S_1 := (S'\setminus\{x_i\}) \cup \{x_i^*\}$. Now,~$|S_1|=|S'|\le B'$
	and because we defined that~$w(x_i^*)=1-2^{-t}>w(x_i)$, we conclude that~$PD_{\Tree'}(S_1)>PD_{\Tree'}(S')\ge D'$.\lb
	If~$S'\subsetneq X^*$, then we consider the set~$S_2 := X^*$. Now,~$|S_2|=|X^*|\le B'$\lb and~$PD_{\Tree'}(S_2)\ge PD_{\Tree'}(S')\ge D'$.
	Finally, assume that~$x_i,x_i^*\not\in S'$ and~$x_j\in S'$ for some~$j\ne i$.
	Again,~$w(x_i^*)=1-2^{-t}>w(x_i)$ and we know that the length of the path from~$v$ to~$x_i^*$ and~$x_j$ is the same. Consider the set~$S_3 := (S'\setminus\{x_j\}) \cup \{x_i^*\}$ and observe~$|S_3|=|S'|\le B'$. Moreover, by the above,~$PD_{\Tree'}(S_3)>PD_{\Tree'}(S')\ge D'$.
	This completes the proof.
\end{proof}

\section{Discussion}
In this chapter, we considered \GNAP and three special cases of \GNAP---namely, \NAP[0]{c_i}{1}{2}, \ucNAP, and the case that the phylogenetic tree is a star---and we provided several tractability and intractability results.
Our most important results are as follows.
\GNAP is \Wh{1}-hard with respect to the number of taxa but can be solved with an~\XP-algorithm when parameterized by the number of unique costs plus \sprops,~$\var_c+\var_w$.
We presented some pseudo-polynomial running-time algorithms for~\NAP[a_i]{c_i}{b_i}{2} and \GNAP restricted to trees is a star.
We finally showed that \ucNAP is \NP-hard even in very restricted cases.

Naturally, several open questions remain.
We want to list some of them.

Most notably, we do not known whether~\GNAP is weakly or strongly \NP-hard.
We therefore can not exclude a pseudo-polynomial running-time algorithm for \GNAP yet.
While we showed that \PS can be solved in pseudo-polynomial running-time by Proposition~\ref{prop:Pre-PS-pseudo}, it remains open whether such an algorithm exists even for \ucNAP.

Moreover, it remains open whether the result of Theorem~\ref{thm:GNAP-varc+varw} can be improved so that \GNAP is \FPT with respect to~$\var_c+\var_w$.

In Observation~\ref{obs:GNAP-D=height=1,varw=2}, we presented an easy reduction to show that \GNAP---or more precisely even the special case~\NAP[0]{c_i}{b}{2}---is already \NP-hard when~$D=1$.
In this reduction, we, however, increase the encoding-length of each \sprop.
We therefore wonder if~\NAP[0]{c_i}{b}{2} or even~\GNAP are \FPT when parameterized with~$D$ plus the binary encoding size of the \sprops.
In Proposition~\ref{prop:GNAP-height=1->mckp}, we showed that this holds at least in the case that the phylogenetic tree has a size of~1.

\chapter{Phylogenetic Diversity with Extinction Times}
\label{ctr:TimePD}

\section{Introduction}
In this chapter, we consider an extension of the classic \MPDlong (\MPD) problem, in which species~(taxa) have differing \emph{extinction times}, after which they will die out if they have not been saved.
In~\MPD, we are given a phylogenetic tree $\Tree$ and integers $k$, and~$D$ as input, and we are asked if there exists a subset of $k$ taxa whose phylogenetic diversity is at least $D$.
\MPD is polynomial-time solvable by a greedy algorithm~\cite{steel,Pardi2005}.
However, already the generalization \BNAP of \MPD, in which each taxon has an associated integer cost which would be necessary to be paid to save the taxon, is \NP-hard~\cite{pardi07}.

In this extension of \MPD, the cost of a taxon may, just as well as financial or space capacities, be used to represent the fact that different taxa may take a different amount of time to save from extinction.
However, to the best of our knowledge, previously studied versions of \MPD do not take into account that different taxa may have different amounts of \emph{remaining} time before extinction.
Thus, to ensure that a set of taxa can be saved with the available resources, it is not enough to guarantee that their total cost is below a certain threshold.
One also needs to ensure that there is a schedule under which each taxon is saved before its moment of extinction.

We take the first step in addressing this issue by introducing two extensions of~\BNAP, denoted \tPDslong(\tPDs) and \tPDwslong (\tPDws), in which each taxon has an associated \emph{rescue length} (the amount of time it takes to save the taxon) and also an \emph{extinction time} (the time after which the taxon can not be saved anymore).
In each problem, there is a set of available teams that can work towards saving the taxa; under \tPDs different teams may collaborate on saving a taxon while in \tPDws they may not.

These problems have much in common with machine scheduling problems, insofar as we may think of the taxa as corresponding to jobs with a certain due date and the teams as corresponding to machines. One may think of \tPDs and \tPDws as machine scheduling problems, in which the objective to be maximized is the phylogenetic diversity (as determined by the input tree $\Tree$) of the set of completed tasks (which are the saved taxa).

\paragraph*{Related Work in Scheduling.}
Scheduling problems are denoted by a three-field notation introduced by Graham~et~al.~\cite{graham1979}.
Herein, many problems are written as a triple $\alpha|\beta|\gamma$, where $\alpha$ is the machine environment, $\beta$ are job characteristics and scheduling constraints, and $\gamma$ is the objective function.
% We only need a fragment of this definition, e.g. we only use $\beta=\emptyset$.
For a more detailed view, we refer to~\cite{graham1979,mnich}.

The scheduling problem most closely related to the problems studied in this chapter is $Pt||\sum w_j (1-U_j)$.
In this problem, we are given a set of jobs, each with an integer weight, a processing time, and a due date.
We are also given $t$ identical machines that can each process one job at a time.
The task is to schedule jobs on the available machines in such a way that we maximize the total weight of jobs completed before the due date.
(Here $w_j$ denotes the weight of job $j$, and $U_j = 1$ if job $j$ is not completed on time, and $U_j =0$, otherwise).
This is similar to \tPDws, in the case that the $t$ available teams have identical starting and ending times. We may think of taxa as analogous to jobs, with the extinction times corresponding to due dates.
The key difference is that in \tPDws, rather than maximizing the total weight of the completed jobs, we aim to maximize their phylogenetic diversity, as determined by the input tree $\Tree$. 

% Scheduling $Pt||\sum w_j (1-U_j)$ is the special case of \tPDws in which the phylogenetic tree is a star and the $t$ teams operate the entire time to maximize the weighted jobs that are completed on time.

Although approximation algorithms for scheduling problems are common, parameterized algorithms for scheduling problems are rare~\cite{mnich}.
The most commonly investigated special case of \tPDws is $1||\sum w_j (1-U_j)$---that\linebreak is~$Pt||\sum w_j (1-U_j)$ with only one machine.
% The most commonly investigated special case of \tPDws is $1||\sum w_j (1-U_j)$, in which one machine is given.
This problem is weakly \NP-hard and is solvable in pseudo-polynomial time~\cite{lawler1969,sahni1976},
while the unweighted version~$1||\sum (1-U_j)$ is solvable in polynomial time~\cite{moore1968,maxwell1970,sidney1973}.
Parameters studied for~$1||\sum w_j (1-U_j)$ include the number of different due dates, the number of different processing times, and the number of different weights.
The problem  is \FPT when parameterized by the sum of any two of these parameters~\cite{hermelin}.
When parameterized by one of the latter two parameters, it is \Wh{1}-hard~\cite{heeger2024} but \XP~\cite{hermelin}.
Also a version has been studied in the light of parameterized algorithms in which there are also few distinct deadlines~\cite{heeger2023}.

\paragraph*{Our Contribution.}
With \tPDs and \tPDws we introduce problems in maximizing phylogenetic diversity in which extinction times of taxa may vary.
We further provide a connection to the well-regarded field of scheduling.

Both problems turn out to be \NP-hard; this result is perhaps unsurprising given their close relation to scheduling problems, and we therefore analyze \tPDs and \tPDws within the framework of parameterized complexity.
Our most important results are $\Oh^*(2^{2.443\cdot D + o(D)})$-time algorithms for \tPDs and \tPDws (Section~\ref{sec:timePD-D}),
and a $\Oh^*(2^{6.056\cdot \Dbar + o(\Dbar)})$-time algorithm for \tPDs (Section~\ref{sec:timePD-Dbar}).
Moreover, both problems are \FPT with respect to the available person-hours $H_{\max_{\ex}}$ (Proposition~\ref{prop:timePD-T+maxr}).
A detailed list of known results for \tPDs and \tPDws is given in Table~\ref{tab:results}.
In the table and the rest of the chapter we use the convention~$n := |X|$.

A key challenge for the design of \FPT algorithms for \tPDs and \tPDws is the combination of the tree structure and extinction time constraints. 
Tree-based problems often suit a dynamic programming (DP) approach, where partial solutions are constructed for subtrees of the input tree, (starting with the individual leaves and proceeding to larger subtrees).
Scheduling problems with due dates (such as our extinction times) may also suit a DP approach, where partial solutions are constructed for subsets of tasks with due dates below some bound.
These two DP approaches have a conflicting structure for the desired processing order, which makes designing a DP algorithm for \tPDs and \tPDws more difficult than it may first appear.
Our solution involves the careful use of color coding to reconcile the two approaches.
We believe that this approach may also be applicable to other extensions of \MPD.

\begin{table}[t]
	\centering
	\footnotesize
	\caption{Parameterized complexity results for \tPDs and \tPDws.}
	\label{tab:results}
	\resizebox{\columnwidth}{!}{%
		\myrowcols
		\begin{tabular}{lllll}
			\hline
			Parameter & \multicolumn{2}{c}{\tPDs} & \multicolumn{2}{c}{\tPDws}\\
			\hline
			Diversity $D$ & \FPT $\Oh^*(2^{2.443\cdot D + o(D)})$ & Thm.~\ref{thm:timePD-D} & \FPT $\Oh^*(2^{2.443\cdot D + o(D)})$ & Thm.~\ref{thm:timePD-D}\\
			Diversity loss $\Dbar$ \; & \FPT $\Oh^*(2^{6.056\cdot \Dbar + o(\Dbar)})$ & Thm.~\ref{thm:timePD-DBar} & \NP-hard even if $\Dbar = 0$ & \cite{GSJ1976}\\
			\# Taxa $n$ & \FPT $\Oh^*(2^{n})$  & Prop.~\ref{prop:timePD-X}\ref{prop:timePD-X-s} & \FPT $\Oh^*(n!)$  & Prop.~\ref{prop:timePD-X}\ref{prop:timePD-X-ws}\\
			\# Teams $|T|$ & \NP-hard even if $|T| = 1$  & Prop.~\ref{prop:timePD-Knapsack}\ref{thm:timePD-KP-NP} \; & \NP-hard even if $|T| = 1$  & Prop.~\ref{prop:timePD-Knapsack}\ref{thm:timePD-KP-NP}\\
			Solution size $k$ & \Wh{1}-hard & Prop.~\ref{prop:timePD-Knapsack}\ref{thm:timePD-KP-W1} & \Wh{1}-hard & Prop.~\ref{prop:timePD-Knapsack}\ref{thm:timePD-KP-W1}\\
			\hline
			Max time $\max_{\ex}$ \; & \textit{open} & & \textit{open} & \\
			Unique times $\var_{\ex}$ \; & \NP-hard even if $\var_{\ex}=1$ & Prop.~\ref{prop:timePD-Knapsack}\ref{thm:timePD-KP-NP} & \NP-hard even if $\var_{\ex}=1$ & Prop.~\ref{prop:timePD-Knapsack}\ref{thm:timePD-KP-NP}\\
			Unique lengths $\var_\ell$ \; & \textit{open} & & \Wh{1}-hard & \cite{heeger2024}\\
			%	$\max_\ell$ & poly for 1 & poly for 1 & \\
			%	$\var_\ell$ &  &  & \\
			%	solution size $|S|$ & XP & XP & \\
			%	& $\Oh^*(n^{|S|})$ & $\Oh^*((n\cdot t)^{|S|})$ & Brute Force\\
			%	max weight $\max_w$ & P\FPT on stars & \NP-hard even if 1 \cite{garey1978} & Reduction from Parity\\
			\hline
			Person-hours $H_{\max}$ & \FPT $\Oh^*((|T|+1)^{2\max_{\ex}})$ & Prop.~\ref{prop:timePD-T+maxr}\ref{prop:timePD-s-T^maxr} & \FPT $\Oh^*(3^{|T|+\max_{\ex}})$ & Prop.~\ref{prop:timePD-T+maxr}\ref{prop:timePD-ws-T+maxr}\\
			& \FPT $\Oh^*((H_{\max_{\ex}})^{2\var_{\ex}})$ & Prop.~\ref{prop:timePD-T+maxr}\ref{prop:timePD-s-T*maxr} & & \\
			$\var_\ell+\var_{\ex}$ & \XP $\Oh(n^{2\var_\ell\cdot\var_{\ex}+1})$ & Prop.~\ref{prop:timePD-varl+varr} & \textit{open} & \\
			\hline
			%$D+\var_{\ex}$ &  &  & \\
			%$\Dbar+\var_{\ex}$ &  &  & \\
		\end{tabular}
	}
	
\end{table}

\paragraph{Structure of the Chapter.}
In Section~\ref{sec:timePD-prelim}, we formally define the two problems \tPDs and \tPDws and prove some simple initial results.
In Section~\ref{sec:timePD-D}, we introduce an \FPT algorithm for \tPDs and \tPDws parameterized by the target $D$, and in Section~\ref{sec:timePD-Dbar}, we give an \FPT algorithm for \tPDs parameterized by the acceptable loss of diversity~$\Dbar$.
In Section~\ref{sec:timePD-other}, we prove a number of other parameterized algorithms.
In Section~\ref{sec:timePD-discussion}, we provide a brief outlook on future research ideas.

\section{Preliminaries}
\label{sec:timePD-prelim}
In this section, we present the formal definition of the problems, and the parameterization.
We further start with some preliminary observations.

\subsection{Definitions}

\paragraph*{Problem Definitions and Parameterizations.}
In the following,
we are given a set of taxa $X$ and a phylogenetic~$X$-tree $\Tree$, which will be used to calculate the phylogenetic diversity of any subset of taxa $A \subseteq X$.
In addition, we are given an integer \emph{extinction time $\ex(x)$} for each taxon~$x \in X$ representing the amount of time remaining to save that taxon before it goes extinct, and an integer \emph{rescue length} $\ell(x)$ representing how much time needs to be spent in order to save that taxon. 
Thus, we need to spend $\ell(x)$ units of time before $\ex(x)$ units of time have elapsed, if we wish to save $x$ from extinction.

In addition, we are given a set of teams  $T = \{t_1,\dots, t_{|T|}\}$ that are available to work on saving taxa, where each team $t_i$ is represented by a pair of integers $(s_i,e_i)$ with $0 \leq s_i < e_i$.
Here $s_i$ and $e_i$ represent the starting time and the ending time of team $t_i$, respectively.
Thus, team $t_i$ is available for $(e_i-s_i)$~units of time, starting after $s_i$ time steps.
For convenience, we refer to the units of time in this paper as \emph{person-hours}, although of course in practice the units of time may be days or weeks.

We define $\mathcal{H}_T : = \{(i,j) \in \mathbb{N}^2 \mid t_i \in T, s_i < j \leq e_i\}$.
That is, $\mathcal{H}_T$ is the set of all pairs $(i,j)$ where team $t_i$ is available to work in \emph{timeslot} $j$.
For $j \in [\var_{\ex}]$, we define~$H_j := |\{(i,j')\in \mathcal{H}_T : j' \le \ex_j\}|$. That is, $H_j$ is the number of person-hours available until time $\ex_j$. 

Then, a \emph{$T$-schedule} is a function $f:\mathcal{H}_T\rightarrow A\cup\{\textsc{none}\}$ for a set of taxa~$A \subseteq X$.
That is, a~$T$-schedule is a mapping from the available timeslots for each team to the taxa in $A$ (or $\textsc{none}$).
Intuitively, $f$ shows which teams will work to save which taxa in $A$, and at which times.
See Figure~\ref{fig:timePD-example-scheduling} for an example.
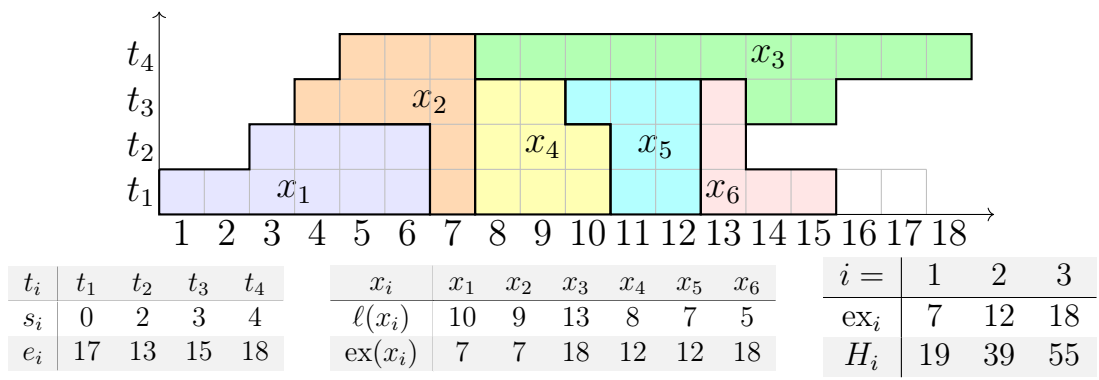
\begin{figure}[t]
	\centering
	\begin{tikzpicture}[scale=0.6,every node/.style={scale=1.2}]
		\tikzstyle{txt}=[circle,fill=none,draw=none,inner sep=0pt]
		
		\fill[blue!10] (0,0) rectangle (6,1);
		\fill[blue!10] (2,1) rectangle (6,2);
		\node[txt] at (3,0.5) {$x_1$};
		
		\fill[orange!30] (3,2) rectangle (8,3);
		\fill[orange!30] (4,3) rectangle (7,4);
		\fill[orange!30] (6,0) rectangle (8,2);
		\node[txt] at (6,2.5) {$x_2$};
		
		\fill[green!30] (7,4) rectangle (18,3);
		\fill[green!30] (13,2) rectangle (15,3);
		\node[txt] at (13.5,3.5) {$x_3$};
		
		\fill[yellow!30] (7,0) rectangle (9,3);
		\fill[yellow!30] (9,0) rectangle (10,2);
		\node[txt] at (8.5,1.5) {$x_4$};
		
		\fill[cyan!30] (9,2) rectangle (12,3);
		\fill[cyan!30] (10,0) rectangle (12,2);
		\node[txt] at (11,1.5) {$x_5$};
		
		\fill[red!10] (12,0) rectangle (15,1);
		\fill[red!10] (12,1) rectangle (13,3);
		\node[txt] at (12.5,0.5) {$x_6$};

		\draw[gray!50] (0,1) -- (17,1);
		\draw[gray!50] (2,2) -- (15,2);
		\draw[gray!50] (3,3) -- (18,3);
		\draw[gray!50] (4,4) -- (18,4);
		
		\foreach \i in {1,...,3}
		\draw[gray!50] (\i,0) -- (\i,\i);
		
		\foreach \i in {4,...,13}
		\draw[gray!50] (\i,0) -- (\i,4);
		
		\foreach \i in {14,...,17}
		\draw[gray!50] (\i,0) -- (\i,1);
		
		\foreach \i in {14,15}
		\draw[gray!50] (\i,2) -- (\i,3);
		
		\foreach \i in {14,...,18}
		\draw[gray!50] (\i,3) -- (\i,4);
		
		\foreach \i in {1,...,18}
		\node[txt] at (\i-.5,-.4) {\i};
		
		\foreach \i in {1,...,4}
		\node[txt] at (-.4,\i-.5) {$t_\i$};
		
		\draw[->] (0,0) -- (0,4.5);
		\draw[->] (0,0) -- (18.5,0);
		
		\draw[thick] (0,1) -- (2,1) -- (2,2) -- (6,2) -- (6,0) -- (0,0) -- (0,1);
		\draw[thick] (3,2) -- (3,3) -- (4,3) -- (4,4) -- (7,4) -- (7,0) -- (6,0) -- (6,2) -- (3,2);
		\draw[thick] (7,3) -- (7,4) -- (18,4) -- (18,3) -- (15,3) -- (15,2) -- (13,2) -- (13,3) -- (7,3);
		\draw[thick] (7,3) -- (7,0) -- (10,0) -- (10,2) -- (9,2) -- (9,3) -- (7,3);
		\draw[thick] (10,0) -- (10,2) -- (9,2) -- (9,3) -- (12,3) -- (12,0) -- (10,0);
		\draw[thick] (12,3) -- (12,0) -- (15,0) -- (15,1) -- (13,1) -- (13,3) -- (12,3);
	\end{tikzpicture}
	\resizebox{.25\columnwidth}{!}{%
		\myrowcols
		\begin{tabular}{c|cccc}
			$t_i$ & $t_1$ & $t_2$ & $t_3$ & $t_4$\\
			\hline
			$s_i$ & 0 & 2 & 3 & 4\\
			$e_i$ & 17 & 13 & 15 & 18
		\end{tabular}
	}
	\;\;
	\resizebox{.4\columnwidth}{!}{%
		\myrowcols
		\begin{tabular}{c|cccccc}
			$x_i$ & $x_1$ & $x_2$ & $x_3$ & $x_4$ & $x_5$ & $x_6$\\
			\hline
			$\ell(x_i)$ & 10 & 9 & 13 & 8 & 7 & 5\\
			$\ex(x_i)$ & 7 & 7 & 18 & 12 & 12 & 18
		\end{tabular}
	}
	\;\;
	\resizebox{.25\columnwidth}{!}{%
		\myrowcols
		\begin{tabular}{c|ccc}
			$i=$ & $1$ & $2$ & $3$\\
			\hline
			$\ex_i$ & 7 & 12 & 18\\
			$H_i$ & 19 & 39 & 55
		\end{tabular}
	}
	\caption{This is a hypothetical valid schedule saving the set of taxa $\{x_1,\dots,x_6\}$.
		Each square marks a tuple $(i,j) \in \mathcal{H}_T$.
	}
	\label{fig:timePD-example-scheduling}
\end{figure}%

We say that a~$T$-schedule~$f$ is \emph{valid} if $\ex(x) \geq j$ for each $x \in A$, and $(i,j) \in \mathcal{H}$ with $f((i,j)) = x$.
That is, $f$ does not assign a team to work on taxon $x$ after its extinction time.
We say that $f$ \emph{saves $A$} if $f$ is valid and $|f^{-1}(x)| \geq \ell(x)$ for all~$x \in A$.
That is, a schedule $f$ saves $A \subseteq X$ if every taxon $x \in A$ has at least $\ell(x)$ person-hours assigned to $x$ by its extinction time $\ex(x)$.

The definition of a valid schedule allows for several teams to be assigned to the same taxon at the same time, so that for example the task of preserving a taxon can be done in half the time by using twice the number of teams. Whether this is realistic or not depends on the nature of the tasks involved in preservation. For instance the task of preparing a new enclosure for animals might be completed faster by several teams working in parallel, whilst the task of rearing infants to adulthood cannot be sped up in the same way. 
Due to these concerns, we will consider two variations of the problem, differentiating in whether a schedule must be \emph{strict}.

A $T$-schedule $f$ saving $A$ is \emph{strict} if $|\{ i \in [|T|] : f((i,j)) = x \}| = 1$ for each~$x \in A$.\lb
That is, there is only one team $t_i$ assigned to save each taxon $x$.
%, and team $t_i$ spends continuous $\ell(x)$ person-hours on $x$.
We state without a proof that we may assume $f^{-1}(x) = \{(i,j+1), (i,j+2), \dots, (i,j+\ell(x)\}$ for some~$j \ge s_i$.
That is, once started, team $t_i$ continuously works on $x$.
We note that in a non-strict schedule $f$, it is even possible that multiple teams may work on the same taxa~$x$ at once.

\newpage
Formally, the problems we regard in this chapter are as follows.
\problemdef{\tPDslong (\tPDs)}{
	A phylogenetic $X$-tree $\Tree = (V,E,\w)$, integers $\ex(x)$ and $\ell(x)$ for each~$x\in X$, a set of teams $T$, and a target diversity $D\in \mathbb{N}_0$}{
	Is there a valid $T$-schedule saving $S$, for some $S\subseteq X$ such that~$PD_\Tree(S)\ge D$}

The problem \tPDws is the same as \tPDs, except for the restriction that the valid $T$-schedule should be strict.

\problemdef{\tPDwslong\\ (\tPDws)}{
	A phylogenetic $X$-tree $\Tree = (V,E,\w)$, integers $\ex(x)$ and $\ell(x)$ for each~$x\in X$, a set of teams $T$, and a target diversity $D\in \mathbb{N}_0$}{
	Is there a strict valid $T$-schedule saving $S$, for some $S\subseteq X$ such that~$PD_\Tree(S)\ge D$}
% \todosi{Discussion: Maybe only write one of the problem definitions. MJ: not for journal version but for conference version it's probably a good idea}
A set~$S$ that satisfies these conditions for an instance~\Instance of \tPDs or \tPDws is called a \emph{solution of \Instance}.

These two problem definitions only differ in the fact that valid~$T$-schedule in the latter needs to be strict while in the former it is not necessary.

Observe that if there is only one team, every valid schedule is also strict.
Thus, an instance $\Instance = (\Tree, \ell, \ex, T, D)$ with $|T|=1$ is a \yes-instance of \tPDs if and only if~$\Instance$ is a \yes-instance of \tPDws.
Lemma~\ref{lem:timePD-scheduleCondition} and \ref{lem:timePD-strictScheduleCondition} elaborate on the conditions in these questions in more detail.

\paragraph{Additional Definitions.}
Now, we introduce some additional definitions that will be helpful for the parameterizations and proofs that follow.

\begin{definition}[$\var_{\ex}$ and the Classes~$Y_i$ and~$Z_i$]
	Let $\var_{\ex} := |\{\ex(x) : x \in X\}|$ and~$\max_{\ex} := \max\{\ex(x) : x\in X\}$. That is, $\var_{\ex}$ is the number of different extinction times for taxa in $X$, and $\max_{\ex}$ is the latest extinction time. 
	
	Let $\ex_1 < \ex_2 < \dots < \ex_{\var_{\ex}} = \max_{\ex}$ denote the elements of~$\ex(X)$.
	For each~$j\in[\var_{\ex}]$ the \emph{class $Y_j\subseteq X$} is the set of taxa $x$ with $\ex(x)=\ex_j$ and we define~$Z_j = Y_1 \cup \dots \cup Y_j$.
	Further, we define $\ex^*(x) = j$ for each $x\in Y_j$.
\end{definition}

Along similar lines to $\var_{\ex}$ and $\max_{\ex}$, we define~$\var_\ell := |\{\ell(x) : x\in X\}|$\lb and~$\max_\ell := |\{\ell(x) : x\in X\}|$.
That is, $\var_\ell$ is the number of different rescue lengths of taxa, and $\max_\ell$ is the largest rescue length.
We let $\ell_1 < \ell_2 < \dots < \ell_{\var_\ell} = \max_\ell$ denote the elements of $\{\ell(x) \mid x \in X\}$.

Given an instance $\Instance$ of \tPDs or \tPDws with target diversity $D$, we define $\Dbar : = \PD(X)-D = \sum_{e \in E}\w(e) - D$.
Thus, $\Dbar$ is the acceptable loss of diversity---if we save a set of taxa $A \subseteq X$ with $\PD(A)\geq D$, then the amount of diversity we lose from $\Tree$ as a whole is at most $\Dbar$.

We define $\Hbar j$ to be $\sum_{x\in Z_j} \ell(x) - H_j$ for each $j\in [\var_{\ex}]$.
That is, $\Hbar j$ is the difference between the number of person-hours needed to save all taxa in $Z_j$, and the number of person-hours available to save them.

% \todoji{For now have ommited this characterization for \tPDws, but we will probably need it somewhere: 
% For a set $Q$ and an integer $t$ we say that $a: X \to [t] \times \mathbb{N}, a(x) = (a_1(x),a_2(x))$ is an \textit{assignment of $Q$} if $a$ is injective.
% The \textit{jobs $J_i$ of team $i$} is the set of taxa $x$ with $a_1(x)=i$ for an assignment. The \textit{first $q$ jobs of team $i\in [t]$}, denoted by $J_i^q$, is the subset taxa $x$ in $J_i$ with $a_2(x)\le q$.
% We say that the \textit{$t$ teams can save a subset $Q$ of taxa} if there is an assignment $a$ of $Q$ such that $s_{a_1(x)} + \sum_{\hat x\in J_{a_1(x)}^{a_2(x)}} \ell_{\hat x} \le \min\{ \ex(x), e_{a_1(x)} \}$ for each $x\in X$.}

\subsection{Observations}
In the following, we present some easy observations.

\begin{lemma}
	\label{lem:timePD-scheduleCondition}
	There exists a valid $T$-schedule saving a set of taxa $A \subseteq X$ if and only\linebreak if $\sum_{x\in A\cap Z_j} \ell(x) \le H_j$ for each $j\in [\var_{\ex}]$.
\end{lemma}
\begin{proof}
	Recall that $\ex_j$ is the $j$th extinction time and $\ex^*(x) = j$ for each $x\in Y_j$
	
	Suppose first that $f$ is a valid schedule saving $A$.
	Then, for each $x\in A$, the set~$f^{-1}(x)$ contains at least $\ell(x)$ pairs $(i,j')$ with $j' \leq \ex(x)$. It follows that $f^{-1}(Z_j)$ contains at least $\sum_{x \in A\cap Z_j} \ell(x)$ pairs $(i,j')$ with $j' \leq \ex_j$, and so for each $j\in [\var_{\ex}]$, we conclude $\sum_{x\in A\cap Z_j} \ell(x) \le |\{(i,j')\in \mathcal{H}_T \mid j' \le \ex_j\}| = H_j$.
	
	Conversely, suppose $\sum_{x\in A\cap Z_i} \ell(x) \le H_j$ for each $j\in [\var_{\ex}]$. 
	Then order the elements of $\mathcal{H}_T$ such that $(i',j')$ appears before $(i'',j'')$ if $j' < j''$, and order the elements of $A$ such that $x$ appears before $y$ if $\ex^*(x) < \ex^*(y)$ (thus, all elements of~$A\cap Y_j$ appear before all elements of $A\cap Y_{j+1}$).
	Now define $f:\mathcal{H}_T \rightarrow A\cup\{\textsc{none}\}$ by repeatedly choosing the first available taxon $x$ of $A$, and assigning it to the first available $\ell(x)$ elements of $\mathcal{H}_T$.
	Then, the first $\sum_{x\in A\cap Z_j} \ell(x) \le H_j$ elements of $\mathcal{H}_T$ are used to save taxa in $Z_j$, and these elements of $\mathcal{H}_T$ are all of the form $(i,j')$ for~$j' \leq \ex_j$.
	It follows that $f$ is a valid schedule saving $A$.
	%	\qed
\end{proof}

By Lemma~\ref{lem:timePD-scheduleCondition}, we have that $\Instance$ is a \yes-instance of \tPDs (or \tPDws with $|T|=1$) if and only if there exists a set $S\subseteq X$ with $\PD(S)\geq D$ such that~$\sum_{x\in A\cap Z_i} \ell(x) \le H_i$ for each $i\in [\var_{\ex}]$.
The following lemma follows from the definitions of valid and strict schedules.

\begin{lemma}
	\label{lem:timePD-strictScheduleCondition}
	There exists a strict valid $T$-schedule saving $A \subseteq X$ if and only if there is a partition of $A$ into sets $A_1\dots, A_{|T|}$ such that for each $i \in [|T|]$, there is a valid~$\{t_i\}$-schedule saving $A_i$. 
	%$\sum_{x\in A\cap Z_j} \ell(x) \le H_j$ for each $j\in [\var_{\ex}]$.
\end{lemma}

For the parameterized point of view,
we first show that \tPDs and \tPDws are \NP-hard, even for quite restricted instances. 
Our proof adapts the reduction of Richard Karp used to show that the scheduling problem $1||\sum w_j(1-U_j)$ is \NP-hard~\cite{karp}.
While Karp gives a reduction from \KP, we use a reduction from $k$-\SubSum.
This result was also observed by~\cite{pardi07,hartmann} for \BNAP---\MPD with integer cost on taxa---which is a special case of \tPDs or \tPDws.

%Proofs of theorems, lemmas and propositions marked with $\star$ are proven in the appendix.
\begin{proposition}[\cite{karp,pardi07,hartmann}]
	\label{prop:timePD-Knapsack}
	~
	\begin{propEnum}%[(a)]
		\item\label{thm:timePD-KP-NP}\tPDs and \tPDws are \NP-hard.
		\item\label{thm:timePD-KP-W1}It is \Wh{1}-hard with respect to $k$ to decide whether an instance of \tPDs or \tPDws has a solution in which $k$ taxa are saved.
	\end{propEnum}
	Both statements hold even if the tree in the input is a star and $|T|=\var_{\ex}=1$.
\end{proposition}
\begin{proof}
	We reduce from $k$-\SubSum, %(Karp reduced from \KP)
	which is \NP-hard and \Wh{1}-hard when parameterized by $k$, the size of the solution~\cite{downey}.
	In $k$-\SubSum we are given a multiset of integers $\mathcal{Z} = \{z_1,\dots,z_{|\mathcal{Z}|}\}$ and two integers $k$ and~$G$.
	It is asked whether there is a set $S\subseteq \mathcal{Z}$ of size $k$ such that $\sum_{z\in S} z = G$.
	
	\proofpara{Reduction}
	Let $(\mathcal{Z},G)$ be an instance of~$k$-\SubSum.
	Let $Q$ be an integer which is bigger than the sum of all elements in $\mathcal{Z}$.
	
	We define an instance $\Instance = (\Tree, \ell, \ex, T, D)$ with
	a star tree $\Tree$ that has a root $\rho$ and a set of leaves $X = \{ x_1,\dots,x_{|\mathcal{Z}|} \}$.
	For each $x_i\in X$, we set edge-weights $\w(\rho x_i)$ and\lb
	the rescue length $\ell(x_i)$ to be $z_i + Q$.
	Further, define an extinction time\lb of~$\ex(x_i) := G + kQ$, which is the same for each taxon and
	the only team operates from $s_1 := 0$ to $e_1 := G + kQ$.
	Finally, we set $D := G + kQ$.
	
	\proofpara{Correctness}
	The reduction is computed in polynomial time for a suitable~$Q$.
	Because there is only one team, $\Instance$ is a \yes-instance of \tPDs if and only if \Instance is a \yes-instance of \tPDws.
	It remains to show that any solution for $\Instance$ saves exactly $k$ taxa, and that $\Instance$ is a \yes-instance of \tPDs if and only if $(\mathcal{Z},G)$ is a \yes-instance of~$k$-\SubSum.
	
	So let $(\mathcal{Z},G)$ be a \yes-instance of $k$-\SubSum with solution $S$.
	Define a set of taxa $S' := \{ x_i\in X \mid z_i \in S \}$.
	Clearly, $S'$ and~$S$ are of size~$k$.
	Then,
	$$
	\PD(S') = \sum_{x_i\in S'} \w(\rho x_i) = \sum_{z_i\in S} (z_i + Q) = kQ + \sum_{z_i\in S} z_i = G + kQ = D,
	$$
	and
	analogously $\sum_{x_i\in S'} \ell(x_i) = \sum_{z_i\in S} (z_i + Q) = G + kQ = H_{1}$.
	Consequently, $S'$ is a solution for %the \yes-instance 
	\Instance.
	
	Let \Instance be a \yes-instance with solution $S'$.
	Because
	$$
	G+kQ = D \le \PD(S') = \sum_{x_i\in S'} \w(\rho x_i) = \sum_{z_i\in S} (z_i + Q) = \sum_{x_i\in S'} \ell(x_i) \le G+kQ,
	$$
	we conclude that $G + kQ  = \sum_{z_i\in S} (z_i + Q)$, from which it follows that $|S'| = k$ and~$\sum_{z_i\in S} z_i = G$. 
	Therefore, $S := \{ z_i \mid x_i \in S' \}$ is a solution of $(\mathcal{Z},G)$.
	%	\qed
\end{proof}

The scheduling variant \textsc{P3||C${}_{\max}$} in which it is asked for whether the given jobs can be scheduled on three parallel machines such that the \emph{makespan} (the maximum time taken on any machine) is at most C${}_{\max}$, is strongly \NP-hard \cite{GSJ1976,garey1978}.
\textsc{P3||C${}_{\max}$} can be seen as a special case of \textsc{P3||$\sum w_j(1-U_j)$} and therefore \tPDws in which there are three teams, the tree is a star on which all weights are 1, and we have to save every taxon in the time defined by the makespan.

\begin{proposition}[\cite{GSJ1976}]
	\label{prop:timePD-ws-Dbar}
	\tPDws is strongly \NP-hard even if the tree is a star, every taxon has to be saved ($\Dbar=0$), there are three teams, $\max_\w=1$\lb and~$\var_{\ex}=1$.
\end{proposition}

Our final observation in this section is that for an instance of \tPDs and a set of taxa~$S\subseteq X$,
we can compute the diversity $\PD(S)$ and check whether there is a valid $T$-schedule saving $S$ in polynomial time (by Lemma~\ref{lem:timePD-scheduleCondition}).
Therefore, checking each subset of $X$ yields an $\Oh^*(2^{n})$-time-algorithm for \tPDs.

For an instance of \tPDws, in $\Oh^*(n!)$ time we can also guess the order of the taxa and then assign taxa to the first team until they have no more capacity.
Afterward, assign the next taxa to the second team and so on.
The instance is a \yes-instance if and only if for some order this assignment yields a strict valid~$T$-schedule of a set of taxa with diversity at least~$D$.
We conclude the following.

\begin{proposition}
	\label{prop:timePD-X}
	~
	\begin{propEnum}
		\item\label{prop:timePD-X-s}\tPDs can be solved in $2^{n} \cdot \poly(|\Instance|)$ time.
		\item\label{prop:timePD-X-ws}\tPDws can be solved in $n! \cdot \poly(|\Instance|)$ time.
	\end{propEnum}
\end{proposition}

\section{The Diversity $D$}
\label{sec:timePD-D}
In this section, we show that \tPDs and \tPDws are \FPT when parameterized by the threshold of diversity $D$.
When one tries to use the standard approach with a dynamic programming algorithm over the vertices of the tree, one will struggle with the question of how to handle the different extinction times of the taxa.
While it is straightforward how to handle taxa of a specific extinction time, already with a second, it is not trivial how to combine sub-solutions.
In order to overcome these issues, we use the technique of color coding in addition to dynamic programming.

In the following, we consider colored versions of the problems, \ctPDs and \ctPDws.
In these, additionally to the input of the respective uncolored variant, we are given a coloring~$c$ which assigns each edge $e\in E$ \emph{a set of colors $c(e)$}: A subset of $[D]$ of size $\w(e)$.
For a taxon $x\in X$, we define $c(x)$ to be the union of colors of the edges on a path from the root to~$x$.
Further, for a set $S$ of taxa, we define~$c(S)$ to be~$\bigcup_{x\in S} c(x)$.
In \ctPDs (respectively, \ctPDws), we ask whether there is a subset $S$ of taxa and a (strictly) valid $T$-schedule saving $S$ such that~$c(S) = [D]$.

Next, we show that \ctPDs and \ctPDws are \FPT with respect to~$D$.
The difficulty of combining sub-solutions for different extinction times also arises in these colored versions of the problem.
However, the color enables us to consider the tree as a whole.
We select taxa with ordered extinction time such that we are able to check whether there is a (strictly) valid $T$-schedule which can indeed save the selected set of taxa.

\begin{lemma}
	\label{lem:timePD-cs-D}
	\ctPDs can be solved in $\Oh(2^D\cdot D\cdot \var_{\ex}^2 \cdot n)$~time.
\end{lemma}
\begin{proof}
	Let $\Instance = (\Tree, \ell, \ex, T, D, c)$ be an instance of \ctPDs.
	For~$p\in[\var_{\ex}]$, we call a set $S$ of taxa \emph{$p$-compatible} if $\ell(S\cap Z_q) \le H_i$ for each $q\in [p]$.
	For a set of colors~$C\subseteq [D]$ and $p\in[\var_{\ex}]$, we call a set $S$ of taxa \emph{$(C,p)$-feasible} if
	\begin{itemize}%[F a)]
		\item $C$ is a subset of $c(S)$,
		\item $S$ is a subset of $Z_p$, and
		\item $S$ is $p$-compatible.
	\end{itemize}
	Intuitively, the $(C,p)$-feasible sets~$S$ consist of taxa which overall cover the colors~$C$ if saved and for which a valid~$T$-schedule saving~$S$ exists.
	
	\proofpara{Table definition}
	In the following dynamic programming algorithm with table $\DP$, for each $C\subseteq [D]$ and $p\in [\var_{\ex}]$ we want entry $\DP[C,p]$ to store the minimum length~$\ell(S)$ of a $(C,p)$-feasible set of taxa $S\subseteq X$, with $\DP[C,p] = \infty$ if no $(C,p)$-feasible set exists.

	\proofpara{Algorithm}
	As a base case, we store $\DP[\emptyset, p] = 0$ for each $p\in [\var_{\ex}]$.
	%Further, we store $\DP[c(x),j] = \psi(\ell(x),H_j)$ if $x\in Y_j$ and $\infty$ otherwise.
	%\todosi{That is just one option, but not necessarily the value. I am not so amazed by that.... Maybe we can fix it with $\emptyset$}
	
	To compute further values, we use the recurrence
	\begin{equation}
		\label{eqn:split-D}
		\DP[C,p] = 
		\min_{\substack{
				x\in Z_p,
				c(x)\cap C\ne \emptyset}}
		\psi(\DP[C\setminus c(x),\ex^*(x)] + \ell(x), H_{\ex^*(x)}).
	\end{equation}
	Recall $\ex^*(x) = q$ for $x\in Y_q$; and $\psi(a,b) = a$ if $a\le b$ and otherwise $\psi(a,b) = \infty$.
	
	We return \yes if $\DP[[D],\var_{\ex}] \le \max_{\ex}$, and otherwise we return \no.
	
	\proofpara{Correctness}
	Let \Instance be an instance of \ctPDs.
	From the definition of $(C,p)$-feasible sets it follows that there exists a $([D],\var_{\ex})$-feasible set of taxa $S$ if and only if $\Instance$ is a \yes-instance of \ctPDs.
	It remains to show that $\DP[C,p]$ stores the smallest length of a $(C,p)$-feasible set $S$.
	
	As the set $S=\emptyset$ is $(\emptyset,p)$-feasible for each $p\in [\var_{\ex}]$, the base case is correct.
	
	As an induction hypothesis, assume that for a fixed and non-empty set of colors~$C\subseteq [D]$ and a fixed integer $p\in [\var_{\ex}]$, entry $\DP[K,q]$ stores the correct value for each $K\subsetneq C$ and $q \le p$.
	We prove that $\DP[C,p]$ stores the correct value by showing 
	that if a $(C,p)$-feasible set of taxa $S$ exists then $\DP[C,p]\le \ell(S)$ and 
	that if $\DP[C,p] = a < \infty$ then there is a $(C,p)$-feasible set $S$ with $\ell(S) = a$.
	
	Let $S\subseteq X$ be a $(C,p)$-feasible set of taxa.
	Observe that if $c(x)\cap C = \emptyset$ for a taxon~$x\in S$ then also $S\setminus \{x\}$ is $(C,p)$-feasible.
	So, we assume that $c(x)\cap C \ne \emptyset$ for each $x\in S$.
	Let $y\in S$ be a taxon such that $\ex(y) \ge \ex(x)$ for each $x\in S$.
	%the taxon $y$ is considered on the right side of Equation~\ref{eqn:split-D}.
	Observe that the set $S\setminus \{y\}$ is $(C\setminus c(y),\ex^*(y))$-feasible.
	Thus, because $y\in S\subseteq Z_p$ and $c(y)\cap C \ne \emptyset$, $\DP[C\setminus c(y),\ex^*(y)] \le \ell(S\setminus \{y\}) = \ell(S) - \ell(y)$ by the induction hypothesis.
	We conclude $\DP[C\setminus c(y),\ex^*(y)] + \ell(y) \le \ell(S) \le H_{\ex^*(y)}$ because $S$ is~$\ex^*(y)$-compatible.
	Hence, $\DP[C,p] \le \DP[C\setminus c(y),\ex^*(y)] + \ell(y) \le \ell(S)$.
	
	Conversely, suppose  $\DP[C,p] < \infty$.
	Then, by Recurrence~(\ref{eqn:split-D}), a taxon~$x\in Z_p$ exists such that  $c(x)\cap C \ne \emptyset$ and $\DP[C,p] = \DP[C\setminus c(x),\ex^*(x)] + \ell(x)$.
	By the\linebreak induction hypothesis, there is a $(C\setminus c(x), \ex^*(x))$-feasible set $S'$ such\linebreak that~$\DP[C\setminus c(x),\ex^*(x)] = \ell(S')$.
	Because $c(x)\cap(C\setminus c(x))$ is empty, $x$ is not in~$S'$.
	Further, $\DP[C,p] = \DP[C\setminus c(x),\ex^*(x)] + \ell(x) = \ell(S') + \ell(x) = \ell(S'\cup\{x\})$ and $S'\cup\{x\}$ is the desired $(C,p)$-feasible set.

	\proofpara{Running time}
	The table has $\Oh(2^D \cdot \var_{\ex})$ entries.
	%In a preprocessing step we can compute the sets $c(x)$ for all $x\in X$ in $\Oh(n^2\cdot ??)$ time.\todos{Please check!}
	For each $x\in X$ and every set of colors $C$, we can compute whether $c(x)$ and $C$ have a non-empty intersection in~$\Oh(D)$ time.
	Then, we can compute the set $C\setminus c(x)$ in $\Oh(D)$ time.
	Each entry stores an integer that is at most $H_{\max_{\ex}}$ (or $\infty$).
	In our RAM model, the addition and the comparison in $\psi$ can be done in constant time, 
	and so the right side of Recurrence~(\ref{eqn:split-D}) can be computed in $\Oh(D \cdot \var_{\ex} \cdot n)$ time.
	Altogether, we can compute a solution in~$\Oh(2^D\cdot D\cdot \var_{\ex}^2 \cdot n)$ time.
	%	\qed
\end{proof}

For the algorithm in Lemma~\ref{lem:timePD-cs-D} it is crucial that 
the $(C,j)$-feasibility of a set $S$ depends only on $\ell(S)$ and the values $H_1,\dots, H_j$.
%we can compute the person-hours we have in each time interval.
This is not the case for \ctPDws, as the available times of the teams impact a possible schedule.
Our solution for this issue is that we divide the colors (representing the diversity) that are supposed to be saved and delegate the colors to specific teams.
The problem then can be solved individually for the teams.
This division and delegation happens in the Recurrence~(\ref{eqn:ws-D}).

\begin{lemma}
	\label{lem:timePD-cws-D}
	\ctPDws can be solved in $\Oh^*(2^D\cdot D^3)$~time.
\end{lemma}
\begin{proof}
	\proofpara{Table definition}
	We define a dynamic programming algorithm with three tables $\DP_0$, $\DP_1$ and $\DP_2$.
	Entries of table $\DP_0$ are computed with the idea of Lemma~\ref{lem:timePD-cs-D}.
	$\DP_1$ is an auxiliary table.
	A final solution can be found in table $\DP_2$, then.
	For this lemma, similar to $H_i$, we define $H_i^{(q)} := |\{(i,j) \in \mathcal{H}_{T} \mid j\le \ex_q \}|$.
	That is, $H_i^{(q)}$ are the person-hours of team $t_i$ until $\ex_q$.
	
	Let \Instance be an instance of \ctPDws.
	We define terms similar to Lemma~\ref{lem:timePD-cs-D}.
	For a team $t_i\in T$ (and an integer $p\in[\var_{\ex}]$), a set $S\subseteq X$ is \emph{$t_i$-compatible} (respectively, \emph{$(p,t_i)$-compatible}) if $\ell(S\cap Z_q) \le H_i^{(q)}$ for each $q\in [\var_{\ex}]$ ($q\in [p]$).
	Inspired by Lemma~\ref{lem:timePD-strictScheduleCondition}, for a $i\in [|T|]$ we call a set $S$ of taxa \emph{$[i]$-compatible} if there is a partition~$S_1,\dots,S_i$ of $S$ such that $S_j$ is $t_j$-compatible for each $j\in [i]$.
	For a set of colors $C\subseteq [D]$, an integer $p\in[\var_{\ex}]$ and a team $t_i \in T$, we call a set $S$ of taxa~\emph{$(C,p,t_i)$-feasible} if
	\begin{itemize}%[F a)]
		\item $S$ is a subset of $Z_p$,
		\item $S$ is $(p,t_i)$-compatible,
		\item $C$ is a subset of $c(S)$, and
		\item ($S=\emptyset$ or $S\cap Y_p \ne \emptyset$).
	\end{itemize}
	
	Formally, for a set of colors $C\subseteq [D]$, an integer $p\in [\var_{\ex}]$ and a team $t_i\in T$ we want that
	\begin{itemize}
		\item in $\DP_0[C,p,i]$ the minimum length $\ell(S)$ of a $(C,p,t_i)$-feasible set of taxa $S\subseteq X$ is stored,
		\item in $\DP_1[C,i]$ stores~1 if there is a $t_i$-compatible set $S\subseteq X$ with~$c(S) \supseteq C$ and 0 otherwise, and
		\item in $\DP_2[C,i]$ stores~1 if there is a $[i]$-compatible set $S\subseteq X$ with~$c(S) \supseteq C$ and 0 otherwise.
	\end{itemize}
	
	\proofpara{Algorithm}
	%	The base cases are as follows.
	As a base case,
	we store $\DP_0[\emptyset, p, i] = 0$ for each $p\in [\var_{\ex}]$ and $i\in [|T|]$.
	
	To compute further values of $\DP_0$, we use the recurrence
	\begin{equation}
		\DP_0[C,p,i] = \min_{x\in Y_p, c(x)\cap C\ne \emptyset, q\le p}
		\psi(\DP[C\setminus c(x),q,i] + \ell(x), H_i^{(q)}).
	\end{equation}
	Recall that $\psi(a,b) = a$ if $a\le b$ and otherwise $\psi(a,b) = \infty$.
	Once all values in $\DP_0$\lb are computed, we compute $z_{C,i} := \min_{q\in [\var_{\ex}]} \DP[C,q,i]$
	and set $\DP_1[C,i] = 1$\linebreak if~$z_{C,i} \le H_i^{(\max_{\ex})}$ and $\DP_1[C,i] = 0$, otherwise.
	
	Finally, we define $\DP_2[C,1] := \DP_1[C,1]$ and compute further values with
	\begin{equation}
		\label{eqn:ws-D}
		\DP_2[C,i+1] = \min_{C' \subseteq C} \DP_2[C',i] \cdot \DP_1[C\setminus C',i+1].
	\end{equation}
	
	We return \yes if $\DP_2[[D],|T|] = 1$, and otherwise, we return \no.
	
	\proofpara{Correctness}
	Analogously as in Lemma~\ref{lem:timePD-cs-D}, one can prove that the entries of $\DP_0$ store the correct value.
	It directly follows by the definition that the entries of $\DP_1$ store the correct value.
	Likewise $\DP_2[C,1]$ stores the correct value for colors $C\subseteq [D]$ be the definition.
	
	Fix a set of colors $C\subseteq [D]$ and an integer $i\in [|T|-1]$.
	We assume as an induction hypothesis that entry $\DP_2[K,j]$ stores the correct value for each $K\subsetneq C$\lb and $j\le i$.
	To conclude that $\DP_2[C,j+1]$ stores the correct value, we show\lb that 
	$\DP_2[C,i+1] \le \ell(S)$ for each $[i+1]$-compatible set $S\subseteq X$ with $c(S) \supseteq C$ and
	if $\DP_2[C,i+1] = a < \infty$ then there is an $[i+1]$-compatible set $S\subseteq X$\lb with $c(S) \supseteq C$ and $\ell(S) = a$.
	
	Let $S\subseteq X$ be an $[i+1]$-compatible set with $c(S) \supseteq C$.
	Let $S_1,\dots,S_{i+1}$ be a partition of $S$ such that $S_j$ is $t_j$-compatible for each $j\in [i+1]$.
	Define $\hat S$ to be the union of the~$S_j$ and $\hat C := C \cap c(\hat S)$.
	Then, $\hat S$ is $[i]$-compatible with~$c(\hat S) \supseteq \hat C$.
	Therefore, entry $\DP_2[\hat C,i]$ stores at least~$\ell(\hat S)$ by the induction hypothesis.
	Because~$c(S_{i+1}) \supseteq C \setminus \hat C$
	we conclude that $\DP_1[C\setminus \hat C,i+1] \le \ell(S_{i+1})$.
	We conclude that $\DP_2[C,i+1] \le \DP_2[\hat C,i] + \DP_1[C\setminus \hat C,i+1] \le \ell(\hat S) + \ell(S_{i+1}) = \ell(S)$.
	
	Let $\DP_2[C,i+1]$ store the value $a<\infty$.
	By Recurrence~(\ref{eqn:ws-D}), there is a set of colors $C'\subseteq C$ such that $a = \DP_2[C',i] + \DP_1[C\setminus C',i+1]$.
	By the induction hypothesis, we conclude then that there is
	an $[i]$-compatible set $S_1$ and
	a $t_{i+1}$-compatible set~$S_2$
	such that $\DP_2[C',i] = \ell(S_1)$ and $\DP_1[C\setminus C',i+1] = \ell(S_2)$.
	Assume first that~$S_1$ and $S_2$ are disjoint.
	We define $S := S_1 \cup S_2$ and conclude with the disjointness that~$\ell(S) = \ell(S_1) \cup \ell(S_2) = a$.
	Further, we conclude that $S$ is $[i+1]$-compatible and~$c(S) = c(S_1) \cup c(S) = C' \cup (C\setminus C') = C$
	such that $S$ is the desired set.
	It remains to argue that non-disjoint $S_1$ and $S_2$ contradict the optimality of $C'$.
	
	Now, let $A$ be the intersection of $S_1$ and $S_2$ and define $\hat C := c(S_1)$ and~$\hat S := S_2 \setminus A$.
	Observe that $\hat S$ is a $t_{i+1}$-compatible set with $c(\hat S) \supseteq (C\setminus C') \setminus c(A) = C \setminus \hat C$.
	Therefore,
	\begin{eqnarray*}
		\DP_2[\hat C,i] + \DP_1[C\setminus \hat C,i+1]
		 &\le& \ell(S_1) + \ell(\hat S)\\
		 &=& \ell(S_1) + \ell(S_2) - \ell(A)\\
		 &=& \DP_2[C',i] + \DP_1[C\setminus C',i+1] - \ell(A),
	\end{eqnarray*}
	which contradicts the optimallity of $C'$, unless~$A$ is empty.

	\proofpara{Running time}
	By Lemma~\ref{lem:timePD-cs-D}, in $\Oh(2^D\cdot D\cdot \var_{\ex}^2 \cdot n)$ time all entries of $\DP_0$ are computed.
	Table $\DP_1$ contains $2^D \cdot |T|$ entries which in~$\Oh(\var_{\ex})$ time can be computed, each.
	
	To compute the values of table $\DP_2$ in \Recc{eqn:ws-D}, we use convolutions.
	Readers unfamiliar with this topic, we refer to \cite[Sec. 10.3]{cygan}.
	We define functions $f_i, g_i: 2^{[D]} \to \{0,1\}$ where
	$f_i(C) := \DP_2[C,i]$ and $g_i(C) := \DP_1[C,i]$.
	Then, we can express Recurrence~(\ref{eqn:ws-D}) as $\DP_2[C,i+1] = f_{i+1}(C) = (f_i * g_{i+1})(C)$.
	Therefore, for each $i\in [|T|]$ we can compute all values of $\DP_2[\cdot,i] = f_i(\cdot)$ in $\Oh(2^D \cdot D^3)$ time~\cite[Thm.~10.15]{cygan}.
	
	Altogether, in $\Oh(2^D\cdot D^3 \cdot |T| \cdot \var_{\ex}^2 \cdot n)$ time, we can compute a solution of an instance of \ctPDws.
	%	\qed
\end{proof}

Our following procedure
in a nutshell
is to use standard color coding techniques to show that in $\Oh^*(2^{\Oh(D)})$~time
one can reduce an instance $\Instance$ of \tPDs (respectively, \tPDws) to $z \in \Oh^*(2^{\Oh(D)})$ instances $\Instance_1,\dots,\Instance_z$ of \ctPDs (\ctPDws),
which can be solved using Lemma~\ref{lem:timePD-cs-D} and Lemma~\ref{lem:timePD-cws-D}.
In particular, if $S$ is a solution for $\Instance$ then $c(S) = D$ under the coloring $c$ in at least one $\Instance_i$.
Then, \Instance is a \yes-instance, if and only if any $\Instance_i$ is a \yes-instance.

\begin{theorem}
	\label{thm:timePD-D}
	\tPDs and \tPDws can be solved in 
	$\Oh^*(2^{2.443\cdot D + o(D)})$~time.
\end{theorem}

For an overview of color coding, we refer the reader to~\cite[Sec.~5.2~and~5.6]{cygan}.

The key idea is that we construct a family $\mathcal{C}$ of colorings on the edges of $\Tree$, where each edge $e\in E$ is assigned a subset of $[D]$ of size $\w(e)$.
Using these, we generate~$|\mathcal{C}|$ instances of the colored version of the problem, which we then solve in $\Oh^*(2^D\cdot D)$ time.
%The colorings are constructed in such a manner that $\Instance$ is a \yes-instance if and only if at least one of the constructed \ctPDs instances is a \yes-instance.
Central for this will be the concept of a perfect hash family as defined in Definition~\ref{def:perfectHashFamily}.

\begin{proof}
	We focus on \tPDs and omit the analogous proof for \tPDws.
	
	\proofpara{Reduction}
	Let $\Instance = (\Tree, \ex, \ell, T, D)$ be an instance of \tPDs.
	We assume that for each taxon $x\in X$ some valid $T$-schedule saving $\{x\}$ exists, as we can delete $x$ from \Tree, otherwise.
	Therefore, if there is an edge $e$ with $\w(e) \ge D$, then~$\{x\}$ is a solution for each $x\in \off(e)$ and we have a trivial \yes-instance.
	Hence, we assume that $\max_\w < D$.
	
	Now, order the edges $e_1, \dots, e_m$ of $\Tree$ arbitrarily and
	define integers $W_0 := 0$ and~$W_j := \sum_{i=1}^{j} \w(e_{i})$ for each $j\in [m]$.
	We set $W := W_m$.
	Let $\mathcal{F}$ be a $(W, D)$-perfect hash family.
	Now, we define a family of colorings $\mathcal{C}$ as follows.
	For every~$f \in \mathcal{F}$, let~$c_f$ be the coloring such that $c_f(e_j) := \{f(W_{j-1}+1), \dots, f(W_j)\}$ for each~$e_j \in E(\Tree)$.
	
	For each $c_f \in \mathcal{C}$,
	let $\Instance_{c_f} = (\Tree, \ex, \ell, T, D, c_f)$ be the corresponding instance of \ctPDs.
	Now, solve every instance $\Instance_{c_f}$ using the algorithm described in Lemma~\ref{lem:timePD-cws-D}, and return \yes if and only if $\Instance_{c_f}$ is a \yes-instance for some $c_f \in \mathcal{C}$.
	
	\proofpara{Correctness}
	For any subset of edges $E'$ with $\w(E') \geq D$, there is a corresponding subset of $[W]$ of size at least $D$.
	Since $\mathcal{F}$ is a $(W, D)$-perfect hash family, it follows that $c_f(E') = [D]$, for some $c_f \in \mathcal{C}$.
	Thus, in particular, if $S\subseteq X$ is a solution for instance~$\Instance$, then $c_f(S) = [D]$, for some $c_f \in \mathcal{C}$, where $c_f(S)$ is the set of colors assigned by $c_f$ to the edges between the root and a taxon of $S$.
	It follows that one of the constructed instances of \ctPDs is a \yes-instance.
	
	Conversely, it is easy to see that any solution $S$ for one of the constructed instances of \ctPDs is also a solution for $\Instance$.

	\proofpara{Running Time}
	The construction of $\mathcal{C}$ takes $e^D D^{\Oh(\log D)} \cdot W \log W$ time, and for each coloring $c \in \mathcal{C}$, the construction of instance~$\Instance_c$ of \ctPDs takes time polynomial in $|\Instance|$.
	Solving instances of \ctPDs takes $\Oh^*(2^D\cdot D)$~time, and the number of instances is $|\mathcal{C}| = e^D D^{\Oh(\log D)} \cdot \log W$.
	
	Thus, the total running time is 
	$\Oh^*(e^D D^{\Oh(\log D)} \log W \cdot (W + (2^D\cdot D)))$.
	This sim\-pli\-fies to $\Oh^*((2e)^D \cdot 2^{\Oh(\log^2(D))})$, because $W = \PD(X) < 2n\cdot D$.
	%	\qed
\end{proof}
%}

\section{The Acceptable Loss of Diversity \Dbar}
\label{sec:timePD-Dbar}
% \todosi{Add information what happens if the root has degree 1. MJ: Let's specify in the preliminaries that root has outdegree (at least) 2.}
In this section, we show that \tPDs is \FPT with respect to the acceptable loss of phylogenetic diversity \Dbar.
Recall that \Dbar is defined as $\PD(X)-D$.
% This result shows that \tPDs differs from \tPDws as by Proposition~\ref{prop:timePD-ws-Dbar}
For a set of taxa~$A \subseteq X$, define $E_d(A) := \{ e\in E \mid \off(e)\subseteq A \} = E \setminus E_\Tree(X\setminus A)$.
That is, $E_d(A)$ is the set of edges that \emph{do not} have offspring in $X\setminus A$. 
Then, one may think of \tPDs as the problem of finding a set of taxa $S$ such that there is a valid~$T$-schedule saving $S$ and $\w_\Sigma(E_d(X\setminus S))$ is at most $\Dbar$.

In order to show the desired result, we again use the techniques of color coding and dynamic programming.
For an easier understanding of the following definitions consider Figure~\ref{fig:timePD-example-Dbar-definitions}.
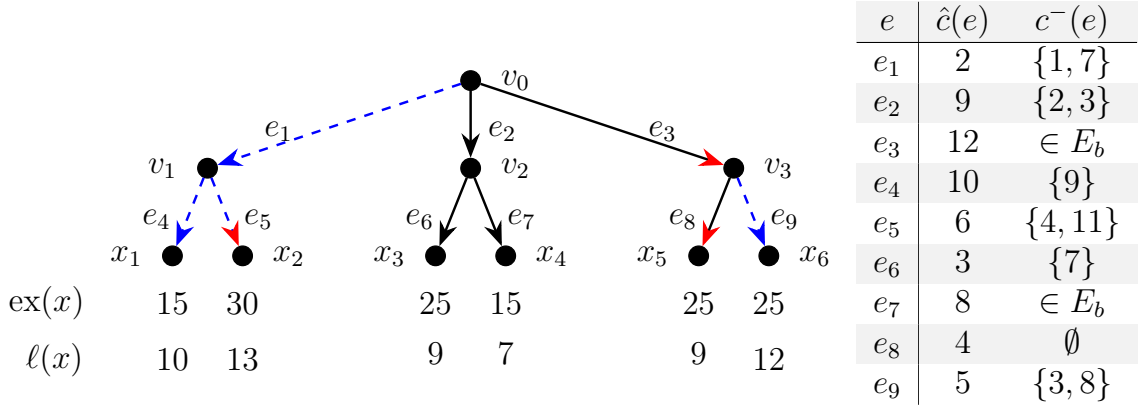
\begin{figure}[t]
	\centering
	\resizebox{1.1\columnwidth}{!}{%
	\begin{tikzpicture}[scale=1,every node/.style={scale=0.9}]
		\tikzstyle{txt}=[circle,fill=white,draw=white,inner sep=0pt]
		\tikzstyle{nde}=[circle,fill=black,draw=black,inner sep=2.5pt]
		
		\node[txt] at (11,8.6) {
			\myrowcols
			\begin{tabular}{c|cc}
				$e$ & $\hat c(e)$ & $c^-(e)$ \\
				\hline
				$e_1$ & 2 & $\{1,7\}$ \\
				$e_2$ & 9 & $\{2,3\}$ \\
				$e_3$ & 12 & $\in E_b$ \\
				$e_4$ & 10 & $\{9\}$ \\
				$e_5$ & 6 & $\{4,11\}$ \\
				$e_6$ & 3 & $\{7\}$ \\
				$e_7$ & 8 & $\in E_b$ \\
				$e_8$ & 4 & $\emptyset$ \\
				$e_9$ & 5 & $\{3,8\}$ \\
			\end{tabular}
		};
		
		\node[nde] (v0) at (5,10) {};
		\node[nde] (v1) at (2,9) {};
		\node[nde] (v2) at (5,9) {};
		\node[nde] (v3) at (8,9) {};
		\node[nde] (v4) at (1.6,8) {};
		\node[nde] (v5) at (2.4,8) {};
		\node[nde] (v6) at (4.6,8) {};
		\node[nde] (v7) at (5.4,8) {};
		\node[nde] (v8) at (7.6,8) {};
		\node[nde] (v9) at (8.4,8) {};
		
		\node[txt,xshift=9mm,yshift=-10mm] (c1) [above=of v1] {$e_1$};
		\node[txt,xshift=4mm,yshift=-10mm] (c2) [above=of v2] {$e_2$};
		\node[txt,xshift=-9mm,yshift=-10mm] (c3) [above=of v3] {$e_3$};
		
		\node[txt,xshift=-2mm,yshift=-10mm] (c4) [above=of v4] {$e_4$};
		\node[txt,xshift=2mm,yshift=-10mm] (c5) [above=of v5] {$e_5$};
		\node[txt,xshift=-2mm,yshift=-10mm] (c6) [above=of v6] {$e_6$};
		\node[txt,xshift=2mm,yshift=-10mm] (c7) [above=of v7] {$e_7$};
		\node[txt,xshift=-2mm,yshift=-10mm] (c8) [above=of v8] {$e_8$};
		\node[txt,xshift=2mm,yshift=-10mm] (c9) [above=of v9] {$e_9$};
		
		\node[txt,yshift=9mm] (r4) [below=of v4] {$15$};
		\node[txt,yshift=9mm] (r5) [below=of v5] {$30$};
		\node[txt,yshift=9mm] (r6) [below=of v6] {$25$};
		\node[txt,yshift=9mm] (r7) [below=of v7] {$15$};
		\node[txt,yshift=9mm] (r8) [below=of v8] {$25$};
		\node[txt,yshift=9mm] (r9) [below=of v9] {$25$};
		
		\node[txt,yshift=9mm] (l4) [below=of r4] {$10$};
		\node[txt,yshift=9mm] (l5) [below=of r5] {$13$};
		\node[txt,yshift=9mm] (l6) [below=of r6] {$9$};
		\node[txt,yshift=9mm] (l7) [below=of r7] {$7$};
		\node[txt,yshift=9mm] (l8) [below=of r8] {$9$};
		\node[txt,yshift=9mm] (l9) [below=of r9] {$12$};
		
		\node[txt,xshift=3mm] (l'4) [left=of l4] {$\ell(x)$};
		\node[txt,xshift=3mm] (r'4) [left=of r4] {$\ex(x)$};
		
		\node[txt,xshift=-9mm] (t0) [right=of v0] {$v_0$};
		\node[txt,xshift=9mm] (t1) [left=of v1] {$v_1$};
		\node[txt,xshift=-9mm] (t2) [right=of v2] {$v_2$};
		\node[txt,xshift=-9mm] (t3) [right=of v3] {$v_3$};
		\node[txt,xshift=9mm] (t4) [left=of v4] {$x_1$};
		\node[txt,xshift=-9mm] (t5) [right=of v5] {$x_2$};
		\node[txt,xshift=9mm] (t6) [left=of v6] {$x_3$};
		\node[txt,xshift=-9mm] (t7) [right=of v7] {$x_4$};
		\node[txt,xshift=9mm] (t8) [left=of v8] {$x_5$};
		\node[txt,xshift=-9mm] (t9) [right=of v9] {$x_6$};
		
		\draw[thick,dashed,blue,arrows = {-Stealth[length=8pt]}] (v0) to (v1);
		\draw[thick,arrows = {-Stealth[length=8pt]}] (v0) to (v2);
		\draw[thick,arrows = {-Stealth[red,length=8pt]}] (v0) to (v3);
		\draw[thick,dashed,blue,arrows = {-Stealth[length=8pt]}] (v1) to (v4);
		\draw[thick,dashed,blue,arrows = {-Stealth[red,length=8pt]}] (v1) to (v5);
		\draw[thick,arrows = {-Stealth[length=8pt]}] (v2) to (v6);
		\draw[thick,arrows = {-Stealth[length=8pt]}] (v2) to (v7);
		\draw[thick,arrows = {-Stealth[red,length=8pt]}] (v3) to (v8);
		\draw[thick,dashed,blue,arrows = {-Stealth[length=8pt]}] (v3) to (v9);
	\end{tikzpicture}
	}
	
	\caption{A hypothetical extinction-6-colored $X$-tree is shown.
	The anchored taxa set~$\mathcal A := \{(x_1,v_1,e_5), (x_2,v_0,e_3), (x_6,v_3,e_8)\}$
	with~$c(E^+(\mathcal A)) = c(\{ e_1, e_4, e_5, e_9 \}) = [11]$ 
	and~$\hat c(E_s(\mathcal A)) = \hat c(\{ e_3, e_5, e_8 \}) = \{ 4, 6, 12 \}$
	is color-respectful.
	The edges in $E^+(\mathcal A)$ are blue and dashed. The edges in $E_s(\mathcal A)$ have a red arrow.
	In contrast to~$\mathcal{A}$,
	the anchored taxa set~$\mathcal A' := \{(x_1,v_0,e_3), (x_2,v_1,e_4), (x_6,v_3,e_8)\}$
	does not have a valid ordering.
	}
	\label{fig:timePD-example-Dbar-definitions}
\end{figure}%
% \todosi{We defined parent- and sibling-edge in the prelims. I think it is sufficient to define it here, as it is not used elsewhere.}

We start by defining a colored version of the problem.
With $E_{\le \Dbar}$ (or $E_{> \Dbar}$) we denote the set of edges $e\in E$ with $\w(e)\le \Dbar$ (or $\w(e) > \Dbar$).
An \emph{extinction-$\Dbar$-colored}~$X$-tree is a phylogenetic $X$-tree $\Tree = (V,E,\w,\hat c,c^-)$,
where $\hat c$ assigns each edge~$e\in E$ a \emph{key-color} $\hat c(e) \in [2\Dbar]$ and
$c^-$ assigns edges $e\in E_{\le \Dbar}$ a set of colors~$c^-(e) \subseteq [2\Dbar]\setminus \{\hat c(e)\}$ of size $\w(e)-1$.
With $c(e)$ we denote the union of the two sets $c^-(e) \cup \{\hat c(e)\}$ for each edge $e\in E_{\le \Dbar}$.

Observe that while in Section~\ref{sec:timePD-D} we wanted the set of edges with an offspring in~$S$ to use every color at least once, here we want that the edges in $E_d(S)$ use \emph{at most} a certain number of colors and therefore each color at most once.

% Note that the set of colors $c(E_d(A))$ cannot be determined strictly from the sets $c(E_d(\{x\}))$ for each $x \in X$.
% In particular, if two taxa $x_1$ and $x_2$ do not share a parent then $E_d(\{x_1,x_2\}) = E_d(\{x_1\}) \cup E_d(\{x_2\})$,
% but if $x_1$ and $x_2$ are the only two children of $v$ then $E_d(\{x_1,x_2\})$ also contains the incoming edge of~$v$.
% This presents a challenge when designing a dynamic programming algorithm that adds taxa to a partial solution one taxon at a time.
% To get around this, we introduce the following concept.

The key idea behind our approach is as follows. 
We seek a solution $S$ in which not only the edges of $E_d(X\setminus S)$ use each color at most once, but also every highest edge in $E_d(X\setminus S)$ has a sibling edge with a key-color not in $c(E_d(X\setminus S))$.
Suppose such a set $S$ exists, and let $x_1,\dots, x_{|X\setminus S|}$ be some ordering of $X \setminus S$.
Observe now that for any $i \in [|X\setminus S|]$, the set of edges in $E_d(\{x_1,\dots, x_{i}\})\setminus E_d(\{x_1,.\dots, x_{i-1}\})$ form a path in $\Tree$ from some vertex $v_i$ to $x_i$.
Furthermore, as the incoming edge of $v_i$ is not in  $E_d(\{x_1,.\dots, x_{i}\})$, there is an outgoing edge $e_i$ of $v_i$ not in  $E_d(\{x_1,.\dots, x_{i}\})$.
We may assume  $\hat c(e_i) \notin E_d(\{x_1,.\dots, x_{i}\})$, either because $e_i \in E_d(X\setminus S)$ (and $E_d(X\setminus S)$ uses each color at most once), or because $e_i \notin E_d(X\setminus S)$ and so $e_i$ is a sibling edge of a higest edge in $E_d(X\setminus S)$.

Taking the tuple $(x_i,v_i,e_i)$ for each $i \in [|X\setminus S|]$ gives us an \emph{anchored taxa set} that gives us all the information we need about $E_d(X\setminus S)$.
Formally speaking,
an \emph{anchored taxa set} $\mathcal A$ is a set of tuples~$(x,v,e)$, where~$x\in X$ is a taxon, $v$ is a strict ancestor of $x$ and $e$ is an outgoing edge of $v$ with $x\not\in \off(e)$.
We plan not to select taxa individually but to define an anchored set of taxa.
If a tuple~$(x,v,e)$ is selected, then we want to let taxon~$x$ go extinct and we also want to arrange that the diversity of the edges on the path from $v$ and $x$ gets lost, but the edge~$e$ outgoing of~$v$ has an offspring which gets saved, or will be selected to go extinct later.
In the rest of the section, whenever we refer to a tuple $(x,v,e)$, we always assume $x\in \off(v)$ and $e$ is an outgoing edge of $v$ with~$x\not\in \off(e)$.
We denote with $X(\mathcal A) := \{ x \mid (x,v,e) \in \mathcal A \}$ the taxa of an anchored taxa set.
For a phylogenetic $X$-tree \Tree, a taxon $x\in X$, and a vertex $v \in \anc(x)$, we denote $\Pvx$ to be the set of edges on the path from~$v$ and~$x$.
%Informally, one may think of $\Pvx$ as representing the set of edges we expect to add to $E_d(X(\mcA))$ when taxon $x$ is added to $X(\mcA)$. 

For an anchored taxa set $\mcA$,
we define two edge sets $E^+(\mcA) := \bigcup_{(x,v,e) \in \mcA} \Pvx$ and~$E_s(\mcA) := \{ e \mid (x,v,e) \in \mcA \}$.
Informally, $E^+(\mcA)$ is the set of edges that connect the taxa of the anchored taxa set with the anchors, and
the edges in $E_s(\mcA)$ are sibling-edges of the topmost edges of $\Pvx$ for each $(x,v,e) \in \mcA$.
These sibling-edges~$e$ may or may not be in $E_d(X(\mcA))$, depending on whether $e$ is part of $P_{u,y}$ for some other tuple~$(y,u,e') \in \mcA$.

A set of edges $E'$ has \emph{unique colors} (or \emph{unique key-colors}), if the sets~$c(e_1)$ and~$c(e_2)$ (respectively, $\{\hat c(e_1)\}$ and $\{\hat c(e_2)\}$)
are disjoint for~$e_1,e_2\in E'$, with~$e_1 \ne e_2$.\lb
An anchored taxa set $\mcA$ has a \emph{valid ordering $(x_1,v_1,e_1),\dots,(x_{|\mcA|},v_{|\mcA|},e_{|\mcA|})$}\lb if $(x_i,v_i,e_i) \in \mcA$, $\ex(x_i) \le \ex(x_j)$ and $\hat c(e_j) \not\in c(E^+(\mcA_j))$ for each pair $i,j\in [|\mcA|]$, with $i\le j$.
With a valid ordering we want to enforce that these tuples are selected in an order such that when letting taxa~$x_1, \dots, x_{|\mcA|}$ go extinct then at most~$\w(E^+(\mcA))$ diversity is lost.
We define $\mcA_j := \{ t_i \in \mcA \mid i\le j \}$.
An anchored taxa set $\mcA$ is \emph{color-respectful} if 
\begin{enumerate}
	\itemsep-.35em
	\item[CR a)]\label{it:Ca} $E^+(\mcA)$ has unique colors,
	\item[CR b)]\label{it:Cb} $E_s(\mcA)$ has unique key-colors,
	\item[CR c)]\label{it:Cc} $E^+(\mcA)$ and $E_{> \Dbar}$ are disjoint,
	\item[CR d)]\label{it:Cd} $\mcA$ has a valid ordering, and
	\item[CR e)]\label{it:Ce} $\Pvx$ and $P_{u,y}$ are disjoint for any tuples $(x,v,e),(y,u,e') \in \mcA$.
\end{enumerate}

The existence of a color-respectful anchored taxa set will be used to show that an instance of \tPDs is a \yes-instance. To formally show this, we first define a colored version of the problem.

\problemdef{\cBartPDslong\Dbar (\cBartPDs\Dbar)}{
	An extinction-$\Dbar$-colored $X$-tree $\Tree = (V,E,\w,\hat c,c^-)$,
	integers $\ell(x)$ and $\ex(x)$ for every taxon $x\in X$,
	and a set of teams $T$}{
	Is there an anchored taxa set $\mcA$ such that $\mcA$ is color-respectful, $|c(E^+(\mcA))| \le \Dbar$ and there is a valid $T$-schedule saving $X\setminus X(\mcA)$}

The following lemma shows how \cBartPDs\Dbar is relevant to \tPDs.

\begin{lemma}\label{lem:timePD-coloredYes}
	Let $\Instance' = (\Tree', \ex, \ell, T, D)$ with $\Tree' = (V,E,\w,\hat c,c^-)$, be an instance of \cBartPDs\Dbar
	and let $\Instance = (\Tree, \ex, \ell, T, D)$ with $\Tree = (V,E,\w)$ be the instance of \tPDs induced by $\Instance'$ when omitting the coloring.
	If $\Instance'$ is a \yes-instance of \cBartPDs\Dbar, then $\Instance$ is a \yes-instance of \tPDs. 
	\end{lemma}
\begin{proof}
	Let $\mcA$ be a solution for $\Instance'$. As there exists a $T$-schedule saving the set of taxa~$S := X\setminus X(\mcA)$, it is sufficient to show that $\PD(S) \ge D$.
	
	Since every edge in $\Tree$ either has an offspring in $S$ or is in $E_d(X\setminus S)$, it follows that $\w(E_d(X\setminus S)) = \PD(X) - \PD(S)$.
	Thus, in order to show $\PD(S) \ge D$ it remains to show that $\w(E_d(X\setminus S)) \leq \PD(X) - D = \Dbar$.
	
	We prove that $E_d(X(\mcA_j))\subseteq E^+(\mcA_j)$ for each $j \in [|\mcA|]$ by an induction on $j$.
	
	For the base case $j= 1$, we have that $X(\mcA_1) = \{x_1\}$ and $E_d(\{x_1\})$ consists of the single incoming edge of $x$.
	As this edge is in $P_{v_1,x_1}$, the claim is satisfied.
	
	For the induction step, assume $j >1$ and $E_d(X(\mcA_{j-1}))\subseteq E^+(\mcA_{j-1})$.
	It suffices to show $E_d(X(\mcA_j)) \setminus E_d(X(\mcA_{j-1})) \subseteq P_{v_j,x_j}$,
	as $E^+(\mcA_j) =  E^+(\mcA_{j-1}) \cup P_{v_j,x_j}$.
	To see this, consider the edge $e_j$.
	Because $\mcA$ has a valid ordering, $\hat c(e_j) \not\in c(E^+(\mcA_j))$, which implies that $e_j$ is not in~$E^+(\mcA_j)$.
	By the inductive hypothesis, we conclude~$E_d(X(\mcA_{j-1}))\subseteq E^+(\mcA_{j-1}) \subseteq E^+(\mcA_j)$ and so $e_j \notin E_d(X(\mcA_{j-1}))$.
	Thus, $e_j$ has an offspring $z$ which is not in $X(\mcA_{j-1})$.
	By definition, $x_j \notin \off(e_j)$.
	So, we conclude that $z \notin X(\mcA_{j-1}) \cup \{x_j\} = X(\mcA_j)$. 
	Furthermore, $z \in \off(e')$ for any edge $e'$ incoming at an ancestor of $v_j$, and so $e' \notin E_d(X(\mcA_j))$ for any such edge $e'$.
	It follows that the only edges in $E_d(X(\mcA_j)) \setminus E_d(X(\mcA_{j-1}))$ must be between $v_j$ and $x_j$, that is, in $P_{v_j,x_j}$.
	We conclude that $E_d(X(\mcA_j))\subseteq E_d(X(\mcA_{j-1}))\cup P_{v_j,x_j} \subseteq E^+(\mcA_{j-1}) \cup P_{v_j,x_j} = E^+(\mcA_j)$.
	
	Letting $j = |\mcA|$, we have $E_d(X(\mcA)) \subseteq E^+(\mcA)$.
	Thereby, we can conclude that~$\w(E_d(X(\mcA))) \leq \w(E^+(\mcA)) = \sum_{e \in E^+(\mcA)}|c(e)| = |c(E^+(\mcA))|$, where the last equality holds because  $E^+(\mcA)$ has unique colors.
	We have $\w(E_d(X(\mcA))) \leq \Dbar$ since~$|c(E^+(\mcA))|\leq \Dbar$ and so $\PD(S) \ge D$, as required.
	%	\qed
\end{proof}

\medskip

In the following,
we continue to define a few more things.
Then, we present Algorithm~$(\Dbar)$ for solving \cBartPDs\Dbar.
Afterward, we analyze the correctness and the running time of Algorithm~$(\Dbar)$.
And finally, we show how to reduce instances of \tPDs to instances of \cBartPDs\Dbar to conclude the desired \FPT-result.
\medskip

Define $\Hbar{p} := \ell(Z_p) - H_p$ to be the number of person-hours one would need additionally to be able to save all taxa in $Z_p$ (when not regarding the time constraints~$\ex_1,\dots,\ex_{p-1}$).
A set of taxa $A\subseteq X$ is \emph{ex-$q$-compatible} for $q\in[\var_{\ex}]$ if~$\ell(A\cap Z_p) \ge \Hbar{p}$ for each $p\in [q-1]$. 
Note that for there to be a valid $T$-schedule saving $S\subseteq X$, it must be satisfied that $\ell(S\cap Z_p) \leq H_p$, and so~$\ell((X\setminus S)\cap Z_p) \geq \Hbar{p}$, for each $p \in[\var_{\ex}]$.
Thus, if $A$ is ex-$q$-compatible then there is a valid $T$-schedule saving~$(X\setminus A)\cap Z_{q-1}$.
Observe that this schedule does not necessarily save~$(X\setminus A)\cap Z_{q}$.

In our dynamic programming algorithm, we need to track the existence of color-respectful anchored taxa sets using particular sets of colors. To this end, we define the notion of \emph{ex-$(C_1,C_2,q)$-feasibility}.
For sets of colors $C_1,C_2\subseteq [2\Dbar]$ and an\lb integer $q\in[\var_{\ex}]$, an anchored taxa set $\mathcal A$ is \emph{ex-$(C_1,C_2,q)$-feasible} if 
\begin{enumerate}
	\itemsep-.35em
	\item[F a)]\label{it:Fa} $\mcA$ is color-respectful,
	\item[F b)]\label{it:Fb} $c(E^+(\mcA))$ is a subset of $C_1$,
	\item[F c)]\label{it:Fc} $\hat c(e)$ is in $C_2 \cup c(E^+(\mcA \setminus \{(x,v,e)\}))$ for each $(x,v,e)\in \mcA$,
	\item[F d)]\label{it:Fd} $X(\mcA)$ is a subset of $Z_q$, and
	\item[F e)]\label{it:Fe} $X(\mcA)$ is ex-$q$-compatible.
\end{enumerate}

Intuitively, we want anchored taxa sets~$\mcA$ to be ex-$(C_1,C_2,q)$-feasible if it is possible to save~$X \setminus X(\mcA)$ and only lose the colors in~$C_1$ and have the colors in~$C_2$ available for later actions.
As an example, observe that the anchored taxa set $\mathcal A$ in Figure~\ref{fig:timePD-example-Dbar-definitions} is ex-$(C_1,C_2,q)$-feasible for $C_1 = [11]$, $C_2 = \{12\}$, and $q=3$ if and only if $\mathcal A$ is $q$-compatible (which we don't know as the teams are not given). This would be the case if and only if~$\Hbar{1} \ge 10$, $\Hbar{2} \ge 22$, and $\Hbar{3} \ge 35$.

\medskip

Our algorithm calculates for each combination $C_1$, $C_2$, and $q$ the maximum rescue length of the taxa of an ex-$(C_1,C_2,q)$-feasible anchored taxa set.
In order to calculate these values recursively, we declare which tuples $(x,v,e)$ are suitable for consideration in the recurrence of the algorithm.

A tuple $(x,v,e)$ is \emph{$(C_1,C_2)$-good} for sets of colors $C_1,C_2\subseteq [2\Dbar]$ if
\begin{enumerate}
	\itemsep-.35em
	\item[G a)]\label{it:Ga} \Pvx has unique colors,
	\item[G b)]\label{it:Gb} $c(\Pvx)$ is a subset of $C_1$, 
	\item[G c)]\label{it:Gc} $\hat{c}(e)$ is in $C_2$, and 
	\item[G d)]\label{it:Gd} $\Pvx$ and $E_{> \Dbar}$ are disjoint.
\end{enumerate}

Let $\mathcal X_{(q)}$ be the set of tuples $(x,v,e)$ such that $x\in Z_q$, $\Pvx$ has unique colors, $\Pvx \subseteq E_{\le \Dbar}$ and $\hat c(e) \not\in c(\Pvx)$.
Disjoint sets of colors $C_1,C_2\subseteq [2\Dbar]$ are \emph{$q$-grounding},
if $c(\Pvx) \not\subseteq C_1$ or $\hat c(e) \not\in C_2$ for every tuple $(x,v,e)\in \mathcal X_{(q)}$.
In Figure~\ref{fig:timePD-example-Dbar-definitions} the set~$\mathcal X_{(1)}$ would be $\{(x_1,v_1,e_5),(x_1,v_0,e_3)\}$---because both paths $P_{v_0,x_4}$ and $P_{v_1,x_4}$ contain the edge~$e_7\in E_{>\Dbar}$, and $\hat c(e_2) = 9 \in c(e_4)$ so $(x_1,v_0,e_2)$ is not contained in~$\mathcal X_{(1)}$.
Therefore, the 1-grounding colors are any disjoint sets of colors $C_1$ and $C_2$ with
($c(e_4) \not\subseteq C_1$ or $\hat c(e_5) = 6 \not\in C_2$)
and
($c(\{e_1,e_4\}) \not\subseteq C_1$ or $\hat c(e_3) = 12 \not\in C_2$).
One non-trivial example are the sets $C_1 = \{6,7,\dots,11\}$ and~$C_2 = [5] \cup \{12\}$.
The empty sets~$C_1$ and~$C_2$ are~$q$-grounding for each~$q$.

\medskip

\medskip
\noindent\textbf{Algorithm~$\mathbf{(\Dbar)}$.}\\
Let \Instance be an instance of \cBartPDs\Dbar.
In the following, we define a dynamic programming algorithm with table~$\DP$.
For all disjoint sets of colors $C_1, C_2 \subseteq [2\Dbar]$\lb with~$|C_1|\leq \Dbar$, and all $q \in [\var_{\ex}]$,
we compute a value $\DP[C_1,C_2,q]$.
The entries\lb of $\DP$ are computed 
in some order such that $\DP[C_1',C_2',q']$ is computed\lb before $\DP[C_1'',C_2'',q'']$ if $C_1' \subsetneq C_1''$.
We want $\DP[C_1,C_2,q]$ to store the \emph{maximum} rescue length $\ell(X(\mcA))$ of an anchored taxa set $\mcA$ that is  \mbox{ex-$(C_1,C_2,q)$-feasible}.
If there exists no ex-$(C_1,C_2,q)$-feasible anchored taxa set, we want $\DP[C_1,C_2,q]$ to store $-\infty$.
\smallskip

As a base case,
given $q\in [\var_{\ex}]$ and $q$-grounding sets of colors $C_1$ and $C_2$.
If~$\Hbar{p} \le 0$ for each $p \in [q-1]$, we store $\DP[C_1,C_2,q] = 0$.
Otherwise, if $C_1$ and $C_2$ are $q$-grounding but $\Hbar{p} > 0$ for any $p \in [q-1]$, we store $\DP[C_1,C_2,q] = -\infty$.

To compute further values, we use the recurrence
\begin{equation}
	\label{eqn:split-Dbar}
	\DP[C_1,C_2,q] =
	\max_{(x,v,e)} \DP[C_1',C_2',\ex^*(x)] + \ell(x).
\end{equation}

% Here $C_1' := (C_1\setminus c(\Pvx)) \setminus \{\hat c(\exv)\}$ and $C_2' := (C_2\cup c(\Pvx)) \setminus \{\hat c(\exv)\}$.
Here, $C_1' := C_1\setminus c(\Pvx)$ and $C_2' := (C_2\cup c(\Pvx)) \setminus \{\hat c(e)\}$.
The maximum in \Recc{eqn:split-Dbar} is taken over $(C_1,C_2)$-good tuples $(x,v,e)$ satisfying $x\in Z_q$\linebreak
and~$\DP[C_1',C_2',\ex^*(x)] + \ell(x) \ge \max\{\Hbar{\ex^*(x)},\Hbar{\ex^*(x)+1},\dots,\Hbar{q-1}\}$.
If no such tuple exists, we set $\DP[C_1,C_2,q] = -\infty$.
Recall that $\ex^*(x) = q$ for each $x\in Y_q$.

We return \yes if $\DP[C_1,C_2,\var_{\ex}]$ stores at least~$\Hbar{\var_{\ex}}$ for some disjoint sets of colors~$C_1,C_2\subseteq [2\Dbar]$ with $|C_1| \leq \Dbar$.
Otherwise, if $\DP[C_1,C_2,\var_{\ex}] \ge \Hbar{\var_{\ex}}$ for all disjoint sets of colors~$C_1$ and $C_2$, we return \no.

\smallskip
We continue to analyze the algorithm by showing that $\DP[C_1,C_2,q]$ stores the largest rescue length $\ell(X(\mcA))$ of an ex-$(C_1,C_2,q)$-feasible anchored taxa set~$\mcA$.
We start with the base case.
\begin{lemma}
	\label{lem:timePD-basic-case}
	For each $q\in [\var_{\ex}]$ and all $q$-grounding sets of colors~$C_1$ and $C_2$,
	$\DP[C_1,C_2,q]$ stores
	the largest rescue length $\ell(X(\mcA))$ of an ex-$(C_1,C_2,q)$-feasible anchored taxa set $\mcA$ if such an $\mcA$ exists, and $-\infty$, otherwise.
\end{lemma}
\begin{proof}
	Let $q\in [\var_{\ex}]$ and 
	let $C_1$ and $C_2$ be $q$-grounding sets of colors,
	and suppose for a contradiction that some non-empty anchored taxa set $\mcA$ is ex-$(C_1,C_2,q)$-feasible.
	Let $(x,v,e)$ be the last tuple in a valid ordering of $\mcA$.
	Such an ordering exists because $\mcA$ is color-respectful.
	Observe that $\hat c(e) \not\in c(E^+(\mcA))$, and in\lb particular~$\hat c(e) \not\in c(\Pvx)$.
	As $\mcA$ satisfies CR a), CR c) and F d), we also have that~$c(\Pvx)$ is uniquely colored, $\Pvx \subseteq E_{\le \Dbar}$ and $x\in Z_q$.
	Thus, $(x,v,e)$ is in~$\mathcal X_{(q)}$.
	Because $\mcA$ satisfies F b) we have $c(\Pvx) \subseteq C_1$,
	and because $\mcA$ satisfies F c) we\lb have~$\hat C_2 \cup c(e) \in c(E^+(\mcA))$, which together with $\hat c(e) \not\in c(E^+(\mcA))$ implies $\hat c(e) \in C_2$.
	But then $(x,v,e)$ is a tuple in $\mathcal X_{(q)}$ with~$c(\Pvx) \subseteq C_1$ and $\hat c(e) \in C_2$, a contradiction to the assumption that $C_1$ and $C_2$ are $q$-grounding.	 
	It follows that no non-empty anchored taxa set is ex-$(C_1,C_2,q)$-feasible.
	
	Observe that the empty set is ex-$(C_1,C_2,q)$-feasible if and only if
	the empty set is ex-$q$-compatible, as the empty set trivially satisfies all other conditions.\linebreak
	Since $\ell(\emptyset\cap Z_p)=\ell(\emptyset)=0$ for each $p\in [\var_{\ex}]$,
	we observe that the empty set is ex-$q$-compatible if and only if $\Hbar{p} \le 0$ for each $p\in [q-1]$.
	Exactly in these cases, entry $\DP[C_1,C_2,q]$ stores at least~$\ell(\emptyset) = 0$.
	%	\qed
\end{proof}

In the following, we show an induction over $|C_1|$.
Observe that $C_1=\emptyset$ and any set of colors $C_2\subseteq [2\Dbar]$ are $q$-grounding for each $q \in [\var_{\ex}]$.
Therefore, Lemma~\ref{lem:timePD-basic-case} is the base case of the induction.

As an induction hypothesis, we assume that for a fixed set of colors $C_1\subseteq [2\Dbar]$ and a fixed $q\in [\var_{\ex}]$, for each $K_1\subsetneq C_1$, $K_2 \subseteq [2\Dbar]\setminus K_1$, and $p \in [q]$, entry~$\DP[K_1,K_2,p]$ stores the largest rescue length $\ell(X(\mcA))$ of an ex-$(K_1,K_2,p)$-feasible anchored taxa set $\mcA$.

In Lemma~\ref{lem:timePD-corr-recurr-A} and Lemma~\ref{lem:timePD-corr-recurr-B}, we proceed to show that with this hypothesis we can conclude that $\DP[C_1,C_2,p]$ stores the desired value.

\begin{lemma}
	\label{lem:timePD-corr-recurr-A}
	If an anchored taxa set $\mcA$ is ex-$(C_1,C_2,q)$-feasible for disjoint sets of colors~$C_1$ and $C_2$ that are not $q$-grounding, then $\DP[C_1,C_2,q]\ge \ell(X(\mcA))$.
\end{lemma}
\begin{proof}
	Let $C_1$ and $C_2$ be disjoint sets of colors that are not $q$-grounding.
	Furthermore, let $\mcA$ be an ex-$(C_1,C_2,q)$-feasible anchored taxa set.
	
	We prove the claim by an induction on $|\mcA|$.
	First, consider $\mcA$ being empty.
	Since~$C_1$ and $C_2$ are not $q$-grounding, there is a tuple $(x,v,e)\in \mathcal X_{(q)}$ such\lb that~$\Pvx \subseteq C_1$ and $\hat c(e) \in C_2$.
	Because $(x,v,e) \in \mathcal X_{(q)}$, the anchored taxa\lb set $\mcA' := \{(x,v,e)\}$ is color-respectful, and by construction, $\mcA$ satisfies F~b), F~c), and~F~d).
	It only remains to show that $\{x\}$ is ex-$q$-compatible to prove that $\mcA'$ is ex-$(C_1,C_2,q)$-feasible.
	Since $\mcA = \emptyset$ is ex-$(C_1,C_2,q)$-feasible, $\ell(x) > \ell(\emptyset) \ge \Hbar{p}$ for each $p\in [q-1]$.
	Thus, $\mcA'$ is ex-$(C_1,C_2,q)$-feasible.
	This is sufficient to show that~$\DP[C_1,C_2,q] \ge \ell(x)$.
	So, we may assume that $\mcA$ is not empty.
	
	Now, let $(x,v,e)$ be the last tuple in a valid ordering of $\mcA$.
	Such an ordering exists because $\mcA$ is color-respectful.
	Define an anchored taxa set $\mcA'$ resulting from removing $(x,v,e)$ from \mcA.
	
	\medskip
	We first show that $\mcA'$ is ex-$(C_1',C_2',\ex^*(x))$-feasible.
	\proofpart{\hyperref[it:Fa]{F a) [$\mcA$ is color-respectful]}}
	Because $\mcA$ is color-respectful and $x$ is the taxon of the largest order, we directly conclude that $\mcA'$ satisfies all properties of color-respectfulness.
	\proofpart{\hyperref[it:Fb]{F b) [$c(E^+(\mcA)) \subseteq C_1$]}}
	Because $E^+(\mcA') = E^+(\mcA) \setminus \Pvx$ and $E^+(\mcA)$ has unique colors, we conclude\lb
	that $c(E^+(\mcA')) = c(E^+(\mcA)) \setminus c(\Pvx) \subseteq C_1 \setminus c(\Pvx)= C_1'$.
	\proofpart{\hyperref[it:Fc]{F c) [$\hat c(e) \in C_2 \cup c(E^+(\mcA \setminus \{(x,v,e)\}))$ for each $(x,v,e)\in \mcA$]}}
	Fix a tuple $(y,u,e')\in \mcA'$.
	We want to show that $\hat c(e')$ is in
	\begin{eqnarray*}
		&& C_2' \cup c(E^+(\mcA'\setminus \{(y,u,e')\}))\\
		&=& ((C_2 \cup c(\Pvx)) \setminus \{\hat c(e)\}) \cup (c(E^+(\mcA\setminus \{(y,u,e')\})) \setminus c(\Pvx))\\
		&=& (C_2 \cup c(E^+(\mcA\setminus \{(y,u,e')\})) \setminus \{\hat c(e)\}
	\end{eqnarray*}
	It holds $\hat c(e') \in C_2 \cup c(E^+(\mcA\setminus \{(y,u,e')\}))$, because $\mcA$ satisfies F c).
	Since $E_s(\mcA)$ has unique key-colors, $\hat c(e') \ne \hat c(e)$.
	Therefore, $\hat c(e') \in C_2' \cup c(E^+(\mcA'\setminus \{(y,u,e')\}))$.
	\proofpart{\hyperref[it:Fd]{F d) [$X(\mcA) \subseteq Z_q$]}}
	$X(\mcA')\subseteq X(\mcA)\subseteq Z_{\ex^*(x)}$ by the definition of $x$.
	\proofpart{\hyperref[it:Fe]{F e) [$X(\mcA)$ is ex-$q$-compatible]}}
	Because $X(\mcA)$ is ex-$q$-compatible, $\ell(X(\mcA') \cap Z_p) = \ell(X(\mcA) \cap Z_p) \ge \Hbar{p}$ for\lb each~$p < \ex^*(x)$.
	Consequently, $X(\mcA')$ is ex-$\ex^*(x)$-compatible.

	\medskip
	Next, we show that $(x,v,e)$ is $(C_1,C_2)$-good.
	\proofpart{\hyperref[it:Ga]{G a) [\Pvx has unique colors]}}
	Because $E^+(\mcA)$ has unique colors, $\Pvx \subseteq E^+(\mcA)$ has unique colors, too.
	\proofpart{\hyperref[it:Gb]{G b) [$c(\Pvx) \subseteq C_1$]}}
	$c(\Pvx) \subseteq c(E^+(\mcA)) \subseteq C_1$ because $\mcA$ satisfies F b).
	\proofpart{\hyperref[it:Gc]{G c) [$\hat{c}(e) \in C_2$]}}
	Because $\mcA$ satisfies F c) and $X(\mcA)$ has a valid ordering,
	$\hat c(e) \in C_2 \cup c(E^+(\mcA'))$ and
	$\hat c(e) \not\in c(E^+(\mcA))$.
	Therefore, especially $\hat c(e) \not\in c(E^+(\mcA')) \subseteq c(E^+(\mcA))$ and we conclude that
	$\hat c(e) \in C_2$.
	\proofpart{\hyperref[it:Gd]{G d) [$\Pvx$ and $E_{> \Dbar}$ are disjoint]}}
	Because $E^+(\mcA)$ has an empty intersection with $E_{> \Dbar}$, also $\Pvx \subseteq E^+(\mcA)$ has empty intersection with $E_{> \Dbar}$.
	
	\medskip
	Since $\mcA'$ is ex-$(C_1',C_2',\ex^*(x))$-feasible and $|\mcA'| < |\mcA|$, we have by the inductive hypothesis that $\DP[C_1',C_2',\ex^*(x)] \ge \ell(X(\mcA'))$.
	Furthermore $\DP[C_1',C_2',\ex^*(x)] + \ell(x) \ge \ell(X(\mcA')) + \ell(x) \ge \ell(X(\mcA))$. As $X(\mcA)$ is ex-$(C_1,C_2,q)$-compatible we have  $ \ell(X(\mcA) = \ell(X(\mcA)\cap Z_p)  \geq \Hbar{p}$ for each $p$ with $\ex^*(x) \le p < q$, and so $\DP[C_1',C_2',\ex^*(x)] + \ell(x)  \ge \max\{\Hbar{\ex^*(x)},\Hbar{\ex^*(x)+1},\dots,\Hbar{q-1}\}$.
	This together with the fact that $(x,v,e)$ is $(C_1,C_2)$-good implies that $(x,v)$ satisfies the conditions of Recurrence~(\ref{eqn:split-Dbar}), from which we can conclude that $\DP[C_1,C_2,q] \geq \DP[C_1',C_2',\ex^*(x)] + \ell(x) \geq \ell(X(\mcA))$.
	%	\qed
\end{proof}

\begin{lemma}
	\label{lem:timePD-corr-recurr-B}
	If $\DP[C_1,C_2,q] = a \ge 0$ for some disjoint sets of colors $C_1$ and $C_2$, then there is an ex-$(C_1,C_2,q)$-feasible anchored taxa set $\mcA$ with $\ell(X(\mcA)) = a$.
\end{lemma}
\begin{proof}
	Assume that $\DP[C_1,C_2,q] = 0$.
	% 	Then $\DP[C_1,C_2,q]$ needs to be in the base case as $\ell(x)>0$ for each $x\in X$ and so after an application of the recurrence, $\DP[C_1,C_2,q]$ would be positive or $-\infty$.
	Then, $C_1$ and $C_2$ must be $q$-grounding, as otherwise the algorithm would apply Recurrence (\ref{eqn:split-Dbar}), and then~$\DP[C_1,C_2,q]$ would store either $-\infty$ or a value which is at least $\ell(x) > 0$ for some $x\in X$.
	Lemma~\ref{lem:timePD-basic-case} shows the correctness of this case.
	
	Now, assume that $\DP[C_1,C_2,q] = a > 0$.
	By Recurrence (\ref{eqn:split-Dbar}), there is a $(C_1,C_2)$-good tuple $(x,v,e)$ such that $x\in Z_q$ and $\DP[C_1,C_2,q] = \DP' + \ell(x)$.
	Here, we define $\DP' := \DP[C_1',C_2',q']$ where (as in Recurrence (\ref{eqn:split-Dbar})) $C_1' := C_1\setminus c(\Pvx)$, $C_2' := (C_2\cup c(\Pvx)) \setminus \{\hat c(e)\}$, and $q' := \ex^*(x)$.
	Further, we know that $\DP' + \ell(x) \ge \max\{\Hbar{q'},\Hbar{q'+1},\dots,\Hbar{q-1}\}$.
	We conclude by the induction hypothesis that there is an ex-$(C_1',C_2',q')$-feasible anchored taxa set~$\mcA'$ such that $\ell(X(\mcA')) = \DP'$.
	Define $\mcA := \mcA' \cup \{(x,v,e)\}$.
	Clearly, $a = \DP' + \ell(x) = \ell(X(\mcA))$.
	It remains to show that $\mcA$ is an ex-$(C_1,C_2,q)$-feasible set.
	
	\medskip
	\proofpart{\hyperref[it:Fa]{F a) [$\mcA$ is color-respectful]}}
	We show that $\mcA$ is color-respectful.
	
	\proofpart{\hyperref[it:Ca]{CR a) [$E^+(\mcA)$ has unique colors]}}
	As $\mcA'$ satisfies CR a) and F b), $E^+(\mcA')$ has unique colors and $c(E^+(\mcA')) \subseteq C_1' = C_1\setminus c(\Pvx)$.
	Consequently, the colors of $E^+(\mcA')$ and $\Pvx$ are disjoint.
	We conclude with the $(C_1,C_2)$-goodness of $(x,v,e)$, that \Pvx has unique colors and so also $E^+(\mcA)$.
	\proofpart{\hyperref[it:Cb]{CR b) [$E_s(\mcA)$ has unique key-colors]}}
	As $\mcA'$ satisfies CR b), $E_s(\mcA')$ has unique key-colors.
	Observe $E_s(\mcA) = E_s(\mcA') \cup \{e\}$.
	Because $\mcA'$ satisfies F c), $\hat c(E_s(\mcA')) \subseteq C_2' \cup c(E^+(\mcA')) \subseteq C_1' \cup C_2'$.
	Further, because $\hat c(e) \in C_2$ and $C_1\cap C_2 = \emptyset$ we have  $\hat c(e) \notin C_1$, which implies
	$\hat c(e) \not\in C_1\cup (C_2 \setminus \{\hat c(e)\}) = C_1'\cup C_2'$. 
	Therefore, $E_s(\mcA)$ has unique key-colors.
	\proofpart{\hyperref[it:Cc]{CR c) [$E^+(\mcA)$ and $E_{> \Dbar}$ are disjoint]}}
	$E^+(\mcA')$ and $E_{>\Dbar}$ are disjoint because $\mcA'$ is color-respectful.
	Further, because $(x,v,e)$ is $(C_1,C_2)$-good, $\Pvx$ and $E_{>\Dbar}$ are disjoint and so $E^+(\mcA') \subseteq E_{\le \Dbar}$.
	\proofpart{\hyperref[it:Cd]{CR d) [$\mcA$ has a valid ordering]}}
	$\mcA'$ has a valid ordering $(y_1,u_1,e_1),\dots,(y_{|\mcA'|},u_{|\mcA'|},e_{|\mcA'|})$.
	As $X(\mcA')\subseteq Z_{\ex^*(x)}$ we conclude $\ex(y) \le \ex(x)$ for each $y\in A'$.
	Further, $\hat c(e) \in C_2$ and $c(\Pvx) \subseteq C_1$ because $(x,v,e)$ is $(C_1,C_2)$-good.
	Since $\mcA'$ satisfies F b), $c(E^+(A')) \subseteq C_1' \subseteq C_1$.
	Because $C_1$ and $C_2$ are disjoint, we conclude that $\hat c(e) \not \in c(E^+(\mcA))$.
	We conclude that\lb $(y_1,u_1,e_1),\dots,(y_{|\mcA'|},u_{|\mcA'|},e_{|\mcA'|}),(x,v,e)$ is a valid ordering of $\mcA$.
	\proofpart{\hyperref[it:Ce]{CR e) [$\Pvx$ and $P_{u,y}$ are disjoint for any $(x,v,e),(y,u,e') \in \mcA$]}}
	We know that $\mcA'$ satisfies CR e).
	Further, as already seen, $c(E^+(\mcA'))$ and $c(\Pvx)$ are disjoint and therefore also \Pvx and $P_{u,y}$ are disjoint for each $(y,u,e') \in \mcA'$.\linebreak 
	It follows that $\mcA$ is color-respectful.
	
	\medskip
	\proofpart{\hyperref[it:Fb]{F b) [$c(E^+(\mcA)) \subseteq C_1$]}}
	$c(E^+(\mcA)) = c(E^+(\mcA')) \cup c(\Pvx) \subseteq C_1' \cup c(\Pvx) = C_1$.
	\proofpart{\hyperref[it:Fc]{F c) [$\hat c(e) \in C_2 \cup c(E^+(\mcA \setminus \{(x,v,e)\}))$ for each $(x,v,e)\in \mcA$]}}
	Observe $\hat c(e) \in C_2$ because $(x,v,e)$ is $(C_1,C_2)$-good.
	Each $(y,u,e') \in \mcA'$ satisfies $\hat c(e') \subseteq C_2' \cup c(E^+(\mcA'\setminus \{(y,u,e')\}))$, as $\mcA'$ satisfies F c).
	Observe $x\ne y$\lb and $C_2' \subseteq C_2 \cup c(\Pvx)$ and therefore $\hat c(e') \subseteq C_2 \cup c(E^+(\mcA\setminus \{(y,u,e')\}))$.
	\proofpart{\hyperref[it:Fd]{F d) [$X(\mcA) \subseteq Z_q$]}}
	$X(\mcA)\subseteq Z_q$ because $X(\mcA')\subseteq Z_q$ and $x\in Z_q$.
	\proofpart{\hyperref[it:Fe]{F e) [$X(\mcA)$ is ex-$q$-compatible]}}
	Because $\mcA'$ is ex-$q'$-compatible, we conclude that $\ell(X(\mcA)\cap Z_p) = \ell(X(\mcA')\cap Z_p) \ge \Hbar{p}$ for each $p\in [q'-1]$.
	Then, with $\DP' + \ell(x) \ge \max\{\Hbar{q'},\Hbar{q'+1},\dots,\Hbar{q-1}\}$ we conclude that $\ell(X(\mcA)\cap Z_p) = \ell(X(\mcA)) \ge \Hbar{p}$ for each $p\in \{q',\dots,q-1\}$.
	Therefore, $X(\mcA)$ is ex-$q$-compatible and ex-$(C_1,C_2,q)$-feasible.
	%	\qed
\end{proof}

\begin{lemma}
	An instance  \Instance of \cBartPDs\Dbar is a \yes-instance if and only if disjoint sets of colors $C_1,C_2 \subseteq [2\Dbar]$ with $|C_1| \leq \Dbar$ and an  ex-$(C_1,C_2,\var_{\ex})$-feasible anchored taxa set $\mcA$ with $\ell(X(\mcA)) \ge \Hbar{\var_{\ex}}$ exist.
	% 	A set $S\subseteq X$ is a solution of an instance \Instance of \cBartPDs\Dbar if and only if
	% 	the set~$A := X\setminus S$ is ex-$(C_1,C_2,\var_{\ex})$-feasible with $\ell(A) \ge \Hbar{\var_{\ex}}$ for disjoint sets of colors $C_1,C_2 \subseteq [2\Dbar]$.
\end{lemma}
\begin{proof}
	Suppose first that \Instance  is a \yes-instance of \cBartPDs\Dbar. That is, there exists a color-respectful anchored taxa set $\mcA$ with $|c(E^+(\mcA))| \leq \Dbar$ such that there is a valid $T$-schedule saving $S := X\setminus X(\mcA)$.
	%Let a set of taxa $S\subseteq X$ be a solution for \Instance and define $A := X\setminus S$.
	Let $C_1 := c(E^+(\mcA))$
	and let $C_2$ be the key-colors of the edges $E_s(\mcA)$ that are not in $C_1$.
	We show that $\mcA$ is ex-$(C_1,C_2,\var_{\ex})$-feasible.
	First, $\mcA$ satisfies F a), F b), and F d) by definition.
	\proofpart{\hyperref[it:Fc]{F c) [$\hat c(e) \in C_2 \cup c(E^+(\mcA \setminus \{(x,v,e)\}))$ for each $(x,v,e)\in \mcA$]}}
	By the definition, $\hat c( E(\mcA)) \subseteq C_2 \cup C_1 = C_2 \cup c(E^+(\mcA))$.
	Fix a tuple $(x,v,e) \in \mcA$.
	Because $\mcA$ has a valid ordering, we conclude that $\hat c(e)\not\in c(\Pvx)$.
	It follows that if~$\hat c(e)$ is in $c(E^+(\mcA))$, then explicitly $\hat c(e) \in c(E^+(\mcA\setminus \{(x,v,e)\}))$.
	Therefore, $\mcA$ satisfies~F~c).
	\proofpart{\hyperref[it:Fe]{F e) [$X(\mcA)$ is ex-$q$-compatible]}}
	Fix some $q\in [\var_{\ex}]$.
	Since $S$ has a valid schedule, we conclude $\ell(S\cap Z_q) \le H_q$ for each $q\in [\var_{\ex}]$.
	Consequently, $\ell(X(\mcA)\cap Z_q) = \ell(Z_q) - \ell(S\cap Z_q) \ge \Hbar{q}$.
	Likewise with $\ell(S) \le H_{\var_{\ex}}$, we conclude that $\ell(X(\mcA)) \ge \Hbar{\var_{\ex}}$.
	
	\medskip
	For the converse, suppose that $\mcA$ is an ex-$(C_1,C_2,\var_{\ex})$-feasible anchored taxa set with $\ell(X(\mcA)) \ge \Hbar{\var_{\ex}}$ for disjoint sets of colors $C_1,C_2 \subseteq [2\Dbar]$ with $|C_1|\leq \Dbar$.\lb
	Observe that the rescue time of all taxa in $S=X\setminus X(\mcA)$ that are in $Z_q$\lb is $\ell(S \cap Z_q) = \ell((X\setminus X(\mcA)) \cap Z_q) = \ell(Z_q) - \ell(X(\mcA)\cap Z_q)$.
	Because $\mcA$ is\lb ex-$\var_{\ex}$-compatible, $\ell(X(\mcA)\cap Z_q) \ge \Hbar{q}$ and so
	$\ell(S \cap Z_q) \le \ell(Z_q) - \Hbar{q} = H_q$. 
	Then, by Lemma~\ref{lem:timePD-scheduleCondition}, there is a valid $T$-schedule saving $S$.
	By definition, $\mcA$ is color-respectful
	and $|c(E^+(\mcA))| \leq |C_1| \leq \Dbar$.
	% 	Therefore $S$ is a solution for the instance \Instance.0
	Therefore, \Instance  is a \yes-instance of \cBartPDs\Dbar
	%	\qed
\end{proof}

\begin{lemma}
	\label{lem:timePD-rt-DBar}
	Algorithm~$(\Dbar)$ computes solutions for instances of \cBartPDs\Dbar in time $\Oh^*(9^\Dbar \cdot \Dbar)$.
\end{lemma}
\begin{proof}
	Since $\DP[C_1,C_2,q]$ is computed for $C_1$ and $C_2$ disjoint subsets of $[2\Dbar]$,
	the table $\DP$ contains $3^{2\Dbar} \cdot \var_{\ex}$ entries.
	
	For given $q\in [\var_{\ex}]$ and color sets~$C_1$ and~$C_2$, we can compute whether~$C_1$ and~$C_2$ are $q$-grounding by testing the conditions for each tuple $(x,v,e)$.
	Therefore, the test can be done in $\Oh(\Dbar \cdot n^3)$ time.
	Here, the $\Dbar$-factor in the running time comes from testing whether $\Pvx \subseteq C_1$.
	
	Analogously, for a given $q\in [\var_{\ex}]$ and sets of colors $C_1$ and $C_2$, we can compute whether a tuple $(x,v,e)$ satisfies the conditions of Recurrence~(\ref{eqn:split-Dbar}) in 
	$\Oh(\Dbar \cdot n)$ time in our RAM model.
	We, however, want to note that we add two numbers of size~$\max\{\Hbar{1}, \Hbar{2}, \dots, \Hbar{\var_{\ex}}\}$.
\end{proof}

\begin{figure}[t]
	\centering
	\begin{tikzpicture}[scale=1,every node/.style={scale=0.9}]
		\tikzstyle{txt}=[circle,fill=white,draw=white,inner sep=0pt]
		\tikzstyle{nde}=[circle,fill=black,draw=black,inner sep=2.5pt]
		
		\node[nde] (v0) at (5,10) {};
		\node[nde] (v1) at (2,9) {};
		\node[nde] (v2) at (5,9) {};
		\node[nde] (v3) at (8,9) {};
		\node[nde] (v4) at (1.6,8) {};
		\node[nde] (v5) at (2.4,8) {};
		\node[nde] (v6) at (4.2,8) {};
		\node[nde] (v7) at (5,8) {};
		\node[nde] (v7') at (5.8,8) {};
		\node[nde] (v8) at (7.6,8) {};
		\node[nde] (v9) at (8.4,8) {};
		
		\node[txt,xshift=9mm,yshift=-10mm] (c1) [above=of v1] {$e_3$};
		
		\node[txt,xshift=-1mm,yshift=-10mm] (c6) [above=of v6] {$e_1$};
		\node[txt,xshift=1mm,yshift=-10mm] (c8) [above=of v7'] {$e_2$};
		\node[txt,xshift=-2mm,yshift=-10mm] (c8) [above=of v8] {$e_4$};
		
		\node[txt,yshift=9mm] (r4) [below=of v4] {$15$};
		\node[txt,yshift=9mm] (r5) [below=of v5] {$30$};
		\node[txt,yshift=9mm] (r6) [below=of v6] {$25$};
		\node[txt,yshift=9mm] (r7) [below=of v7] {$15$};
		\node[txt,yshift=9mm] (r7) [below=of v7'] {$30$};
		\node[txt,yshift=9mm] (r8) [below=of v8] {$25$};
		\node[txt,yshift=9mm] (r9) [below=of v9] {$30$};
		\node[txt,xshift=3mm] (r'4) [left=of r4] {$\ex(x)$};
		
		\node[txt,xshift=-9mm] (t0) [right=of v0] {$v_3$};
		\node[txt,xshift=9mm] (t1) [left=of v1] {};
		\node[txt,xshift=-9mm] (t2) [right=of v2] {$v_1=v_2$};
		\node[txt,xshift=-9mm] (t3) [right=of v3] {$v_4$};
		
		\node[txt,xshift=10mm] (t6) [left=of v6] {$x_2$};
		\node[txt,xshift=-11mm] (t7) [right=of v7] {$x_1$};
		\node[txt,xshift=-10mm] (t7) [right=of v7'] {$x_3$};
		\node[txt,xshift=-9mm] (t9) [right=of v9] {$x_4$};
		\textbf{}
		\draw[thick,arrows = {-Stealth[red,length=8pt]}] (v0) to (v1);
		\draw[thick,dashed,blue,arrows = {-Stealth[length=8pt]}] (v0) to (v2);
		\draw[thick,arrows = {-Stealth[length=8pt]}] (v0) to (v3);
		\draw[thick,arrows = {-Stealth[length=8pt]}] (v1) to (v4);
		\draw[thick,arrows = {-Stealth[length=8pt]}] (v1) to (v5);
		\draw[thick,dashed,blue,arrows = {-Stealth[red,length=8pt]}] (v2) to (v6);
		\draw[thick,dashed,blue,arrows = {-Stealth[length=8pt]}] (v2) to (v7);
		\draw[thick,dashed,blue,arrows = {-Stealth[red,length=8pt]}] (v2) to (v7');
		\draw[thick,arrows = {-Stealth[red,length=8pt]}] (v3) to (v8);
		\draw[thick,dashed,blue,arrows = {-Stealth[length=8pt]}] (v3) to (v9);
	\end{tikzpicture}
	\caption{
	Example of the construction of an anchored taxa set $\mcA$ (Theorem~\ref{thm:timePD-DBar}).
	Here given $X\setminus S = \{x_1,x_2,x_3,x_4\}$, the constructed $\mathcal A$ is $\{ (x_i,v_i,e_i) \mid i \in [4] \}$.
	Blue dashed edges are in $E_d(X\setminus S) = E^+(\mcA)$ and edges with red arrow are in $E_s(\mcA)$.
	$\mcA_1$ could have been~$\{ (x_1,v_1,e_2) \}$ provoking that $(x_1,v_1,e_2)$ must have been replaced in step $i=2$.
	}
	\label{fig:timePD-anchored-taxa-set}
\end{figure}
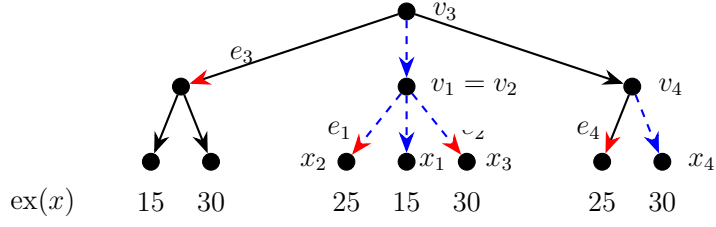%

\begin{theorem}
	\label{thm:timePD-DBar}
	\tPDs can be solved in
	$\Oh^*(2^{6.056 \cdot \Dbar + o(\Dbar)})$~time.
\end{theorem}

Just like in the previous section,
the key idea is that we construct a family of colorings $\mathcal{C}$ on the edges of $\Tree$, where each edge $e\in E$ is assigned a key-color and additionally for each $e\in E_{\le \Dbar}$ a subset $c^-(e)$ of $[2\Dbar]$ of size $\w(e)-1$.
Using these, we generate~$|\mathcal{C}|$ instances of \cBartPDs\Dbar, which with Algorithm~$(\Dbar)$ we solve in $\Oh^*(9^\Dbar \cdot \Dbar)$ time.
The colorings are constructed in such a manner that $\Instance$ is a \yes-instance if and only if at least one of the constructed \cBartPDs\Dbar instances is a \yes-instance.
Central for this is the concept of a perfect hash family as defined in Definition~\ref{def:perfectHashFamily}.

\begin{proof}[Proof of Theorem~\ref{thm:timePD-DBar}]
	\proofpara{Reduction}
	Let an instance $\Instance = (\Tree, \ex, \ell, T, D)$ of \tPDs be given.
	Let $|E_{\le \Dbar}| = m_1$ and $|E_{> \Dbar}| = m - m_1$.
	Order the edges $e_1, \dots, e_{m_1}$ of~$E_{\le \Dbar}$
	and the edges $e_{m_1+1}, \dots, e_{m}$ of~$E_{> \Dbar}$, arbitrarily.
	We define integers~$W_0 := m$ and~$W_j := m + \sum_{i=1}^{j} (\w(e_{i}) - 1)$ for each $j\in [m_1]$.
	
	Let $\mathcal{F}$ be a $(W_m, 2\Dbar)$-perfect hash family.
	Now, we define the family~$\mathcal{C}$ of colorings as follows.
	For every $f \in \mathcal{F}$,
	let $\hat c_f$ and $c^-_f$ be colorings such that $\hat c_f(e_j) = f(j)$ for each $e_j \in E(\Tree)$;
	and $c^-_f (e_j) := \{f(W_{j-1}+1), \dots, f(W_j)\}$ for each~$e_j \in E_{\le \Dbar}$.
	
	For each $c_f \in \mathcal{C}$,
	let $\Instance_{f} = (\Tree, \ex, \ell, T, D, \hat c_f, c^-_f)$ be the corresponding instance of \cBartPDs\Dbar.
	Now, solve each instance $\Instance_{f}$ using Algorithm~$(\Dbar)$ and return \yes if $\Instance_{f}$ is a \yes-instance for some $c_f \in \mathcal{C}$.
	Otherwise, return \no.
	
	\proofpara{Correctness}
	If one of the  constructed instances of \cBartPDs\Dbar is a \yes-instance, then $\Instance$ is a \yes-instance of \tPDs, by Lemma~\ref{lem:timePD-coloredYes}.
	
	For the converse, 
	if $S\subseteq X$ is a solution for $\Instance$, then $\w(E') \ge D$ where $E'$ is the set of edges which have at least one offspring in $S$.
	Therefore, $\w(E\setminus E') \le \w(E) - D = \Dbar$.
	Define $E_1 := E\setminus E' = E_d(X\setminus S)$ and let $E_2'$ be the set of edges $e\in E'$ which have a sibling-edge in $E_1$.
	If $e$ and~$e'$ are sibling-edges and both are in $E_2'$, then remove~$e'$ from $E_2'$.
	Continuously repeat the previous step to receive $E_2$ in which any two edges~$e,e' \in E_2$ are not sibling-edges.
	Observe $|E_2| \le |E_1| \le \w(E_1) \le \Dbar$.
	
	Let $Z \subseteq [W_m]$ be the set of colors of $E_1$ and the key-colors of $E_2$. 
	More precisely, $Z:= \{ j\in [m] \mid e_j \in E_1 \cup E_2\} \cup \{W_{j-1} +1, \dots, W_j \mid e_j \in E_1\}$.
	Since $\mathcal{F}$ is a~$(W_m, 2\Dbar)$-perfect hash family and $|Z| = \w(E_1) + |E_2| \leq 2\Dbar$, there exists a function $f \in \mathcal{F}$ such that $f(z) \neq f(z')$ for distinct $z, z' \in Z$.
	It follows that $E_1$ has unique colors, $E_2$ has unique key-colors, and $c(E_1)\cap \hat c(E_2) = \emptyset$ in the constructed instance $\Instance_f$.
	
	Let $C_1  = c(E_1)$ and $C_2 = \hat c(E_2)$.
	We claim that instance $\Instance_f$ has a color-respectful anchored taxa set $\mcA$ with $E^+(\mcA) = E_1$, $E_s(\mcA) \subseteq E_1 \cup E_2$, and $X(\mcA) = X\setminus S$.\linebreak
	Since $|c(E_1)| = |C_1| \leq \Dbar$ and there is a valid $T$-schedule saving $S$, this implies that~$\Instance_f$ is a \yes-instance of \cBartPDs\Dbar.
	
	We define $\mcA := \{(x_i,v_i,e_i) \mid 1 \leq i \leq |X \setminus S|\}$, where $x_i$, $v_i$, and $e_i$ are constructed iteratively as follows.
	Let $x_1,\dots, x_{|X\setminus S|}$ be an ordering of the set of taxa $X\setminus S$ such that $\ex(x_i) \leq \ex(x_j)$ if $i \leq j$.
	Define $E_d^{(i)} := E_d(\{x_1,\dots, x_i\})$ for each  $1 \leq i \leq |X \setminus S|$ and let $E_d^{(0)}$ be empty.
	For each $i \in [|X \setminus S|]$,
	let $e' := v_iw_i$ be the highest edge in~$E_d^{(i)} \setminus E_d^{(i-1)}$.
	This completes the construction of $x_i$ and $v_i$. Construct $e_i$ as follows.
	If $v_j \neq v_i$ for each $j > i$, then let $e_i$ be an arbitrary outgoing edge of $v_i$ not in $E_d^{(i)}$.
	Such an edge exists as $v_iw_i$ is the topmost edge of $E_d^{(i)} \setminus E_d^{(i-1)}$ and therefore not all edges outgoing of $v_i$ are in $E_d^{(i)}$.
	Otherwise, let $j_i$ be the minimum integer such that~$j_i>i$ and $v_{j_i}=v_i$.
	Then, let $e_i :=v_iw_{j_i}$.
	Note that $e_i$ is not in $E_d^{(i)}$ since $e_i$ is the topmost edge of $E_d^{(j_i)} \setminus E_d^{(j_i-1)}$.
	See Figure~\ref{fig:timePD-anchored-taxa-set} for an example.
	
	It remains to show that $\mcA$ is color-respectful.
	$\mcA$ satisfies CR a) as by construction $E^+(\mcA) = \bigcup_{i=1}^{|X \setminus S|} P_{v_i,x_i} = E_d(X\setminus S) = E_1$.
	Similarly, $\mcA$ satisfies CR b) because~$E_s(\mcA) \subseteq E_1 \cup E_2$.
	$\mcA$ satisfies CR c) by the fact that $\w(E_1) \leq \Dbar$.
	$\mcA$ satisfies CR d), as~$(x_1,v_1,e_1),\dots, (x_{|\mcA|},v_{|\mcA|},e_{|\mcA|})$ is a valid ordering of $\mcA$.
	
	To see that $\mcA$ satisfies CR e), consider any $(x_i,v_i,e_i), (x_j,v_j,e_j)$ for $i < j$. 
	As vertex~$v_j$ is an ancestor of $x_j$, every edge in $P_{v_j,x_j}$ has an offspring not in $\{x_1,\dots, x_i\}$ and therefore $P_{v_j,x_j}$ is disjoint from $E_d(x_1,\dots, x_i)$, which contains all edges of $P_{v_i,x_i}$. It follows that $P_{v_i,x_i}$ and $P_{v_j,x_j}$ are disjoint.
	
	\proofpara{Running Time}
	The construction of $\mathcal{C}$ takes $e^{2\Dbar} (2\Dbar)^{\Oh(\log {(2\Dbar)})} \cdot W_m \log W_m$ time, and for each $c \in \mathcal{C}$ the construction of each instance of \cBartPDs\Dbar takes time polynomial in $|\Instance|$. By Lemma~\ref{lem:timePD-rt-DBar}, solving an instance of \cBartPDs\Dbar takes~$\Oh^*(9^{\Dbar}\cdot {2\Dbar})$~time, and $|\mathcal{C}| = e^{2\Dbar} (2\Dbar)^{\Oh(\log {(2\Dbar)})} \cdot \log W_m$ is the number of constructed instances of \cBartPDs\Dbar.
	
	Thus,
	$\Oh^*(e^{2\Dbar} (2\Dbar)^{\Oh(\log {(2\Dbar)})} \log W_m \cdot (W_m + (9^{\Dbar}\cdot {2\Dbar})))$
	is the total running time.
	Because $W_m = \w(E_{\le \Dbar}) + |E_{>\Dbar}| \le 
	|E|\cdot\Dbar \le 2n \cdot \Dbar$,
	%4n\cdot \Dbar + 2n$ 
	the running time further simplifies
	to~$\Oh^*((3e)^{2\Dbar} \cdot 2^{\Oh(\log^2(\Dbar))}) = \Oh^*(2^{6.056 \cdot \Dbar + o(\Dbar)})$.
	%	\qed
\end{proof}

\section{Further Parameterized Complexity Results}
\label{sec:timePD-other}

\subsection{Parameterization by the Available Person-hours}
In this subsection, we consider parameterization by the available person-hours~$H_{\max_{\ex}}$.
Observe we may assume that $|T|,\max_{\ex} \le H_{\max_{\ex}} \le |T| \cdot \max_{\ex}$.

\begin{proposition}
	\label{prop:timePD-T+maxr}
	\tPDs and \tPDws are \FPT when parameterized by the available person-hours $H_{\max_{\ex}}$.
	More precisely,
	\begin{propEnum}%[(a)]
		\item\label{prop:timePD-s-T^maxr} \tPDs can be solved in $\Oh((|T| + 1)^{2\cdot\max_{\ex}} \cdot n)$ time,
		\item\label{prop:timePD-s-T*maxr} \tPDs can be solved in $\Oh((H_{\max_{\ex}})^{2\var_{\ex}} \cdot n)$ time, and
		\item\label{prop:timePD-ws-T+maxr} \tPDws can be solved in $\Oh(3^{|T| \cdot \max_{\ex}} \cdot |T| \cdot n)$ time.
	\end{propEnum}
\end{proposition}
\begin{proof}[Proof of Proposition~\ref{prop:timePD-T+maxr}\ref{prop:timePD-s-T^maxr}]
	\proofpara{Table definition}
	Let \Instance be an instance of \tPDs and
	let~$\myvec{a} = (a_1,\dots,a_{\max_{\ex}})$ be a $\max_{\ex}$-dimensional vector with $a_i \in [|T|]_0$.
	
	We define a dynamic programming algorithm with a table $\DP$.
	In entry $\DP[v,\myvec{a},b]$, for vertex~$v$ and integer $b\in \{0,1\}$, we store 0 if $b=0$;
	if $b = 1$, then we store the maximum value 
	of $\PDsub{\Tree_v}(S)$ for a non-empty subset~$S \subseteq \off(v)$
	that can be saved when having $a_j$ available teams in timeslot $j\in [\max_{\ex}]$.
	That is, if for every $j \in [\var_{\ex}]$ we have~$\sum_{x \in S \cap Z_j} \ell(x) \leq \sum_{i=1}^{\ex_j} a_i$.
	If there is no such non-empty $S$, we store $-\infty$.
	We define an auxiliary table $\DP'$ in which in entry $\DP'[v,i,\myvec{a},b]$ we only consider non-empty sets~$S \subseteq \off(u_1) \cup \dots \cup \off(u_i)$ where~$u_1, \dots, u_z$ are the children of~$v$.
	
	\proofpara{Algorithm}
	As a base case, for each taxon $x$
	we store $\DP[x,\myvec{a},b] = 0$ if $\sum_{i=1}^{\ex(x)} a_i \ge \ell(x)$ or if~$b=0$.
	Otherwise, we store $\DP[x,\myvec{a},1] = -\infty$.
	
	Let $v$ be an internal vertex with children $u_1,\dots,u_z$.
	
	We define $\DP'[v,\myvec{a},b,1] := \DP[u_1,\myvec{a},b] + \w(v u_1) \cdot b$
	and we compute further values with the recurrence
	\begin{equation}
		\label{eqn:s-T^maxr}
		\DP'[v,i+1,\myvec{a},b] =
		\max_{\myvec{a'}, b_1,b_2}
		\DP'[v,i,\myvec{a'},b_1] + \DP[u_{i+1},\myvec{a}-\myvec{a'},b_2] + 
		\w(v u_{i+1}) \cdot b_2.
	\end{equation}
	Here, we select $\myvec{a'}$, $b_1$, and~$b_2$ such that $a_j' \in [a_j]_0$ for each $j\in [\max_{\ex}]$ and $b_1,b_2 \in [b]_0$.
	We finally set $\DP[v,\myvec{a},b] := \DP'[v,z,\myvec{a},b]$.

	We return \yes if $\DP[\rho,\myvec{a}^*,1] \ge D$ where $\rho$ is the root of the phylogenetic tree \Tree and $a_i^*$ is the number of teams $t_j = (s_j,e_j)$ with $s_j< i\le e_j$. Otherwise, return \no.

	\proofpara{Correctness}
	In vector $\myvec{a}^*$ the number of available teams per timeslot is stored, such that \Instance is a \yes-instance if and only if $\DP[\rho,\myvec{a}^*,1] \ge D$. It remains to show that $\DP$ and~$\DP'$ stores the intended value.
	
	If $x$ is a leaf then the subtree $\Tree_x$ rooted at $x$ does not contain edges and\linebreak so~$\PDsub{\Tree_x}(\{x\}) = \PDsub{\Tree_x}(\emptyset) = 0$.
	As $S=\{x\}$ is the only non-empty set of taxa with~$S \subseteq \off(x)$, in the case that $b=1$, we need to ensure that in $\myvec{a}$, enough person-hours are declared for $\{x\}$, which is $\sum_{i=1}^{\ex(x)} a_i \ge \ell(x)$. Thus, the base case is correct.
	
	Let $v$ be a vertex with children $u_1,\dots,u_z$.
	To see that $\DP'[v,i,\myvec{a},b]$, for $i\in [q]$, stores the correct value, observe that the diversity of edge $vu_i$ can be added if and only if at least one taxon in $\off(u_i)$ survives, which is happening if and only if $b_2=1$ (or $b=1$ in the case of $i=1$).
	Further, the available person-hours at $v$ can naturally be divided between the children $v$.
	Therefore, $\DP'$ stores the correct value.
	The correctness of $\DP[v,\myvec{a},b]$ directly follows from the correctness of $\DP'[v,z,\myvec{a},b]$.

	\proofpara{Running time}
	The tables contain $\Oh(n \cdot (|T| + 1)^{\max_{\ex}})$ entries.
	For a leaf $x$, the value of $\DP[x,\myvec{a},1]$ can be computed in constant time, in our RAM model.
	For an internal vertex~$v$ with~$z$ children, we need to copy the value of $\DP'[v,z,\myvec{a},b]$ in constant time.
		
	In Recurrence~(\ref{eqn:s-T^maxr}), there are $\Oh((|T| + 1)^{\max_{\ex}})$ options to choose $\myvec{a'}$ and at most~$4$ options to choose $b_1$ and $b_2$,
	and so the values of the entries can be computed in~$\Oh((|T| + 1)^{\max_{\ex}} \cdot n)$ time.
	
	Finally, we need to iterate over the options for $\myvec{a^*}$.
	Altogether, we can compute a solution for an instance of \tPDs in $\Oh((|T| + 1)^{2\cdot\max_{\ex}} \cdot n)$ time.
	
	We note that the table entries store values of~$\Oh(D)$ and therefore the running time is also feasible in more restricted RAM models.
\end{proof}
\begin{proof}[Proof of Proposition~\ref{prop:timePD-T+maxr}\ref{prop:timePD-s-T*maxr}]
	\proofpara{Table definition}
	Let $\myvec{a} = (a_1,\dots,a_{\var_{\ex}})$ be a $\var_{\ex}$-dimensional vector with values $a_i \in [H_i]_0$.
	
	We define a dynamic programming algorithm with a table $\DP$.
	In entry $\DP[v,\myvec{a},b]$, for a vertex~$v$ and an integer $b\in \{0,1\}$, we store 0 if $b=0$;
	if $b = 1$, then we store the maximum value of $\PDsub{\Tree_v}(S)$ for a subset $S$ of offspring of $v$ such that $\sum_{x \in S \cap Z_j} \leq a_j$ for all $j \in [\var_{\ex}]$.
	We define an auxiliary table $\DP'$ in which in entry $\DP[v,i,\myvec{a},b]$ we only consider non-empty sets~$S \subseteq \off(u_1) \cup \dots \cup \off(u_i)$ where~$u_1, \dots, u_z$ are the children of~$v$.
	
	\proofpara{Algorithm}
	As a base case, 
	for each leaf $x$ we store $\DP[x,\myvec{a},b] = 0$ if $b=0$ or~$a_i \ge \ell(x)$, for each $i \ge \ex(x)$.
	Otherwise, we store $\DP[x,\myvec{a},1] = -\infty$.
	
	Let $v$ be an internal vertex with children $u_1,\dots,u_z$.
	
	We define $\DP'[v,1,\myvec{a},b] := \DP[u_1,\myvec{a},b] + \w(v u_1) \cdot b$
	and we compute further values with the recurrence
	\begin{equation}
		\label{eqn:T*maxr}
		\DP'[v,i+1,\myvec{a},b] =
		\max_{\myvec{a'}, b_1,b_2} \DP'[v,i,\myvec{a'},b_1] + \DP[u_{i+1},\myvec{a}-\myvec{a'},b_2] + 
		\w(v u_{i+1}) \cdot b_2.
	\end{equation}
	Here, we select $\myvec{a'}$, $b_1$, and~$b_2$ such that $a_j' \in [a_j]_0$ for each $j\in [\max_{\ex}]$ and $b_1,b_2 \in [b]_0$.
	We finally set $\DP[v,\myvec{a},b] := \DP'[v,z,\myvec{a},b]$.

	We return \yes if $\DP[\rho,\myvec{a}^*,1] \ge D$ where $\rho$ is the root of the given phylogenetic tree \Tree and $a_i^* = H_i$. Otherwise, we return \no.

	\proofpara{Correctness}
	In the base case, it is enough to check whether $a_i \ge \ell(x)$ for $i\ge \ex(x)$.
	A visualization is given in Figure~\ref{fig:timePD-example-4b}.
	\begin{figure}[t]
		\centering
		\begin{tikzpicture}[scale=0.75,every node/.style={scale=1.2}]
			\tikzstyle{txt}=[circle,fill=none,draw=none,inner sep=0pt]
			
			\fill[red!10] (0,1) rectangle (3,3);
			\fill[red!10] (2,0) rectangle (4,1);
			\node[txt] at (2,1.5) {$x_1$};
			
			\fill[blue!10] (3,1) rectangle (4,3);
			\fill[blue!10] (4,0) rectangle (5,2);
			\node[txt] at (4,1.5) {$x_2$};
			
			\fill[green!30] (4,2) rectangle (5,3);
			\fill[green!30] (5,0) rectangle (7,3);
			\node[txt] at (5.5,1.5) {$x_3$};
			
			\fill[yellow!30] (7,0) rectangle (9,3);
			\fill[yellow!30] (9,0) rectangle (10,2);
			\node[txt] at (8.5,1.5) {$x_4$};
			
			\fill[cyan!30] (9,2) rectangle (11,3);
			\fill[cyan!30] (10,0) rectangle (12,2);
			\node[txt] at (10.5,1.5) {$x_5$};
			
			\fill[orange!30] (12,0) rectangle (15,2);
			\node[txt] at (13.5,1.5) {$x_6$};

			\draw[gray!50] (0,1) -- (15,1);
			\draw[gray!50] (0,2) -- (15,2);
			\draw[gray!50] (0,3) -- (11,3);
			
			\draw[gray!50] (1,1) -- (1,3);
			
			\foreach \i in {2,...,11}
			\draw[gray!50] (\i,0) -- (\i,3);
			
			\foreach \i in {12,...,15}
			\draw[gray!50] (\i,0) -- (\i,2);
			
			\foreach \i in {1,...,15}
			\node[txt] at (\i-.5,-.4) {\i};
			
			\foreach \i in {1,...,3}
			\node[txt] at (-.4,\i-.5) {$t_\i$};
			
			\draw[->] (0,0) -- (0,3.5);
			\draw[->] (0,0) -- (15.5,0);
			
			\foreach \i in {4,7,15}
			\draw[dashed,red] (\i,-.5) -- (\i,3.5);
			
			\node[txt,red] at (3.5,3.25) {$\ex_1$};
			\node[txt,red] at (6.5,3.25) {$\ex_2$};
			\node[txt,red] at (14.5,3.25) {$\ex_3$};
			
			\draw[thick] (0,1) -- (0,3) -- (3,3) -- (3,1) -- (4,1) -- (4,0) -- (2,0) -- (2,1) -- (0,1);
			\draw[thick] (3,3) -- (3,1) -- (4,1) -- (4,0) -- (5,0) -- (5,2) -- (4,2) -- (4,3) -- (3,3);
			\draw[thick] (5,0) -- (5,2) -- (4,2) -- (4,3) -- (7,3) -- (7,0) -- (5,0);
			\draw[thick] (7,3) -- (7,0) -- (10,0) -- (10,2) -- (9,2) -- (9,3) -- (7,3);
			\draw[thick] (10,0) -- (10,2) -- (9,2) -- (9,3) -- (11,3) -- (11,2) -- (12,2) -- (12,0) -- (10,0);
			\draw[thick] (12,2) -- (12,0) -- (15,0) -- (15,2) -- (12,2);
		\end{tikzpicture}
		\resizebox{.2\columnwidth}{!}{%
			\myrowcols
			\begin{tabular}{c|ccc}
				$t_i$ & $t_1$ & $t_2$ & $t_3$\\
				\hline
				$s_i$ & 2 & 0 & 0\\
				$e_i$ & 15 & 15 & 11
			\end{tabular}
		}
		\;\;
		\resizebox{.45\columnwidth}{!}{%
			\myrowcols
			\begin{tabular}{c|cccccc}
				$x_i$ & $x_1$ & $x_2$ & $x_3$ & $x_4$ & $x_5$ & $x_6$\\
				\hline
				$\ell(x_i)$ & 8 & 4 & 7 & 8 & 6 & 6\\
				$\ex(x_i)$ & 4 & 7 & 7 & 15 & 15 & 15
			\end{tabular}
		}
		\;\;
		\resizebox{.2\columnwidth}{!}{%
			\myrowcols
			\begin{tabular}{c|ccc}
				$i=$ & $1$ & $2$ & $3$\\
				\hline
				$\ex_i$ & 4 & 7 & 15\\
				$H_i$ & 10 & 19 & 49
			\end{tabular}
		}
		\caption{This is a hypothetical schedule of taxa $x_1,\dots,x_6$.
			This schedule is only possible, as $x_2$ is already started before $\ex_1$.
			Thus, it is not sufficient to assign $x_2,x_3\in Y_2$ the person-hours between $\ex_1$ and $\ex_2$, which are $H_2-H_1=9$.}
		\label{fig:timePD-example-4b}
	\end{figure}
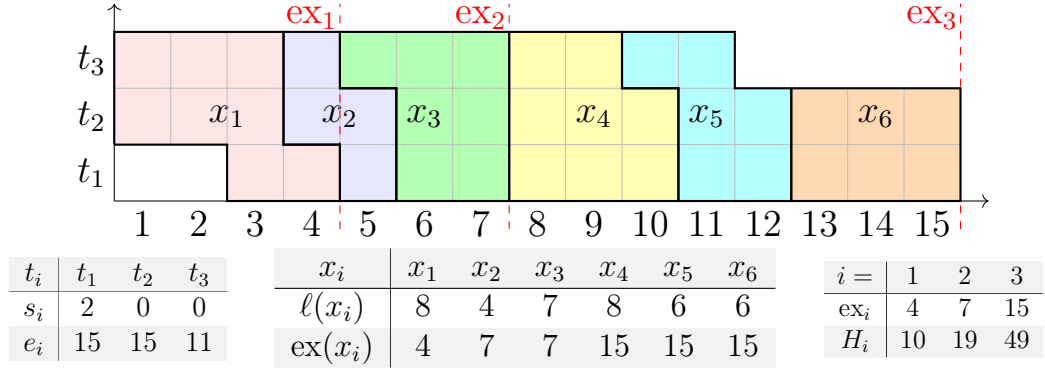%
	
	This algorithm is similar to the algorithm behind Proposition~\ref{prop:timePD-T+maxr}\ref{prop:timePD-s-T^maxr}.
	Instead of remembering the available teams each timeslot, in $\myvec{a}$,
	we store the available person-hours per unique remaining time.
	Therefore, the correctness proof is analogous to the correctness proof of Proposition~\ref{prop:timePD-T+maxr}\ref{prop:timePD-s-T^maxr}.
	
	\proofpara{Running time}
	The tables contain $\Oh(n \cdot (H_{\max_{\ex}})^{\var_{\ex}})$ entries. Each entry in $\DP$ can be computed in constant time.
	In Recurrence~(\ref{eqn:T*maxr}), there are $\Oh((H_{\max_{\ex}})^{\var_{\ex}})$ options to choose $\myvec{a'}$ and at most~$4$ to choose $b_1$ and $b_2$,
	and so the value of each entry can be computed in $\Oh((H_{\max_{\ex}})^{\var_{\ex}} \cdot n)$ time.
	% 	Computing $\myvec{a^*}$ does not require iterating over op\todos{What could that be?}.
	Altogether, we can compute a solution for \tPDs in $\Oh((H_{\max_{\ex}})^{2\var_{\ex}} \cdot n)$~time.
	%	\qed
\end{proof}
\begin{proof}[Proof of Proposition~\ref{prop:timePD-T+maxr}\ref{prop:timePD-ws-T+maxr}]
	\proofpara{Table definition}
	Let $\myvec{A} = (A_1,\dots,A_{\max_{\ex}})$ be a $\max_{\ex}$-tuple in which $A_i$ are subsets of $T$.
	
	We define a dynamic programming algorithm with a table $\DP$.
	In entry $\DP[v,\myvec{A},b]$, for a vertex~$v$ and an integer $b\in \{0,1\}$, we store 0 if $b=0$;
	if $b = 1$, then we store the maximum value of $\PDsub{\Tree_v}(S)$ for a subset of taxa $S \subseteq \off(v)$ that can be saved using only teams from $A_j$ 
	% 	and otherwise
	% 	the greatest possible diversity in the subtree rooted at $v$, if the teams $A_j$ available teams 
	at each timeslot $j\in [\max_{\ex}]$.
	We define an auxiliary table $\DP'$, in which, in entry $\DP'[v,i,\myvec{a},b]$ we only consider non-empty sets~$S \subseteq \off(u_1) \cup \dots \cup \off(u_i)$ where~$u_1, \dots, u_z$ are the children of~$v$.
	
	\proofpara{Algorithm}
	As a base case,
	for each leaf $x$ we store $\DP[x,\myvec{A},b] = 0$ if $b=0$ or
	if there is a team $t_j\in T$ and an integer $i\in [\max_{\ex} - \ell(x)]_0$ such that $t_j \in A_{i+i'}$ for each~$i' \in [\ell(x)]$.
	Otherwise, we store $\DP[x,\myvec{A},1] = -\infty$.
	
	Let $v$ be an internal vertex with children $u_1,\dots,u_z$.
	We compute the entry with the equation $\DP'[v,1,\myvec{A},b] := \DP[u_1,\myvec{A},b] + \w(v u_1) \cdot b$
	and we compute further values with the recurrence
	\begin{equation}
		\label{eqn:T+maxr}
		\DP'[v,i+1,\myvec{A},b] =
		\max_{\myvec{A'}, b_1,b_2} \DP'[v,i,\myvec{A'},b_1] + \DP[u_{i+1},\myvec{A}-\myvec{A'},b_2] + 
		\w(v u_{i+1}) \cdot b_2.
	\end{equation}
	Here, we select $\myvec{A'}$, $b_1$, and~$b_2$ such that $A_j'\subseteq A_j$ for each $j\in [\max_{\ex}]$ and $b_1,b_2 \in [b]_0$.
	We finally set $\DP[v,\myvec{A},b] := \DP'[v,z,\myvec{A},b]$.

	We return \yes if $\DP[\rho,\myvec{A^*},1] \ge D$ where $\rho$ is the root of the given phylogenetic tree \Tree and $A_i^*$ for $i\in[\max_{\ex}]$ contains team $t_j =(s_j,e_j) \in T$ if and only if~$s_j<i\le e_j$. Otherwise, we return \no.
	
	\newpage
	\proofpara{Correctness}
	If $\DP$ stores the intended value, $\Instance$ is a \yes-instance if and only if~$\DP[\rho,\myvec{A^*},1] \ge D$, where $\rho$ is the root of the phylogenetic tree \Tree.
	It remains to prove that $\DP$ stores the intended value.
	
	For a leaf $x$, if $\DP[x,\myvec{A},b]=0$ then either $b=0$ or there are $\ell(x)$ consecutive $A_i$s in which a team $t_j$ occurs.
	Thus, $x$ can be saved and $\PDsub{\Tree_x}(\{x\}) = 0$.
	Likewise, we show it the other way round and the base case is correct.
	
	The correctness of the recurrence can be shown analogously to the correctness of Proposition~\ref{prop:timePD-T+maxr}\ref{prop:timePD-s-T^maxr}.

	\proofpara{Running time}
	Observe that for each $i\in [\max_{\ex}]$ there are $2^{|T|}$ options to select~$A_i$. Thus, there are $(2^{|T|})^{\max_{\ex}} = 2^{|T| \cdot \max_{\ex}}$ options for~$\myvec{A}$.
	Therefore, the tables contain~$\Oh(n \cdot 2^{|T| \cdot \max_{\ex}})$ entries.
	For a leaf $x$, we need to iterate over the teams and the timeslots to check whether there is a team contained in $\ell(x)$ consecutive $A_i$s. Thus, we can compute all entries of $\DP$ in $\Oh(2^{|T| \cdot \max_{\ex}} \cdot |T|^{1+\max_{\ex}} \cdot n)$ time.
	
	Observe that in Recurrence~(\ref{eqn:T+maxr}), $A_j'$ is a subset of $A_j$, so there are $\Oh(3^{|T| \cdot \max_{\ex}})$ viable options for $\myvec{A}$ and $\myvec{A'}$. Thus, we conclude that all entries of $\DP'$ can be computed in $\Oh(3^{|T| \cdot \max_{\ex}} \cdot n)$.
	
	Consequently, because $|T|^{1+\max_{\ex}} \le 1.5^{|T|\cdot\max_{\ex}} \cdot |T|$ we can compute a solution for an instance of \tPDws in $\Oh(3^{|T| \cdot \max_{\ex}} \cdot |T| \cdot n)$ time.
	%	\qed
\end{proof}

\subsection{Few Rescue Lengths and Remaining Times}
The scheduling problem $1||\sum w_j(1-U_j)$ is \FPT with respect to $\var_\ell + \var_{\ex}$~\cite{hermelin}.
In this subsection, we describe a dynamic programming algorithm, similar to the approach in Theorem~\ref{thm:GNAP-varc+varw}, to prove that \tPDs is at least \XP when parameterized by $\var_\ell + \var_{\ex}$.
Here, $\var_\ell$ and $\var_{\ex}$ are the number of unique needed rescue lengths and unique remaining times, respectively.

Let \Instance be an instance of \tPDs,
let $\ell(X) = \{\ell_1,\dots,\ell_{\var_\ell}\}$ for each $\ell_i < \ell_{i+1}$ with $i\in [\var_\ell-1]$ be the unique rescue lengths, and let $\ex(X) = \{\ex_1,\dots,\ex_{\var_{\ex}}\}$ with~$\ex_i < \ex_{i+1}$ for each $i\in [\var_{\ex}-1]$ be the unique remaining times.

\begin{proposition}
	\label{prop:timePD-varl+varr}
	\tPDs can be solved in $\Oh(n^{2\var_\ell \cdot \var_{\ex}} \cdot (n + \var_\ell\cdot \var_{\ex}^2))$ time.
\end{proposition}
\begin{proof}
	\proofpara{Table definition}
	By~$\myvec{A}$ we denote an integer-matrix of size $\var_\ell \times \var_{\ex}$.
	Further, we denote $A_{i,j}$ to be the entry in row $i\in [\var_\ell]$ and column $j\in [\var_{\ex}]$ of~$\myvec{A}$.
	With~$\myvec{A}_{(i,j) +z}$ we denote the matrix resulting from~$\myvec{A}$ in which in row $i$ and column~$j$, the value~$z$ is added.
	
	We define a dynamic programming algorithm with table $\DP$.
	In entry~$\DP[v,\myvec A,b]$, we want to store $0$ if $b=0$.
	If~$b=1$, store in~$\DP[v,\myvec A,b]$ the maximum diversity that can be achieved in the subtree rooted at $v$ in which at most $A_{i,j}$ taxa $x$ are chosen with $\ell(x) = \ell_i$ and~$\ex(x) = \ex_j$.
	We define an auxiliary table $\DP'$ in which in entry~$\DP[v,\myvec{A},b,i]$ we only consider the first $i$ children of $v$.
	
	\proofpara{Algorithm}
	For each leaf $x$ with $\ell(x) = \ell_i$ and $\ex(x) = \ex_j$, we store $\DP[x,\myvec A,b] = 0$ if $b=0$ or $A_{i,j} > 0$.
	Otherwise, we store $\DP[x,\myvec A,b] = -\infty$.
	
	Let $v$ be an internal vertex with children $u_1,\dots,u_z$.
	
	We define $\DP'[v,1,\myvec{A},b] := \DP[u_1,\myvec{A},b] + \w(v u_1) \cdot b$
	and we compute further values with the recurrence
	\begin{equation}
		\label{eqn:varl+valr}
		\DP'[v,i+1,\myvec A,b] = \max_{\myvec{B},b_1,b_2} \DP'[v,i,\myvec{B},b_1] + \DP[u_{i+1},\myvec{A}-\myvec{B},b_2] + \w(v u_{i+1})\cdot b_2.
	\end{equation}
	Here, we select $\myvec{B}$ such that $B_{i,j}\le A_{i,j}$ for each $i\le [\var_\ell]$, $j\le [\var_{\ex}]$
	and $b_1$, and~$b_2$ are selected to be in $[b]_0$.
	We finally set $\DP[v,\myvec{A},b] := \DP'[v,z,\myvec{A},b]$.
	
	We return \yes if $\DP[\rho,\myvec{A}] \ge D$ for the root $\rho$ of \Tree and a matrix $\myvec{A}$ such that~$(\ell_1,\dots,\ell_{\var_\ell}) \cdot \myvec A \cdot \myvec 1_i \le H_i$ for each $i\in [\var_{\ex}]$.
	Here, $\myvec 1_i$ is a $\var_{\ex}$-dimensional vector in which the first $i$ positions are 1 and the remaining are 0 for each $i\in [\var_i]$.
	
	\proofpara{Correctness}
	By the definition of $\DP$, we see that an instance is a \yes-instance of \tPDs if and only if $\DP[\rho,\myvec{A}] \ge D$ for the root $\rho$ of \Tree and a matrix $\myvec{A}$ such that~$(\ell_1,\dots,\ell_{\var_\ell}) \cdot \myvec A \cdot \myvec 1_i \le H_i$ for each $i\in [\var_{\ex}]$.
	It only remains to show that $\DP$ stores the correct value.
	
	Observe that a leaf $x$ can be saved if and only if we are allowed to save a taxon with $\ell(x) = \ell_i$ and $\ex(x) = \ex_j$.
	Thus $\DP[x,\myvec{A},b]$ should store 0 if $b=0$ or if $A_{i,j}\ge 1$.
	The other direction follows as well.
	The correctness of the recurrence can be shown analogously to the correctness of Proposition~\ref{prop:timePD-T+maxr}\ref{prop:timePD-s-T^maxr}.
	
	\proofpara{Running time}
	The matrix $\myvec{A}$ contains $\var_\ell\cdot \var_{\ex}$ entries with integers in $[n]_0$.
	We can assume that if $A_{i,j} = n$ for some $i\le [\var_\ell]$, $j\le [\var_{\ex}]$,
	then $A_{p,q} = 0$ for each~$p\ne i$ and $q\ne j$.
	Such that there are $n^{\var_{\ex}\cdot \var_{\ex}} + n$ options for a matrix $\myvec{A}$.
	Therefore, the tables $\DP$ and $\DP'$ contain $\Oh(n\cdot n^{\var_{\ex}\cdot \var_{\ex}})$ entries.
	
	Each entry in $\DP$ can be computed in linear time.
	To compute entries of~$\DP'$, in Recurrence~(\ref{eqn:varl+valr}) we iterate over possible matrices $\myvec{B}$ and booleans $b_1,b_2$.
	Therefore, we can compute each entry in $\DP'$ in $\Oh(n^{2\var_{\ex}\cdot \var_{\ex}+1})$ time.
	
	For the initialization, we need to iterate over possible matrices $\myvec{A}$ and compute whether~$(\ell_1,\dots,\ell_{\var_\ell}) \cdot \myvec A \cdot \myvec 1_i \le H_i$ in $\Oh(n^{\var_{\ex}\cdot \var_{\ex}} \cdot \var_\ell\cdot \var_{\ex}^2)$ time.
	That proves the overall running time.
	%	\qed
\end{proof}

\subsection{Pseudo-polynomial Running Time on Stars}
In this subsection, we show that \tPDs can be solved in pseudo-polynomial time if the input tree is a star.
In the light of Proposition~\ref{prop:timePD-ws-Dbar}, such a result is unlikely to hold for \tPDws. % can be solved by an algorithm with pseudo-polynomial running time.

\begin{proposition}
	\label{prop:timePD-stars}
	If the given input tree is a star, \tPDs can be solved
	\begin{propEnum}
		\item in $\Oh((\max_{\ex})^2 \cdot n)$ time,
		\item in $\Oh(D^2 \cdot n)$ time,
		\item in $\Oh(\Dbar^2 \cdot n)$ time, or
		\item in $\Oh((\max_\w)^2 \cdot n^3)$ time, where $\max_\w$ is the maximum edge weight.
	\end{propEnum}
\end{proposition}
\begin{proof}
	Let $\Instance$ be an instance of \tPDs in which the given phylogenetic tree is a star.
	%We assume that there is an order of the leaves $x_1,\dots,x_{n}$ with $\ex(x_i) \le \ex(x_{i+1})$ for each $i\in [n-1]$.
	With the help of \KP, we solve the problem separately on the different classes $Y_i$ and use a dynamic programming algorithm to combine the solutions.
	In \KP, we are given a set of items $N$, each with a weight $\w_i$ and a profit $p_i$, a capacity $C$, and a desired profit $P$.
	The question is whether we can select a set of the items $S \in N$ such that the sum of profits $p_i$ in $S$ is at least $P$ while the sum of weights is not exceeding $C$.

	A solution for an instance of \KP can be found in~$\Oh(C \cdot |N|)$ time~\cite{weingartner,rehs} with a dynamic programming algorithm, indexing solutions by the\lb capacity~$C' \in [C]$ and the number of items $i \in [|N|]$, and storing the maximum profit for a subset of the first $i$ items that has total weight of at most~$C'$.
	We note that such an algorithm also computes the maximum profit for every capacity~$C'\in [C]$, and so we may assume that in $\Oh(C \cdot |N|)$ time it is possible to construct a table $DP_a$ such that $DP_a[i,C']$ stores the maximum $P$ such that $(Y_i,\w_j',p_j,C,P)$ is a \yes-instance of \KP. 
	
	Along similar lines, there exists a dynamic programming algorithm with a running time of~$\Oh(P \cdot |N|^2)$~\cite{weingartner,rehs}, where entries are indexed by~$P'\in [P]$\lb and~$i \in [|N|]$, and we store the minimum capacity $C'$ for which there is a subset of the first $i$ items with a total profit of at least $P'$.
	Adapting this algorithm, we receive a running time of~$\Oh(\Pbar \cdot |N|^2)$, where $\overline P = \sum_{a_i\in N} p_i - P$ and $\Cbar = \sum_{a_i\in N} \w_i - C$.
	In such an algorithm, we index solutions by $\Pbar' \leq \Pbar$ and $i \leq N$, and store the \emph{maximum} total weight of a subset of the first $i$ items whose total profit is \emph{at most} $\Pbar'$.
	If such a set exists with a weight of at least $\Cbar$ for $\Pbar' = \Pbar$, then the complement set is a solution for the \KP instance.

	\proofpara{Algorithms}
	We describe the algorithm with a running time of $\Oh((\max_{\ex})^2 \cdot n)$ and omit the similar cases for the other variants of the algorithm.
	
	Recall that~$Y_i$ is the set of taxa with~$\ex(x) = i$.
	In the following, for each set~$Y_i$ in~\Instance, we consider an instance of \KP with item set~$Y_i$ in which $x_j\in Y_i$ with incoming edge $e_j$ has a weight of $\w_j':=\ell(x_j)$ and a profit of~$p_j:=\w(e_j)$.
	We define a table $\DP$, in which for any $i \in [var_r]$ and $C \in [\ex_i]$, entry $\DP[i,C]$ stores the maximum desired profit $P$ such that $(Y_i,\w_j',p_j,C,P)$ is a \yes-instance of \KP. 
	
	We define a table $\DP'$ in which we combine the sub-solutions on $Z_i$, next.
	As a base case, we store $\DP'[1,C] := \DP[1,C]$ for each $C\in [\ex_1]$, $P\in [D]$, and~$\Pbar\in [\Dbar]$.
	To compute further values, we can use the recurrence
	\begin{eqnarray}
		\DP'[i+1,C] = 
		\max_{C' \in [C]_0} \DP'[i,C'] + \DP[i+1,C-C']
	\end{eqnarray}
	Finally, we return \yes if $\DP'[\var_{\ex},\max_{\ex}] \ge D$.

	\proofpara{Correctness}
	For convenience, we will write $\w(S)$ to denote $\sum_{x\in S} \w(\rho x)$ for a set~$S$ of taxa.
	Observe that $\w(S) = \PD(S)$, as $\Tree$ is a star.
	% We show the correctness for $\max_{\ex}$. For the other parameters it is shown analogously.
	The correctness of the values in $\DP$ follows from the correctness of the \KP-algorithms.
	Assume that for some $i\in [\var_{\ex}]$ the correct value is stored in $\DP'[i,C]$ for each $C\in [D]$.
	
	Let $\DP'[i+1,C]$ store $a\ge 0$.
	Then,
	by the construction 
	there is an integer $C'$ such that~$a = a_1 + a_2 = \DP'[i,C'] + \DP[i+1,C-C']$.
	Consequently, there are sets of taxa $S_1\subseteq Z_i$ and~$S_2\subseteq Y_{i+1}$ such that $\ell(S_1) + \ell(S_2) = C' + (C-C') = C$ and~$\w(S_1) + \w(S_2) = a_1 + a_2 = a$.
	Therefore, $S := S_1 \cup S_2$ is a set with $\ell(S) = C$ and~$\DP'[i+1,C] = \w(S)$.
	
	Conversely, let $S\subseteq Z_{i+1}$ be a set of taxa with $\ell(S) = C$.
	Define $S_1 := S\cap Z_{i}$ and~$S_2 := S\cap Y_{i+1}$ and let $C'$ be $\ell(S_1)$.
	We conclude that
	\begin{eqnarray*}
		\w(S) = \w(S_1) + \w(S_2) \ge \DP'[i,C'] + \DP[i+1,C-C'].
	\end{eqnarray*}

	\proofpara{Running time}
	Note that the algorithm solving the \KP-instance is also a dynamic programming algorithm, such that all the values of the table $\DP$ are computed in $\Oh(\max_{\ex} \cdot n)$ time.
	
	The table $\DP'$ has $\var_{\ex}\cdot \max_{\ex}$ entries.
	To compute a value with the recurrence, we need to check the $\Oh(\max_{\ex})$ options to select $C$ and add two numbers of size at most $D$, and so we can compute all values of $\DP'$ in $\Oh((\max_{\ex})^2 \cdot \var_{\ex})$ time.
	
	For the other running times:
	We observe that \KP can be computed in~$\Oh(D^2 \cdot n)$, and in $\Oh(\Dbar^2 \cdot n)$ time.
	The rest follows analogously.
	
	We can assume that $D \le n \cdot \max_\w$, as we are dealing with a trivial \no-instance otherwise.
	Therefore, in $\Oh((n \cdot \max_\w)^2 \cdot n) = \Oh((\max_\w)^2 \cdot n^3)$~time we can solve this special case of \tPDs.
	%	\qed
\end{proof}

\section{Discussion}
\label{sec:timePD-discussion}
With \tPDs and \tPDws, we introduced two \NP-hard generalizations of \MPD in which taxa may have distinct extinction times.
Most relevantly, we have proven that both problems are \FPT when parameterized by the target diversity $D$ and that \tPDs is also \FPT when parameterized by the acceptable loss of phylogenetic diversity $\Dbar$.
We have further proven that both problems are \FPT when parameterized by the available person-hours.

\tPDs is \NP-hard, but it remains an open question whether \tPDs is solvable in pseudo-polynomial time.
Indeed, we do not know if \tPDs or \tPDws can be solved in polynomial time even when the maximum rescue length needed to save a taxon is 2.
We note that the scheduling problem $1||\sum w_j (1-U_j)$ is \Wh 1-hard when parameterized by the unique number of processing times~\cite{heeger2024}; this implies that even on stars \tPDws is \Wh 1-hard when parameterized by the unique number of rescue lengths.

We further ask whether the  $\Oh(2^{n})$ running time for \tPDs (Proposition~\ref{prop:timePD-X}) can be improved to $\Oh(2^{o(n)})$, or whether this bound can be shown to be tight under \SETH or \ETH.
It also remains an open question whether  \tPDs or \tPDws are \FPT with respect to the largest extinction time $\max_{\ex}$.

We have not regarded kernelization algorithms. An interesting open question therefore is whether \tPDs or \tPDs admit a kernelization of polynomial size with respect to $D$ or $\Dbar$.

\chapter{Phylogenetic Diversity with Ecological Dependencies}
\label{ctr:FoodWebs}

\section{Introduction}
As we have seen in the previous chapters, the inherently limited amount of resources that one may devote to the task of saving taxa, necessitates decisions on which conservation strategies to pursue.
To support such decisions, one needs to incorporate quantitative information on the possible impact and the success likelihood of conservation strategies.
In this context, one task is to compute an optimal conservation strategy in the light of this information.

To find a conservation strategy with the best positive impact, one would ideally aim to maximize the functional diversity of the surviving taxa (species).
However, measuring this diversity even is impossible in many scenarios~\cite{MPC+18}.

Luckily, \MPD can be computed in polynomial time~\cite{steel,Pardi2005} and therefore became the standard in measuring biodiversity~\cite{crozier}.
Computing an optimal conservation strategy becomes much more difficult, however, when the success likelihood of a strategy is included in the model.
One way to achieve this is to add concrete survival probabilities for protected taxa, leading in its most general form to the \NP-hard \GNAPLong~\cite{hartmann,GNAP}, which has been observed in Chapter~\ref{ctr:GNAP}.
This problem formulation, however, still has a central drawback: It ignores that the survival of some taxa may also depend on the survival of other taxa.
This aspect was first considered by Moulton et~al.~\cite{moulton} in the \PDDlong~(\PDD)~problem.

Dependencies of taxa can take any thinkable form.
One model of dependencies are so called \emph{food-webs} in which the relationship between predators and prey are presented.
Food-webs are especially relevant in the observation of an ecosystem because, in them, one easily sees the role of a taxon within the outside world and the overall flow of biomass~\cite{cirtwill,lindeman1942}.

Moulton et al.~\cite{moulton} showed that \PDD~can be solved by the greedy algorithm if the objective of maximizing phylogenetic diversity agrees with the viability constraint in a precise technical sense.
Later, \PDD was conjectured to be \NP-hard in~\cite{spillner}.
This conjecture was confirmed by Faller et al.~\cite{faller}, who showed that \PDD is \NP-hard even if the food-web~\Food is a tree.
Further, Faller~et~al.~\cite{faller} considered \sPDD, the special case where the phylogenetic tree is restricted to be a star, and showed that \sPDD is \NP-hard even for food-webs which have a bipartite graph as underlying graph.
Finally, polynomial-time algorithms were provided for very restricted special cases, for example for \PDD when the food-web is a \emph{directed} tree~\cite{faller}.

\paragraph{Our contribution.}
Because \PDD has been shown to be \NP-hard already on very restricted instances~\cite{faller}, we turn to parameterized complexity in order to overcome this intractability.
Here, we consider the most natural parameters related to the solution, such as the solution size~$k$ and the threshold of diversity~$D$, and parameters that describe the structure of the input food-web~\Food.
We formally consider the decision problem, where we ask for the existence of a viable solution with diversity at least~$D$, but our algorithms actually solve the optimization problem as well.

Our most important results are as follows.
In Theorem~\ref{thm:PDD-k+height}, we prove that \PDD is \FPT when parameterized with the solution size~$k$ plus the height of the phylogenetic tree~$\Tree$.
This also implies that \PDD is \FPT with respect to~$D$, the diversity threshold.
However, both problems, \PDD and \sPDD, are unlikely to admit a kernel of polynomial size when parameterized by~$D$.
We also consider the dual parameter~$\Dbar$, that is, the amount of diversity that is lost from~$\Tree$, and show that \PDD is \Wh{1}-hard with respect to~$\Dbar$.

We then consider the structure of the food-web.
In particular, we consider the special case that each connected component of the food-web~$\Food$ is a complete digraph---so called cluster graphs.
As we will show, this case is structurally equivalent to the case that each connected component of~$\Food$ is a star with one source vertex.
Thus, this case describes a particularly simple dependency structure, where taxa are either completely independent or have a common source.
We further show that \PDD is \NP-hard in this special case while~\sPDD admits an \FPT-algorithm when parameterized by the vertex deletion distance to cluster graphs.
Our results thus yield structured classes of food-webs where the complexity of \sPDD and~\PDD strongly differ.
Finally, we show that \sPDD is \FPT with respect to the treewidth of the food-web and therefore can be solved in polynomial time if the food-web is a tree (Theorem~\ref{thm:PDD-tw}).
Our result disproves a conjecture of Faller et al.~\cite[Conjecture~4.2]{faller} stating that \sPDD is NP-hard even when the food-web is a tree.
Again, this result shows that~\sPDD~can be substantially easier than \PDD{} on some structured classes of food-webs.

Table~\ref{tab:PDD-results} gives an overview over the results in this chapter.
Here, Dom.-Source stands for Dominant-Source; a graph class for DAGs \Food in which there is a single source~$\sigma$ and an edge~$\sigma x$ in \Food for each~$x\in V(\Food) \setminus \{\sigma\}$.
Figure~\ref{fig:PDD-PDD-results} gives results for structural parameters and sets these parameters into relation.
\begin{table}[t]\centering
	\caption{An overview over the parameterized complexity results for \PDD and \sPDD.
		``D.t. $\tau$'' stands for ``Distance to~$\tau$''---the number of vertices that need to be removed to obtain graph class~$\tau$.}
	
	\resizebox{\columnwidth}{!}{%
		\myrowcols
		\begin{tabular}{l ll ll}
			\hline
			Parameter & \multicolumn{2}{c}{\sPDD} & \multicolumn{2}{c}{\PDD} \\
			\hline
			Budget $k$ & \FPT & Thm.~\ref{thm:PDD-k-stars} & \XP & Obs.~\ref{obs:PDD-k-XP}\\
			Diversity $D$ & \FPT & Thm.~\ref{thm:PDD-D} & \FPT & Thm.~\ref{thm:PDD-D}\\
			& no poly kernel & Thm.~\ref{thm:PDD-D-kernel-PDss} & no poly kernel & Thm.~\ref{thm:PDD-D-kernel-PDts}\\
			Species-loss $\kbar$ & \Wh{1}-hard, \XP & Prop.~\ref{prop:PDD-Dbar},~Obs.~\ref{obs:PDD-k-XP}~~ & \Wh{1}-hard, \XP & Prop.~\ref{prop:PDD-Dbar},~Obs.~\ref{obs:PDD-k-XP}\\
			Diversity-loss $\Dbar$ & \Wh{1}-hard, \XP & Prop.~\ref{prop:PDD-Dbar},~Obs.~\ref{obs:PDD-k-XP} & \Wh{1}-hard, \XP & Prop.~\ref{prop:PDD-Dbar},~Obs.~\ref{obs:PDD-k-XP}\\
			\hline
			D.t. Cluster & \FPT & Thm.~\ref{thm:PDD-dist-cluster-FPT} & \NP-h for 0 & Thm.~\ref{thm:PDD-dist-cluster-hardness}\\
			D.t. Co-Cluster & \FPT & Thm.~\ref{thm:PDD-co-cluster} & \FPT & Thm.~\ref{thm:PDD-co-cluster}\\
			Treewidth & \FPT & Thm.~\ref{thm:PDD-tw} & \NP-h for 1 & \cite{faller}\\
			Max Leaf \# & \FPT & Thm.~\ref{thm:PDD-tw} & \NP-h for 2 & Cor.~\ref{cor:PDD-maxleaf}\\
			D.t. Dom.-Source & \NP-h for 1 & Prop.~\ref{prop:PDD-DS} & \NP-h for 1 & Prop.~\ref{prop:PDD-DS}\\
			D.t. Bipartite & \NP-h for 0 & \cite{faller} & \NP-h for 0 & \cite{faller}\\
			Max Degree & \NP-h for 3 & \cite{faller} & \NP-h for 3 & \cite{faller}\\
			\hline
		\end{tabular}
	}
	\label{tab:PDD-results}
\end{table}

\paragraph*{Structure of the Chapter.}
In Section~\ref{sec:PDD-prelims}, we formally define \PDDlong and prove some simple initial results.
In Section~\ref{sec:PDD-k} and Section~\ref{sec:PDD-D}, we consider parameterization by the budget~$k$ and the threshold of diversity~$D$, the two integers of the input.
In Section~\ref{sec:PDD-kbar}, we consider \PDD with respect to the number of taxa that go extinct and the acceptable loss of diversity.
In Section~\ref{sec:PDD-structural}, we consider parameterization by structural parameters of the food-web. 
Finally, in Section~\ref{sec:PDD-discussion}, we discuss future research ideas.

\section{Preliminaries}
\label{sec:PDD-prelims}
In this section, we present the formal definition of the problems, and the parameterization.
We further start with some preliminary observations.

\subsection{Definitions}

\subparagraph*{Food-Webs.}
For a given set of taxa $X$, a \emph{food-web~$\Food=(X,E)$ on $X$} is a directed acyclic graph.
If $xy$ is an edge of $E$ then $x$ is \emph{prey} of $y$ and $y$ is a \emph{predator} of~$x$.
The set of prey and predators of $x$ is denoted with $\prey{x}$ and $\predators{x}$, respectively.
A taxon $x$ with an empty set of prey is a \emph{source} and $\sources$ denotes the set of sources in the food-web~\Food.

For a given taxon $x\in X$, we define $X_{\le x}$ to be the set of taxa $X$ which can reach~$x$ in \Food.
Analogously, $X_{\ge x}$ is the set of taxa that can be reached from $x$ in \Food.

For a given food-web~\Food and a set~$Z \subseteq X$ of taxa,
a set of taxa~$A\subseteq Z$ is \emph{$Z$-viable} if~$\sourcespersonal{\Food[A]} \subseteq \sourcespersonal{\Food[Z]}$.
A set of taxa~$A\subseteq X$ is \emph{viable} if $A$ is $X$-viable.
In other words, a set~$A \subseteq Z \subseteq X$ is~$Z$-viable or viable if each vertex with an in-degree of~0 in~$\Food[A]$ also has in-degree~0 in~$\Food[Z]$ or in~\Food, respectively.

\subparagraph*{Problem Definitions and Parameterizations.}
Formally, the main problem we regard in this chapter is defined as follows.

\problemdef{\PDDlong (\PDD)}{
	A phylogenetic~$X$-tree~$\Tree$, a food-web~$\Food$ on $X$, and integers~$k$ and~$D$}{
	Is there a viable set~$S\subseteq X$ such that~$|S|\le k$, and~$\PD(S)\ge D$}

Additionally, in \sPDDlong~(\sPDD) we consider the special case of \PDD in which the phylogenetic~$X$-tree~$\Tree$ is a star.

Throughout this chapter, we adopt the common convention that $n$ is the number of taxa~$|X|$
and we let $m$ denote the number of edges in the food-web $|E(\Food)|$.
Observe that \Tree has $\Oh(n)$~edges.
Such a relation does not necessarily hold for~$m$.

For an instance $\Instance = (\Tree,\Food,k,D)$ of \PDD, we define the parameter~$\Dbar$ to\lb be~$\PD(X)-D = \sum_{e \in E}\w(e) - D$.
Informally, $\Dbar$ is the acceptable loss of diversity: If we save a set of taxa $A \subseteq X$ with $\PD(A)\geq D$, then the total amount of diversity we lose from $\Tree$ is at most $\Dbar$.
Similarly, we define~$\kbar := |X| - k$.
That is, \kbar~is the minimum number of species that need to become extinct.

\subsection{Preliminary Observations}
We present some observations and reduction rules which we use throughout this chapter.

\begin{observation}
	\label{obs:PDD-viable}
	Let $\Food$ be a food-web.
	A set~$A \subseteq X$ is viable if and only if there are edges~$E_A \subseteq E(\Food)$ such that every connected component in the graph~$(A,E_A)$ is a tree with the root in $\sources$.
\end{observation}
\begin{proof}
	If~$A$ is viable then $\sourcespersonal{\Food[A]}$ is a subset of \sources.
	It follows that for each taxon~$x \in A$, either~$x$ is a source in~\Food or~$A$ contains a prey~$y$ of~$x$.
	
	Conversely, if a graph~$(A,E_A)$ exists in which all connected components are trees, then explicitly the sources of~$\Food[A]$ are a subset of~\sources.
\end{proof}

\begin{observation}
	\label{obs:PDD-solution-size}
	Let $\Instance = (\Tree,\Food,k,D)$ be a \yes-instance of \PDD.
	Then, unless~$k > n$, a viable set~$S \subseteq X$ with~$\PD(S) \ge D$ exists which has a size of exactly~$k$.
\end{observation}
\begin{proof}
	Let~$S$ be a solution for \Instance.
	If~$S$ has a size of~$k$, nothing remains to show.
	Otherwise, observe that $S\cup \{x\}$ is viable and $\PD(S\cup \{x\}) \ge \PD(S)$ for each taxon~$x \in (\predators{S} \cup \sources) \setminus S$.
	Because $(\predators{S} \cup \sources) \setminus S$ is non-empty, unless $S = X$,
	we conclude that~$S\cup \{x\}$ is a solution and iteratively, there is a solution of size~$k$. 
\end{proof}

\begin{observation}
	\label{obs:PDD-one-source}
	Let $\Instance = (\Tree,\Food,k,D)$ be an instance of \PDD.
	In $\Oh(|\Instance|^2)$~time, an equivalent instance $\Instance' = (\Tree',\Food',k',D')$ of \PDD with~$D' \in \Oh(D)$ and only one source in $\Food'$, can be computed.
\end{observation}
\begin{proof}
	\proofpara{Construction}
	Let $\Instance = (\Tree,\Food,k,D)$ be an instance of \PDD.
	Add a new taxon~$\star$ to \Food and add edges from~$\star$ to each taxon~$x$ of~$\sources$ to obtain $\Food'$.
	To obtain $\Tree'$, add $\star$ as a child to the root $\rho$ of \Tree and set $\w'(\rho \star) = D+1$ and~$\w'(e) = \w(e)$ for each $e\in E(\Tree)$.
	Finally, set $k' := k+1$ and $D' := 2\cdot D + 1$.
	
	\proofpara{Correctness}
	All steps can be performed in $\Oh(|\Instance|^2)$~time.
	Because $S\subseteq X$ is a solution for~\Instance if and only if $S\cup \{\star\}$ is a solution for $\Instance'$,
	the instance $\Instance' = (\Tree',\Food',k+1,2\cdot D+1)$ is a \yes-instance of \PDD if and only if \Instance is a \yes-instance of~\PDD.
\end{proof}

\begin{rr}
	\label{rr:each-taxon-savable}
	Let~$R \subseteq X$ be the set of taxa which have a distance of at least~$k$ to every source.
	Then, set~$\Food' := \Food - R$ and~$\Tree' := \Tree - R$.
\end{rr}
\begin{lemma}
	\label{lem:PDD-each-taxon-savable}
	Reduction Rule~\ref{rr:each-taxon-savable} is correct and in $\Oh(n+m)$~time can be applied exhaustively.
\end{lemma}
\begin{proof}
	By definition, each viable set of taxa which has a size of~$k$ is disjoint from~$R$.
	Therefore, the set $R$ is disjoint from every solution.
	With a breadth-first search, the set $R$ can be found in~$\Oh(n+m)$ time .
	This is also the total running time, since one application of the rule is exhaustive.
\end{proof}

\begin{rr}
	\label{rr:maxw<D}
	Apply Reduction Rule~\ref{rr:each-taxon-savable} exhaustively.
	If~$\max_\w \ge D$ return~\yes.
\end{rr}
After Reduction Rule~\ref{rr:each-taxon-savable} has been applied exhaustively,
for any taxon~$x\in X$ there is a viable set~$S_x$ of size at most~$k$ with~$x \in S_x$.
If edge~$e$ has a weight of at least~$D$, then for each taxon~$x$ which is an offspring of~$e$,
the set~$S_x$ is viable, has a size of at most~$k$, and~$\PD(S_x) \ge \PD(\{x\}) \ge D$.
So,~$S_x$ is a solution.

\begin{rr}
	\label{rr:redundant-edges}
	Given an instance~$\Instance = (\Tree,\Food,k,D)$ of \PDD with\lb edges~$vw, uw \in E(\Food)$ for taxa~$v$,~$w$ and each~$u\in \prey{v}$.
	If~$v$ is not a source, then remove $vw$ from $E(\Food)$.
\end{rr}
\begin{lemma}
	\label{lem:redundant-edges}
	\Cref{rr:redundant-edges} is correct and can be applied exhaustively in $\Oh(n^3)$~time.
\end{lemma}
\begin{proof}
	\proofpara{Correctness}
	If $\Instance'$ is a \yes-instance, then so is $\Instance$.
	
	Conversely, let $\Instance$ be a \yes-instance of \PDD with solution $S$.
	If $v\not\in S$, then $S$ is also a solution for instance $\Instance'$.
	If $v\in S$ then because $S$ is viable in $\Food$, some vertex~$u$ of $\prey{v}$ is in~$S$.
	Consequently, $S$ is also viable in $\Food - \{vw\}$, as~$w$ still could be fed by~$u$ (if $w\in S$).
	
	\proofpara{Running time}
	For two taxa~$v$ and~$w$, we can check~$\prey{v} \subseteq \prey{w}$ in~$\Oh(n)$~time.
	Consequently, an exhaustive application of Reduction Rule~\ref{rr:redundant-edges} takes $\Oh(n^3)$~time.
\end{proof}

\section{The Solution Size $k$}
\label{sec:PDD-k}
In this section, we consider parameterization by the size of the solution $k$.
First, we observe that \PDD is \XP when parameterized by $k$ and \kbar.
Recall that~$\kbar := n - k$ is the minimum number of taxa which need to go extinct.
In Section~\ref{sec:PDD-k-stars}, we show that \sPDD is \FPT with respect to $k$.
We generalize this result in Section~\ref{sec:PDD-k+height} by showing that \PDD is \FPT when parameterized by $k+\height_{\Tree}$.

\begin{observation}
	\label{obs:PDD-k-XP}
	\PDD can be solved in $\Oh(n^{k + 2})$ and $\Oh(n^{\kbar + 2})$~time.
\end{observation}
\begin{proof}
	\proofpara{Algorithm}
	Iterate over the sets $S\subseteq X$ of size $k$.
	Return \yes if there is a viable set~$S$ with $\PD(S) \ge D$.
	Return \no if there is no such set.
	
	\proofpara{Correctness and Running time}
	The correctness of the algorithm follows from Observation~\ref{obs:PDD-solution-size}.
	Checking whether a set $S$ is viable and has diversity of at least $D$ can be done $\Oh(n^2)$~time.
	The claim follows because there are $\binom{n}{k} = \binom{n}{n-k} = \binom{n}{\kbar}$ subsets of $X$ of size~$k$.
\end{proof}

\subsection{s-PDD With $k$}
\label{sec:PDD-k-stars}
We now show that \sPDD is \FPT when parameterized by the size of the solution~$k$.

\begin{theorem} \label{thm:PDD-k-stars}
	\sPDD can be solved in $\Oh(2^{3.03 k + o(k)} \cdot nm \cdot \log n)$~time.
\end{theorem}

In order to prove this theorem, we color the taxa and require that a solution should contain at most one taxon of each color.
Formally, the auxiliary problem which we consider is defined as follows.
In \cksPDDlong~(\cksPDD), alongside the usual input~$(\Tree,\Food,k,D)$ of \sPDD, we are given a vertex-coloring~$c: X\to [k]$ which assigns each taxon a \emph{color}~$c(x) \in [k]$.
We ask for whether there is a viable set $S \subseteq X$ of taxa such that~$\PD(S) \ge D$ and $c(S)$ is \emph{colorful}.
A set~$c(S)$ is colorful if~$c$ is injective on~$S$.
Observe that each colorful set~$S$ satisfies $|S| \le k$.
We continue to show how to\lb solve \cksPDD before we apply tools of the color coding toolbox to extend this result to the uncolored version.

\begin{lemma}
	\label{lem:PDD-k-stars}
	\cksPDD can be solved in $\Oh(3^k \cdot n \cdot m)$~time.
\end{lemma}
\begin{proof}
	\proofpara{Table definition}
	Let $\Instance = (\Tree,\Food,k,D,c)$ be an instance of \cksPDD, and by Observation~\ref{obs:PDD-one-source} we assume that~$\star \in X$ is the only source in \Food.
	
	Given~$x\in X$, a set of colors $C\subseteq [k]$, and a set of taxa $X'\subseteq X$:
	A set~$S \subseteq X'$ is~\emph{$(C,X')$-feasible} if
	\begin{itemize}%[a)]
		\itemsep-.35em
		\item\label{it:ka}$c(S) = C$,
		\item\label{it:kb}$c(S)$ is colorful, and
		\item\label{it:kc}$S$ is $X'$-viable.
	\end{itemize}
	
	We define a dynamic programming algorithm with tables $\DP$ and $\DP'$.
	For a taxon~$x\in X$ and a set of colors~$C\subseteq [k]$, we want entry~$\DP[x,C]$ to store the maximum~$\PD(S)$ of~$(C,X_{\ge x})$-feasible sets~$S$.
	Recall that $X_{\ge x}$ is the set of taxa which $x$ can reach in $\Food$.
	If no~$(C,X_{\ge x})$-feasible set $S \subseteq X'$ exists, we want $\DP[x,C]$ to store~$-\infty$.
	In other words, in $\DP[x,C]$ we store the biggest phylogenetic diversity of a set~$S$ which is $X_{\ge x}$-viable and $c$ bijectively maps~$S$ to~$C$.
	
	For a taxon $x$, let $y_1,\dots,y_q$ be an arbitrary but fixed order of $\predators{x}$.
	In the auxiliary table~$\DP'$, we want entry $\DP'[x,p,C]$ for~$p\in [q]$, and $C \subseteq [k]$ to store the maximum~$\PD(S)$ of~$(C,X')$-feasible sets $S \subseteq X'$, where $X' = \{x\} \cup X_{\ge y_1} \cup \dots \cup X_{\ge y_p}$.
	If no~$(C,X')$-feasible set $S \subseteq X'$ exists, we want $\DP'[x,p,C]$ to store~$-\infty$.
	
	\proofpara{Algorithm}
	As a base case, for each~$x\in X$ and~$p\in [|\predators{x}|]$
	let entries~$\DP[x,\emptyset]$ and~$\DP[x,p,\emptyset]$ store~0
	and
	let entry~$\DP[x,C]$ store~$-\infty$ if~$C$ is non-empty and~$c(x)$ does not occur in~$C$.
	For each~$x\in X$ with~$\predators{x} = \emptyset$, we store~$\w(\rho x)$ in $\DP[x,\{c(x)\}]$. Recall that~$\rho x$ is an edge because~\Tree is a star.
	
	Fix $x \in X$.
	For every~$Z \subseteq C \setminus \{c(x)\}$, we set~$\DP'[x,1,\{c(x)\} \cup Z] := \DP[y_1,Z]$.
	To compute further values, once $\DP'[x,q,Z]$ for each $q\in [p]$, and every $Z \subseteq C$ is computed, for~$Z \subseteq C \setminus \{c(x)\}$ we use the recurrence
	\begin{equation}
		\label{eqn:recurrence-k}
		\DP'[x,p+1,\{c(x)\} \cup Z] :=
		\max_{Z'\subseteq Z}
		\DP'[x,p,\{c(x)\} \cup Z\setminus Z']
		+
		\DP[y_{p+1},Z'].
	\end{equation}
	
	Finally, we set $\DP[x,C] := \DP'[x,q,C]$ for every~$C \subseteq [k]$.
	
	We return \yes if~$\DP[\star,C]$ stores~1 for some $C \subseteq [k]$.
	Otherwise, we return \no.

	\proofpara{Correctness}
	The base cases are correct.

	The tables are computed first for taxa further away from the source and with increasing size of $C$.
	Assume that for a fixed taxon $x$ with predators~$y_1, \dots, y_q$ and a fixed $p\in [q]$, the entries~$\DP[x',Z]$ and~$\DP'[x,p',Z]$, for each~$x' \in \predators{x}$, for each $p'\in [p]$, and every $Z \subseteq [k]$, store the desired value.
	Fix a set~$C \subseteq [k]$ with $c(x) \in C$.
	We show that if~$\DP'[x,p+1,C]$ stores~$d$ then there is a~$(C,X')$-feasible set $S \subseteq X' \cup X_{\ge y_{p+1}}$ for $X' := \{x\} \cup X_{\ge y_1} \cup \dots \cup X_{\ge y_{p}}$ with $\PD(S) = d$.
	Afterward, we show that if $S \subseteq X' \cup X_{\ge y_{p+1}}$ with $\PD(S) = d$
	is a~$(C,X')$-feasible set then~$\DP'[x,p+1,C]$ stores at least~$d$. 
	
	If~$\DP'[x,p+1,C] = d > 0$, then, by \Recc{eqn:recurrence-k}, a set $Z \subseteq C \setminus \{c(x)\}$ exists such that $\DP'[x,p,C\setminus Z] = d_x$ and $\DP[y_{p+1},Z] = d_y$ with $d = d_x + d_y$.
	Therefore, there is a~$(C\setminus Z,X')$-feasible set~$S_x \subseteq X'$ with $\PD(S_x) = d_x$ and a~$(Z,X_{\ge y_{p+1}})$-feasible set~$S_y \subseteq X_{\ge y_{p+1}}$ with~$\PD(S_y) = d_y$.
	Define $S := S_x \cup S_y$ and observe that~$\PD(S) = d$, because~$\Tree$ is a star, and~$c(S_x)$ and~$c(S_y)$ are disjoint and therefore~$S_x$ and~$S_y$.
	It remains to show that $S$ is a~$(C,X' \cup X_{\ge y_{p+1}})$-feasible set.
	First, observe that
	because $C\setminus Z$ and $Z$ are disjoint, we conclude that $c(S)$ is colorful.
	Then, $c(S) = c(S_x) \cup c(S_y) = C\setminus Z \cup Z = C$ where the first equation is satisfied because~$c(S)$ is colorful.
	The taxa $x$ and $y_{p+1}$ are the only sources in~$\Food[X_{\ge x}]$ and $\Food[X_{\ge y_{p+1}}]$, respectively.
	Therefore, $x$ is in~$S_x$ and~$y_{p+1}$ is in~$S_y$, unless $S_y$ is empty.
	If~$S_y = \emptyset$ then $S = S_x$ and $S$ is $X' \cup X_{\ge y_{p+1}}$-viable because $S$ is $X'$-viable.
	Otherwise, if~$S_y$ is non-empty then because $S_y$ is~$X_{\ge y_{p+1}}$-viable, we conclude $\sourcespersonal{\Food[S_y]} = \{ y_{p+1} \}$.
	As~$x\in S$ and~$y_{p+1} \in \predators{x}$, we conclude $\sourcespersonal{\Food[S]} = \{x\}$ and so $S$ is~$X' \cup X_{\ge y_{p+1}}$-viable.
	Therefore, $S$ is a~$(C,X' \cup X_{\ge y_{p+1}})$-feasible set.
	
	Conversely, let $S \subseteq X' \cup X_{\ge y_{p+1}}$ be a non-empty~$(C,X' \cup X_{\ge y_{p+1}})$-feasible set with~$\PD(S) = d$.
	Observe that $X'$ and $X_{\ge y_{p+1}}$ are not necessarily disjoint.
	We define $S_y$ to be the set of taxa of $X_{\ge y_{p+1}}$ which are connected to $y_{p+1}$ in $\Food[X_{\ge y_{p+1}}]$.
	Further, define~$Z := c(S_y)$ and define $S_x := S \setminus S_y$.
	As $c(S)$ is colorful, especially~$c(S_x)$ and $c(S_y)$ are colorful.
	Thus, $S_y$ is a~$(Z,X_{\ge y_{p+1}})$-feasible time.
	Further, we conclude that~$c(S_x) = C \setminus c(S_y) = C \setminus Z$.
	As $\sourcespersonal{\Food[S]} = \sourcespersonal{\Food[X' \cup X_{\ge y_{p+1}}]} = \{x\}$, we conclude $x\in S$.
	Because \Food is acyclic and~$y_{i+1}$ is a predator of~$x$, we conclude $x$ is not in~$X_{\ge y_{p+1}}$ and so~$x$ is in~$S_x$.
	Each vertex of $S$ which can reach $y_{p+1}$ in~$\Food[S]$ is in $F_{\ge y_{p+1}}$ and therefore in~$S_y$.
	Consequently, because $S$ is $X' \cup X_{\ge y_{p+1}}$-viable we conclude $\sourcespersonal{\Food[S_x]} = \{x\}$.
	Thus, $S_x$ is~$(C \setminus Z, X')$-feasible.
	Therefore, $\DP[y_{p+1},Z] = \PD(S_x)$ and $\DP'[x,p,C\setminus Z] = \PD(S_y)$.
	Hence, $\DP'[x,p+1,C]$ stores at least $\PD(S)$.

	\proofpara{Running time}
	The base cases can be checked in $\Oh(k)$ time.
	As each $c\in [k]$ in \Recc{eqn:recurrence-k} can either be in $Z'$, in $\{c(x)\} \cup Z \setminus Z'$ or in $[k] \setminus (\{c(x)\} \cup Z)$, all entries of the tables can be computed in $\Oh(3^k \cdot n \cdot m)$~time.
	
	We note that the table entries store values of~$\Oh(D)$ and therefore the running time is also feasible in more restricted RAM models.
\end{proof}

For the following proof we construct a perfect hash family $\mathcal{H}$ which is defined in \Cref{def:perfectHashFamily},
Central for proving Theorem~\ref{thm:PDD-k-stars} is to define and solve an instance of~\cksPDD for each function in $\mathcal{H}$.

\begin{proof}[Proof of Theorem~\ref{thm:PDD-k-stars}]
	\proofpara{Reduction}
	Let $\Instance = (\Tree, \Food, k, D)$ be an instance of \PDD.
	We assume that \Food only has one source, by~Observation~\ref{obs:PDD-one-source}.
	
	Let $x_1, \dots, x_{n}$ be an order of the taxa.
	Compute an $(n, k)$-perfect hash family $\mathcal{H}$.
	For every $f \in \mathcal{H}$,
	let $c_f$ be a coloring such that
	$c_f(x_j) = f(j)$ for each~$x_j \in X$.
	
	For every $f \in \mathcal{H}$,
	construct an instance $\Instance_{f} = (\Tree, \Food, k, D, c_f)$ of \cksPDD and solve~$\Instance_{f}$ using~Lemma~\ref{lem:PDD-k-stars} and
	return \yes if and only if $\Instance_{f}$ is a \yes-instance for some~$f \in \mathcal{H}$.

	\proofpara{Correctness}
	We show that if \Instance has a solution $S$ then there is an~$f \in \mathcal{H}$ such that $\Instance_f$ is a \yes-instance of \cksPDD.
	Let $S$ be a solution for \Instance.
	Thus, $S$ is viable, $\PD(S) \ge D$, and~$S$ has a size of at most~$k$.
	We may assume $|S| = k$ by Observation~\ref{obs:PDD-solution-size}.
	By the definition of~$(n, k)$-perfect hash families, there exists a function~$f\in \mathcal{H}$
	such that~$c_f(S)$ is colorful.
	So, $S$ is a solution for $\Instance_f$.
	
	Conversely,
	a solution of $\Instance_f$ for any~$f\in \mathcal{H}$ is a solution for~\Instance.

	\proofpara{Running Time}
	The instances $\Instance_f$ can be constructed in $e^k k^{\Oh(\log k)} \cdot n \log n$ time.
	An instance of \cksPDD can be solved in $\Oh(3^k \cdot n \cdot m)$~time, and the number of instances is~$|\mathcal{C}| = e^k k^{\Oh(\log k)} \cdot \log n$.
	Thus, the total running time is 
	$\Oh^*(e^k k^{\Oh(\log k)} \log n \cdot (3^k\cdot nm))$ which simplifies to $\Oh((3e)^k \cdot 2^{\Oh(\log^2(k))} \cdot nm \cdot \log n)$.
\end{proof}

\subsection{PDD With $k+\height_{\Tree}$}
\label{sec:PDD-k+height}
In this subsection, we generalize the result of the previous subsection by showing that \PDD is \FPT when parameterized with the size of the solution~$k$ plus~$\height_{\Tree}$, the height of the phylogenetic tree.
This algorithm uses the techniques of color coding, data reduction by reduction rules, and the enumeration of trees.

\begin{theorem} \label{thm:PDD-k+height}
	\PDD can be solved in~$\Oh^*(K^K \cdot 2^{3.03K + o(K)})$~time. Herein, we\linebreak write~$K := k\cdot \height_{\Tree}$.
\end{theorem}

We define \emph{pattern-trees}~$\Tree_P = (V_P,E_P,c_P)$ to be a tree~$(V_P,E_P)$ with a vertex-coloring $c_P: V_P \to [k\cdot \height_{\Tree}]$.
Recall that~$\spannbaum{Y}$ is the spanning tree of the vertices in~$Y$.
To show the result of Theorem~\ref{sec:PDD-k+height}, we use a subroutine for solving the following problem.

In \PDDplong~(\PDDp), we are given alongside the usual input~$(\Tree,\Food,k,D)$ of \PDD a pattern-tree~$\Tree_P = (V_P,E_P,c_P)$, and a vertex-coloring~$c: V(\Tree) \to [k\cdot \height_{\Tree}]$.
We ask whether there is a viable set $S \subseteq X$ of taxa such that~$S$ has a size of at most~$k$, $c(\spannbaum{ S\cup\{\rho\} })$ is colorful, and $\spannbaum{ S\cup\{\rho\} }$ and $\Tree_P$ are \emph{color-equal}.
That is, there is an edge~$uv$ of $\spannbaum{ S\cup\{\rho\} }$ with~$c(u) = c_u$ and $c(v) = c_v$ if and only if there is an edge $u'v'$ of $\Tree_P$ with~$c(u') = c_u$ and~$c(v') = c_v$.
Informally, given a pattern-tree, we want that it matches the colors of the spanning tree induced by the root and a solution.

Next, we present reduction rules with which we can reduce the phylogenetic tree in an instance of \PDDp to be a star which subsequently can be solved with~Theorem~\ref{thm:PDD-k-stars}.
Afterward, we show how to apply this knowledge to compute a solution for \PDD.

\begin{rr}
	\label{rr:edge-original}
	Let~$uv$ be an edge of~$\Tree$.
	If there is no edge $u'v'\in E_P$ with~$c_P(u') = c(u)$ and $c_P(v') = c(v)$, then set~$\Tree' := \Tree - \desc(v)$ and~$\Food' := \Food - \off(v)$.
\end{rr}
\begin{lemma}
	\label{lem:PDD-edge-original}
	Reduction Rule~\ref{rr:edge-original} is correct and can be applied exhaustively\linebreak in~$\Oh(n^3)$~time.
\end{lemma}
\begin{proof}
	\proofpara{Correctness}
	Assume that $S\subseteq X$ is a solution of the instance of \PDDp.
	As there is no edge $u'v'\in E_P$ with $c_P(u') = c(u)$ and $c_P(v') = c(v)$ we conclude that~$S\cap \desc(v) = \emptyset$ and so the reduction rule is safe.
	
	\proofpara{Running Time}
	To check whether Reduction Rule~\ref{rr:edge-original} can be applied, we need to iterate over both $E(\Tree)$ and $E_P$.
	Therefore, a single application can be executed in $\Oh(n^2)$~time.
	In each application of Reduction Rule~\ref{rr:edge-original} we remove at least one vertex so that an exhaustive application can be computed in $\Oh(n^3)$~time.
\end{proof}

\begin{rr}
	\label{rr:edge-pattern}
	Let~$u'v'$ be an edge of~$\Tree_P$.
	For each vertex $u\in V(\Tree)$ with~$c(u) = c_P(u')$ such that~$u$ has no child~$v$ with $c(v) = c_P(v')$,
	set~$\Tree' := \Tree - \desc(v)$ and~$\Food' := \Food - \off(v)$.
\end{rr}
\begin{lemma}
	\label{lem:PDD-edge-pattern}
	Reduction Rule~\ref{rr:edge-pattern} is correct and can be applied exhaustively\linebreak in $\Oh(n^3)$~time.
\end{lemma}
\begin{proof}
	\proofpara{Correctness}
	Let $S$ be a solution for an instance of \PDDp.
	The spanning tree $\spannbaum{ S\cup\{\rho\} }$ contains exactly one vertex $w$ of color $c(u)$.
	As~$c(w) = c_P(u')$ we conclude that $w$ has a child $w'$ and $c(w') = c(v')$.
	Consequently,~$S\cap \desc(u) = \emptyset$ and~$w \ne u$.
	
	\proofpara{Running Time}
	Like in Reduction Rule~\ref{rr:edge-original}, iterate over the edges of \Tree and $\Tree_P$.
	Each application either removes at least one vertex or concludes that the reduction rule is applied exhaustively.
\end{proof}

%The previous reduction rules potentially removed 
%\begin{rr}
%	\label{rr:food-web}
%	If \Tree has leaf set $L$ and \Food has vertex set $X$ with $L \subsetneq X$,
%	then compute the set $R \subseteq X$ of vertices $v$ which on every path from $v$ to a source of \Food contain a vertex of~$X\setminus L$,
%	set~$\Tree' := \Tree - R$ and~$\Food' := \Food - R$.
%\end{rr}
%\begin{lemma}
%	\label{lem:PDD-rr-food-web}
%	Reduction Rule~\ref{rr:food-web} is correct and can be applied exhaustively in $\Oh(n^2 \cdot m)$~time.
%\end{lemma}
%\begin{proof}
%	\proofpara{Correctness}
%	By definition, each set $S\subseteq X$ with $S\cap R \ne \emptyset$ is not viable.
%	Therefore, Reduction Rule~\ref{rr:food-web} is correct.
%	
%	\proofpara{Running Time}
%	Iterate over the taxa~$x\in X\setminus L$ and compute whether there is a path from~$x$ to some vertex of $\sources$ in $\Food[X\setminus L]$.
%	If not add $x$ to $R$.
%	With breadth-first search, this computation can be done in~$\Oh(m\cdot n^2)$~time, where the quadratic factor arises from the iteration over $X\setminus L$ and $\sources$.
%	%
%	A single application is exhaustive.\todos{True??}
%\end{proof}

Observe that the application of the previous two reduction rules may create leaves in the phylogenetic tree which are not taxa.
We can safely remove these from the tree.
Now, let us come to our final reduction rule.

\begin{rr}
	\label{rr:internal-vertex}
	Apply Reduction Rules~\ref{rr:edge-original} and~\ref{rr:edge-pattern} exhaustively.
	
	Let~$\rho$ be the root of~\Tree and let~$\rho_P$ be the root of~$\Tree_P$.
	Let $v'$ be a grand-child of $\rho_P$ and let~$u'$ be the parent of~$v'$.
	\begin{propEnum}%[1.]
		\item For each vertex $u$ of $\Tree$ with~$c(u) = c_P(u')$,
		add edges $\rho v$ to $\Tree$ for every child $v$ of $u$.
		\item Set the weight of $\rho v$ to $\w(uv)$ if~$c(v) \ne c_P(v')$,
		or~$\w(uv)+\w(\rho u)$ if~$c(v) = c_P(v')$.
		\item Add edges~$\rho_P w'$ to $\Tree_P$ for every child~$w'$ of~$u'$.
		\item Set~$\Tree_P' := \Tree_P - u'$ and~$\Tree' := \Tree - u$.
	\end{propEnum}
\end{rr}
\begin{figure}[t]
\centering
\begin{tikzpicture}[scale=0.8,every node/.style={scale=0.7}]
	\node[draw,fill=red,inner sep=3pt,circle] (root) at (10,10) {};
	
	\node[draw,fill=blue,inner sep=3pt,circle,yshift=5mm,xshift=-15mm] (c1) [below= of root] {};
	\node[draw,fill=green,inner sep=3pt,circle,yshift=5mm,xshift=0mm] (c2) [below= of root] {};
	\node[draw,fill=orange,inner sep=3pt,circle,yshift=5mm,xshift=15mm] (c3) [below= of root] {};
	
	\node[draw,fill=green!40!black!50!,inner sep=3pt,circle,yshift=5mm,xshift=0mm] (c11) [below= of c1] {};
	
	\node[draw,fill=cyan,inner sep=3pt,circle,yshift=5mm,xshift=-10mm] (c31) [below= of c3] {};
	\node[draw,fill=yellow,inner sep=3pt,circle,yshift=5mm,xshift=0mm] (c32) [below= of c3] {};
	
	\node[draw,fill=gray,inner sep=3pt,circle,yshift=8mm,xshift=0mm] (c321) [below= of c32] {};

	\draw[-{Stealth[length=6pt]}] (root) -> (c1);
	\draw[-{Stealth[length=6pt]}] (root) -> (c2);
	\draw[-{Stealth[length=6pt]}] (root) -> (c3);
	
	\draw[-{Stealth[length=6pt]}] (c1) -> (c11);
	\draw[-{Stealth[length=6pt]}] (c3) -> (c31);
	\draw[-{Stealth[length=6pt]}] (c3) -> (c32);
	
	\draw[-{Stealth[length=6pt]}] (c32) -> (c321);
	
	\node[xshift=14mm] [left= of c3] {$u'$};
	\node[xshift=14mm] [left= of c32] {$v'$};
	\node[xshift=5mm] [left= of root] {(1)};
\end{tikzpicture}
\begin{tikzpicture}[scale=0.8,every node/.style={scale=0.7}]
	\node[draw,fill=red,inner sep=3pt,circle] (root) at (10,10) {};
	
	\node[draw,fill=blue,inner sep=3pt,circle,yshift=5mm,xshift=-25mm] (c1) [below= of root] {};
	\node[draw,fill=blue,inner sep=3pt,circle,yshift=5mm,xshift=-15mm] (c4) [below= of root] {};
	\node[draw,fill=green,inner sep=3pt,circle,yshift=5mm,xshift=-5mm] (c2) [below= of root] {};
	\node[draw,fill=orange,inner sep=3pt,circle,yshift=5mm,xshift=5mm] (c3) [below= of root] {};
	\node[draw,fill=orange,inner sep=3pt,circle,yshift=5mm,xshift=30mm] (c5) [below= of root] {};
	
	\node[draw,fill=green!40!black!50!,inner sep=3pt,circle,yshift=5mm,xshift=0mm] (c11) [below= of c1] {};
	
	\node[draw,fill=green!40!black!50!,inner sep=3pt,circle,yshift=5mm,xshift=-3mm] (c41) [below= of c4] {};
	\node[draw,fill=green!40!black!50!,inner sep=3pt,circle,yshift=5mm,xshift=3mm] (c42) [below= of c4] {};
	
	\node[draw,fill=cyan,inner sep=3pt,circle,yshift=5mm,xshift=-8mm] (c30) [below= of c3] {};
	\node[draw,fill=cyan,inner sep=3pt,circle,yshift=5mm,xshift=0mm] (c31) [below= of c3] {};
	\node[draw,fill=yellow,inner sep=3pt,circle,yshift=5mm,xshift=8mm] (c32) [below= of c3] {};
	
	\node[draw,fill=cyan,inner sep=3pt,circle,yshift=5mm,xshift=-8mm] (c51) [below= of c5] {};
	\node[draw,fill=yellow,inner sep=3pt,circle,yshift=5mm,xshift=0mm] (c52) [below= of c5] {};
	
	\node[draw,fill=gray,inner sep=3pt,circle,yshift=8mm,xshift=-5mm] (c321) [below= of c32] {};
	\node[draw,fill=gray,inner sep=3pt,circle,yshift=8mm,xshift=5mm] (c322) [below= of c32] {};
	
	\node[draw,fill=gray,inner sep=3pt,circle,yshift=8mm,xshift=0mm] (c521) [below= of c52] {};
	
	\draw[-{Stealth[length=6pt]}] (root) -> node[above] {6} (c1);
	\draw[-{Stealth[length=6pt]}] (root) -> node[right] {1} (c4);
	\draw[-{Stealth[length=6pt]}] (root) -> node[right] {3} (c2);
	\draw[-{Stealth[length=6pt]}] (root) -> node[right] {1} (c3);
	\draw[-{Stealth[length=6pt]}] (root) -> node[left,xshift=-3mm] {2} (c5);
	
	\draw[-{Stealth[length=6pt]}] (c1) -> node[left] {4} (c11);
	\draw[-{Stealth[length=6pt]}] (c4) -> node[left] {5} (c41);
	\draw[-{Stealth[length=6pt]}] (c4) -> node[right] {2} (c42);
	
	\draw[-{Stealth[length=6pt]}] (c3) -> node[left] {4} (c30);
	\draw[-{Stealth[length=6pt]}] (c3) -> node[left] {2} (c31);
	\draw[-{Stealth[length=6pt]}] (c3) -> node[right] {3} (c32);
	
	\draw[-{Stealth[length=6pt]}] (c32) -> node[left] {1} (c321);
	\draw[-{Stealth[length=6pt]}] (c32) -> node[left] {1} (c322);
	
	\draw[-{Stealth[length=6pt]}] (c5) -> node[left] {2} (c51);
	\draw[-{Stealth[length=6pt]}] (c5) -> node[right] {2} (c52);
	
	\draw[-{Stealth[length=6pt]}] (c52) -> node[left] {2} (c521);
\end{tikzpicture}
\begin{tikzpicture}[scale=0.8,every node/.style={scale=0.7}]
	\node[draw,fill=red,inner sep=3pt,circle] (root) at (10,10) {};
	
	\node[draw,fill=blue,inner sep=3pt,circle,yshift=5mm,xshift=-15mm] (c1) [below= of root] {};
	\node[draw,fill=green,inner sep=3pt,circle,yshift=5mm,xshift=-5mm] (c2) [below= of root] {};
	
	\node[draw,fill=green!40!black!50!,inner sep=3pt,circle,yshift=5mm,xshift=0mm] (c11) [below= of c1] {};
	
	\node[draw,fill=cyan,inner sep=3pt,circle,yshift=-7mm,xshift=3mm] (c31) [below= of root] {};
	\node[draw,fill=yellow,inner sep=3pt,circle,yshift=-7mm,xshift=10mm] (c32) [below= of root] {};
	
	\node[draw,fill=gray,inner sep=3pt,circle,yshift=8mm,xshift=0mm] (c321) [below= of c32] {};
	
	\draw[-{Stealth[length=6pt]}] (root) -> (c1);
	\draw[-{Stealth[length=6pt]}] (root) -> (c2);
	
	\draw[-{Stealth[length=6pt]}] (c1) -> (c11);
	\draw[-{Stealth[length=6pt]}] (root) to[bend left=8] (c31);
	\draw[-{Stealth[length=6pt]}] (root) to[bend left=15] (c32);
	
	\draw[-{Stealth[length=6pt]}] (c32) -> (c321);
	
	\node[xshift=5mm] [left= of root] {(2)};
	\draw (8.25,10.2) -> (8.25,6.9);
\end{tikzpicture}
\begin{tikzpicture}[scale=0.8,every node/.style={scale=0.7}]
	\node[draw,fill=red,inner sep=3pt,circle] (root) at (10,10) {};
	
	\node[draw,fill=blue,inner sep=3pt,circle,yshift=5mm,xshift=-25mm] (c1) [below= of root] {};
	\node[draw,fill=blue,inner sep=3pt,circle,yshift=5mm,xshift=-15mm] (c4) [below= of root] {};
	\node[draw,fill=green,inner sep=3pt,circle,yshift=5mm,xshift=-5mm] (c2) [below= of root] {};
	
	\node[draw,fill=green!40!black!50!,inner sep=3pt,circle,yshift=5mm,xshift=0mm] (c11) [below= of c1] {};
	
	\node[draw,fill=green!40!black!50!,inner sep=3pt,circle,yshift=5mm,xshift=-3mm] (c41) [below= of c4] {};
	\node[draw,fill=green!40!black!50!,inner sep=3pt,circle,yshift=5mm,xshift=3mm] (c42) [below= of c4] {};
	
	\node[draw,fill=cyan,inner sep=3pt,circle,yshift=5mm,xshift=-8mm] (c30) [below= of c3] {};
	\node[draw,fill=cyan,inner sep=3pt,circle,yshift=5mm,xshift=0mm] (c31) [below= of c3] {};
	\node[draw,fill=yellow,inner sep=3pt,circle,yshift=5mm,xshift=8mm] (c32) [below= of c3] {};
	
	\node[draw,fill=cyan,inner sep=3pt,circle,yshift=5mm,xshift=-8mm] (c51) [below= of c5] {};
	\node[draw,fill=yellow,inner sep=3pt,circle,yshift=5mm,xshift=0mm] (c52) [below= of c5] {};
	
	\node[draw,fill=gray,inner sep=3pt,circle,yshift=8mm,xshift=-5mm] (c321) [below= of c32] {};
	\node[draw,fill=gray,inner sep=3pt,circle,yshift=8mm,xshift=5mm] (c322) [below= of c32] {};
	
	\node[draw,fill=gray,inner sep=3pt,circle,yshift=8mm,xshift=0mm] (c521) [below= of c52] {};
	
	\draw[-{Stealth[length=6pt]}] (root) -> node[above] {6} (c1);
	\draw[-{Stealth[length=6pt]}] (root) -> node[left] {1} (c4);
	\draw[-{Stealth[length=6pt]}] (root) -> node[left] {3} (c2);
	
	\draw[-{Stealth[length=6pt]}] (c1) -> node[left] {4} (c11);
	\draw[-{Stealth[length=6pt]}] (c4) -> node[left] {5} (c41);
	\draw[-{Stealth[length=6pt]}] (c4) -> node[right] {2} (c42);
	
	\draw[-{Stealth[length=6pt]}] (c32) -> node[left] {1} (c321);
	\draw[-{Stealth[length=6pt]}] (c32) -> node[left] {1} (c322);
	
	\draw[-{Stealth[length=6pt]}] (root) to[bend left=5] node[right] {4} (c30);
	\draw[-{Stealth[length=6pt]}] (root) to[bend left=10] node[right] {2} (c31);
	\draw[-{Stealth[length=6pt]}] (root) to[bend left=15] node[right] {4} (c32);
	
	\draw[-{Stealth[length=6pt]}] (root) to[bend left=20] node[right] {2} (c51);
	\draw[-{Stealth[length=6pt]}] (root) to[bend left=25] node[right] {4} (c52);
	
	\draw[-{Stealth[length=6pt]}] (c52) -> node[left] {2} (c521);
\end{tikzpicture}
\caption{An example for Reduction Rule~\ref{rr:internal-vertex}.
	(1) An instance of \PDDp.
	(2)~The instance of (1) after an application of Reduction Rule~\ref{rr:internal-vertex} to the marked vertices. 
	In both instances, the pattern-tree is on the left and the phylogenetic tree is on the right.
	}
\label{fig:PDD-rr-internal-vertex}
\end{figure}
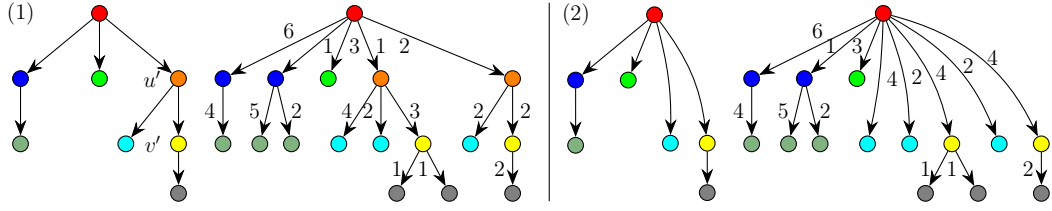%
Figure~\ref{fig:PDD-rr-internal-vertex} depicts an application of Reduction Rule~\ref{rr:internal-vertex}.

\begin{lemma}
	\label{lem:PDD-internal-vertex}
	Reduction Rule~\ref{rr:internal-vertex} is correct and can be applied exhaustively\linebreak in $\Oh(n^3)$~time.
\end{lemma}
\begin{proof}
	\proofpara{Correctness}
	Assume that \Instance is a \yes-instance of \PDDp with solution~$S$.
	Because $\spannbaum{ S\cup\{\rho\} }$ and $\Tree_P$ are color-equal also $\spannbaumsub{\Tree'}{ S\cup\{\rho\} }$ and $\Tree_P'$ are color-equal.
	Let~$u^*$ and~$w_1$ be the unique vertices in $\spannbaum{ S\cup\{\rho\} }$ with $c(u^*) = c_P(u')$ and $c(w_1) = c_P(v')$.
	Let~$w_2,\dots,w_\ell$ be the other children of $u^*$.
	Because $\PDsub{\Tree'}(S)$ is the sum of the weights of the edges of $\spannbaumsub{\Tree'}{ S\cup\{\rho\} }$, we conclude
	$$
	\PDsub{\Tree'}(S) = \PD(S) - (\w(\rho u^*) + \sum_{i=1}^\ell \w(u^* w_i)) + \sum_{i=1}^\ell \w'(\rho w_i).
	$$
	Since $\w'(\rho w_1) = \w(\rho u^*) + \w(u^* w_1)$ and $\w'(\rho w_i) = \w(u^* w_i)$ for $i\in [\ell] \setminus \{1\}$,
	we conclude that~$\PDsub{\Tree'}(S) = \PD(S) \ge D$.
	Therefore, $S$ is a solution for~$\Instance'$.
	
	The other direction is shown analogously.
	
	\proofpara{Running Time}
	For a given grand-child~$v'$ of~$\rho_P$, one needs to perform~$\Oh(n)$ color-checks and add~$\Oh(n)$ edges to the tree.
	As the reduction rule can be applied at most~$|\Tree_P| \in \Oh(K) \in \Oh(n)$~times,
	an exhaustive application takes~$\Oh(n^2)$~time.
	So, the predominant factor in the running time is the exhaustive application of the other reduction rules.
\end{proof}

With these reduction rules, we can reduce the phylogenetic tree of a given instance of \PDDp to only be a star and then solve \PDDp by applying~Theorem~\ref{thm:PDD-k-stars}.

\begin{lemma}
	\label{lem:PDD-k+height}
	\PDDp can be solved in~$\Oh(3^k \cdot n \cdot m + n^3)$~time.
\end{lemma}
\begin{proof}
	\proofpara{Algorithm}
	Let an instance~$\Instance = (\Tree, \Food, k, D, \Tree_P = (V_P,E_P,c_P), c)$ of \PDDp be given.
	If there is a vertex $v\in V_P$ and $c_P(v) \not\in c(V(\Tree))$, then return \no.
	If $c(\rho) \ne c_P(\rho_P)$ where $\rho$ and $\rho_P$ are the roots of \Tree and $\Tree_P$ respectively, return \no.
	
	Apply~Reduction Rule~\ref{rr:internal-vertex} exhaustively.
	Then, both $\Tree_P$ and \Tree~are stars.
	Return \yes if and only if
	$(\Tree',\Food',k,D,c)$ is a \yes-instance of \cksPDD.

	\proofpara{Correctness}
	If $\Tree_P$ contains a vertex~$v$ with $c_P(v) \not\in c(V(\Tree))$,
	or if $c(\rho) \ne c_P(\rho_P)$,
	then~\Instance is a \no-instance.
	Then, the correctness follows by Lemma~\ref{lem:PDD-internal-vertex} and Lemma~\ref{lem:PDD-k-stars}.

	\proofpara{Running time}
	Reduction Rule~\ref{rr:internal-vertex} can be applied exhaustively in $\Oh(n^3)$~time.
	With the application of Lemma~\ref{lem:PDD-k-stars}, the overall running time is $\Oh(3^k \cdot n \cdot m + n^3)$.
\end{proof}

Now, we have everything to prove Theorem~\ref{thm:PDD-k+height}.
We reduce from \PDD to \PDDp and apply Lemma~\ref{lem:PDD-k+height}.
For this, we use the fact that there are~$n^{n-2}$ labeled directed trees on~$n$ vertices~\cite{Shor95} which can be enumerated in $\Oh(n^{n-2})$~time~\cite{beyer}.
To solve an instance \Instance of \PDD, we check each of these trees as a pattern-tree for a given coloring of the phylogenetic tree.
These colorings are defined with a perfect hash family as defined in \Cref{def:perfectHashFamily}.
Recall that $K=k\cdot \height_{\Tree}$.
\begin{proof}[Proof of Theorem~\ref{thm:PDD-k+height}]
	\proofpara{Algorithm}
	Let $\Instance = (\Tree, \Food, k, D)$ be an instance of \PDD.
	Let the vertices of~\Tree be ordered as~$v_1, \dots, v_{|V(\Tree)|}$.
	Iterate over $i \in [\min\{K,|V(\Tree)|\}]$.
	Compute a~$(|V(\Tree)|,i)$-perfect hash family~$\mathcal{H}_i$.
	Compute the set~$\mathcal{P}_i$ of labeled directed trees with~$i$ vertices.
	
	For every~$\Tree_P = (V_P, E_P, c_P) \in \mathcal{P}_i$, proceed as follows.
	Assume that the labels of~$\Tree_P$ are in $[i]$.
	For every $f \in \mathcal{H}_i$,
	let $c_f$ be a coloring such that $c_f(v_j) = f(j)$ for each $v_j \in V(\Tree)$.
	
	For every $f \in \mathcal{H}_i$,
	solve instance $\Instance_{\Tree_P,f} := (\Tree, \Food, k, D, \Tree_P, c_f)$ of \PDDp using~Lemma~\ref{lem:PDD-k+height}.
	Return \yes if and only if~$\Instance_{\Tree_P,f}$ is a \yes-instance for some~$f \in \mathcal{H}_i$ and some~$\Tree_P \in \mathcal{P}_i$.

	\proofpara{Correctness}
	Any solution of an instance $\Instance_{\Tree_P,f}$ of \PDDp is a solution for $\Instance$.

	Conversely, we show that if $S$ is a solution for \Instance, then there are~$\Tree_P$ and~$f$ such that~$\Instance_{\Tree_P,f}$ is a \yes-instance of \PDDp.
	So let~$S$ be a viable set of taxa with $|S| \le k$ and $\PD(S) \ge D$.
	Let $V^* \subseteq V(\Tree)$ be the set of vertices~$v$ that have offspring in $S$.
	It follows that~$|V^*| \le \height_{\Tree} \cdot |S| \le K$.
	Then, there is a hash function~$f \in \mathcal{H}_{V^*}$ mapping $V^*$ bijectively to $[|V^*|]$.
	Consequently, $\mathcal{P}_{|V^*|}$ contains a tree~$\Tree_P$ which is isomorphic to~$\Tree[V^*]$ with labels~$c_f$.
	Hence, $\Instance_{\Tree_P,f}$ is a \yes-instance of \PDDp.

	\proofpara{Running Time}
	For a fixed $i\in [K]$, the set $\mathcal{H}_i$ contains $e^i i^{\Oh(\log i)} \cdot \log n$ hash functions and the set~$\mathcal{P}_i$ contains $\Oh(i^{i-2})$ labeled trees.
	In $\Oh(i^{i-2} \cdot n\log n)$~time, both of the sets can be computed.
	Given a hash function and a tree,
	each instance~$\Instance_{\Tree_P,f}$ of \PDDp is constructed in $\Oh(n)$~time and can be solved in~$\Oh(3^k \cdot n^3)$~time.
	Thus, the overall running time is~$\Oh(K \cdot e^{K} K^{K - 2 + \Oh(\log K)} \cdot 3^k \cdot n^3 \log n)$,
	which summarizes to~$\Oh(K^K \cdot 2^{1.44K + 1.58k + o(K)} \cdot n^3 \log n)$.
\end{proof}

\section{The Diversity $D$}
\label{sec:PDD-D}
In this section, we consider parameterization with the required threshold of diversity~$D$.
As the edge-weights are integers, we conclude that we can return \yes if~$k\ge D$ or if the height of the phylogenetic tree~\Tree is at least $D$, after Reduction Rule~\ref{rr:each-taxon-savable} has been applied exhaustively. Otherwise,~$k+\height_{\Tree}\in \Oh(D)$ and thus the \FPT-algorithm for $k+\height_{\Tree}$ (Theorem~\ref{thm:PDD-k+height}) directly gives an \FPT-algorithm for \PDD in that case.

In Section~\ref{sec:PDD-FPT-D}, we present a faster \FPT-algorithm for the parameter~$D$.
Afterward, we show that it is unlikely that \PDD admits a polynomial kernel for~$D$, even in very restricted cases.

\subsection{FPT-Algorithm For $D$}
\label{sec:PDD-FPT-D}
By Theorem~\ref{thm:PDD-k+height}, \PDD is \FPT with respect to the desired diversity~$D$.
Here, we present another algorithm with a faster running time.
To obtain this algorithm, we implicitly subdivide edges of the phylogenetic tree according to their edge weights.
We then use color coding on the vertices of the subdivided tree.
\begin{theorem}
	\label{thm:PDD-D}
	\PDD can be solved in $\Oh(2^{3.03(2D + k) + o(D)} \cdot nm + n^2)$~time.
\end{theorem}

We define \cDPDDlong (\cDPDD), a colored version of \PDD, as follows.
In addition to the usual input of \PDD, we receive two colorings~$c$ and~$\hat c$ which assign each edge $e\in E(\Tree)$ a subset $c(e)$ of $[D]$, called \emph{a set of colors}, which is of size $\w(e)$ and each taxon $x\in X$ a \emph{color} $\hat c(x) \in [k]$.
We extend the function~$c$ to also assign color sets to taxa $x\in X$ by defining $c(x) := \bigcup_{e\in E'} c(e)$ where $E'$ is the set of edges with $x\in \off(e)$.
In~\cDPDD, we ask whether there is a viable set~$S\subseteq X$ of taxa such that~$c(S) = [D]$, and~$\hat c$ is colorful.

In the following we show how to solve \cDPDD and then we show how to apply standard color coding techniques to reduce from \PDD to \cDPDD.

Finding a solution for \cDPDD can be done with techniques similar to the ones used to prove~Lemma~\ref{lem:PDD-k-stars}.
\begin{lemma}
	\label{lem:PDD-colored-D}
	\cDPDD can be solved in $\Oh(3^{D+k} \cdot n\cdot m)$~time.
\end{lemma}
\begin{proof}
	\proofpara{Table definition}
	Let $\Instance = (\Tree,\Food,k,D,c,\hat c)$ be an instance of \cDPDD and we assume that~$\star \in X$ is the only source in \Food by Observation~\ref{obs:PDD-one-source}.
	
	Given~$x\in X$, sets of colors $C_1\subseteq [D]$, $C_2\subseteq [k]$, and a set of taxa $X'\subseteq X$,
	a set~$S \subseteq X' \subseteq X$ is \emph{$(C_1,C_2,X')$-feasible} if
	\begin{itemize}%[a)]
		\itemsep-.35em
		\item\label{it:Da}$C_1$ is a subset of $c(S)$,
		\item\label{it:Db}$\hat c(S) = C_2$,
		\item\label{it:Dc}$\hat c(S)$ is colorful, and
		\item\label{it:Dd}$S$ is $X'$-viable.
	\end{itemize}
	
	We define a dynamic programming algorithm with tables $\DP$ and $\DP'$.
	We want entry $\DP[x,C_1,C_2]$, for~$x\in X$ and sets of colors $C_1\subseteq [D]$, $C_2\subseteq [k]$, to store~1 if a~$(C_1,C_2,X_{\ge x})$-feasible set~$S$ exists.
	Otherwise, we want $\DP[x,C_1,C_2]$ to store~0.
	
	For a taxon $x \in X$, let $y_1,\dots,y_q$ be an arbitrary but fixed order of $\predators{x}$.
	In the auxiliary table we want entry $\DP'[x,p,C_1,C_2]$, for~$p\in [q]$, $C_1\subseteq [D]$, and $C_2\subseteq [k]$, to store~1 if a~$(C_1,C_2,X')$-feasible set $S \subseteq X'$ exists, where $X' := \{x\} \cup X_{\le y_1} \cup \dots \cup X_{\le y_p}$.
	If no~$(C_1,C_2,X')$-feasible set $S \subseteq X'$ exists, we want $\DP'[x,p,C_1,C_2]$ to store~0.
	
	\proofpara{Algorithm}
	As a base case, for each~$x\in X$ and each~$p\in [|\predators{x}|]$ let~$\DP[x,\emptyset,\emptyset]$ and~$\DP'[x,p,\emptyset,\emptyset]$ store~1.
	Further, let~$\DP[x,C_1,C_2]$ and~$\DP'[x,p,C_1,C_2]$ store~0 if~$C_1 \not\subseteq c(x)$, or~$\hat c(x) \not\in C_2$, or~$|C_2| > |X_{\ge x}|$ for each~$x\in X$, every~$C_1\subseteq [D]$, and every $C_2\subseteq [k]$.
	For each taxon~$x\in X$ with no predators, we store~1 in $\DP[x,C_1,C_2]$ if $C_1=C_2=\emptyset$ or if $C_1 \subseteq c(x)$ and $C_2 = \{\hat c(x)\}$.
	Otherwise, we store~0.
	
	Fix a taxon $x\in X$.
	Assume that~$\DP[y,C_1,C_2]$ is computed for each~$y\in \predators{x}$, every~$C_1\subseteq [D]$, and every $C_2\subseteq [k]$.
	For~$C_1 \subseteq [D] \setminus c(x)$ and~$C_2 \subseteq [k] \setminus \{\hat c(x)\}$, we set
	\begin{equation}
		\label{eqn:recurrence-D-pre}
		\DP'[x,1,c(x) \cup C_1,\{\hat c(x)\} \cup C_2] := \DP[y_1,C_1,C_2].
	\end{equation}
	
	Fix an integer $p\in [q]$.
	We assume that $\DP'[x,p,C_1,C_2]$ is computed for every~$C_1\subseteq [D]$, and every $C_2\subseteq [k]$.
	Then, for $C_1 \subseteq [D] \setminus c(x)$ and~$C_2 \subseteq [k] \setminus \{\hat c(x)\}$ we use the recurrence
	\begin{eqnarray}
		\label{eqn:recurrence-D}
		& & \DP'[x,p+1,c(x) \cup C_1,\{\hat c(x)\} \cup C_2]\\
		&:=&	\max_{C_1'\subseteq C_1, C_2'\subseteq C_2}
				\DP'[x,p,c(x) \cup C_1\setminus C_1',\{\hat c(x)\} \cup C_2\setminus C_2']
				\cdot \DP[y_{p+1},C_1',C_2'].
		\nonumber
	\end{eqnarray}
	Finally, set $\DP[x,C_1,C_2] := \DP'[x,q,C_1,C_2]$ for every~$C_1\subseteq [D]$, and every~$C_2\subseteq [k]$.
	
	Return \yes if~$\DP[\star,[D],C_2]$ stores~1 for some~$C_2 \subseteq [k]$.
	Otherwise, return \no.
	
	\proofpara{Correctness}
	The correctness can be shown analogously to the correctness of\linebreak Lemma~\ref{lem:PDD-k-stars}.
	
	\proofpara{Running time}
	The tables $\DP$ and $\DP'$ contain $\Oh(2^{D+k} \cdot n\cdot m)$ entries.
	Each entry in the base case can be computed in $\Oh(D^2 \cdot n)$~time.
	In Recurrence~(\ref{eqn:recurrence-D}), each color can occur either in $C_1'$, in $c(x) \cup C_1 \setminus C_1'$ or not in $c(x) \cup C_1$.
	Likewise, with $\{\hat c(x)\} \cup Z$ and $Z'$.
	Therefore, all values in Recurrence~(\ref{eqn:recurrence-D}) can be computed in $\Oh(3^{D+k} \cdot n\cdot m)$ time
	which is thus also the overall running time.
\end{proof}

We can now prove~Theorem~\ref{thm:PDD-D} by showing how the standard \PDD can be reduced to \cDPDD via color coding.
This proof is analogous to the proof of Theorem~\ref{thm:PDD-k-stars}.

\begin{proof}[Proof of Theorem~\ref{thm:PDD-D}]
	\proofpara{Reduction}
	Let $\Instance = (\Tree, \Food, k, D)$ be an instance of \PDD.
	We assume that \Food only has one source by~Observation~\ref{obs:PDD-one-source}
	and that $\max_\w < D$ by~Reduction Rule~\ref{rr:maxw<D}.
	
	Let~$e_1, \dots, e_{|E(\Tree)|}$ and~$x_1, \dots, x_{n}$ be an order of the edges and taxa of~\Tree, respectively.
	We define integers~$W_0 := 0$ and $W_j := \sum_{i=1}^{j} \w(e_{i})$ for each $j\in [|E(\Tree)|]$.
	Set~$W := W_{|E(\Tree)|}$.
	
	Compute a $(W, D)$-perfect hash family $\mathcal{H}_D$ and an~$(n,k)$-perfect hash family $\mathcal{H}_k$.  
	For every $g \in \mathcal{H}_k$,
	let $c_{g,2}$ be a coloring such that $c_{g,2}(x_j) = g(j)$ for each~$x_j \in X$.
	For every $f \in \mathcal{H}_D$,
	let $c_{f,1}$ be a coloring such that $c_{f,1}(e_j) = \{f(W_{j-1}+1), \dots, f(W_j)\}$ for each $e_j \in E(\Tree)$.
	
	For hash functions $f \in \mathcal{H}_D$ and~$g\in \mathcal{H}_k$,
	construct an instance $\Instance_{f,g}$ of \cDPDD with $\Instance_{f,g} := (\Tree, \Food, k, D, c_{f,1}, c_{g,2})$.
	Solve every instance~$\Instance_{f,g}$ using~Lemma~\ref{lem:PDD-colored-D} and return \yes if and only if $\Instance_{f,g}$ is a \yes-instance for some~$f \in \mathcal{H}_D$, and some $g\in \mathcal{H}_k$.

	\proofpara{Correctness}
	We first show that if \Instance is a \yes instance then $\Instance_{f,g}$ is a \yes-instance for some~$f \in \mathcal{H}_D$, $g\in \mathcal{H}_k$.
	For any set of edges $E'$ with $\w(E') \ge D$, there is a corresponding subset of~$[W]$ of size at least $D$.
	Since $\mathcal{H}_D$ is a $(W, D)$-perfect hash family,~$c_{f,1}(E') = [D]$, for some $f \in \mathcal{H}_D$.
	Analogously, for each set $X'$ of taxa of size~$k$ there is a hash function~$g\in \mathcal{H}_k$ such that $c_{g,2}(X') = [k]$.
	Thus in particular, if~$S\subseteq X$ is a solution of size~$k$ for instance $\Instance$, then $c_{f,1}(S) = [D]$, for some $f \in \mathcal{H}_D$ and $c_{g,2}(S) = [k]$, for some $g \in \mathcal{H}_k$.
	It follows that one of the constructed instances of \cDPDD is a \yes-instance.
	
	Conversely, a solution for $\Instance_{f,g}$ for some~$f \in \mathcal{H}_D$, and some $g\in \mathcal{H}_k$ is also a solution for $\Instance$.

	\proofpara{Running Time}
	We apply Reduction Rule~\ref{rr:maxw<D} and Observation~\ref{obs:PDD-one-source} in $\Oh(n^2)$~time.\lb
	By Observation~\ref{obs:PDD-one-source},~$k' \le k+1$ and~$D' \le 2D+1$.
	We can construct~$\mathcal{H}_D$\lb and~$\mathcal{H}_k$ in~$e^{D'} {D'}^{\Oh(\log {D'})} \cdot W \log W$ time.
	By Lemma~\ref{lem:PDD-colored-D}, instances of \cDPDD\lb can be solved in~$\Oh(3^{{D'}+{k'}}\cdot nm)$~time each, and the number of instances\lb
	is $|\mathcal{H}_{D'}| \cdot |\mathcal{H}_{k'}| = e^{D'} {D'}^{\Oh(\log {D'})} \cdot \log W \cdot e^{k'} {k'}^{\Oh(\log {k'})} \cdot \log n \in e^{D'+k' + o(D)} \cdot \log W$.
	
	Thus, the total running time is 
	$\Oh(e^{2D+k + o(D)} \cdot \log W \cdot (W + 3^{2D+k}\cdot nm))$.
	This simplifies to $\Oh((3e)^{2D + k + o(D)} \cdot nm + n^2)$,
	as~$W = \PD(X) < 2n\cdot D$.
\end{proof}

\subsection{No Poly Kernels For $D$}
In this subsection, first, we give a reduction from \SC to \sPDD to show that \sPDD does not admit a polynomial kernelization algorithm when parameterized by~$D$,
assuming~\NPcoNPpoly~(Theorem~\ref{thm:PDD-D-kernel-PDss}).
This result also holds for \PDD.
However, we afterward provide a compression from \textsc{Graph Motif} to \PDD to show the following stronger result.
Even if $\Food$, is a directed forest, \PDD does not admit a polynomial kernelization algorithm when parameterized by~$D$,
assuming~\NPcoNPpoly~(Theorem~\ref{thm:PDD-D-kernel-PDts}).

In the following, we reduce from \SC to \sPDD.
In \SC, an input consists of a universe $\mathcal{U}$, a family $\mathcal{Q}$ of subsets of $\mathcal{U}$, and an integer $k$.
It is asked whether there exists a sub-family of sets $\mathcal{Q'} \subseteq \mathcal{Q}$ such that $\mathcal{Q'}$ has a carnality of at most $k$ and the union of~$\mathcal{Q'}$ covers the entire universe.
Assuming \NPcoNPpoly, \SC does not admit a polynomial kernel when parameterized by the size of the universe $|\mathcal{U}|$~\cite{dom,cygan}.

Our reduction from \SC to \sPDD is similar to the reduction from \VC to \sPDD presented in~\cite{faller}.

\begin{theorem} \label{thm:PDD-D-kernel-PDss}
	\sPDD does not admit a polynomial kernelization algorithm with respect to $D$,
	assuming \NPcoNPpoly.
\end{theorem}
\begin{proof}
	\proofpara{Reduction}
	Let an instance $\Instance = (\mathcal{Q},\mathcal{U},k)$ of \SC be given.
	We define an instance $\Instance' = (\Tree,\Food,k',D)$ of \sPDD as follows.
	Let \Tree be a star with root $\rho$ and leaves~$X := \mathcal{Q} \, \cup \, \mathcal{U}$.
	We set $\w(\rho Q) = 1$ for each $Q\in \mathcal{Q}$ and $\w(\rho u) = 2$ for each~$u\in \mathcal{U}$.
	Further, the food-web \Food is a graph with vertices $X$ and we add an edge~$Qu$ to \Food if and only if~$u\in Q$ for~$u\in \mathcal{U}$ and~$Q\in \mathcal{Q}$.
	Finally, we set $k' := k + |\mathcal{U}|$ and $D := k + 2|\mathcal{U}|$.
	
	\proofpara{Correctness}
	We may assume that $k\le |\mathcal{U}|$.
	Therefore, $D$ is bounded in $|\mathcal{U}|$.
	
	Now, assume that $\mathcal{Q'}$ is a solution for \Instance.
	If necessary, add further sets to $\mathcal{Q'}$ until $|\mathcal{Q'}| = k$.
	Because $\mathcal{Q'}$ is a solution for \Instance, for each $u\in\mathcal{U}$ there is a $Q\in\mathcal{Q'}$ such that $u\in Q$.
	Hence, $S := \mathcal{Q'} \cup \mathcal{U}$ is viable, has a size of $|S| = k + |\mathcal{U}| = k'$ and~$\PD(S) = |\mathcal{Q'}| + 2|\mathcal{U}| = k + 2|\mathcal{U}|$.
	
	Conversely,
	let $S$ be a solution for instance $\Instance'$ and assume $|S| = k'$.
	Let $S_{\mathcal{Q}}$ be the intersection of $S$ with $\mathcal{Q}$.
	Define $a := |S_{\mathcal{Q}}|$ and $b := |S|-a$.
	Then, we have~$\PD(S) = a + 2b = |S| + b$ and~$\PD(S) \ge k + 2|\mathcal{U}| = k' + |\mathcal{U}|$.
	We conclude that $b \ge |{\mathcal{U}}|$ and so $a = |S|-b \le k'-|{\mathcal{U}}| = k$.
	Since $S$ is viable, for each $u\in \mathcal{U}$ there is a $Q\in S_{\mathcal{Q}}$ with $u\in Q$.
	Therefore, $S_{\mathcal{Q}}$ is a solution for \Instance.
\end{proof}

For \PDD we want to show the non-existence of a polynomial kernel even in the case that the food-web is restricted to a forest.
Here, we define a cross-composition from \textsc{Graph Motif}, a problem in which one is given a graph $G$ with vertex-coloring~$\chi$ and a multiset of colors $M$.
It is asked whether $G$ has a connected set of vertices whose multiset of colors equals $M$.
\textsc{Graph Motif} was shown to be \NP-hard even on trees~\cite{lacroix}.
It remains \NP-hard to compute a solution for an instance of \textsc{Graph Motif} on trees of maximum vertex-degree three, even if each item appears at most once in~$M$, and there is a color $c^*\in M$ that only one vertex of $G$ takes~\cite{fellowsmotif}.

\begin{theorem} \label{thm:PDD-D-kernel-PDts}
	Even if the given food-web~$\Food$ is a directed forest,
	\PDD does not admit a polynomial kernelization algorithm with respect to $D$,
	assuming \NPcoNPpoly.
\end{theorem}
\begin{proof}
	We describe a cross-composition from \textsc{Graph Motif} to \PDD.
	We refer readers unfamiliar to cross-compositions to~\cite[Chapter~15.1]{cygan}.
	
	\proofpara{Cross-Composition}
	Fix a set of colors $M$.
	Let $\Instance_1,\dots,\Instance_q$ be instances of \textsc{Graph Motif} with $\Instance_i = (G_i = (V_i,E_i),M)$ such that all $G_i$ are trees with maximum vertex-degree of three.
	Let $v_i$ be the only vertex of $V_i$ with $\chi(v_i) = c^* \in M$.
	
	We define an instance $\Instance = (\Tree,\Food,k,D)$ of \PDD as follows.
	Let \Tree be a tree with vertex set $\{\rho\} \cup M \cup X$ and leaves $X := \bigcup_{i=1}^q V_i$.
	Add an edge $\rho c$ to \Tree for each~$c\in M$ and add an edge $cv$ if and only if $\chi(v) = c$.
	Each edge has a weight of~1.
	Orient each edge in $G_i$ away from~$v_i$ to obtain $H_i$ for each $i\in [q]$.
	Let \Food be the food-web which is the union of $H_1,\dots,H_q$.
	We define $k := |M|$ and $D := 2\cdot|M|$.

	\proofpara{Correctness}
	We show that some instance of $\Instance_1,\dots,\Instance_q$ is a \yes-instance of \textsc{Graph Motif} if and only if \Instance is a \yes-instance of \PDD.
	
	Let $\Instance_i$ be a \yes-instance of \textsc{Graph Motif} for an $i\in [q]$.
	Consequently, there is a set of vertices $S\subseteq V_i$ such that $|S| = |M|$, $\chi(S) = M$ and $G_i[S]$ is connected.
	We conclude that~$v_i\in S$ because~$v_i$ is the only vertex with $\chi(v_i) = c^*$.
	Thus, $H_i[S]$ is a connected subtree of $H_i$ which contains $v_i$, the only source in~$H_i$.
	We conclude that~$S$ is viable in \Instance.
	By definition,~$|S| = |M|$ and because each vertex in~$S$ has another color we conclude that~$\PD(S) = 2 \cdot |M|$.
	Hence, $S$ is a solution for \Instance.
	
	Conversely, let \Instance be a \yes-instance of \PDD.
	Consequently, there is a viable set~$S$ of taxa such that $|S| \le |M|$ and $\PD(S) \ge 2 \cdot |M|$.
	We say that a taxon $x\in X$ has color~$c\in M$ if~$cx$ is an edge in \Tree.
	Observe that~$\PD(A\cup\{x\}) \le \PD(A) + 2$ for any set of taxa $A$.
	Further,~$\PD(A\cup\{x\}) = \PD(A) + 2$ if and only if~$x$ has a color that none of the taxa of~$A$ has.
	We conclude that the taxa in $S$ have unique colors.
	Fix~$j\in [q]$ such that $v_j\in S$.
	The index~$j$ is uniquely defined because there is a unique taxon in $S$ with the color $c^*$.
	Because the vertices~$v_1,\dots,v_q$ are the only sources in \Food and $S$ is viable, we conclude $u\in V_i$ can not occur in $S$ for~$i\ne j$.
	Therefore, we conclude $S\subseteq V_j$.
	Because $S$ is viable, $H_i[S]$ is connected and therefore also~$G_i[S]$.
	Thus, $S$ is a solution for instance $\Instance_j$ of \textsc{Graph Motif}.
\end{proof}

\section{The Loss of Species $\kbar$ and Diversity $\Dbar$}
\label{sec:PDD-kbar}
\subsection{Hardness For the Acceptable Diversity Loss $\Dbar$}
In some instances, the diversity threshold~$D$ may be very large. Then, however, the acceptable loss of diversity~\Dbar would be relatively small.
Recall, $\Dbar$ is defined as $\PD(X) - D$.
Encouraged by this observation, recently, several problems in maximizing phylogenetic diversity have been studied with respect to the acceptable diversity loss~\cite{MAPPD,TimePD}.
In this section, we show that, unfortunately, \sPDD is~$\Wh 1$-hard with respect to $\Dbar$ even if edge-weights are at most two.

To show this result, we reduce from \rbnb.
In \rbnb, the input is an undirected bipartite graph $G$ with vertex bipartition~$V=V_r \cup V_b$ and an integer $k$.
The question is whether there is a set $S\subseteq V_r$ of size at least $k$ such that the neighborhood of $V_r \setminus S$ is $V_b$.
\rbnb is \Wh 1-hard when parameterized by the size of the solution $k$~\cite{downey}.
\begin{proposition}
	\label{prop:PDD-Dbar}
	\sPDD is $\Wh 1$-hard with respect to $\Dbar$, even if $\max_\w=2$.
\end{proposition}
\begin{proof}
	\proofpara{Reduction}
	Let $\Instance := (G=(V=V_r \cup V_b,E),k)$ be an instance of \rbnb.
	%	and let the red vertices be labeled $V_r = \{v_1,\dots,v_{|V_r|}\}$
	%	Let $M$ be a sufficiently large integer, i.e. $M > |V_r|$, and
	%	let $W_1,\dots,W_{|V_r|}$ be vertex-sets of size $M$.
	%	Define $W$ to be the union $W_1 \cup \dots \cup W_{|V_r|}$.
	We construct an instance $\Instance' = (\Tree, \Food, k', D)$ of \sPDD as follows.
	Let \Tree be a star with root $\rho\not\in V$ and leaves $V$.
	In \Tree, an edge $e = \rho u$ has a weight of~1 if $u\in V_r$ and
	otherwise~$\w(e) = 2$, if $u\in V_b$.
	Define a food-web \Food with vertices~$V$ and
	for each edge~$\{u,v\}\in E$, and every pair of vertices $u\in V_b$, $v\in V_r$, add an edge~$uv$ to~\Food.
	Finally, set $k' := |V| - k$ and~$D := 2\cdot |V_b| + |V_r| - k$, or equivalently~$\kbar = \Dbar = k$.
	
	\proofpara{Correctness}
	The reduction can be computed in polynomial time.
	We show that if \Instance is a \yes-instance of \rbnb then~$\Instance'$ is a \yes-instance of \PDD. Afterward, we show the converse.
	
	Assume that $\Instance$ is a \yes-instance of \rbnb.
	Therefore, there is a set~$S\subseteq V_r$ of size at least $k$ such that $N_G(V_r\setminus S) = V_b$.
	We assume $|S| = k$ as $N_G(V_r\setminus S) = V_b$ still holds if we shrink $S$.
	We define $S' := V\setminus S$ and show that $S'$ is a solution for $\Instance'$.
	The size of~$S'$ is $|V\setminus S| = |V| - |S| = k'$.
	Further, $\PD(S) = 2\cdot |V_b| + |V_r\setminus S| = 2\cdot |V_b| + |V_r| - k = D$.
	By definition, the vertices in $V_r$ are sources.
	Further, because $S$ is a solution for \Instance, each vertex of~$V_b$ has a neighbor in $V_r\setminus S$.
	So, $S'$ is viable and $\Instance'$ is a \yes-instance of \sPDD.
	
	Conversely, let $S'\subseteq V$ be a solution for instance $\Instance'$ of \sPDD.
	Without loss of generality, $S'$ contains $r$ vertices from $V_r$ and $b$ vertices of $V_b$.
	Consequently, $|V| - k \ge |S'| = b + r$ and $2\cdot |V_b| + |V_r| - k = D \le \PD(S') = 2b + r$.
	Therefore, $r \le |V| - k - b$ and so~$2b \ge 2\cdot |V_b| + |V_r| - k - r \ge 2\cdot |V_b| + |V_r| - k - (|V| - k - b) = |V_b| + b$.
	We conclude $b = |V_b|$ and $V_b \subseteq S'$.
	Further, $r = |V_r| - k$.
	We define $S := V_r \setminus S'$ and conclude $|S| = |V_r| - r = k$.
	Because $S'$ is viable, each vertex in $V_b$ has a neighbor in $S'\setminus V_b$.
	Therefore, $S$ is a solution for the \yes-instance \Instance of \rbnb.
\end{proof}

\subsection{Parameter $\kbar$ When Each Taxon Has At Most One Prey}
It is known that \PDD remains \NP-hard if the food-web is acyclic and every vertex has at most one prey~\cite{faller}.
By Proposition~\ref{prop:PDD-Dbar} we know that \PDD is \Wh{1}-hard when parameterized by the acceptable loss of diversity $\Dbar$ or the minimum number of extincting taxa $\kbar$.
In the following, we show that in the special case that each taxon has at most one prey, \PDD is \FPT when parameterized with~$\kbar$.
Observe that each taxon has at most one prey if and only if the food-web is an out-forest.

\begin{proposition} \label{prop:PDD-kbar+indeg}
	\PDD can be solved in time~$\Oh(2^{3\kbar + o(\kbar)}n\log n + n^5)$, if each taxon has at most one prey.
\end{proposition}

In order to show the claimed result we again resort to the technique of color coding.
We define a colored version of the problem, show how to solve it,
and at the end of the section, we show how to reduce instances of the uncolored problem to the colored version.

We define \ckPDDlong (\ckPDD), a colored version of the problem, as follows.
Additionally to the usual input $\Tree$, $\Food$, $k$, and $D$ of \PDD, we receive a coloring $c$ which assigns each taxon $x\in X$ either~0 or~1.
In \ckPDD, we ask whether there is a set~$S\subseteq X$ that holds each of the following
\begin{propEnum}
	\item the size of $S$ is at most $k$,
	\item the phylogenetic diversity of $S$ is at least $D$,
	\item $S$ is viable,
	\item each taxon $x\in X\setminus S$ satisfies $c(x)=0$,
	\item each neighbor $y$ of $X\setminus S$ in the food-web \Food satisfies $c(y) = 1$, and
	\item If~$\off(v) \subseteq X\setminus S$ then $\off(u) \subseteq X\setminus S$ or there is a taxon $y\in \off(u)$ with~$c(y) = 1$ for each edge $uv\in E(\Tree)$.
\end{propEnum}

Observe that if the color $c(x)$ of a taxon $x\in X$ is~1, then $x$ has to be saved,
while~$x$ may be saved or may go extinct if the color $c(x)$ is~0.

\begin{observation} \label{obs:PDD-equivalent-extinction-1}
	Given a \yes-instance $\Instance = (\Tree,\Food,k,D,c)$ of \ckPDD with\lb solution~$S$ and
	a (possibly internal) vertex $w\in V(\Tree)$ with $c(x)=0$ for\lb each $x\in \off(w)$.
	Then either $\off(w) \subseteq S$ or $\off(w) \subseteq X\setminus S$.
\end{observation}
\begin{proof}
	If $\off(w) \subseteq S$, we are done.
	
	Let there be a taxon $x\in \off(w)$ which is not in $S$.
	Let $v$ be the descendant of~$x$ (possibly~$v=x$) such that $\off(v) \subseteq X\setminus S$
	and $\off(u) \cap S \ne \emptyset$ for the parent~$u$ of~$v$.
	Because~$S$ is a solution, $\off(v) \subseteq X\setminus S$, and $\off(u) \cap S \ne \emptyset$, we conclude that there is a taxon~$y\in \off(u)$ with~$c(y) = 1$.
	Therefore, $w$ is a descendant of $v$ and we conclude that~$\off(w) \subseteq \off(v) \subseteq X\setminus S$.
\end{proof}

\begin{definition}
	~
	\begin{propEnum}
		\item[a)] Given a set of taxa and a coloring $c: X\to \{0,1\}$, we define~$c^{-1}(0)$ and $c^{-1}(1)$ to be the subsets of $X$ that $c$ maps to~0 or~1, respectively.
		\item[b)] Given a food-web \Food and a coloring $c$, we define $\Food_{c,0}$ to be the underlying undirected graph of \Food induced by $c^{-1}(0)$.
	\end{propEnum}
\end{definition}

\begin{observation}
	\label{obs:PDD-equivalent-extinction-2}
	Given a \yes-instance $\Instance = (\Tree,\Food,k,D,c)$ of \ckPDD\linebreak with solution~$S$ and
	let $C$ be a connected component of $\Food_{c,0}$.
	Then either~$V(C) \subseteq S$\linebreak or~$V(C) \subseteq X\setminus S$.
\end{observation}
\begin{proof}
	If $V(C) \subseteq S$, we are done.
	
	Since $S$ is a solution, each neighbor $y$ of $X\setminus S$ in the food-web satisfies $c(y) = 1$.
	Therefore, if there is a taxon $x\in V(C)$ which is not in $S$,
	then each neighbor of $x$ is in $X\setminus S$ or in~$c^{-1}(1)$.
	By definition, $C$ is a connected component and $V(C) \subseteq c^{-1}(0)$.
	We conclude that~$V(C) \subseteq X\setminus S$.
\end{proof}

\begin{lemma} \label{lem:PDD-c-kbar+outdeg}
	\ckPDD can be solved in $\Oh(n^4)$ time if each taxon has at most one~prey.
\end{lemma}
\begin{proof}
	\proofpara{Intuition}
	We reduce an instance $\Instance = (\Tree, \Food, k,D,c)$ of \ckPDD to an instance of \KP, in which the threshold of profit is limited in $\kbar$.
	In \KP we are given a set of items~$\mathcal{A}$, a cost-function $c: \mathcal{A}\to \mathbb{N}$, a value-function $d: \mathcal{A}\to \mathbb{N}$, and two integers $B$---the budget---and~$D$---the required value.
	It is asked whether there is a set $A \subseteq \mathcal{A}$ such that~$c_\Sigma(A) \le B$ and~$d_\Sigma(A) \ge D$.
	\KP is \NP-hard~\cite{karp} and can be solved in $\Oh(D\cdot |\mathcal A|^2)$ time~\cite{rehs}.
	
	\proofpara{Algorithm}
	Compute the set $Z$ of edges $uv$ of \Tree which hold that $\off(v) \subseteq c^{-1}(0)$ and there is a taxon $y\in \off(u)$ with $c(y) = 1$.
	For each $e=uv\in Z$, define integers~$p_v := |\off(v)|$ and~$q_v := \w(e) + \sum_{e'\in E(\Tree_v)} \w(e')$, where~$\Tree_v$ is the subtree of \Tree rooted at $v$.
	Intuitively, $p_v$ is the number of taxa and $q_v$ is the diversity that would be lost if all taxa in $\off(v)$ would go extinct.
	
	Define a graph $G$ on the vertex set $V_G := \{ v \mid uv \in Z \}$.
	Compute the connected components $C_1,\dots,C_\ell$ of $\Food_{c,0}$.
	Iterate over $C_i$.
	Compute the edges $u_1v_1, \dots, u_qv_q \in Z$ with $\off(v_j) \cap V(C_i) \ne \emptyset$ for $j\in [q]$.
	Make $v_1,\dots,v_q$ a clique in $G$.
	
	Compute the set $\mathcal A$ of connected components of $G$.
	
	Iterate over $v\in V_G$.
	If there is a taxon $y\in c^{-1}(1)$ and in \Food there is a path from some taxon $x\in \off(v)$ to $y$,
	then remove the connected component $A$ from $\mathcal A$.
	
	We define a \KP instance $\Instance' := (\mathcal A,c',d,B,D')$ as follows.
	The set of items is~$\mathcal A$,
	we define the budget~$B$ as $\Dbar = \PD(X) - D$ and the desired profit~$D'$ as $\kbar$.
	The cost- and value-function are defined as follows.
	For an item $A\in \mathcal A$, we define
	$c'(A)$ to be $\sum_{v\in V(A)} q_v$, and
	$d(A)$ to be $\sum_{v\in V(A)} p_v$.

	If $\Instance'$ is a \yes-instance of \KP, we return \yes.
	Otherwise, we return \no.

	\proofpara{Correctness}
	Observe that $\{\off(v) \mid uv\in Z \}$ and $\{C_1,\dots,C_\ell\} =: \mathcal{C}_\Food$ are both partitions of~$c^{-1}(0)$.
	We define the set~$X(A) := \{ x\in X \mid x\in \off(v), v\in V(A) \}$, for connected components $A \in \mathcal{A}$ of $G$.
	We show that the instance \Instance of \ckPDD is a \yes-instance if and only if the instance $\Instance'$  of \KP is a \yes-instance.
	
	Let $\Instance'$ be a \yes-instance of \KP.
	So, there is a set $S\subseteq \mathcal{A}$ with $d_\Sigma(S) \ge \kbar$ and $c'_\Sigma(S) \le \PD(X) - D$.
	Let $P\subseteq V_G$ be the union of the vertices $V(A)$ for~$A$ in~$S$ and
	let $Q \subseteq X$ to be the set union of taxa $X(A)$ for~$A$ in~$S$.
	We want to show that $R := X \setminus Q$ is a solution of the instance \Instance of \ckPDD.
	Because $d_\Sigma(S) \le \kbar$ we conclude $\kbar \ge \sum_{v\in P} p_v = \sum_{v\in P} |\off(v)| \ge |Q|$.
	Consequently, the size of $R$ is at most~$|X| - \kbar = k$.
	By the definition of $Z$, we know that each $v\in P$ has a parent $u$ which satisfys that there is a taxon $y\in \off(u)$ with $c(y) = 1$.
	Consequently,
	\begin{eqnarray}
		\PD(R) = \PD(X) - \sum_{v\in P} q_v = \PD(X) - d_\Sigma(S) \ge \PD(X) - \Dbar = D.
	\end{eqnarray}
	Let $x\in Q$ and $y \in X$ be taxa which satisfy that $y$ can be reached from $x$ in \Food.
	Because~$x$ is in~$Q$, there is a connected component $A\in S$ of $G$ and a vertex $v\in V(A)$ that satisfies~$x \in \off(v)$.
	By the construction, we know that each taxon that can be reached from~$x$, including~$y$, are colored with~0 because~$A$ would have been removed from~$\mathcal A$, otherwise.
	Consequently, $y$ is in the same connected component as~$x$ in~$\Food_{c,0}$ and thus there is a vertex~$v' \in V(A)$ that satisfies~$y\in \off(v')$.
	Therefore, $y$ is also in~$Q$ and we conclude that $R$ is viable.
	By the definition of~$\mathcal{A}$, we conclude $Q \subseteq c^{-1}(0)$.
	By the construction the taxa of a connected component of $\Food_{c,0}$ are a subset of the union $\bigcup_{v\in V(A)} \off(v)$ for some unique~$A\in \mathcal{A}$.
	Therefore, the neighbors of~$Q$ in~\Food are colored with~1.
	So,~$R$ is a solution of the \yes-instance~\Instance of~\ckPDD.

	Conversely, let \Instance be a \yes-instance of \ckPDD with solution $S$.
	Let~$Z$ be the set of edges as defined in the algorithm
	and let $C_1,\dots,C_\ell$ be the connected components of $\Food_{c,0}$.
	Let $\mathcal{A}$ be the set of connected components of $G$.
	By Observations~\ref{obs:PDD-equivalent-extinction-1} and~\ref{obs:PDD-equivalent-extinction-2}, we conclude that for each connected component $A$ in $\mathcal{A}$ either $X(A) \subseteq S$ or~$X(A) \subseteq X\setminus S$.
	Now, let $\mathcal{S} \subseteq \mathcal{A}$ be the set of connected components $A$ of $G$ that satisfy $X(A) \subseteq X\setminus S$.
	We observe that $X\setminus S$ has a size of $\kbar$ and $\PD(S) \ge D$, as~$S$ is a solution of \Instance
	and so $d_\Sigma(\mathcal{S}) = \sum_{A\in \mathcal{S}, v\in A} p_v = \sum_{A\in \mathcal{S}, v\in A} |\off(v)| = |X\setminus S| \ge \kbar = D'$.
	Further, $c'_\Sigma(\mathcal{S}) = \sum_{A\in \mathcal{S}, v\in A} q_v = \PD(X) - \PD(S) \ge \PD(X) - D = B$.
	Hence, $\mathcal{S}$ is a solution to the \yes-instance $\Instance'$ of \KP.

	\proofpara{Running time}
	Observe that $|E(\Food)| \in \Oh(n)$ because each taxon has at most one prey.
	The size of $Z$ is at most $n$ and therefore also the size of $|\mathcal{A}|$.
	All steps in the reduction can be computed in $\Oh(n^4)$ time.
	Defining the instance $\Instance'$ is done in~$\Oh(n)$~time.
	Computing whether~$\Instance'$ is a \yes-instance of \KP can be done in~$\Oh(\kbar \cdot n^2)$~time.
	Therefore, the overall running time is $\Oh(n^4)$.
\end{proof}

It remains to show how to utilize the result of Lemma~\ref{lem:PDD-c-kbar+outdeg} to compute a solution for an instance of \PDD.
Other than in the proofs of Theorems~\ref{thm:PDD-k-stars},~\ref{thm:PDD-k+height}, and~\ref{thm:PDD-D}, we resort on the concept of $(n,k)$-universal sets instead of $(n,k)$-perfect hash functions.

\begin{proof}[Proof (of Proposition~\ref{prop:PDD-kbar+indeg})]
	\proofpara{Algorithm}
	Let $\Instance = (\Tree,\Food,k,D)$ be an instance of \PDD.
	Let $x_1,\dots,x_n \in X$ be an arbitrary order of the taxa.
	
	Compute an $(n,3\kbar)$-universal set~$\mathcal{U}$.
	Iterate over $A \in \mathcal{U}$ and define $2$-color\-ings~$c_A: X\to \{0,1\}$ by setting $c_A(x_i) := 1$ if and only if $i\in A$.
	Then, solve the instances $\Instance_A := (\Tree, \Food, k, D, c_A)$ of \ckPDD with Lemma~\ref{lem:PDD-c-kbar+outdeg}.
	Return \yes if there is an $A \in \mathcal{U}$, such that $\Instance_A$ is a \yes-instance.
	Otherwise, return \no.
	
	\proofpara{Correctness}
	If $\Instance_A$ for some $A\in \mathcal{U}$ is a \yes-instance of \ckPDD and let~$S\subseteq X$ be a solution.
	Then, by the definition, $S$ is viable, $|S| \le k$, and $\PD(S) \ge D$.
	Therefore, instance \Instance is a \yes-instance of \PDD with solution $S$.
	
	Assume that \Instance is a \yes-instance of \PDD with solution $S$.
	Define~$Y := X\setminus S$---the species that die out.
	Observe that $|Y| = |X| - |S| \ge |X| - k = \kbar$.
	Let $u_1v_1,\dots,u_\ell v_\ell$ be the edges in~$E(\Tree)$ that satisfy $\off(v_i) \subseteq Y$ and there is a taxon $\bar x_i \in \off(u_i) \setminus Y$.
	Define $Z_1 := \{\bar x_1, \dots, \bar x_\ell\}$ and observe $|Z_1| \le \ell \le |Y| \le \kbar$.
	Now, let $Z_2$ be the set of neighbors of~$Y$ in~$\Food$.
	Because each taxon has at most one prey, we conclude that if $x\in Y$ then also $\predators{x} \subseteq Y$.
	Therefore, each taxon in $Z_2$ is the prey of at least one taxon of $Y$.
	Consequently, $|Z_2| \le |Y| \le \kbar$.
	Define $Z := Y \cup Z_1 \cup Z_2$ and by the previous, we know $|Z| \le 3\kbar$.
	If necessary, add taxa to $Z$, such that contains $3\kbar$ taxa.
	Define sets~$Y' := \{i \mid x_i \in Y\}$ and $Z' := \{i \mid x_i \in Z\}$.
	
	Because $\mathcal{U}$ is an $(n,3\kbar)$-universal set, $\{A \cap S \mid A \in \mathcal U\}$ contains all $2^{3\kbar}$ subsets of~$S$ for any $S\subseteq [n]$ of size $3\kbar$.
	Thus, there is an $A^* \in \mathcal{U}$ such that $A \cap Z' = Z'\setminus Y'$.
	Therefore, $c_{A^*}$ maps each $y\in Y$ to 0, each $z\in Z_1 \cup Z_2$ to 1 and each other taxon to any value of $\{0,1\}$.
	We conclude that the instance $\Instance_{A^*} := (\Tree, \Food, k, D, c_{A^*})$ is a \yes-instance of \ckPDD.

	\proofpara{Running Time}
	$(n,3\kbar)$-universal sets of size $2^{3\kbar}(3\kbar)^{\Oh(\log (3\kbar))}\log n = 2^{3\kbar}2^{\Oh(\log^2 (\kbar))}\log n$ are constructed in time $2^{3\kbar}(3\kbar)^{\Oh(\log (3\kbar))}n\log n = 2^{3\kbar} 2^{\Oh(\log^2 (\kbar))}n\log n$~\cite{Naor1995SplittersAN}.
	
	For a given $A\in \mathcal{U}$, we can construct $\Instance_A$ in $\Oh(|A|) = \Oh(n)$ time.
	By Lemma~\ref{lem:PDD-c-kbar+outdeg}, we can compute whether $\Instance_A$ is a \yes-instance in $\Oh(n^4)$ time.
	Therefore, the overall running time is~$\Oh(2^{3\kbar} 2^{\Oh(\log^2 (\kbar))}n\log n + n^5) = \Oh(2^{3\kbar + o(\kbar)}n\log n + n^5)$.
\end{proof}

\section{Structural Parameters}
\label{sec:PDD-structural}
In this section, we study how the structure of the food-web affects the complexity of \sPDD and \PDD.
Figure~\ref{fig:PDD-PDD-results} and Table~\ref{tab:PDD-results} depict a summary of these complexity~results.
\begin{figure}
	\tikzstyle{para}=[rectangle,draw=black,minimum height=.8cm,fill=gray!10,rounded corners=1mm, on grid]
	
	\newcommand{\tworows}[2]{\begin{tabular}{c}{#1}\\{#2}\end{tabular}}
	\newcommand{\threerows}[3]{\begin{tabular}{c}{#1}\\{#2}\\{#3}\end{tabular}}
	\newcommand{\distto}[1]{\tworows{Distance to}{#1}}
	\newcommand{\disttoc}[2]{\threerows{Distance to}{#1}{#2}}
	
	\DeclareRobustCommand{\tikzdot}[1]{\tikz[baseline=-0.6ex]{\node[draw,fill=#1,inner sep=2pt,circle] at (0,0) {};}}
	\DeclareRobustCommand{\tikzdottc}[2]{\tikz[baseline=-0.6ex]{\node[draw,diagonal fill={#1}{#2},inner sep=2pt,circle] at (0,0) {};}}
	
	\definecolor{r}{rgb}{1, 0.6, 0.6}
	\definecolor{g}{rgb}{0.8, 1, 0.7}
	\definecolor{grey}{rgb}{0.9453, 0.9453, 0.9453}
	\definecolor{grey2}{rgb}{0.85, 0.85, 0.85}
	
	\tikzset{
		diagonal fill/.style 2 args={fill=#2, path picture={
				\fill[#1, sharp corners] (path picture bounding box.south west) -|
				(path picture bounding box.north east) -- cycle;}},
		reversed diagonal fill/.style 2 args={fill=#2, path picture={
				\fill[#1, sharp corners] (path picture bounding box.north west) |- 
				(path picture bounding box.south east) -- cycle;}}
	}
	\centering
	\begin{tikzpicture}[node distance=3*0.45cm and 4.8*0.38cm, every node/.style={scale=0.7}]
		\linespread{1}
		% first
		\node[para,fill=g] (vc) {Minimum Vertex Cover};
		\node[para, diagonal fill=gr, xshift=38mm] (ml) [right=of vc] {Max Leaf \#};
		\node[para, xshift=-40mm,fill=g] (dc) [left=of vc] {\distto{Clique}};
		
		% second
		\node[para, fill=r, xshift=-25mm] (dsd) [below= of dc] {\distto{Source-Dominant}}
		edge[bend left=10] (dc);
		\node[para, diagonal fill=gr, xshift=11mm] (dcc) [below= of dc] {\distto{Cluster}}
		edge (dc)
		edge[bend left=10] (vc);
		\node[para, fill=g,xshift=41mm] (dcl) [below= of dc] {\distto{Co-Cluster}}
		edge[bend right=10] (dc)
		edge (vc);
		\node[para, diagonal fill=gr, xshift=8mm] (ddp) [below=of vc] {\distto{disjoint Paths}}
		edge (vc)
		edge (ml);
		\node[para,diagonal fill=gr] (fes) [below =of ml] {\tworows{Feedback}{Edge Set}}
		edge (ml);
		\node[para, xshift=2mm, diagonal fill=gr] (bw) [below right=of ml] {Bandwidth}
		edge (ml);
		\node[para, xshift=5mm, yshift=0mm,diagonal fill=gr] (td) [right=of ddp] {Treedepth}
		edge[bend right=28] (vc);
		
		% third
		\node[para, diagonal fill=gr] (fvs) [below= of ddp] {\tworows{Feedback}{Vertex Set}}
		edge (ddp)
		edge[bend right=5] (fes);
		\node[para,fill=r] (mxd) [below= of bw] {\tworows{Maximum}{Degree}}
		edge (bw);
		\node[para, xshift=8mm, diagonal fill=gr] (pw) [below= of td] {Pathwidth}
		edge (ddp)
		edge (td)
		edge[bend right=8] (bw);
		
		% fourth
		\node[para, yshift=0mm, diagonal fill=gr] (tw) [below=of pw] {Treewidth}
		edge (fvs)
		edge (pw);
		\node[para, xshift=5mm, fill=r] (dbp) [below left= of fvs] {\distto{Bipartite}}
		edge (fvs);

		\node[para, yshift=-20mm, diagonal fill={white}{grey2}] [below= of dc] {\tworows{{\PDD}\;\;\;\;\;\;\;\;\;\;\;}{\;\;\;\;\;\;\;\;\;\;\sPDD}};
	\end{tikzpicture}
	\caption{
		This figure depicts the relationship between a structural parameter of the food-web and the complexity of solving \PDD and \sPDD.
		A parameter~$\pi$ is marked
		green~(\tikzdot{g}) if \PDD admits an \FPT-algorithm with respect to~$\pi$,
		red~(\tikzdot{r}) if \sPDD is \NP-hard for constant values of~$\pi$,
		or red and green~(\tikzdottc gr) if \PDD is \NP-hard for constant values of~$\pi$ while \sPDD admits an \FPT-algorithm with respect to~$\pi$.
		Two parameters~$\pi_1$ and~$\pi_2$ are connected with an edge if in every graph the parameter~$\pi_1$ further up van be bounded by a function in $\pi_2$.
		A more in-depth look into the hierarchy of graph-parameters can be found in~\cite{graphparameters}.
	}
	\label{fig:PDD-PDD-results}
\end{figure}

Some of these results are already shown by Faller et al.~\cite{faller}.
These include that \PDD remains \NP-hard on instances in which the food-web is a bipartite graph~\cite{faller} and \sPDD is already \NP-hard if the food-web is a tree of height three~\cite{faller}.
Further, Faller et al.~\cite{faller} also gave a reduction from \VC to \PDD.
Because \VC is \NP-hard for graphs of maximum degree three~\cite{mohar}, also the constructed food-web has a maximum degree of three.

In the next subsection, we have an in-depth look into the parameterization by the distance to cluster, also called cluster vertex deletion number~($\distclust$) of the food-web.
There, we show that \PDD is \NP-hard even if the underlying undirected graph of the food-web is a cluster graph but \sPDD is \FPT when parameterized by~$\distclust$.
Afterward, we show that \PDD is \FPT when parameterized by the distance to co-cluster~$(\distcoclust)$
and that \sPDD is \FPT with respect to the treewith~$(\tw)$ of the food-web.
Consequently, \sPDD can be solved in polynomial time if the food-web is a tree, disproving Conjecture~4.2. in~\cite{faller}.

\subsection{Distance to Cluster}
In this subsection, we consider the special case that given an instance of \PDD or \sPDD, we need to remove only a few vertices from the food-web \Food to obtain a cluster graph in the underlying undirected graph.
Recall that a graph is a cluster graph if the existence of edges~$\{u,v\}$ and~$\{v,w\}$ imply the existence of the edge~$\{u,w\}$.

Because \Food is acyclic, we have the following: Every clique in~\Food has exactly one vertex $v_0 \in V(C)$ such that $v_0 \in \prey{v}$ for each~$v\in V(C)\setminus \{v_0\}$.
In fact, after applying Reduction Rule~\ref{rr:redundant-edges} exhaustively to any cluster graph, every connected component is a star in which all edges are directed away from the center.

In this subsection, we further use the following classification of instances.
Recall that~$\spannbaum{Y}$ is the spanning tree of the vertices in~$Y$.
\begin{definition}
An instance $\Instance = (\Tree, \Food, k, D)$ of \PDD or \sPDD is \emph{source-separating} if $\spannbaum{\{\rho\} \cup \sources}$ and $\spannbaum{\{\rho\} \cup (X\setminus \sources)}$ only have $\rho$ as common vertex.
Here, $\rho$ is the root of \Tree.
\end{definition}
Figure~\ref{fig:PDD-dist-cluster-flow} depicts in (1) an example of a source-separating instance.

In Theorem~\ref{thm:PDD-dist-cluster-FPT}, we show that \sPDD is \FPT with respect to the distance to cluster of the food-web.
Afterward, we show that \PDD, however, is \NP-hard on source-separating instances in which the food-web has a cluster graph as underlying graph.
Then, we show that \PDD can be solved in polynomial time in a further special case.

\begin{theorem}
	\label{thm:PDD-dist-cluster-FPT}
	\sPDD can be solved in $\Oh(6^d \cdot n^2 \cdot m \cdot k^2)$~time, when we are given a set~$Y\subseteq X$ of size $d$ such that $\Food-Y$ is a cluster graph.
\end{theorem}
To prove this theorem, we first consider a special case in the following auxiliary lemma.
\begin{lemma}
	\label{lem:PDD-dist-cluster-FPT}
	Given an instance $\Instance = (\Tree,\Food,k,D)$ of \sPDD and a set $Y\subseteq X$ of size~$d$ such that $\Food-Y$ is a source-dominant graph, we can compute whether there is a viable set $S\cup Y$ with~$|S \cup Y|\le k$ and $\PD(S\cup Y) \ge D$ in~$\Oh(3^d \cdot n \cdot k^2)$~time.
\end{lemma}
\begin{proof}
	We provide a dynamic programming algorithm. 
	Let $C_1,\dots,C_c$ be the connected components of $\Food-Y$ and
	let $x_i$ be the only source in $C_i$ for each~$i\in [c]$.
	
	\proofpara{Table definition}
	A set $S \subseteq X\setminus Y$ of taxa is \emph{$(\ell,Z)$-feasible}, if $|S| \le \ell$ and $S\cup Z$ is viable, for an integer~$\ell$ and a set~$Z$ of taxa.
	The dynamic programming algorithm has tables $\DP$ and $\DP_i$ for each~$i\in [c]$.
	We want entry $\DP[i,\ell,Z]$, for~$i\in [c]$, $\ell\in [k]_0$, and~$Z \subseteq Y$, to store the largest phylogenetic diversity~$\PD(S)$ of an $(\ell,Z)$-feasible set $S \subseteq C_1 \cup \dots \cup C_{i}$, and $-\infty$ if no such set $S$ exists.

	The table entries~$\DP_i[j,b,\ell,Z]$ additionally have a dimension~$b$ with $b\in \{0,1\}$.
	For~$b=0$, an entry $\DP_i[j,b,\ell,Z]$ stores the largest phylogenetic diversity~$\PD(S)$ of an $(\ell,Z)$-feasible set $S \subseteq \{ x_1^{(i)}, \dots, x_{j}^{(i)} \}$.
	For $b=1$, additionally some vertex~$v_{j'}^{(i)}$ with~$j' < j$ needs to be contained in $S$.
	
	\proofpara{Algorithm}
	Iterate over the edges of $\Food$.
	For each edge $uv\in E(\Food)$ with $u,v\in Y$, remove all edges incoming at $v$, including $uv$, from $E(\Food)$. After this removal,~$v$ is a new source.
	
	We initialize the base cases of $\DP_i$ by setting $\DP_i[j,0,0,Z]:=0$ for each $i\in [c]$, each~$j\in [|C_i|]$, and every~$Z \subseteq \sources$.
	Moreover,~$\DP_i[1,b,\ell,Z]:=\w(\rho v_1^{(i)})$ if~$\ell \ge 1$ and~$Z \subseteq \predators{v_1^{(i)}} \cup \sources$; and $\DP_i[1,b,\ell,Z]:=-\infty$, otherwise.
	
	To compute further values for $j\in [|C_i|-1]$, $b\in \{0,1\}$, and $\ell \in [k]$ we use the recurrences
	\begin{equation}
		\label{eqn:recurrence-dist-cluster-1}
		\DP_i[j+1, b, \ell, Z] = \max\{ \DP_i[j, b, \ell, Z], \DP_i[j, b', \ell-1, Z \setminus \predators{v_{j+1}^{{(i)}}}] + \w(\rho v_{j+1}^{{(i)}})\},
	\end{equation}
	where~$b' = 0$ if there is an edge from a vertex in~$Y$ to~$x_{j+1}^{(i)}$ and otherwise~$b' = 1$.
	
	Finally, we set $\DP[1,\ell,Z] := \DP_1[|C_1|,0,\ell,Z]$ and compute further values with
	\begin{equation}
		\label{eqn:recurrence-dist-cluster-2}
		\DP[i+1, \ell, Z] = \max_{Z' \subseteq Z, \ell' \in [\ell]_0} \DP[i, \ell', Z'] + \DP_{i+1}[|C_{i+1}|, 0, \ell-\ell', Z\setminus Z'].
	\end{equation}
	
	There is a viable set $S\cup Y$ with $|S \cup Y|\le k$ and $\PD(S\cup Y) \ge D$ if and only if~$\DP[c, k-|Y|, Z] \ge D - \PD(Y)$.
	
	\proofpara{Correctness}
	Assume that $\DP$ stores the intended values.
	Then, there is an $(\ell,Z)$-feasible set $S \subseteq X\setminus Y$, if~$\DP[c, k-|Y|, Z] \ge D - \PD(Y)$.
	First, this implies that~$S\cup Y$ is viable.
	Moreover, since $S$ has size at most $k-|Y|$, we obtain $|S \cup Y| \le k$.
	Finally, because~\Tree is a star and~$S$ and~$Y$ are disjoint, $\PD(S) \ge D - \PD(Y)$ implies~$\PD(S \cup Y) \ge D$.
	The converse direction can be shown analogously.
	
	It remains to show that $\DP$ and $\DP_i$ store the right values.
	The base cases are correct.
	Towards the correctness of \Recc{eqn:recurrence-dist-cluster-1},
	as an induction hypothesis, assume that $\DP_i[j,b,\ell,Z]$ stores the desired value %the biggest phylogenetic diversity of an $(\ell,Z)$-feasible set $S\subseteq \{ x_1^{(i)}, \dots, x_{j}^{(i)} \}$ with $S \ne \emptyset$ if $b=1$
	for a fixed $j\in [|C_i|-1]$, each $i\in [c]$, $b\in \{0,1\}$, $\ell\in [k]_0$ and every~$Z\subseteq Y$.
	Let $\DP_i[j+1,b,\ell,Z]$ store $d$.
	We show that there is an $(\ell,Z)$-feasible set $S \subseteq \{ x_1^{(i)}, \dots, x_{j+1}^{(i)} \}$ with~$\PD(S) = d$.
	We conclude with \Recc{eqn:recurrence-dist-cluster-1} that 
	$\DP_i[j, b, \ell, Z]$ stores~$d$ or 
	$\DP_i[j, 1, \ell-1, Z \setminus \predators{v_{j+1}^{{(i)}}}]$ stores~$d - \w(\rho v_{j+1}^{{(i)}})$.
	If $\DP_i[j, b, \ell, Z]$ stores~$d$, then because there is an $(\ell,Z)$-feasible set $S \subseteq \{ x_1^{(i)}, \dots, x_{j}^{(i)} \} \subseteq \{ x_1^{(i)}, \dots, x_{j+1}^{(i)} \}$ we are done.
	If~$\DP_i[j, 1, \ell-1, Z \setminus \predators{v_{j+1}^{{(i)}}}]$ stores $d - \w(\rho v_{j+1}^{{(i)}})$ then there is an 
	$(\ell-1,Z \setminus \predators{v_{j+1}^{{(i)}}})$-feasible set $S' \subseteq \{ x_1^{(i)}, \dots, x_{j}^{(i)} \}$ containing $x_1^{(i)}$ or some $x_{j'}^{(i)} \in \predators{Y}$.
	Consequently, also $S' \cup \{x_{j+1}^{(i)}\}$ is $(\ell,Z)$-feasible and~$\PD(S' \cup \{x_{j+1}^{(i)}\}) = d$.
	
	Now conversely, let $S \subseteq \{ x_1^{(i)}, \dots, x_{j+1}^{(i)} \}$ be an $(\ell,Z)$-feasible set.
	We show that~$\DP_i[j+1,b,\ell,Z]$ stores at least $\PD(S)$.
	If $S \subseteq \{ x_1^{(i)}, \dots, x_{j}^{(i)} \}$ then we know from the induction hypothesis that $\DP_i[j,b,\ell,Z]$ stores $\PD(S)$ and\lb therefore also~$\DP_i[j+1,b,\ell,Z]$ stores $\PD(S)$.
	If $x_{j+1}^{(i)} \in S$, then $S$\lb contains $x_{1}^{(i)}$ or some~$x_{j'}^{(i)} \in \predators{Y}$.
	Define $S' := S \setminus \{x_{j+1}^{(i)}\}$.
	Then, $|S'| = \ell-1$,\lb and $S' \cup (Z \setminus \predators{x_{j+1}^{(i)}})$ is viable because $S$ is $(\ell,Z)$-feasible.
	Consequently,\lb $\DP_i[j,1,\ell-1,Z \setminus \predators{x_{j+1}^{(i)}}] \ge \PD(S') = \PD(S) - \w(\rho x_{j+1}^{(i)})$ and
	therefore,\lb $\DP_i[j+1,b,\ell,Z] \ge \PD(S)$.

	Now, we focus on the correctness of \Recc{eqn:recurrence-dist-cluster-2}.
	Let $\DP[i+1,\ell,Z]$ store $d$.
	We show that there is an $(\ell,Z)$-feasible set $S\subseteq C_1 \cup \dots \cup C_{i+1}$ with $\PD(S) = d$.
	Because $\DP[i+1,\ell,Z]$ stores $d$, by \Recc{eqn:recurrence-dist-cluster-2}, there are $Z'\subseteq Z$ and $\ell' \in [\ell]_0$ such that $\DP[i, \ell', Z'] = d_1$, $\DP_{i+1}[|C_{i+1}|, 0, \ell-\ell', Z\setminus Z'] = d_2$, and $d_1 + d_2 = d$.
	By the induction hypothesis, there is an~$(\ell',Z')$-feasible set $S_1 \subseteq C_1 \cup \dots \cup C_{i}$
	and\lb an~$(\ell-\ell',Z\setminus Z')$-feasible set $S_2 \subseteq C_{i+1}$
	such that~$\PD(S_1) = d_1$ and~$\PD(S_2) = d_2$.
	Then, $S := S_1 \cup S_2$ satisfies $|S| \le |S_1| + |S_2| \le \ell' + (\ell - \ell') = \ell$.
	Further, because $Y$ has no outgoing edges, $Z' \subseteq \predators{S_1} \cup \sources$ and~$Z\setminus Z' \subseteq \predators{S_2} \cup \sources$.
	Therefore, $Z \subseteq \predators{S} \cup \sources$ and $S\cup Z$ is viable.
	We conclude that~$S$ is the desired set.
	
	Let $S\subseteq C_1 \cup \dots \cup C_{i+1}$ be an $(\ell,Z)$-feasible set with $\PD(S) = d$.
	We show that~$\DP[i+1,\ell,Z] \ge d$.
	Define $S_1 := S \cap (C_1 \cup \dots \cup C_{i})$ and~$Z' := \predators{S_1} \cap Z$.
	We conclude that $S_1\cap Z'$ is viable.
	Then, $S_1$ is $(\ell',Z')$-feasible, where~$\ell' := |S_1|$.
	Define~$S_2 := S \cap C_{i+1} = S \setminus S_1$.
	Because $S \cup Z$ is viable and $Z$ does not have outgoing edges, we know that $Z \subseteq \predators{S} \cup \sources$.
	So, $Z \setminus Z' \subseteq \predators{S_2} \cup \sources$ and because $|S_2| = |S| - |S_1| = \ell - \ell'$,
	we conclude that $S_2$ is $(\ell - \ell',Z \setminus Z')$-feasible.
	Consequently, $\DP[i,\ell',Z'] \ge \PD(S_1)$ and $\DP_{i+1}[|C_{i+1}|,\ell-\ell',Z\setminus Z'] \ge \PD(S_2)$.
	Hence, $\DP[i+1,\ell,Z] \ge \PD(S_1) + \PD(S_2) = \PD(S)$ because \Tree is a star.
	
	\proofpara{Running time}
	The tables $\DP$ and $\DP_i$ for $i\in [c]$ have $\Oh(2^d \cdot n \cdot k)$ entries in total.
	Whether one of the base cases applies can be checked in linear time.
	We can compute the set $Z \setminus \predators{x}$ for any given $Z\subseteq Y$ and $x\in X$ in $\Oh(d^2)$ time.
	Therefore, the~$\Oh(2^d \cdot n \cdot k)$ times we need to apply \Recc{eqn:recurrence-D-pre} consume $\Oh(2^d d^2 \cdot n \cdot k)$ time in total.
	In \Recc{eqn:recurrence-D}, each $x\in Y$ can be in $Z'$, in $Z\setminus Z'$ or in $Y\setminus Z$ so that we can compute all the table entries of $\DP$ in~$\Oh(3^d \cdot n \cdot k^2)$ which is also the overall running time.
\end{proof}

\begin{proof}[Proof (of Theorem~\ref{thm:PDD-dist-cluster-FPT})]
	\proofpara{Algorithm}
	We iterate over all the subsets $Z$ of $Y$ and
	define~$R_0 := Y \setminus Z$.
	We want that $Z$ is the set of taxa that are surviving while the taxa in $R_0$ are going extinct.
	Compute the set $R \subseteq X$ of taxa $r$ for which in $\Food$ each path from $r$ to $s$ contains a taxon of~$R_0$, for each $s \in \sources$.
	Compute~$\Tree_Z := \Tree - R$ and $\Food_Z := \Food - R$.
	Continue with the next~$Z$ if~$Z \cap R \ne \emptyset$.
	
	With Lemma~\ref{lem:PDD-dist-cluster-FPT}, compute whether there is a viable set~$S' = S\cup Z$ in instance~$\Instance = (\Tree_Z, \Food_Z, k, D)$ such that $|S'| \le k$ and $\PD(S') \ge D$.
	Return \yes, if such a set exists.
	Otherwise, continue with the next subset~$Z$ of $Y$.
	After iterating over all subsets $Z$ of $Y$, return \no.
	
	\proofpara{Correctness}
	Let $Z$ be a given subset of $Y$.
	Assume that there is a solution~$S$ with~$Y\cap S = Z$.
	Let $x$ be a vertex for which each path from $x$ to a source~$s$ of~\Food contains a vertex of $Y \setminus Z$.
	Because $S$ is viable, $x$ is not in $S$.
	
	The rest of the correctness follows by Lemma~\ref{lem:PDD-dist-cluster-FPT}.
	
	\proofpara{Running time}
	For each subset $Z$ of $Y$, we can compute $\Tree_Z$ and $\Food_Z$ in $\Oh(m\cdot n)$ time.
	By Lemma~\ref{lem:PDD-dist-cluster-FPT}, we can compute whether there is a solution $S$ with $Y\cap S = Z$ in~$\Oh(3^d \cdot n \cdot k^2)$~time.
	Therefore, the overall running time is $\Oh(6^d \cdot n^2 \cdot m \cdot k^2)$.
\end{proof}

Next, we show that \PDD, in contrast to \sPDD, is \NP-hard even when the food-web is restricted to have a cluster graph as an underlying undirected graph.
We obtain this hardness
by a reduction from \VC on cubic graphs.
In cubic graphs, the degree of each vertex is \emph{exactly} three.
Recall, in \VC we are given an undirected graph~$G = (V,E)$ and an integer~$k$ and we ask whether a set~$C\subseteq V$ of size at most~$k$ exists such that~$u\in C$ or~$v\in C$ for each edge~$\{u,v\} \in E$.
The set~$C$ is called a \emph{vertex cover}.
\VC is \NP-hard on cubic graphs~\cite{mohar}.

\begin{theorem}
	\label{thm:PDD-dist-cluster-hardness}
	\PDD is \NP-hard
	on source-separating instances in which the food-web is a cluster graph
	and each connected component has at most four vertices.
\end{theorem}
\begin{proof}
	\proofpara{Reduction}
	Let $(G,k)$ be an instance of \VC, where $G = (V,E)$ is cubic.
	We define an instance $\Instance = (\Tree,\Food,k',D)$ of \PDD as follows.
	Let $\Tree$ have a root $\rho$.
	For each vertex $v\in V$, we add a child $v$ of $\rho$.
	For each edge $e=\{u,v\}\in E$, we add a child~$e$ of~$\rho$ and two children~$[u,e]$ and~$[v,e]$ of $e$.
	Let $N$ be a big integer.
	We set the weight of~$\rho e$ to~$N-1$ for each edge~$e$ in $E$.
	All other edges of~\Tree have a weight of~1.
	Additionally, for each edge~$e=\{u,v\}\in E$ we add edges $(u,[u,e])$ and~$(v,[v,e])$ to \Food.
	Finally, we set $k' := |E| + k$ and $D := N\cdot |E| + k$.
	
	\proofpara{Correctness}
	The instance \Instance of \PDD is constructed in polynomial time for a suitable~$N$.
	The leaves in~\Tree are~$V \cup \{ [u,e], [v,e] \mid e=\{u,v\} \in E \}$.
	The sources of~\Food are $V$.
	Therefore, \Instance is source-separating.
	Let~$e_1$, $e_2$, and~$e_3$ be the edges incident with~$v\in V(G)$.
	Each connected component in \Food contains four vertices, $v$, and~$[v,e_i]$ for~$i\in\{1,2,3\}$.\footnote{Observe that we only constructed the instance so that the connected components are stars, but one could add edges~$([v,e_1],[v,e_2])$, $([v,e_1],[v,e_3])$, and~$([v,e_2],[v,e_3])$ for each vertex~$v$, to meet the formal requirement of a cluster graph.}
	
	We show that $(G,k)$ is a \yes-instance of \VC if and only if \Instance is a \yes-instance of \PDD.
	Let $C\subseteq V$ be a vertex cover of $G$ of size at most $k$.
	If necessary, add vertices to $C$ until $|C|=k$.
	For each edge $e\in E$, let~$v_e$ be an endpoint of $e$ that is contained in $C$. Note that~$v_e$ exists since~$C$ is a vertex cover.
	We show that $S := C \cup \{ [v_e,e] \mid e\in E \}$ is a solution for \Instance:
	The size of~$S$ is~$|C| + |E| = k + |E| = k'$.
	By definition, for each taxon $[v_e,e]$ we have $v_e\in C\subseteq S$, so $S$ is viable.
	Further, as~$S$ contains a taxon~$[v_e,e]$ for each edge~$e\in E$, we conclude that $\PD(S) \ge N\cdot |E| + \PD(C) = N\cdot |E| + k = D$.
	Therefore, $S$ is a solution of~\Instance.
	
	Conversely,
	let $S$ be a solution of instance \Instance of \PDD.
	Define $C := S \cap V(G)$ and define~$S' := S \setminus C$.
	Because $\PD(S) \ge D$, we conclude that for each $e\in E$ at least one taxon~$[u,e]$ with $u\in e$ is contained in $S'$.
	Thus, $|S'| \ge |E|$ and we conclude that the size of $C$ is at most~$k$.
	Because $S$ is viable we conclude that $u$ is in~$C$ for each~$[u,e] \in S'$.
	Therefore, $C$ is a vertex cover of size at most~$k$ of~$G$.
\end{proof}

\PDD is \NP-hard even if each connected component in the food-web is a path of length three~\cite{faller}.
So, we wonder if the hardness still holds if we restrict the food-web even more such that each connected component of the food-web contains at most two vertices.
We were not able to show that \PDD is \NP-hard for this special case.
However, in the next proposition, we show that \PDD can be solved in polynomial time if we restrict the instance further to be source-separating as well.

\begin{proposition}
	\label{prop:PDD-dist-cluster-flow}
	\PDD can be solved in $\Oh(k\cdot n^2 \cdot \log^2 n)$ time on source-separating instances, if the food-web only has isolated edges.
\end{proposition}
\begin{figure}[t]
	\centering
	\begin{tikzpicture}[scale=0.8,every node/.style={scale=0.7}]
		\tikzstyle{txt}=[circle,fill=white,draw=white,inner sep=0pt]
		\tikzstyle{nde}=[circle,fill=black,draw=black,inner sep=2.5pt]
		\tikzstyle{ndeg}=[circle,fill=lime,draw=black,inner sep=2.5pt]
		\tikzstyle{ndeo}=[circle,fill=orange,draw=black,inner sep=2.5pt]
		
		\node[nde] (root) at (0.5,0) {};
		
		\node[nde, xshift= -25mm] (u1) [below=of root] {};
		\node[nde] (u2) [below=of root] {};
		\node[nde, xshift= 25mm] (u3) [below=of root] {};
		
		\node[ndeo, xshift= -5mm] (v11) [below=of u1] {};
		\node[nde, xshift= 5mm] (v12) [below=of u1] {};
		
		\node[ndeg, xshift= -10mm] (v21) [below=of u2] {};
		\node[ndeg] (v22) [below=of u2] {};
		\node[ndeg, xshift= 10mm] (v23) [below=of u2] {};
		
		\node[nde, xshift= -5mm] (v31) [below=of u3] {};
		\node[ndeg, xshift= 5mm] (v32) [below=of u3] {};
		
		\node[ndeo, xshift= -10mm] (w121) [below=of v12] {};
		\node[ndeo] (w122) [below=of v12] {};
		\node[ndeo, xshift= 10mm] (w123) [below=of v12] {};
		
		\node[ndeg, xshift= -5mm] (w311) [below=of v31] {};
		\node[ndeg, xshift= 5mm] (w312) [below=of v31] {};
		
		\node[txt, xshift= -13mm] [right=of root] {$\rho$};
		
		\node[txt, yshift= 13mm] [below=of v11] {$x_0$};
		\node[txt, yshift= 13mm] [below=of v21] {$x_4$};
		\node[txt, yshift= 13mm] [below=of v22] {$x_5$};
		\node[txt, yshift= 13mm] [below=of v23] {$x_6$};
		\node[txt, yshift= 13mm] [below=of v32] {$x_9$};
		
		\node[txt, yshift= 13mm] [below=of w121] {$x_1$};
		\node[txt, yshift= 13mm] [below=of w122] {$x_2$};
		\node[txt, yshift= 13mm] [below=of w123] {$x_3$};
		\node[txt, yshift= 13mm] [below=of w311] {$x_7$};
		\node[txt, yshift= 13mm] [below=of w312] {$x_8$};
		
		\node[txt] at (-1,0) {$\Tree=$};
		
		\node[txt] at (-2.1,0) {(1)};
		
		\draw[blue, arrows = {-Stealth[length=8pt]}] (root) to (u1);
		\draw[blue, arrows = {-Stealth[length=8pt]}] (root) to (u2);
		\draw[blue, arrows = {-Stealth[length=8pt]}] (root) to (u3);
		
		\draw[blue, arrows = {-Stealth[length=8pt]}] (u1) to (v11);
		\draw[dotted, arrows = {-Stealth[length=8pt]}] (u1) to (v12);
		
		\draw[dotted, arrows = {-Stealth[length=8pt]}] (u2) to (v21);
		\draw[blue, arrows = {-Stealth[length=8pt]}] (u2) to (v22);
		\draw[blue, arrows = {-Stealth[length=8pt]}] (u2) to (v23);
		
		\draw[dotted, arrows = {-Stealth[length=8pt]}] (u3) to (v31);
		\draw[blue, arrows = {-Stealth[length=8pt]}] (u3) to (v32);
		
		\draw[dotted, arrows = {-Stealth[length=8pt]}] (v12) to (w121);
		\draw[dotted, arrows = {-Stealth[length=8pt]}] (v12) to (w122);
		\draw[dotted, arrows = {-Stealth[length=8pt]}] (v12) to (w123);
		
		\draw[dotted, arrows = {-Stealth[length=8pt]}] (v31) to (w311);
		\draw[dotted, arrows = {-Stealth[length=8pt]}] (v31) to (w312);
		
		\draw[dashed] (3.5,0.5) to (3.5,-5);
		
		\node[txt] at (4,0.25) {$\Food=$};
		
		\node[ndeo] (a1) at (4,-0.5) {};
		\node[ndeo] (b1) at (4.5,-0.5) {};
		\node[ndeg] (a2) at (4,-1.5) {};
		\node[ndeg] (b2) at (4.5,-1.5) {};
		\node[ndeo] (a3) at (4,-2.75) {};
		\node[ndeo] (b3) at (4.5,-2.75) {};
		\node[ndeg] (a4) at (4,-3.75) {};
		\node[ndeg] (b4) at (4.5,-3.75) {};
		\node[ndeg] (a5) at (4,-4.5) {};
		\node[ndeg] (b5) at (4.5,-4.5) {};
		
		\draw[arrows = {-Stealth[length=8pt]}] (a2) to (a1);
		\draw[arrows = {-Stealth[length=8pt]}] (b2) to (b1);
		\draw[arrows = {-Stealth[length=8pt]}] (a4) to (a3);
		\draw[arrows = {-Stealth[length=8pt]}] (b4) to (b3);
		
		\node[txt, yshift= 14mm] [below=of a2] {$x_7$};
		\node[txt, yshift= 14mm] [below=of a4] {$x_6$};
		\node[txt, yshift= 14mm] [below=of a5] {$x_8$};
		\node[txt, yshift= 14mm] [below=of b2] {$x_5$};
		\node[txt, yshift= 14mm] [below=of b4] {$x_9$};
		\node[txt, yshift= 14mm] [below=of b5] {$x_4$};
		
		\node[txt, yshift= -14mm] [above=of a1] {$x_2$};
		\node[txt, yshift= -14mm] [above=of b1] {$x_0$};
		\node[txt, yshift= -14mm] [above=of a3] {$x_3$};
		\node[txt, yshift= -14mm] [above=of b3] {$x_1$};
		
		\draw (5,0.5) to (5,-5);
	\end{tikzpicture}
	\begin{tikzpicture}[scale=0.8,every node/.style={scale=0.7}]
		\tikzstyle{txt}=[circle,fill=white,draw=white,inner sep=0pt]
		\tikzstyle{nde}=[circle,fill=black,draw=black,inner sep=2.5pt]
		\tikzstyle{ndeg}=[circle,fill=lime,draw=black,inner sep=2.5pt]
		\tikzstyle{ndeo}=[circle,fill=orange,draw=black,inner sep=2.5pt]
		
		\node[nde] (root1) at (-1,-.7) {};
		\node[nde, yshift= 7mm] (u1) [below=of root1] {};
		\node[ndeo, xshift= -5mm, yshift= 7mm] (v11) [below=of u1] {};
		\node[nde, xshift= 5mm, yshift= 7mm] (v12) [below=of u1] {};
		
		\node[ndeo, xshift= -10mm, yshift= 5mm] (w121) [below=of v12] {};
		\node[ndeo, yshift= 5mm] (w122) [below=of v12] {};
		\node[ndeo, xshift= 10mm, yshift= 5mm] (w123) [below=of v12] {};
		
		\node[nde] (root2) at (1.5,0) {};
		\node[nde, yshift= 5mm] (u2) [below=of root2] {};
		\node[nde, xshift= 25mm, yshift= 5mm] (u3) [below=of root2] {};
		
		\node[ndeg, xshift= -10mm, yshift= 5mm] (v21) [below=of u2] {};
		\node[ndeg, yshift= 5mm] (v22) [below=of u2] {};
		\node[ndeg, xshift= 10mm, yshift= 5mm] (v23) [below=of u2] {};
		
		\node[nde, xshift= -5mm, yshift= 5mm] (v31) [below=of u3] {};
		\node[ndeg, xshift= 5mm, yshift= 5mm] (v32) [below=of u3] {};
		
		\node[ndeg, xshift= -5mm, yshift= 8mm] (w311) [below=of v31] {};
		\node[ndeg, xshift= 5mm, yshift= 8mm] (w312) [below=of v31] {};
		
		\node[nde] (s) at (-2,-4.5) {};
		\node[nde] (v) at (2,-4.5) {};
		
		\node[txt, xshift= 13mm] [left=of root1] {$\rho_1$};
		\node[txt, xshift= 13mm] [left=of root2] {$t=\rho_2$};
		
		\node[txt, yshift= 13mm] [below=of s] {$s$};
		\node[txt, yshift= 13mm] [below=of v] {$\nu$};
		
		\node[txt, yshift= 14mm] [below=of v11] {$x_0$};
		\node[txt, xshift= 14mm] [left=of v21] {$x_4$};
		\node[txt, xshift= 14mm] [left=of v22] {$x_5$};
		\node[txt, xshift= 14mm] [left=of v23] {$x_6$};
		\node[txt, xshift= 14mm] [left=of v32] {$x_9$};
		
		\node[txt, yshift= 13mm] [below=of w121] {$x_1$};
		\node[txt, yshift= 13mm] [below=of w122] {$x_2$};
		\node[txt, yshift= 13mm] [below=of w123] {$x_3$};
		\node[txt, xshift= 14mm] [left=of w311] {$x_7$};
		\node[txt, yshift= 14mm] [below=of w312] {$x_8$};
		
		\node[txt] at (-2,-.5) {$G=$};
		
		\node[txt] at (-2.1,0) {(2)};
		
		\draw[dotted, arrows = {-Stealth[length=4pt]}, bend left=-20, dotted] (root1) to (u1);
		\draw[blue, arrows = {-Stealth[length=4pt]}, bend left=20, dotted] (u2) to (root2);
		\draw[dotted, arrows = {-Stealth[length=4pt]}, bend left=20, dotted] (u3) to (root2);
		
		\draw[dotted, arrows = {-Stealth[length=4pt]}, bend left=-20, dotted] (u1) to (v11);
		\draw[dotted, arrows = {-Stealth[length=4pt]}, bend left=-20, dotted] (u1) to (v12);
		
		\draw[dotted, arrows = {-Stealth[length=4pt]}, bend left=-20, dotted] (v21) to (u2);
		\draw[dotted, arrows = {-Stealth[length=4pt]}, bend left=-20, dotted] (v22) to (u2);
		\draw[dotted, arrows = {-Stealth[length=4pt]}, bend left=-20, dotted] (v23) to (u2);
		
		\draw[dotted, arrows = {-Stealth[length=4pt]}, bend left=-21, dotted] (v31) to (u3);
		\draw[dotted, arrows = {-Stealth[length=4pt]}, bend left=-20, dotted] (v32) to (u3);
		
		\draw[dotted, arrows = {-Stealth[length=4pt]}, bend left=-20, dotted] (v12) to (w121);
		\draw[dotted, arrows = {-Stealth[length=4pt]}, bend left=-20, dotted] (v12) to (w122);
		\draw[dotted, arrows = {-Stealth[length=4pt]}, bend left=-20, dotted] (v12) to (w123);
		
		\draw[dotted, arrows = {-Stealth[length=4pt]}, bend left=-28, dotted] (w311) to (v31);
		\draw[dotted, arrows = {-Stealth[length=4pt]}, bend left=-20, dotted] (w312) to (v31);
		
		\draw[blue, arrows = {-Stealth[length=4pt]}, bend left=20] (root1) to (u1);
		\draw[blue, arrows = {-Stealth[length=4pt]}, bend left=-20] (u2) to (root2);
		\draw[blue, arrows = {-Stealth[length=4pt]}, bend left=-20] (u3) to (root2);
		
		\draw[blue!80!, arrows = {-Stealth[length=4pt]}, bend left=20] (u1) to (v11);
		\draw[dotted, arrows = {-Stealth[length=4pt]}, bend left=20] (u1) to (v12);
		
		\draw[dotted, arrows = {-Stealth[length=4pt]}, bend left=20] (v21) to (u2);
		\draw[blue, arrows = {-Stealth[length=4pt]}, bend left=20] (v22) to (u2);
		\draw[blue, arrows = {-Stealth[length=4pt]}, bend left=20] (v23) to (u2);
		
		\draw[dotted, arrows = {-Stealth[length=4pt]}, bend left=20] (v31) to (u3);
		\draw[blue, arrows = {-Stealth[length=4pt]}, bend left=22] (v32) to (u3);
		
		\draw[dotted, arrows = {-Stealth[length=4pt]}, bend left=20] (v12) to (w121);
		\draw[dotted, arrows = {-Stealth[length=4pt]}, bend left=20] (v12) to (w122);
		\draw[dotted, arrows = {-Stealth[length=4pt]}, bend left=20] (v12) to (w123);
		
		\draw[dotted, arrows = {-Stealth[length=4pt]}, bend left=20] (w311) to (v31);
		\draw[dotted, arrows = {-Stealth[length=4pt]}, bend left=28] (w312) to (v31);

		\draw[blue, arrows = {-Stealth[length=4pt]}, bend left=-40] (v11) to (v22);
		\draw[dotted, arrows = {-Stealth[length=4pt]}, bend left=-29] (w121) to (v32);
		\draw[dotted, arrows = {-Stealth[length=4pt]}, bend left=-24] (w122) to (w311);
		\draw[dotted, arrows = {-Stealth[length=4pt]}, bend left=-15] (w123) to (v23);

		\draw[blue, arrows = {-Stealth[length=4pt]}] (s) to (v);
		\draw[blue, arrows = {-Stealth[length=4pt]}, bend left=20] (s) to (root1);
		
		\draw[dotted, arrows = {-Stealth[length=4pt]}, bend left=7, dotted] (v) to (v21);
		\draw[dotted, arrows = {-Stealth[length=4pt]}, bend left=7, dotted] (v) to (v22);
		\draw[blue, arrows = {-Stealth[length=4pt]}, bend left=5] (v) to (v23);
		\draw[dotted, arrows = {-Stealth[length=4pt]}, bend left=5, dotted] (v) to (w311);
		\draw[dotted, arrows = {-Stealth[length=4pt]}, bend left=6, dotted] (v) to (w312);
		\draw[blue, arrows = {-Stealth[length=4pt]}, bend left=-42] (v) to (v32);
	\end{tikzpicture}
	\caption{This figure shows in (1) a hypothetical source-separating instance of \PDD.
		The sources are drawn in green and the predators in orange.
		In (2), the instance of \MCNF is shown, as it would have been constructed in Proposition~\ref{prop:PDD-dist-cluster-flow}.
		For the sake of readability, capacities and costs are omitted.
		The blue and solid edges mark how a solution~$\{x_0,x_5,x_6,x_9\}$ transfer to a flow.
	}
	\label{fig:PDD-dist-cluster-flow}
\end{figure}
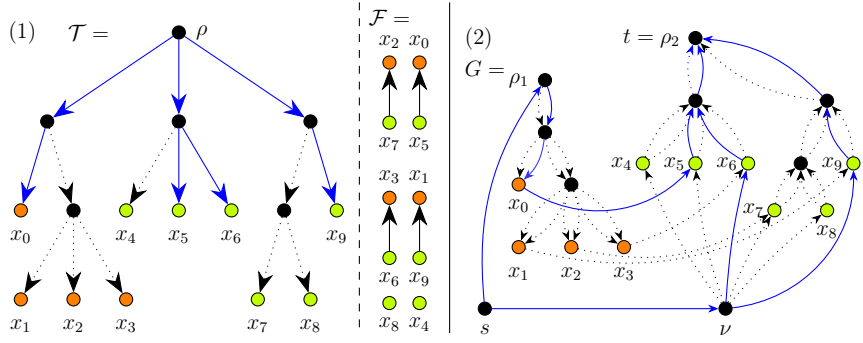%
\begin{proof}
	In this algorithm, we use techniques similar to the way Bordewich~et~al. showed that phylogeny across two trees can be computed efficiently~\cite{bordewich2trees}.
	We reduce a source-separating instance of \PDD to $\Oh(k)$ instances of \MCNFlong (\MCNF), which can be solved in~$\Oh(n^2 \log^2 n)$ time, each~\cite{minoux,networkflows}.
	
	In \MCNF one is given a directed graph $G = (V,E)$ with two vertices~$s,t \in V$---the \emph{source $s$} and the \emph{sink $t$}---and two integers $F$ and $C$, where each edge $e \in E$ has a positive capacity $c(e)$ and a cost $a(e)$.
	The cost may be negative and $G$ may have parallel edges.
	
	A \emph{flow} $f$ assigns each edge $e$ a non-negative integer value~$f(e)$.
	The \emph{cost of a flow~$f$} is $\sum_{e\in E} f(e) \cdot a(e)$.
	A flow $f$ is \emph{$q$-proper} if
	$q = \sum_{w\in V} f(sw) = \sum_{u\in V} f(ut)$,
	$\sum_{u\in V} f(uv) = \sum_{w\in V} f(vw)$ for each vertex $v\in V\setminus \{s,t\}$, and
	$f(e) \le c(e)$ for each edge~$e$.
	For an instance $(G,s,t,F,C,c,a)$ of \MCNF,
	we ask whether~$G$ contains an~$F$-proper flow of cost at most $C$.
	
	\proofpara{Algorithm}
	Let $\Instance = (\Tree, \Food, k, D)$ be a source-separating instance of \PDD.
	Define~$\Tree_1$ and~$\Tree_2$ to be the subtree of $\Tree$ spanning over the vertices $\{\rho\} \cup \sources$, and, respectively, $\{\rho\} \cup (X \setminus \sources)$.
	To avoid confusion, for~$i\in\{1,2\}$ the root of $\Tree_i$ is called $\rho_i$.
	For an example of the following reduction, consider Figure~\ref{fig:PDD-dist-cluster-flow}.
	
	Iterate over~$k' \in [ \lfloor k/2 \rfloor ]_0$.
	We define an instance $\Instance'_{k'} = (G,s,t,F,C,c,a)$ of \MCNF as follows.
	The set of vertices of~$G$ are the vertices of $\Tree_1$ and $\Tree_2$ and two new vertices $s$ and $\nu$.
	The vertex~$s$ is the source and $t:=\rho_2$ is the sink of $G$.
	We add edges $s\rho_1$ and $s\nu$, where $c(s\rho_1) = k'$ and~$c(s\nu) = k-2k'$.
	For each leaf $x$ of~$\Tree_2$, add an edge $\nu x$ of capacity $k$.
	For each edge~$xy$ of \Food, add an edge~$yx$ with a capacity of~1.
	Each edge so far has a cost of 0.
	For each edge $e=uv$ in $\Tree_1$, add two parallel edges $e$ and $e'$
	where $e$ has a capacity of 1 and a cost of $-\w(e)$ and $e'$ has a capacity of~$k-1$ and a cost of~0.
	For each edge $uv$ in $\Tree_2$, add two parallel edges $vu$ and $(vu)'$ (note that the edges are reversed)
	where $vu$ has a capacity of 1 and a cost of~$-\w(uv)$ and $(vu)'$ has a capacity of~$k-1$ and a cost of~0.
	To complete the instance, we set~$F := k-k'$ and $C := -D$.
	
	We compute a solution for $\Instance'_{k'}$ and return \yes if $\Instance'_{k'}$ is a \yes-instance of \MCNF.
	If $\Instance'_{k'}$ is a \no-instance of \MCNF for each $k'$, then we return \no.

	\proofpara{Correctness}
	The correctness is shown similarly as the correctness of the algorithm in~\cite{bordewich2trees}.
	We show first that if \Instance is a \yes-instance of \PDD then there is a $k'$ such that $\Instance'_{k'}$ is a \yes-instance of~\MCNF.
	Afterward, we show that if $\Instance'_{k'}$ is a \yes-instance of \MCNF for specific $k'$ then \Instance is a \yes-instance of \PDD.
	
	Assume that \Instance is a \yes-instance of \PDD with solution $S\subseteq X$.
	If necessary, add vertices to $S$ until $|S| = k$.
	Let $S_1$ be the subset of vertices of $S$ which are not sources.
	Let $k'$ be the size of $S_1$.
	Further, denote the set~$\{ x\in X \mid xy\in E(\Food), y\in S_1 \}$ with~$S_2$.
	Because~$S$ is viable and every connected component in~\Food contains at most two vertices, we conclude $S_2\subseteq S$.
	Then, we define $S_3 := S \setminus (S_1 \cup S_2)$.
	We define a flow~$f$ in~$\Instance'_{k'}$ as follows.
	Let $E_i$ be the set of edges on a path between~$\rho$ and~$S_i$ in~\Tree, for $i \in [3]$.
	We set $f(e) = 1$ for each~$e\in E_1$ and~$f(vu)=1$ for each~$uv \in E_2 \cup E_3$.
	So far we have defined the flow that ensures that the cost is $-D$.
	Now, we ensure that the flow $f$ is $k-k'$-proper.
	For each edge $xy\in E(\Food)$ with~$y\in S_1$ we set~$f(yx) = 1$. 
	Further, for each $x\in S_3$ we set~$f(\nu x) = 1$.
	We then set~$f(s\rho_1) = k'$ and~$f(sv) = k-2k'$.
	For each edge $e$ of $\Tree_1$, we set $f(e') = | \off(e) \cap S_1 | -1$.
	That is, the~$f(e')$ is the number of offspring $e'$ has in $S_1$ minus 1.
	For each edge $e=uv$ of $\Tree_2$ we set~$f((vu)') = | \off(e) \cap (S_2 \cup S_3) | -1$.
	It remains to show that $f$ is $k-k'$-proper.
	
	\proofpara{Claim: The flow $f$ is $k-k'$-proper}\\
	\begin{claimproof}
		One can easily verify that $f(e) \le c(e)$ for each $e\in E(G)$.
		
		Observe that the size of $S_3$ is $|S| - |S_1 \cup S_2| = k - 2k'$.
		The flow leaving $s$ is~$\sum_{u\in V} f(su) = f(s\rho_1) + f(sv) = k' + k - 2k' = k - k'$.
		The flow entering $t=\rho_2$ is~$\sum_{w\in V} f(wt) = |S_2| + |S_3| = k' + (k - 2k') = k - k'$.
		By definition, the flow entering and leaving $\rho_1$ has size $k'$.
		The flow leaving $\nu$ is $|S_3|$ and so equals to the flow entering~$\nu$, $f(s\nu) = k-2k'$.
		Each leaf $x$ has, if~$x\in S$, a flow of 1 incoming and leaving, and 0, otherwise.
		For an internal vertex $v\ne \rho$ of $\Tree$ with children $w_1,\dots,w_z$ and parent $u$, we observe
		$\off(uv) = \bigcup_{i=1}^z \off(vw_i)$.
		With this observation we conclude that $f$ is $k-k'$-proper.
	\end{claimproof}
	
	Conversely, assume now that $f$ is a $k-k'$-proper flow of $\Instance'_{k'}$.
	Let $S$ be the set of vertices that are corresponding to leaves in $\Tree$ and that have a positive entering flow.
	We show that $S$ is a solution for instance \Instance of \PDD.
	Since $f(s\rho_1) \le k'$, we conclude that $|S \cap \off(\rho_1)| \le k'$.
	We conclude~$|S \cap \off(\rho_2)| \le k-2k' + |S \cap \off(\rho_1)| \le k-k'$, because $f(s\nu) \le k-2k'$.
	Thus, the size of $S$ is at most $k$.
	If $x\in \off(\rho_1)$ is in $S$ then $x$ has a positive flow leaving $x$ and we conclude that $f(xy) > 0$ for $yx \in E(\Food)$.
	Therefore, $y$ is in $S$ and $S$ is viable.
	Let $E_1$ be the set of edges $e$ with~$f(e)>0$ and~$a(e) < 0$.
	Recall that $a(e)$ is the cost of $e$.
	Only edges corresponding to an edge in $\Tree$ are in $E_1$.
	Let $E_1'$ be the corresponding edges of $E_1$, especially, edges in $\Tree_2$ are turned around.
	Then, we observe
	\begin{equation}
		\PD(S) \ge \sum_{e\in E_1'} \w(e)
		 =  -\sum_{e\in E_1} f(e)\cdot a(e) \ge D.
	\end{equation}
	Therefore, $S$ is a solution for instance \Instance of \PDD.
	
	\proofpara{Running time}
	For a given $k'$, we can construct the instance $\Instance'_{k'}$ in linear time.
	A solution for an instance of \MCNF can be computed in $\Oh(n^2 \log^2 n)$~time~\cite{minoux,networkflows}.
	An instance $\Instance'_{k'}$ contains at most $2n$ vertices and at most $6n$ edges.
	So, the overall running time is~$\Oh(k\cdot n^2 \cdot \log^2 n)$.
\end{proof}

\subsection{Distance to Co-Cluster}
In this section, we show that \PDD is \FPT with respect to the distance to co-cluster of the food-web.
A  graph is a co-cluster graph if its complement graph is a cluster graph.

In our algorithm, we solve as a subroutine an auxiliary problem, called \HStw, in which we are given a universe $\mathcal{U}$, a family of subsets~$\mathcal{W}$ of $\mathcal{U}$, a $\mathcal{U}$-tree \Tree, and integers $k$ and $D$.
We ask whether there is a set $S \subseteq \mathcal{U}$ of size at most $k$ such that~$\PD(S) \ge D$ and $S\cap W \ne \emptyset$ for each $W\in \mathcal{W}$.

\begin{lemma}
	\label{lem:PDD-HStw}
	\HStw can be solved in $\Oh(3^{|\mathcal{W}|} \cdot n)$ time.
\end{lemma}
\begin{proof}
	We adopt the dynamic programming algorithm with which one can solve \MPD with weighted costs for saving taxa in pseudo-polynomial time~\cite{pardi07}.
	
	\proofpara{Table definition}
	We define two tables $\DP$ and an auxiliary table $\DP'$.
	We want that entry $\DP[v,0,\mathcal{M}]$, for~$v\in V(\Tree)$, and~$\mathcal{M} \subseteq \mathcal{W}$, stores 0 if $\mathcal{M}$ is empty and~$\DP[v,1,\mathcal{M}]$ stores the biggest phylogenetic diversity $\PDsub{\Tree_v}(S)$ of a set $S\subseteq \off(v)$ in the subtree~$\Tree_v$ rooted at $v$ where~$S\cap M \ne \emptyset$ for each $M\in \mathcal{M}$.
	Otherwise, we store $-\infty$.
	For a vertex~$v \in V(\Tree)$ with children~$w_1,\dots,w_j$,
	in $\DP'[v,i,b,\mathcal{M}]$ we only consider sets~$S \subseteq \off(w_1) \cup \dots \cup \off(w_i)$.
	
	\proofpara{Algorithm}
	For a leaf $u \in \mathcal{U}$ of \Tree, in $\DP[u,1,\mathcal{M}]$ store 0 if~$u\in M$ for each $M\in\mathcal{M}$.
	Otherwise, store $-\infty$.
	For a given vertex~$v$ of~\Tree, in $\DP[v,0,\mathcal{M}]$ and $\DP'[v,i,0,\mathcal{M}]$ store 0 if $\mathcal{M} = \emptyset$.
	Otherwise, store $-\infty$.
	
	Let $v$ be an internal vertex of \Tree with children $w_1,\dots,w_z$.
	For each~$b\in \{0,1\}$ and every family~$\mathcal{M}$ of subsets of~$\mathcal{W}$,
	we set $\DP'[v,1,b,\mathcal{M}] := \DP[w_1,b,\mathcal{M}] + b\cdot \w(vw_1)$.
	To compute further values, we use the recurrence
	\begin{equation}
		\label{eqn:recurrence-HStw}
		\DP'[v,i+1,1,\mathcal{M}] =
			\max_{\mathcal{M}', b_1, b_2}
			\DP'[v,i,b_1,\mathcal{M}'] + \DP[w_{i+1},b_2,\mathcal{M}\setminus \mathcal{M}'] + b_2 \cdot \w(vw_{i+1}).
	\end{equation}
	Here, the maximum is taken over $\mathcal{M}' \subseteq \mathcal{M}$ and $b_1,b_2 \in \{0,1\}$ where we additionally require~$b_1+b_2 \ge 1$; and if~$\mathcal{M}' \ne \emptyset$ then~$b_1=1$; and if~$\mathcal{M}' \ne \mathcal{M}$ then~$b_2=1$.
	Finally, we set $\DP[v,b,\mathcal{M}] := \DP'[v,z,b,\mathcal{M}]$.
	
	If $\DP[\rho, 1, \mathcal{W}] \ge D$, return \yes.
	Otherwise, return \no. 
	
	\proofpara{Correctness}
	The base cases are correct.
	Overall, the correctness of Recurrence~(\ref{eqn:recurrence-HStw}) can be shown analogously to~\cite{pardi07}.
	
	\proofpara{Running time}
	As a tree has at most $2n$ vertices, the overall size of $\DP$ and $\DP'$ is~$\Oh(2^{|\mathcal{W}|} \cdot n)$.
	In Recurrence~(\ref{eqn:recurrence-HStw}), there are three options for each $M\in \mathcal{W}$.
	Hence, the the total number of terms considered in the entire table can be computed in~$\Oh(3^{|\mathcal{W}|} \cdot n)$~time.
\end{proof}

In the following, we reduce from \PDD to \HStw.
Herein, we select a subset of the modulator~$Y$ to survive.
Additionally, we select the first taxon~$x_i$ which survives in~$X \setminus Y$.
Because~$\Food - Y$ is a co-cluster graph,~$x_i$ is in a specific independent set~$I \subseteq X$ and any taxon~$X \setminus (I \cup Y)$ feed on~$x_i$.
Then, by saving a taxon~$x_j \in X \setminus (I \cup Y)$, any other taxon in~$X \setminus Y$ has some prey.
Subsequently, a solution is found by Lemma~\ref{lem:PDD-HStw}.

\begin{theorem}
	\label{thm:PDD-co-cluster}
	\PDD can be solved in $\Oh(6^d \cdot n^3)$ time, when we are given a set~$Y\subseteq X$ of size $d$ such that $\Food-Y$ is a co-cluster graph.
\end{theorem}
\begin{proof}
	\proofpara{Algorithm}
	Given an instance $\Instance = (\Tree,\Food,k,D)$ of \PDD.
	Let $x_1,\dots,x_n$ be a topological ordering of $X$ which is induced by \Food.
	Iterate over the subsets $Z$ of $Y$.
	Let $P_Z$ be the sources of~\Food in~$X\setminus Y$ and let $Q_Z$ be $\predators{Z} \setminus Y$, the taxa in $X\setminus Y$ which are being fed by~$Z$.
	Further, define~$R_Z := P_Z \cup Q_Z \subseteq X\setminus Y$.
	Iterate over the vertices $x_i \in R_Z$.
	Let $x_i$ be from the independent set $I$ of the co-cluster graph~$\Food-Y$.
	Iterate over the vertices~$x_j \in X \setminus (Y \cup I)$.
	
	For each set~$Z$, and taxa~$x_i$ and~$x_j$, with Lemma~\ref{lem:PDD-HStw} we compute the optimal solution for the case that~$Z$ is the set of taxa of $Y$ that survive while all taxa of~$Y\setminus Z$ go extinct, $x_i$ is the first taxon in $X\setminus Y$, and $x_j$ the first taxon in $X \setminus (Y \cup I)$ to survive.
	(The special cases that only taxa from $I \cup Y$ or only from $Y$ survive are omitted here.)

	More formally, we define instances $\Instance_{Z,i,j}$ of \HStw for~$Z \subseteq Y$, and~$i,j\in [n]$ with~$i<j$ as follows.
	Let the universe $\mathcal{U}_{i,j}$ be the union of~$\{ x_{i+1}, \dots, x_{j-1} \} \cap I$ and~$\{ x_{j+1}, \dots, x_n \} \setminus Y$.
	In other words, we let the taxa in~$Y$ and in $\{x_1, \dots, x_{i}\}$ and in $\{x_1, \dots, x_{j}\} \setminus I$ go extinct.
	For each taxon~$x\in Z$ compute the set~$\prey{x}$.
	If~$x\not\in \sources$ and~$\prey{x} \cap (Z \cup \{x_i,x_j\}) = \emptyset$, then add $\prey{x} \setminus Y$ to the family of sets $\mathcal{W}_{Z,i,j}$.
	For each edge~$uv\in E(\Tree)$ for which in~$\off(uv)$ a vertex of~$Z \cup \{x_i,x_j\}$ is contained, add an edge~$uw$ for each child~$w$ of~$v$ and remove~$v$ and its incident edges to obtain~$\Tree_{Z,i,j}$.
	Finally, we set $k' := k - |Z| - 2$ and $D' := D - \PD(Z \cup \{x_i,x_j\})$.
	
	Solve~$\Instance_{Z,i,j}$.
	If $\Instance_{Z,i,j}$ is a \yes-instance then return \yes.
	Otherwise, continue with the iteration.
	If $\Instance_{Z,i,j}$ is a \no-instance for every $Z\subseteq Y$, and each $i,j\in [n]$, then return \no.
	
	\proofpara{Correctness}
	We show that if the algorithm returns \yes, then $\Instance$ is a \yes-instance of~\PDD.
	Afterward, we show the converse.
	
	Let the algorithm return \yes.
	Consequently, there is a set $Z \subseteq Y$, and there are taxa~$x_i \in X \setminus Y$ and~$x_j \in X \setminus (Y \cup V(I))$ such that $\Instance_{Z,i,j}$ is a \yes-instance of \HStw.
	As before, $I$ is the independent set such that~$x_i\in V(I)$.
	Consequently, there is a set $S \subseteq \mathcal{U}_{i,j}$
	of size at most $k - |Z| - 2$
	such that~$\PDsub{\Tree_{Z,i,j}}(S) \ge D' = D - \PD(Z \cup \{x_i,x_j\})$
	and $S \cap W \ne \emptyset$ for each~$W\in \mathcal{W}_{Z,i,j}$.
	We show that $S^* := S \cup Z \cup \{x_i,x_j\}$ is a solution for instance \Instance of \PDD.
	Clearly, $|S^*| = |S| + |Z| + 2 \le k$
	and $\PD(S^*) = \PDsub{\Tree_{Z,i,j}}(S) + \PD(Z \cup \{x_i,x_j\}) \ge D$ be-\linebreak cause~$\Tree_{Z,i,j}$ is the tree which  results from $\Tree$ after saving~$Z \cup \{x_i,x_j\}$.
	Further, by definition $x_i$ is a source, or is fed by~$Z$.
	Because $\Food-Y$ is a co-cluster graph and~$x_j$ is not in $I$, the independent set in which $x_i$ is, we conclude that $x_j\in \predators{x_i}$.
	As~$S \cap W \ne \emptyset$ for each $W\in \mathcal{W}_{Z,i,j}$,
	each taxon $x\in Z$ has a prey in $Z \cup \{x_i,x_j\}$ or in~$S$ so that $\prey{x} \cap S^* \ne \emptyset$.
	Therefore, $S^*$ is viable and indeed a solution for \Instance.

	Assume now that $S$ is a solution for instance \Instance of \PDD.
	We define $Z := S \cap Y$ and let $x_i$ and $x_j$ be the taxa in $S\setminus Y$, respectively $S\setminus (Y \cup I)$, with the smallest index.
	As before, $I$ is the independent set of $x_i$.
	We show that instance $\Instance_{Z,i,j}$ of \HStw has $S^* := S \setminus (Z \cup \{x_i,x_j\})$ as a solution.
	Clearly, the size of~$S^*$ is~$|S| - |Z| - 2 \le k'$
	and by the definition of $\Tree_{Z,i,j}$, we also conclude that~$\PDsub{\Tree_{Z,i,j}}(S^*) \ge D'$.
	Let $M \in \mathcal{W}_{Z,i,j}$.
	There is a taxon $z\in Z$ that, by definition, satisfies~$M = \prey{z} \setminus Y$,
	and $z \not\in \sources$,
	and $\prey{z} \cap (Z \cup \{x_i,x_j\}) = \emptyset$.
	Consequently, as $S$ is viable, there is a taxon $x\in S \cap \prey{z}$ so that $S\cap M \ne \emptyset$.
	Hence,~$S^*$ is a solution of instance $\Instance_{Z,i,j}$ of \HStw.

	\proofpara{Running time}
	For a given $Z\subseteq Y$, we can compute
	the topological order $x_1,\dots,x_n$ and the set $R_Y$ in $\Oh(n^2)$ time.
	The iterations over $x_i$ and $x_j$ take $\Oh(n^2)$ time.
	Observe, $|\mathcal{W}_{Z,i,j}| \le |Z|$.
	By Lemma~\ref{lem:PDD-HStw}, checking whether $\Instance_{Z,i,j}$ is a \yes-instance takes $\Oh(3^d n)$~time each.
	So, the overall running time is $\Oh(6^d \cdot n^3)$ time.
\end{proof}

\subsection{Treewidth}
Conjecture~4.2., formulated by Faller et~al.~\cite{faller} supposes that \sPDD remains \NP-hard on instances where the underlying graph of the food-web is a tree.
In this subsection, assuming \PneqNP, we disprove this conjecture by showing that \sPDD can be solved in polynomial time on food-webs which are trees.
We even show a stronger result: \sPDD is \FPT with respect to the treewidth of the food-web.

\begin{theorem}
	\label{thm:PDD-tw}
	\sPDD can be solved in $\Oh(9^\tw \cdot nk)$~time.
\end{theorem}

To show~Theorem~\ref{thm:PDD-tw}, we define a dynamic programming algorithm over a tree-de\-com\-po\-si\-tion of~\Food.
In each bag, we divide the taxa into three sets indicating that they
	a)~``are supposed to go extinct'',
	b)~``will be saved'' but ``still need prey'',
	c)~or ``will be saved'' without restrictions.
The algorithm is similar to the standard treewidth algorithm for \mbox{\DS~\cite{cygan}}.

\begin{proof}
	Let $\Instance = (\Tree, \Food, k, D)$ be an instance of \sPDD.
	We define a dynamic programming algorithm over a nice tree-decomposition~$T$ of~$\Food=(V_\Food,E_\Food)$.
	
	For a node~$t\in T$, let $Q_t$ be the bag associated with~$t$ and let $V_t$ be the union of bags in the subtree of~$T$ rooted at~$t$.
	
	\proofpara{Definition of the Table}
	We index solutions by a partition $R\cup G\cup B$ of $Q_t$, and a non-negative integer $s$.
	For a set of taxa~$Y \subseteq V_t$, we call a vertex $u \in V_t$ \emph{red with respect to $Y$}
	if $u$ is in $Y$ and $u$ has a predator but no prey in $Y$.
	We call $u$ \emph{green with respect to $Y$}
	if $u$ is in $Y$ and
		a)~$u$ is a source in \Food, or
		b)~has prey in $Y$.
	Finally, we call $u$ \emph{black with respect to $Y$}
	if $u$ is not in $Y$.
	
	For a node~$t$ of the tree-decomposition, a partition $R\cup G\cup B$ of $Q_t$, and an integer~$s$, a set of taxa~$Y \subseteq V_t$ is called \emph{$(t,R,G,B,s)$-feasible}, if all the following conditions hold.
	\begin{enumerate}
		\item[(T1)] Each vertex in $Y\setminus Q_t$ is green with respect to $Y$.
		\item[(T2)] The vertices $R\subseteq Q_t$ are red with respect to $Y$.
		%are not leaves in \Net and have an incoming but no outgoing edge in $F$.
		\item[(T3)] The vertices $G\subseteq Q_t$ are green with respect to $Y$. 
		%leaves in \Net or have an outgoing edge in $F$.
		\item[(T4)] The vertices $B\subseteq Q_t$ are black with respect to $Y$.
		%not leaves in \Net and are not incident with an edge of $F$.
		\item[(T5)] The size of $Y$ is $s$.
	\end{enumerate}

	In entry~$\DP[t,A,R,G,B,s]$, we want to store the largest diversity~$\PD(Y)$\lb of~$(t,R,G,B,s)$-feasible sets~$Y$.
	If there is no $(t,R,G,B,s)$-feasible sets~$Y$, we want~$-\infty$ to be stored.
	Let $r$ be the root of the nice tree-decomposition $T$.
	Then, $\DP[r,\emptyset,\emptyset,\emptyset,k]$ stores an optimal diversity.
	So, we return~\yes if $\DP[r,\emptyset,\emptyset,\emptyset,k]\ge D$ and \no otherwise.
	
	In the following, any time a non-defined entry $\DP[t,R,G,B,s]$ is called (in particular, if $s < 0$), we take $\DP[t,R,G,B,s]$ to be $-\infty$.
	
	We regard the different types of nodes of a tree-decomposition separately and discuss their correctness.

	\proofpara{Leaf Node}
	For a leaf~$t$ of~$T$, the bags~$Q_t$ and $V_t$ are empty. We store
	\begin{eqnarray} \label{tw:leaf}
		\DP[t,\emptyset,\emptyset,\emptyset,0] &=& 0.
	\end{eqnarray}
	For all other values, we store $\DP[t,R,G,B,s] = -\infty$.
	
	\Recc{tw:leaf} is correct by definition.

	\proofpara{Introduce Node}
	Suppose that~$t$ is an \emph{introduce node}, that is,~$t$ has a single child~$t'$ with~$Q_t = Q_{t'} \cup \{v\}$.
	We store $\DP[t,R,G,B,s] = \DP[t',R,G,B\setminus \{v\},s]$ if $v\in B$. If~$v\in G$ but~$v$ is not a source in~\Food and~$v$ does not have prey in~$R\cup G$, then we store~$\DP[t,R,G,B,s] = -\infty$. Otherwise, we store
	\begin{eqnarray} \label{tw:insertvertex}
		\DP[t,R,G,B,s] &=& \max_{A\subseteq \predators{v} \cap (R\cup G)} \DP[t',R',G',B,s - 1] + \w(\rho v).
	\end{eqnarray}
	
	Herein, we define $R'$ and $G'$ depending on~$A$ to be~$R' := (R \setminus (\{v\} \cup \predators{v})) \cup A$ and~$G' := (G \cup (\predators{v} \cap (R \cup G))) \setminus (A \cup \{v\})$.

	If $v\in B$, then $v\not\in Y$ for any $(t,R,G,B,s)$-feasible set~$Y$.
	If $v \in G$ but is not a source and has no prey in $R\cup G$, then $v$ is not green with respect to~$Y$ for any~$Y$.
	So, these two cases store the desired value.
	Towards the correctness of \Recc{tw:insertvertex}:
	If $v\in G\cup R$, then we want to select $v$ and therefore we need to add $\w(\rho v)$ to the value of~$\DP[t,R,G,B,s]$.
	Further, by the selection of $v$ we know that the predators of $v$ could be green (but maybe are still stored as red) in $t$, but in $t'$ could be red and therefore we reallocate $\predators{v} \cap (R\cup G)$ and let $A$ be red beforehand.

	\proofpara{Forget Node}
	Suppose that~$t$ is a \emph{forget node}, that is,~$t$ has a single child~$t'$\linebreak with~$Q_t = Q_{t'} \setminus \{v\}$.
	We store
	\begin{eqnarray} \label{tw:forget}
		\DP[t,R,G,B,s] =  \max \{
		\DP[t',R,G\cup\{v\},B,s];
		\DP[t',R,G,B\cup\{v\},s]
		\}.
	\end{eqnarray}

	\Recc{tw:forget} follows from the definition that vertices in $v\in V_t \setminus Q_t$, depending on whether they are chosen or not, are either black or green with respect to~$(t,R,G,B,s)$-feasible sets.

	\proofpara{Join Node}
	Suppose that~$t\in T$ is a \emph{join node}, that is,~$t$ has two children~$t_1$ and~$t_2$ with~$Q_t = Q_{t_1} = Q_{t_2}$.
	We call two partitions $R_1\cup G_1\cup B_1$ and $R_2\cup G_2 \cup B_2$ of~$Q_t$ \emph{qualified} for $R\cup G\cup B$ if $R=(R_1\cup R_2) \setminus (G_1 \cup G_2)$ and $G = G_1\cup G_2$ (and consequently~$B = B_1 \cap B_2$). See~Figure~\ref{fig:PDD-joinColorings}.
	We store
	\begin{equation} \label{tw:join}
		\DP[t,R,G,B,s] = \max_{(\pi_1, \pi_2) \in{\cal Q},s'} ~ \DP[t_1,R_1,G_1,B_1,s'] + \DP[t_2,R_2,G_2,B_2,s-s'],
	\end{equation}
	%maximum is taken over all partitions $R_1\cup G_1\cup B_1$ and $R_2\cup G_2 \cup B_2$  that are qualified for $R,G,B$ and all $s' \in [s]$.
	where ${\cal Q}$ is the set of pairs of partitions $\pi_1 = R_1\cup G_1\cup B_1$ and $\pi_2 = R_2\cup G_2 \cup B_2$ that are qualified for $R \cup G \cup B$, and $s' \in [s]_0$.

	By~Figure~\ref{fig:PDD-joinColorings}, we observe that in \Recc{tw:join} we consider the correct combinations of~$R$,~$G$, and~$B$.

	\newcommand{\darkgreen}{green!50!black}
	\begin{figure}
		\begin{center}
			\begin{tabular}{c|c|c|c|}\cline{2-4}
				 & $B_1$ & \color{red}{$R_1$} & \color{\darkgreen}{$G_1$}\\\hline
				\multicolumn{1}{|l|}{$B_2$} & $B$ & \color{red}{$R$} & \color{\darkgreen}{$G$} \\\hline
				\multicolumn{1}{|l|}{\color{red}{$R_2$}} & \color{red}{$R$} & \color{red}{$R$} & \color{\darkgreen}{$G$} \\\hline
				\multicolumn{1}{|l|}{\color{\darkgreen}{$G_2$}} & \color{\darkgreen}{$G$} & \color{\darkgreen}{$G$} & \color{\darkgreen}{$G$} \\\hline
			\end{tabular} 
			
			\caption{This table shows the relationship between the three partitions $R_1\cup G_1\cup B_1$, $R_2\cup G_2 \cup B_2$ and $R\cup G\cup B$ in the case of a join node, when $R_1\cup G_1\cup B_1$ and $R_2\cup G_2 \cup B_2$ are qualified for $R\cup G\cup B$.
				The table shows which of the sets $R,G$, or~$B$ an element~$v \in Q_t$ will be in, depending on its membership in $R_1,G_1,B_1,R_2,G_2$, and $B_2$.
				For example if $v \in R_1$ and $v \in B_2$, then $v \in R$.}
			\label{fig:PDD-joinColorings}
		\end{center}
	\end{figure}
	
	\proofpara{Running Time}
	The table contains $\Oh(3^\tw \cdot nk)$ entries, as a tree decomposition contains $\Oh(n)$ nodes.
	Each leaf and forget node can be computed in linear time.
	An introduce node can be computed in $\Oh(2^\tw\cdot n)$~time.
	In a join node, considering only the qualified sets, $(\pi_1,\pi_2)$ already define $R$, $G$, and $B$.
	Thus, all join nodes can be computed in $\Oh(9^\tw \cdot nk)$~time,
	which is also the overall running time.
\end{proof}

\subsection{Hardness Results}
\paragraph{Max Leaf Number.}
Based on results of Faller~et~al.~\cite{faller}, in Corollary~\ref{cor:PDD-maxleaf} we show that \PDD is \NP-hard on instances in which the underlying undirected graph of the food-web has a max leaf number of~2.

\begin{corollary}
	\label{cor:PDD-maxleaf}
	\PDD is \NP-hard even if the food-web has max leaf number two.
\end{corollary}
\begin{proof}
	When regarding the reduction of Faller et al.~\cite[Theorem 5.1.]{faller} from \VC to \PDD we observe three things.
	The vertex $x$ has been added to ensue that $y$ functions as a root in their unrooted definition of the problem.
	Therefore, we do not need $x$ for our definition of \PDD, leaving paths with a length of~3.
	Further, some edges in the reduction have a weight of 0 but it would have not caused problems setting them to 1 and multiplying each other edge with a big constant.
	
	\proofpara{Reduction}
	Let $\Instance = (\Tree, \Food, k, D)$ be an instance of \PDD in which the set of taxa is~$X$ and each connected component of \Food is a path with a length of~3.
	We construct an instance~$\Instance' = (\Tree', \Food', k, D')$ of \PDD as follows.
	Let $P^{(0)},P^{(1)},\dots,P^{(q)}$ be the connected components of \Food where $P^{(i)}$ contains the taxa~$\{x_{i,0},x_{i,1},x_{i,2}\}$ and edges~$x_{i,0}x_{i,1}$ and $x_{i,1}x_{i,2}$.
	Let~$N$ and~$M$ be constants bigger than $|X|+1$ and such that $N \cdot q < M$.
	
	Let $Y$ be a set of new taxa~$y_{i,j}$ for~$i\in [q]$ and~$j\in [N]$.
	Our new set of taxa is~$X \cup Y$.
	Multiply each edge-weight in \Tree with $M$ and add the taxa $Y$ as children of the root~$\rho$ to obtain $\Tree'$.
	Set the weight of the edges $\rho y_{i,j}$ to 1 for each $i\in [q]$, and~$j\in [N]$.
	To obtain $\Food'$, we add~$Y$ as vertices to \Food and add edges~$y_{i,j}y_{i,j+1}$,~$x_{i-1,0}y_{i,1}$ and~$y_{i,N}x_{i,2}$ for each~$i\in [q]$, and each~$j\in [N]$.
	Figure~\ref{fig:PDD-maxleaf} depicts an example of how to create~$\Food'$.
	Finally, we set~$k' = k$ and set~$D' := D\cdot M$.
	
	\proofpara{Correctness}
	The reduction can be computed in polynomial time.
	The underlying graph of $\Food'$ is a path and therefore has a max leaf number of two.
	Any solution for~\Instance is also a solution for~$\Instance'$.
	
	Conversely, let $\Instance'$ be a \yes-instance of \PDD with solution $S \subseteq X \cup Y$.
	Define~$Y^{(i)}$ to be the set of the taxa $y_{i,j}$ for $j\in {[N]}$.
	If $x_{i,2}$ is in $S$ but $x_{i,1}$ is not in~$S$ then, because~$S$ is viable, we know that $Y^{(i)} \subseteq S$.
	Observe $\PD(Y^{(i)}) = N$ and~$\PD(x_{i,1}) \ge M$.
	Thus, also~$S' := (S \setminus Y^{(i)}) \cup \{x_{i,0},x_{i,1}\}$ is a solution for~$\Instance'$.
	Therefore, we assume that if $x_{i,j}$ is in~$S$ then also $x_{i,j-1}$ for $j\in \{1,2\}$.
	Define sets $S_X = S\cap X$ and $S_{Y} = S\cap Y$.
	Then, $\PD(S_X)$ is dividable by $M$ and~$\PD(S_{Y}) \le q\cdot N < M$.
	We conclude $\PD(S_X) \ge M\cdot D$ and $S_X$ is a solution for instance \Instance of \PDD.
\end{proof}
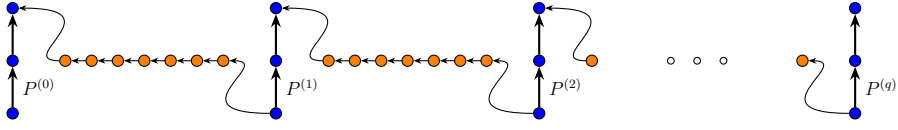
\begin{figure}[t]
	\centering
	\begin{tikzpicture}[scale=0.7,every node/.style={scale=0.6}]
		\tikzstyle{txt}=[circle,fill=white,draw=white,inner sep=0pt]
		\tikzstyle{ndeg}=[circle,fill=blue,draw=black,inner sep=2.5pt]
		\tikzstyle{ndeo}=[circle,fill=orange,draw=black,inner sep=2.5pt]
		\tikzstyle{dot}=[circle,fill=white,draw=black,inner sep=1.5pt]
		
		\foreach \i in {0,...,2}
		\node[ndeg] (a\i) at (0,\i) {};
		\foreach \i in {0,...,2}
		\node[ndeg] (b\i) at (5,\i) {};
		\foreach \i in {0,...,2}
		\node[ndeg] (c\i) at (10,\i) {};
		\foreach \i in {0,...,2}
		\node[ndeg] (d\i) at (16,\i) {};
		
		\node[txt] at (0.5,0.5) {$P^{(0)}$};
		\node[txt] at (5.5,0.5) {$P^{(1)}$};
		\node[txt] at (10.5,0.5) {$P^{(2)}$};
		\node[txt] at (16.5,0.5) {$P^{(q)}$};
		
		\draw[thick,arrows = {-Stealth[length=5pt]}] (a0) to (a1);
		\draw[thick,arrows = {-Stealth[length=5pt]}] (a1) to (a2);
		\draw[thick,arrows = {-Stealth[length=5pt]}] (b0) to (b1);
		\draw[thick,arrows = {-Stealth[length=5pt]}] (b1) to (b2);
		\draw[thick,arrows = {-Stealth[length=5pt]}] (c0) to (c1);
		\draw[thick,arrows = {-Stealth[length=5pt]}] (c1) to (c2);
		\draw[thick,arrows = {-Stealth[length=5pt]}] (d0) to (d1);
		\draw[thick,arrows = {-Stealth[length=5pt]}] (d1) to (d2);
		
		\foreach \i in {0,...,6}
		\node[ndeo] (p\i) at (1+\i/2,1) {};
		\foreach \i in {0,...,6}
		\node[ndeo] (q\i) at (6+\i/2,1) {};
		
		\foreach \i in {0,...,5}
		\draw[arrows = {-Stealth[length=3pt]}] (1.4+\i/2,1) to (p\i);
		\foreach \i in {0,...,5}
		\draw[arrows = {-Stealth[length=3pt]}] (6.4+\i/2,1) to (q\i);
		
		\node[ndeo] (r0) at (11,1) {};
		\node[ndeo] (r6) at (15,1) {};
		
		\foreach \i in {0,...,2}
		\node[dot] at (12.5+\i/2,1) {};
		
		\draw[arrows = {-Stealth[length=3pt]}] (p0) .. controls +(left:1cm) and +(right:2cm) .. (a2);
		\draw[arrows = {-Stealth[length=3pt]}] (q0) .. controls +(left:1cm) and +(right:2cm) .. (b2);
		\draw[arrows = {-Stealth[length=3pt]}] (r0) .. controls +(left:1cm) and +(right:2cm) .. (c2);
		
		\draw[arrows = {-Stealth[length=3pt]}] (b0) .. controls +(left:2cm) and +(right:1cm) .. (p6);
		\draw[arrows = {-Stealth[length=3pt]}] (c0) .. controls +(left:2cm) and +(right:1cm) .. (q6);
		\draw[arrows = {-Stealth[length=3pt]}] (d0) .. controls +(left:2cm) and +(right:1cm) .. (r6);
	\end{tikzpicture}
	\caption{This figure shows an example of the food-web $\Food'$ we reduce to in Corollary~\ref{cor:PDD-maxleaf}.
		The vertices of $X$ are blue and the vertices in $Y$ are orange.
		Here, (despite~$7<|Y|+1$) we use $N=7$.
	}
	\label{fig:PDD-maxleaf}
\end{figure}%

We observe that the construction in the previous corollary creates a food-web with an anti-chain of size $2q+1$.
We, therefore, ask whether \PDD is still \NP-hard if the DAG-width of the food-web is a constant.

\paragraph{Distance to Source-Dominant.}
In this paragraph, we consider \emph{source-dominant} food-webs.
A food-web is source-dominant, if there is a vertex~$v$, which is a source and every other vertex is a predator of~$v$.
Observe that if a directed acyclic graph~$G$ has a clique graph as underlying undirected graph, then~$G$ is source-dominant.

In Theorem~\ref{thm:PDD-dist-cluster-FPT}, we showed that \sPDD is \FPT when parameterized with distance to cluster.
The algorithm can easily be adopted to the distance to stars.
This raises the question of whether \sPDD is also \FPT with respect to the distance to source-dominant.
Here, we shortly want to discuss that such a result is unlikely.

First, if the food-web is a source-dominant, then it is necessary to save the only source, after which all other taxa can be saved without further restrictions.
Therefore, after saving the source, we can run the greedy-algorithm for \MPD.

Now, let~$\Instance = (\Tree, \Food, k, D)$ be an instance of \sPDD.
By Observation~\ref{obs:PDD-one-source}, we can assume that~$\Food$ has a single source~$s$.
For a big integer~$N$, we add~$N$ to~$\w(\rho x)$ for each taxon~$x\in X$.
Then, we add a new taxon~$\star$ and add an edge from~$\star$ to each taxon~$x\in X \setminus \{s\}$.
We finally set~$\w(\rho \star) := 1$ and~$D' := D + N\cdot k$ and obtain an instance~$\Instance'$.
By the construction, it is not possible to save~$\star$ and so~$S\subseteq X$ is a solution for~$\Instance$ if and only if~$S$ is a solution for~$\Instance'$.
Also, observe that~$\Food' - \{s\}$ is source-dominant.

We conclude the following.

\begin{proposition}
	\label{prop:PDD-DS}
	~
	\begin{propEnum}
		\item \PDD can be solved in polynomial time if the food-web is source-dominant.
		\item \sPDD is \NP-hard on instances in which the food-web has a distance to source-dominant of~$1$.
	\end{propEnum}
\end{proposition}

\section{Discussion}
\label{sec:PDD-discussion}
In this chapter, we studied the algorithmic complexity of \PDD and \sPDD with respect to various parameterizations.
\PDD is \FPT when parameterized with the solution size plus the height of the phylogenetic tree.
Consequently, \PDD is \FPT with respect to~$D$, the threshold of diversity.
However, both problems, \PDD and \sPDD, are unlikely to admit a kernel of polynomial size.
Further, unlike some other problems on maximizing phylogenetic diversity~\cite{MAPPD,TimePD}, \PDD probably does not admit an \FPT-algorithm with respect to~\Dbar, the acceptable loss of phylogenetic diversity.

We also considered the structure of the food-web.
Among other results, we showed that \PDD remains \NP-hard even if the food-web is a cluster graph.
On the positive side, we proved that \PDD is \FPT with respect to the number of vertices that need to be removed from the food-web to obtain a co-cluster.
We showed further that \sPDD is \FPT with respect to the treewidth of the food-web and therefore can be solved in polynomial time if the food-web is a tree.

Unsurprisingly,
several interesting questions remain open after our examination of \PDD and \sPDD.
Arguably the most relevant one is whether \PDD is \FPT with respect to~$k$, the size of the solution.
Also, it remains open whether \PDD can be solved in polynomial time if each connected component in the food-web contains at most two vertices.

Clearly, further structural parameterizations can be considered.
We only considered structural parameters which consider the underlying graph.
But parameters which also consider the orientation of edges, such as the longest anti-chain, could give a better view on the structure of the food-web than parameters which only consider the underlying graph.

Liebermann et al.~\cite{lieberman} introduced and analyzed weighted food-webs.
Such a weighted model may provide a more realistic view of a species’ effect on and interaction with other species~\cite{cirtwill}.
Maximizing phylogenetic diversity with respect to a weighted food-web in which one potentially needs to save several prey per predator would be an interesting generalization for our work and has the special case in which one needs to save all prey for each predator.

\chapter{Phylogenetic Diversity in Networks}
\label{ctr:Networks}

\section{Introduction}
Phylogenetic Diversity (PD) as originally proposed by Faith is defined for phylogenetic trees.
Assuming it is not possible to preserve all threatened species, e.g. due to limited resources, we would like to find a subset of species that can be preserved, for which the overall diversity is maximized.
%Optimizing PD has the nice property of also being optimal~\cite{steel,Pardi2005}.

However, when the evolution of the species under interest is also shaped by reticulation events
such as hybrid speciation, lateral gene transfer or recombination,
then the picture is no longer as rosy as for \MPD.
In reticulation events, a single species may inherit genetic material and, thus, features from multiple direct ancestors and its evolution should be represented by a phylogenetic network~\cite{phylogeneticNetworks} rather than a tree.
Several ways of extending the notion of PD for phylogenetic networks have been proposed~\cite{WickeFischer2018,bordewichNetworks}, of which two are called \apPD and \NetPD.

Unfortunately, unlike in phylogenetic trees, the optimization of phylogenetic diversity on phylogenetic trees is \NP-hard for both cases~\cite{bordewichNetworks}.
For this reason, we study the problem from the perspective of parameterized complexity.

\paragraph*{Related Work.}
All-paths phylogenetic diversity as a measure on networks was first introduced in~\cite{WickeFischer2018} under the name `phylogenetic subnet diversity'.
\NetPD has been defined in~\cite{bordewichNetworks}.

Phylogenetic diversity forms the basis of the Shapley Value, a measure that describes how much a \emph{single} species contributes to overall biodiversity. The definition of the Shapley Value involves the phylogenetic diversity of every possible subset of species, and so is difficult to calculate directly.
However, the Shapley Value is equivalent to the Fair Proportion Index on phylogenetic trees~\cite{redding2006incorporating,fuchs2015equality}, and it can be calculated in polynomial time.
In the case of phylogenetic networks, it was shown that this result also extends to Shapley Value based on the all-paths phylogenetic diversity measure.
This is in contrast to the \NP-hardness result---while it is easy to determine the individual species that contributes the most phylogenetic diversity across all sets of species, it is \NP-hard to find a \emph{set} of species for which \apPD or~\NetPD is maximal~\cite{bordewichNetworks}.

The computational complexity of \MAPPD was first studied in~\cite{bordewichNetworks}, where the authors showed that the problem is \NP-hard and cannot be approximated in polynomial time with approximation ratio better than $1-\frac{1}{e}$ unless \PeqNP, but is polynomial-time solvable on level-1-networks.
As \MaxNPD is a generalization of \MAPPD, the hardness result naturally extend.
Moreover, \MaxNPD remains \NP-hard even for the restricted class of phylogenetic networks called \emph{normal} networks~\cite{bordewichNetworks}.
\MaxNPD, on a positive side, is \FPT when parameterized by the number of reticulations~\cite{MaxNPD}, that is the number of vertices with at least two incoming edges.

Extensions of phylogenetic diversity to phylogenetic networks have also been considered, both for splits systems~\cite{spillner,PDSplitSystems,CircularSplitSystems}.

\paragraph*{Our Contribution and Structure of the Chapter.}
We study several parameterizations of the problems \MAPPDlong (\MAPPD) and \MaxNPD.
We formally define both problems in the next section.
In Section~\ref{sec:Net-wpSC} we establish an equivalence between the solution size para\-meter\-i\-za\-tions of \MAPPD and a generalization of \SC that we call \wpSClong.
Consequently, \MAPPD is \Wh{2}-hard.
We also show that \MAPPD is $\Wh 1$-hard with respect to the minimum number of taxa that need to go extinct.
On the positive side, we show in Section~\ref{sec:Net-FPTdiversity} that \MAPPD is fixed-parameter tractable with respect to $D$, the threshold of phylogenetic diversity and also with respect to the acceptable loss in phylogenetic diversity.
Afterward, we turn to structural parameters.
In Section~\ref{sec:Net-FPTtreelike} we give single-exponential fixed-parameter algorithms for \MAPPD with respect to the number of reticulations in the network, and with respect to the treewidth of the underlying graph of the network.
In the case of reticulations, this algorithm is asymptotically tight under \SETH.

Finally, in Section~\ref{sec:Net-reduction-MaxNPD}, we study \MaxNPD and show that \MaxNPD is \NP-hard even on networks with a level of~1.
This result answers an open question from~\cite{bordewichNetworks}.

\begin{table}[t]\centering
	\caption{An overview over the parameterized complexity results for \MAPPD and \MaxNPD.}
	
%	\resizebox{\columnwidth}{!}{%
		\myrowcols
		\begin{tabular}{l ll ll}
			\hline
			Parameter & \multicolumn{2}{c}{\MAPPDlong} & \multicolumn{2}{c}{\MaxNPD} \\
			\hline
			Budget $k$ & \Wh{2}-hard & Thm.~\ref{thm:Net-PSC->MAPPD} & \Wh{2}-hard & Thm.~\ref{thm:Net-PSC->MAPPD}\\
			Diversity $D$ & \FPT & Cor.~\ref{cor:Net-D} & \textit{open} & \\
			Species-loss $\kbar$ & \Wh{1}-hard & Thm.~\ref{thm:Net-kbar} & \Wh{1}-hard & Thm.~\ref{thm:Net-kbar}\\
			Diversity-loss $\Dbar$ & \FPT & Thm.~\ref{thm:Net-Dbar} & \textit{open} & \\
			\hline
			Number of & \FPT & Thm.~\ref{thm:Net-ret} & \FPT & \cite{MaxNPD}\\
			\rowcolor{white}
			reticulations~\ret & poly-kernel \textit{open} & & poly-kernel \textit{open} & \\
			\rowcolor{gray!10!white}
			Number of & \FPT & Thm.~\ref{thm:Net-ret} & \FPT & \cite{MaxNPD}\\
			ret.-edges~\eret & poly-kernel & Thm.~\ref{thm:Net-kernel} & poly-kernel \textit{open} & \\
			Level & \FPT & Thm.~\ref{thm:Net-tw} & \NP-hard for 1 & Thm.~\ref{thm:Net-level-1}\\
			Treewidth & \FPT & Thm.~\ref{thm:Net-tw} & \NP-hard for 2 & Thm.~\ref{thm:Net-level-1}\\
			\hline
		\end{tabular}
%	}
	\label{tab:Net-results}
\end{table}

\section{Preliminaries}
\label{sec:Net-prelim}
In this section, we present the formal definition of the problems, and the parameterization.
We further start with some preliminary observations.

\subsection{Phylogenetic Diversity in Phylogenetic Networks}
\label{sec:Net-PD-def}
We assume that every edge $e$ in a network $\Net = (V,E)$ has an associated \emph{weight} $\w(e)$, which is a positive integer.

\paragraph*{All-Paths Phylogenetic Diversity.}
For a set of taxa $Y$, an edge $e$ is \textit{affected by~$Y$} if $\off(e)\cap Y \ne \emptyset$ and \textit{strictly affected by $Y$} if $\off(e)\subseteq Y$.
The sets $T_Y$ and $E_Y$ are the strictly affected and affected edges by $Y$, respectively.
For a set of taxa $Y$, the \textit{all-paths phylogenetic diversity $\apPD(Y)$ of $Y$} is

\begin{equation}
	\label{eqn:apPDdef}
	\apPD(Y) := \sum_{e\in E_Y} \w(e).
\end{equation}
That is, $\apPD(Y)$ is the total weight of all edges $uv$ in $\Tree$ so that there is a path from~$v$ to a vertex in~$Y$. 
If \Net is a phylogenetic tree, then this definition coincides with the definition of phylogenetic diversity in Equation~(\ref{eqn:PDdef}).

\paragraph*{Network Diversity.}
The $\NetPD$-measure of phylogenetic diversity in networks allows the case that reticulations may not inherit all of the features from every parent.
This is modeled via an \emph{\iprop} $p(e) \in [0,1]$ on each reticulation edge~$e = uv$.
Here, $p(e)$ represents the expected proportion of features present in $u$ that are also present in $v$; or equivalently, $p(e)$ is the probability that a feature in $u$ is inherited by $v$.
We assume these probabilities are distributed independently and identically at random.
Non-reticulation edges can be considered as having an \iprop of~$1$.

For a subset of taxa $Z \subseteq X$,
the measure $\NetPD(Z)$ represents the expected number of distinct features appearing in taxa in~$Z$~\cite{bordewichNetworks}.
For each edge~$uv$, this measure is obtained by multiplying the number~$\w(uv)$ of features developed on the branch~$uv$ (which is assumed to be proportional to the length of the branch)
with the \emph{probability}~$\gam{Z}(uv)$ that a random feature appearing in $v$ or developed on $uv$ will survive when preserving~$Z$.

Formally, we define~$\gam{Z}(uv)$ as follows:

\begin{definition}
	\label{def:gamma}    
	Given a phylogenetic~$X$-network $\Net = (V,E)$ with an edge weight function~$\w: E\to\N$,
	\iprops $p: E \to [0,1]$
	and a set of taxa $Z\subseteq X$.
	We define~$\gam{Z}: E \to [0,1]$ for each edge~$uv\in E$ as follows:
	\begin{itemize}
		\item If $v$ is a leaf, then $\gam{Z}(uv) := 1$ if $v \in Z$, and $\gam{Z}(uv) := 0$, otherwise.\\
		\textbf{Intuition}: The features of $v$ survive if and only if $v$ is preserved by $Z$.
		\item If $v$ is a reticulation with outgoing edge~$vw$, then set~$\gam{Z}(uv) := p(uv)\cdot\gam{Z}(vw)$.\\
		\textbf{Intuition}: The features of $v$ are a mixture of features of its parents and
		the features of $u$ have a certain \iprop $p(uv)$ of being included in this mix and, thereby, survive in preserved descendants of $x$.
		\item If $v$ is a tree vertex with children $x_1,\dots,x_\ell$.\\
		We set~$\gam{Z}(uv) := 1 - \prod_{i=1}^\ell (1-\gam{Z}(vx_i))$.
		In the special case that $v$ has two children, $x_1$ and~$x_2$, this is equivalent to $\gam{Z}(vx_1) + \gam{Z}(vx_2) - \gam{Z}(vx_1)\cdot\gam{Z}(vx_2)$.\\
		\textbf{Intuition}: To lose a feature of $v$, it has to be lost in all children of $v$, which are assumed to be independent events, since all copies of the feature developed independently.
	\end{itemize}
\end{definition}
Further, we only consider values of $p$ on edges incoming to leaves or reticulations, so we may restrict the domain of $p$
to those edges.
%
% \todo{Turning myself around a bit on which types of features to talk about, the intuitive meaning of $\gamma$}
%The intuition behind $\gam[]{Z}(e)$ is as follows:
%if $v$ is a leaf or a reticulation, then $\gam[]{Z}$ indicates the amount of features appearing in a leaf $v$ only appear in $Z$ , and the expected only $p(e)$ of the features arising on a reticulation edge $uv$ will be passed on to $v$.
%For the case when $v$ is a tree node, we note that the expected probability of a feature arising on $e$ appearing in $Z$ is the same as the probability that it is passed on\todo{define or state differently} to $Z$ via one of the child edges $e_1,e_2$. This is the same as $1$ minus the probability the neither edge passes that feature on to $Z$ - that is $1 - (1-\gam[]{Z}(e_1))(1-\gam[]{Z}(e_2))$.
%
We now define the measure~$\NetPD^{(p)}(Z)$ for a subset of taxa $Z$ as follows:
\begin{equation}
	\label{eqn:NetPDdef}
	\NetPD^{(p)}(Z) := \sum_{e \in E}\w(e)\cdot \gam{Z}(e).
\end{equation}

We will omit the superscript~$(p)$ usually, as we do not operate with several \iprops~$p$s simultaneously.

Observe that $\gam{Z}(e)$ and $\NetPD(Z)$ are monotone,
that is, $\gam{Z'}(e) \leq \gam{Z}(e)$ and $\NetPD^p(Z') \leq \NetPD^p(Z)$
for all $Z' \subseteq Z \subseteq X$.

Observe that if a vertex has no reticulation descendants,
then $\gam[p]{Z}(e) = 1$\lb
if~$\off(e) \cap Z \neq \emptyset$, and otherwise $\gam[p]{Z}(e) = 0$.
Therefore, $\NetPD$ coincides\lb
 with~$\apPD$ if $\Net=\Tree$ is a tree.

\subsection{Problem Definitions and Parameterizations.}
This chapter's main object of study are the following two problems, introduced in~\cite{WickeFischer2018,bordewichNetworks}:

\problemdef{\MAPPDlong (\MAPPD)}
{A phylogenetic $X$-network $\Net = (V,E,\lambda)$ and two integers $k$ and $D$}
{Is there a subset $Z\subseteq X$ with a size of at most $k$ and~$\apPD(Z) \ge D$}

\problemdef{\MaxNPD}
{A phylogenetic $X$-network $\Net = (V,E,\lambda)$ with
	\iprops~$p:E\to[0,1]$,
	and two integers $k$ and $D$}
{Is there a subset $Z\subseteq X$ with a size of at most $k$ and $\NetPD(Z) \geq D$}

In Section \ref{sec:Net-wpSC}, we show that there is a strong connection between \MAPPD and the problem \wpSClong, which we define as follows.
\problemdef{\wpSClong (\wpSC)}
{A universe $\mathcal{U}$, a family $\mathcal{F}$ of subsets over $\mathcal{U}$,
	an integer weight $\w(u)$ for each item $u\in\mathcal U$,
	and two integers $k$ and $D$}
{Are there sets $F_1,\dots,F_k\in \mathcal F$ such that
	the sum of the weights of the elements in $L := \bigcup_{i=1}^k F_i$ is at least $D$}
The condition in \wpSC can be formalized as follows: $\sum_{u\in L} \w(u) \ge D$.
Observe, \textsc{Set Cover} is a special case of \wpSC with $D = \w_\Sigma(\mathcal{U})$.

\paragraph{Parameters.}
We examine \MAPPD within the framework of parameterized complexity.
In addition to the parameters $k$ and $D$, that are the number of saved taxa and the threshold of preserved phylogenetic diversity,
we also study the dual parameters which are the minimum number of species that will go extinct $\kbar := |X|-k$ and the acceptable loss of phylogenetic diversity $\Dbar := \apPD(X) - D$.
By \ret, we denote the number of reticulations in \Net, and by \twN we denote the treewidth of the underlying undirected graph of \Net.
By \eret, we denote the number of reticulation-edges that need to be removed such that~\Net is a tree.
More formally, $\eret := \sum_{r \in R} (\deg^-(r) -1) = |E| - |V| + 1$, where~$R$ is the set of reticulations of~\Net and~$\deg^-(v)$ is the in-degree of a vertex~$v$.
By $\max_\w$, we denote the biggest weight of an edge.

Further, we show that \MaxNPD is \NP-hard on level-1-networks.
Recall that the level of a network \Net is the maximum reticulation number of a subgraph~$\Net[V']$ for some~$V' \subseteq V(\Net)$ where we require that the underlying undirected graph of~$\Net[V']$ is biconnected.

In this section, we use the convention that~$n$ is the number of vertices in the network and~$m$ is the number of edges in the network.
Unlike in trees, we can not assume~$n\in \Oh(|X|)$.

\paragraph*{Binary Networks.}
A phylogenetic $X$-network is called \textit{binary} if the root has a degree of~2, each leaf has a degree of~1, and each other vertex has a degree of~3.
In this chapter, we do not assume networks to be binary; in particular, we allow tree vertices to have an in-degree and an out-degree of~$1$. 
Bordewich et~al.~\cite{bordewichNetworks} required that the given network \Net is binary.
In the following, we show that algorithmically, there is hardly any difference for \MAPPD.
\begin{lemma}
	\label{lem:Net-degree}
	When given an instance $(\Net,k,D)$ of \MAPPD, in $\Oh(|E|)$ time an equivalent instance~$(\Net',k',D')$ of \MAPPD with a binary network $\Net'$, $\twwithoutN_{\Net'}=\twN$, and $|E'| \le 2|E|$ can be computed.
\end{lemma}
\begin{proof}
	\proofpara{Algorithm}
	Let $\Instance := (\Net:=(V,E,\w),k,D)$ be an instance of \MAPPD.
	We set $k':=k$ and $D':=D\cdot (|E|+1)$.
	Iterate over $E$ and set $\w'(e) = \w(e) \cdot (|E|+1)$.
	Iterate over the vertices $v$ and compute the degree $\deg(v)$ of $v$.
	If $\deg(v)\in \{1,3\}$ then continue with the next vertex.
	If $\deg(v) = 2$ and the edges $e_1 = uv$ and $e_2 = vw$ are incident with $v$,
	then delete $e_1$ and~$e_2$ and $v$ from $\Net$ and insert the edge $uw$ with weight $\w'(uv)+\w'(vw)$.

	If $v$ is a reticulation with $\deg(v) > 3$ and the edges $e_i = u_i v$ for $i\in [\deg(v)-1]$ and $\hat e = vw$ are incident with $v$.
	Then, replace $v$ with vertices $a_1,\dots,a_{\deg(v)-2}$ and add edges $u_1a_1$, $a_{\deg(v)-2}w$, $a_{i-1} a_i$, and $u_i a_{i-1}$ for $i\in \{2,3,\dots,\deg(v)-1\}$, where the weight of the outgoing edge of $u_i$ is $\w'(u_i v)$ and all other new edges have a weight of~1.
	Otherwise, if $\deg(v) > 3$ and $v$ is a tree vertex or the root, proceed analogously with inverted edges.
	Figure~\ref{fig:Net-degree} depicts this procedure for a reticulation.
	\begin{figure}[t]
		\centering
		\begin{tikzpicture}[scale=0.8,every node/.style={scale=0.8}]
			\node[draw,fill=black,inner sep=2pt,circle] (a) at (4,1) {};
			
			\path[-{Stealth[length=6pt]},draw] (1,3) -> node[below] {$\w_{u_1}$} (a);
			\path[-{Stealth[length=6pt]},draw] (2,3) -> node[above] {$\w_{u_2}$} (a);
			\path[dotted,-{Stealth[length=6pt]},draw] (4.5,3) -> node[right] {$\w_{u_i}$} (a);
			\path[-{Stealth[length=6pt]},draw] (7,3) -> node[right] {$\w_{u_{\deg(v)-1}}$} (a);
			
			\path[-{Stealth[length=6pt]},draw] (a) -> node[right] {$\w_w$} (4,0);
			
			\node at (6.5,0.3) {(1)};
			\node at (3.8,0.8) {$v$};
			
			\path[draw] (7.5,-1) -> (7.5,4);
		\end{tikzpicture}
		\begin{tikzpicture}[scale=0.8,every node/.style={scale=0.8}]
			\node[draw,fill=black,inner sep=2pt,circle] (a) at (1,5) {};
			\node[draw,fill=black,inner sep=2pt,circle] (b) at (2,4) {};
			\node[draw,fill=black,inner sep=2pt,circle] (c) at (3,3) {};
			\node[draw,fill=black,inner sep=2pt,circle] (d) at (4,2) {};
			
			\path[-{Stealth[length=6pt]},draw] (0.5,6) -> node[left] {$\w'_{u_1}$} (a);
			\path[-{Stealth[length=6pt]},draw] (1.5,6) -> node[right] {$\w'_{u_2}$} (a);
			\path[-{Stealth[length=6pt]},draw] (3,6) -> node[right] {$\w'_{u_3}$} (b);
			\path[-{Stealth[length=6pt]},draw] (4.5,6) -> node[right] {$\w'_{u_i}$} (c);
			\path[-{Stealth[length=6pt]},draw] (6,6) -> node[right] {$\w'_{u_{\deg(v)-1}}$} (d);
			
			\path[-{Stealth[length=6pt]},draw] (4,2) -> node[right] {$\w'_w$} (4,1);
			
			\path[-{Stealth[length=6pt]},draw] (a) -> (b);
			\path[dotted,-{Stealth[length=6pt]},draw] (b) -> (c);
			\path[dotted,-{Stealth[length=6pt]},draw] (c) -> (d);
			
			\node at (5.5,2) {(2)};
			
			\node at (0.8,4.8) {$a_1$};
			\node at (1.8,3.8) {$a_2$};
			\node at (2.8,2.8) {$a_i$};
			\node at (3,2) {$a_{\deg(v)-2}$};
		\end{tikzpicture}
		\caption{An example of the construction of $\Net'$ in the case of a reticulation with multiple incoming edges.
			We use the weights $\w'_p:=\w_p \cdot (|E|+1)$.
			Unlabeled edges have a weight of~1.}
		\label{fig:Net-degree}
	\end{figure}
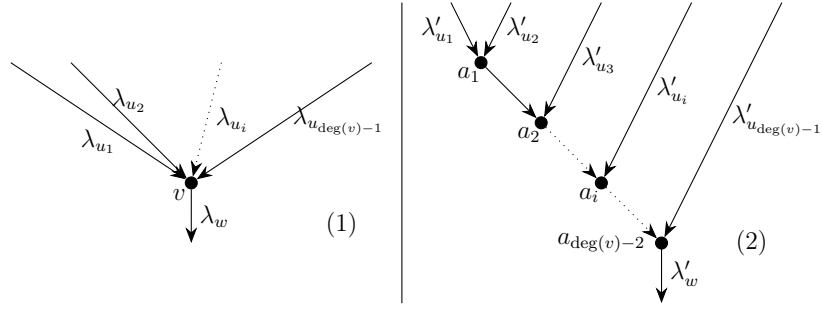%

	\proofpara{Correctness}
	The network $\Net'$ is clearly binary.
	The treewidth is not effected by the operations we did.
	We add $\Oh(\deg(v))$ edges for each vertex $v$ to $\Net'$ such that there are at most $|E| + \sum_{v\in V} \deg(v) = 2|E|$ edges in $\Net'$.

	It remains to show the correctness of the algorithm.
	Let $Y$ be a solution for $\Instance$ and let $e_1,\dots,e_\ell\in E$ be the edges that are affected by $Y$, of which $e_{t+1},\dots,e_\ell$ are incident with a vertex with a degree of~2.
	Thus, $e_1,\dots,e_t$ are also edges in $\Net'$ and there are edges $\hat e_1,\dots,\hat e_{p}$ that are affected by $Y$ with $(|E|+1)\cdot \sum_{i=t+1}^{\ell} \w(e_i) = \sum_{i=1}^{p} \w(\hat e_i)$.
	Consequently,
	\begin{eqnarray*}
		D' &\le& (|E| + 1) \cdot \apPD(Y)\\
		& = & (|E| + 1) \cdot \sum_{i=1}^{\ell} \w(e_i)\\
		& = & \sum_{i=1}^{t} \w'(e_i) + \sum_{i=1}^{p} \w(\hat e_i)\\
		& \le & \apPDsub{\Net'}(Y).
	\end{eqnarray*}

	Analogously, let $Y$ be a solution for $\Instance'$ and let $e_1,\dots,e_t$ be the edges with a weight of higher than~1.
	We can find edges $\hat e_1,\dots,\hat e_p\in E$ that are affected by $Y$ with $\sum_{i=1}^{t} \w'(e_i) = (|E| + 1) \cdot \sum_{i=1}^{p} \w(e_i)$.
	Observe that the edges with a weight of~1 are between the vertices $a_i$ and $a_{i+1}$, such that there are $\Oh(|E|)$ of these edges.
	Thus, $D' \le \sum_{i=1}^{t} \w'(e_i) = (|E| + 1) \cdot \sum_{i=1}^{p} \w(e_i) = \apPD(Y)$.
	Hence, $Y$ is also a solution for instance~$\Instance$.

	\proofpara{Running Time}
	For each vertex $v$, we perform at most $\Oh(\deg(v))$ operations. Therefore, as $|E| = \sum_{v\in V} \deg(v)$, the overall running time is $\Oh(|E|)$.
\end{proof}

\section{MapPD and Item-Weighted Partial Set Cover}
\label{sec:Net-wpSC}
In this section, we showcase a relationship between \MAPPD and \wpSC by presenting reductions in both directions.
Bordewich et al. already presented a similar reduction from \SC to \MAPPD \cite{bordewichNetworks}.
\begin{theorem}
	\label{thm:Net-PSC->MAPPD}
	For every instance $\Instance = (\mathcal{U},\mathcal{F},\w,k,D)$ of \wpSC, 
	\begin{propEnum}
		\item\label{it:P->Mweighted} an equivalent instance~$\Instance' =(\Net,k',D')$ of \MAPPD with $|X|=\ret=|\mathcal F|$ and~$k'=k$ can be computed in time polynomial in $|\mathcal U|+|\mathcal F|$;
		\item\label{it:P->Munweighted} an equivalent instance~$\Instance_2' = (\Net = (V,E,\w'),k',D')$ of \MAPPD in which $k'=k$ and each edge has a weights of~1 can be computed in time polynomial\lb in~$|\mathcal U|+|\mathcal F|+\max_\w$.
	\end{propEnum}
	
\end{theorem}
This theorem has several applications for the complexity of \MAPPD.
Because \SC is \Wh 2-hard with respect to the size of the solution $k$~\cite{downeybook}, \MAPPD is as well. 
% This is quite surprising\todos{Reviewer: I disagree: I would say that problems that are easy on trees and hard in DAGs are not uncommon.}, as
This is in contrast to the fact that
\MAPPD can be solved in polynomial time when the network does not have reticulations and, therefore, is a phylogenetic tree~\cite{steel}.
\begin{corollary} \label{cor:Net-k+maxw}
	\MAPPD is $\Wh 2$-hard when parameterized with $k$, even if $\max_\w=1$.
\end{corollary}

Recall that in \rbnb an undirected bipartite graph $G$ with vertex bipartition $V(G)=V_r \cup V_b$ and an integer $k$ are given.
The question is whether there is a set $S\subseteq V_r$ of size at least $k$ such that each vertex $v$ of $V_b$ has a neighbor in $V_r \setminus S$.
%\footnote{A vertex $u$ is a neighbor of $v$, if there is an edge $\{u,v\}$. The set of neighbors of $v$ is denoted by $N(v)$.} in $V_r\setminus S$.
%
There is a standard reduction from \rbnb to \SC:
Let $V_b$ be the universe.
For each vertex $v\in V_r$ add a set $F_v:=N(v)$ to $\mathcal F$ and finally set $k':=|V_r|-k$.
\rbnb is \Wh 1-hard when parameterized by the size of the solution~\cite{downey}.
Hence, \SC is \Wh 1-hard with respect to $|\mathcal F|-k$, and with Theorem~\ref{thm:Net-PSC->MAPPD}, we conclude as follows.
\begin{theorem} \label{thm:Net-kbar}
	\MAPPD is $\Wh 1$-hard when parameterized with $\kbar=|X|-k$.
\end{theorem}

\MAPPD can be solved in $\Oh^*(2^{|X|})$ with a brute force algorithm that tries every possible subset of species as a solution.
In Theorem~\ref{thm:Net-ret}, we prove that \MAPPD can be solved in $\Oh^*(2^{\ret})$ time.
In order to prove that these algorithms can not be improved significantly, we apply the well-established \texttt{Strong Exponential Time Hypothesis} (\SETH).

Unless \SETH fails, \SC can not be solved in $O^*(2^{\epsilon \cdot |F|})$ time for any $\epsilon<1$ \cite{cyganSETH,lin}.
Thus, Theorem~\ref{thm:Net-PSC->MAPPD} shows that under \SETH, not a lot of hope remains to find faster algorithms for \MAPPD than these two algorithms.
Thus, these two algorithms, with respect to the number of taxa $|X|$ and reticulations $\ret$, for \MAPPD are tight with the lower bounds.
\begin{corollary} \label{cor:Net-X}
	Unless \SETH fails, \MAPPD can not be solved in $2^{\epsilon\cdot |X|} \cdot \poly(|\Instance|)$ time or in $2^{\epsilon\cdot \ret} \cdot \poly(|\Instance|)$ time for any $\epsilon < 1$.
\end{corollary}

\noindent So now, without further ado, we prove Theorem~\ref{thm:Net-PSC->MAPPD}.\\[-.85cm]

\begin{proof}[Proof of Theorem~\ref{thm:Net-PSC->MAPPD}]
	\proofpara{Reduction}
	Let $\Instance = (\mathcal{U},\mathcal{F},k,D)$ be an instance of \wpSC. Let $\mathcal{U}$ consist of the items $u_1,\dots,u_n$ and let $\mathcal{F}$ contain the sets $F_1,\dots,F_m$.
	We may assume that for each $u_i$ there is a set $F_j$ which contains $u_i$.
	We define an\lb instance~$\Instance' = (\Net,k,D')$ of \MAPPD as follows.
	Let $k$ stay unchanged and define~$D' := D\cdot Q + 1$ for $Q := m(n+1)$.
	We define a network~\Net with leaves~$x_1,\dots,x_m$, and further vertices~$\rho,v_1,\dots,v_n,w_1,\dots,w_m$.

	Let the set of edges consist of the edges $\rho v_{i}$ for $i\in [n]$, $w_j x_j$ for $j\in [m]$, and let~$v_i w_j$ be an edge if and only if $u_i \in F_j$.
	We define the weight of $\rho v_i$ to be $\w(u_i) \cdot Q$ for each $i\in [n]$ and 1 for each other edge.
	Figure~\ref{fig:Net-PSC->MAPPD} depicts an example of this reduction.
	\begin{figure}[t]
		\centering
		\begin{tikzpicture}[scale=0.9,every node/.style={scale=0.85}]
			\node (root) at (0,3) {$\rho$};
			\begin{scope}[name prefix = u]
				\node (a) at (-6,1.8) {$v_1$};
				\node (b) at (-3.6,1.8) {$v_2$};
				\node (c) at (-1.2,1.8) {$v_3$};
				\node (d) at (1.2,1.8) {$v_4$};
				\node (e) at (3.6,1.8) {$v_5$};
				\node (f) at (6,1.8) {$v_6$};
			\end{scope}
			\begin{scope}[name prefix = F]
				\node (a) at (-4.5,0) {$w_1$};
				\node (b) at (-2.25,0) {$w_2$};
				\node (c) at (0,0) {$w_3$};
				\node (d) at (2.25,0) {$w_4$};
				\node (e) at (4.5,0) {$w_5$};
			\end{scope}
			\begin{scope}[name prefix = x]
				\node (a) at (-4.5,-1) {$x_1$};
				\node (b) at (-2.25,-1) {$x_2$};
				\node (c) at (0,-1) {$x_3$};
				\node (d) at (2.25,-1) {$x_4$};
				\node (e) at (4.5,-1) {$x_5$};
			\end{scope}
			
			\draw[-{Stealth[length=6pt]}] (root) to node[above]{35} (ua);
			\draw[-{Stealth[length=6pt]}] (root) to node[below]{70} (ub);
			\draw[-{Stealth[length=6pt]}] (root) to node[left]{105} (uc);
			\draw[-{Stealth[length=6pt]}] (root) to node[right]{140} (ud);
			\draw[-{Stealth[length=6pt]}] (root) to node[below]{175} (ue);
			\draw[-{Stealth[length=6pt]}] (root) to node[above]{210} (uf);

			\draw[-{Stealth[length=6pt]}] (ub) to (Fa);
			\draw[-{Stealth[length=6pt]}] (uc) to (Fa);
			\draw[-{Stealth[length=6pt]}] (ud) to (Fa);
			\draw[-{Stealth[length=6pt]}] (ua) to (Fb);
			\draw[-{Stealth[length=6pt]}] (uf) to (Fb);
			\draw[-{Stealth[length=6pt]}] (ua) to (Fc);
			\draw[-{Stealth[length=6pt]}] (uc) to (Fc);
			\draw[-{Stealth[length=6pt]}] (ud) to (Fc);
			\draw[-{Stealth[length=6pt]}] (ub) to (Fd);
			\draw[-{Stealth[length=6pt]}] (ue) to (Fd);
			\draw[-{Stealth[length=6pt]}] (uf) to (Fd);
			\draw[-{Stealth[length=6pt]}] (ua) to (Fe);
			\draw[-{Stealth[length=6pt]}] (uc) to (Fe);
			\draw[-{Stealth[length=6pt]}] (ue) to (Fe);

			\draw[-{Stealth[length=6pt]}] (Fa) to (xa);
			\draw[-{Stealth[length=6pt]}] (Fb) to (xb);
			\draw[-{Stealth[length=6pt]}] (Fc) to (xc);
			\draw[-{Stealth[length=6pt]}] (Fd) to (xd);
			\draw[-{Stealth[length=6pt]}] (Fe) to (xe);
		\end{tikzpicture}
		\caption{This figure depicts the network \Net that we reduce to from the instance\lb $(\mathcal{U}:=\{u_1,\dots,u_6\},\mathcal{F}:=\{F_1,\dots,F_5\},\w,k,D)$ of \wpSC with $\w(u_i)=i$, $F_1:=\{u_2,u_3,u_4\}$, $F_2:=\{u_1,u_6\}$, $F_3:=\{u_1,u_3,u_4\}$, $F_4:=\{u_2,u_5,u_6\}$, $F_5:=\{u_1,u_3,u_5\}$.
			Unlabeled edges have a weight of 1.
			Here $n=6,m=5$ and $Q = 35$.
			The value of $k'$ would be $k$ and $D'$ would be $35D + 1$.}
		\label{fig:Net-PSC->MAPPD}
	\end{figure}
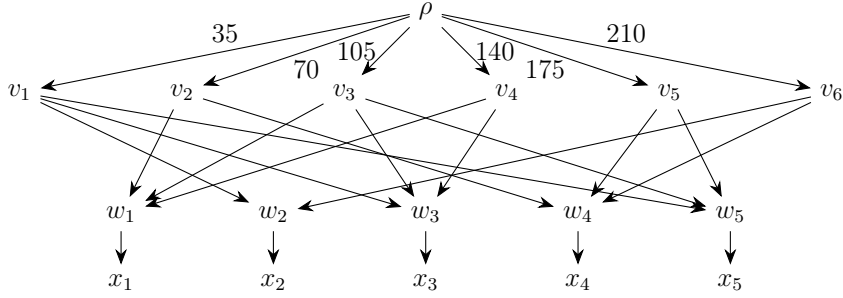%
	
	This completes the construction of instance $\Instance'$ in case~(\ref{it:P->Mweighted}) of the theorem.
	We now describe how to construct an instance $\Instance_2'$ from $\Instance'$  in which the maximum weight of an edge is 1, completing the construction for case~(\ref{it:P->Munweighted}).
	For each edge~$e = \rho v_i$ with~$w(e) > 1$, make $\w(e)-1$ subdivisions and attach a new leaf as the child of each subdividing vertex.
	We call these newly added leaves \emph{false leaves}, in contrast to the other leaves of~\Net, that we call \emph{true leaves}.

	\proofpara{Correctness}
	The instance $\Instance'$ is computed in time polynomial in $|\mathcal U|+|\mathcal F|$ and
	the instance $\Instance_2'$ is computed in time polynomial in $|\mathcal U|+|\mathcal F|+\max_\w$.
	Clearly, in $\Instance'$ we observe $k'=k$ and $|X|=\ret=|\mathcal F|$;
	and in $\Instance_2'$ we observe $k'=k$ and $\max_{\omega'}=1$.
	It remains to show the equivalence of the instances.

	Without loss of generality, let $S:=\{F_{1},\dots,F_{\ell}\}$ with $\ell\le k$ be a solution for the instance \Instance of \wpSC that covers the items $u_1,\dots,u_p$.
	We show that $Y:=\{x_{1},\dots,x_{\ell}\}$ is a solution for the instances $\Instance'$ and $\Instance_2'$ of \MAPPD.
	Clearly, the size of $Y$ is at most~$k$.
	Now, consider the phylogenetic diversity of $Y$ in the network of $\Instance'$.
	Let $\hat E$ be the edges in \Net between two vertices of $v_1,\dots,v_p,w_1,\dots,w_\ell$.
	Then,
	\begin{eqnarray*}
		\apPD(Y) &=&  \sum_{i=1}^\ell \w'(w_i x_i) + \sum_{e\in \hat E} \w'(e) + \sum_{i=1}^p \w'(\rho v_i)
		\ge \w'(w_1 x_1) + \sum_{i=1}^p \w'(\rho v_i)\\
		&=& 1 + \sum_{i=1}^p \w(u_i)\cdot Q
		= 1 + Q \cdot \sum_{i=1}^p \w(u_i)
		\ge 1 + QD = D'.
	\end{eqnarray*}
	Here, the first inequality holds because the sets in $S$ cover the items $u_1,\dots,u_p$ and the last inequality holds because $S$ is a solution.
	It is easy to see that the phylogenetic diversity of the set $Y$ in the network of $\Instance_2'$ is identical.
	Hence, $Y$ is a solution of $\Instance'$ and $\Instance_2'$.

	Without loss of generality, let $Y:=\{x_1,\dots,x_\ell\}$ with $\ell\le k$ be a solution for the instance $\Instance'$ of \MAPPD.
	We show that $S:=\{F_1,\dots,F_\ell\}$ is a solution of \wpSC.
	Clearly, the size of $S$ is at most $k$.
	Without loss of generality, let $v_1,\dots,v_p$ be the ancestors of $Y$ that are children of $\rho$.
	Thus, there is a path from $v_i$ to a taxon $x_j\in Y$, for each $i\in [p]$.
	By the construction of $\Instance'$, we conclude $u_i \in F_j$.
	Again, letting $\hat E$ be the edges between two vertices of $v_1,\dots,v_p,w_1,\dots,w_\ell$, we have 
	$$
	D' \le \apPD(Y) = \sum_{i=1}^\ell \w'(w_i x_i) + \sum_{e\in \hat E} \w'(e) + \sum_{i=1}^p \w'(\rho v_i) \le m + nm + \sum_{i=1}^p \w'(\rho v_i).
	$$
	Consequently, $\sum_{i=1}^p Q\cdot \w(u_i) = \sum_{i=1}^p \w'(\rho v_i) \ge D' - Q = Q(D-1) + 1$.
	We conclude that $\sum_{i=1}^p \w(u_i) \ge D-1 + 1/Q$.
	Since the weights of $u_i$ are integers, it follows that~$\sum_{i=1}^p \w(u_i) \ge D$.

	It remains to show that for each solution $Y$ of instance $\Instance_2'$, an equivalent solution of 
	%the instance \Instance of \wpSC 
	$\Instance'$ exists.
	If $Y$ does not contain false leaves, then 
	%we are done by the previous arguments.
	as previously observed, the phylogenetic diversity of $Y$ is the same in $\Instance'$ as in $\Instance_2'$.
	%
	%Also, if $X\subseteq Y$, then $k\ge m$. Consequently, $\Instance$ is a \yes-instance with trivial solution $F_1,\dots,F_m$.
	%
	Assume now otherwise and let~$z$ be a false leaf in $Y$ and let $p_z$ be the parent of $z$.
	We consider two different cases;
	That~$p_z$ has an offspring in $Y\setminus\{z\}$, or not.
	In the former case, $p_z$ has an offspring in $Y\setminus\{z\}$, we observe $\apPD(Y) = \apPD(Y\setminus\{z\}) + \w(p_z z) = \apPD(Y\setminus\{z\}) + 1$.\lb
	Consequently, we can replace $z$ with any true leaf that is not yet in $Y$ to obtain another solution of $\Instance_2'$.
	In the second case, let the true leaf $x_i$ be an offspring of $p_z$.
	Because of the assumption, $x_i\not\in Y$.
	
	Then, $\apPD(Y) - \w(p_z z) \le \apPD((Y\setminus \{z\}) \cup \{x_i\}) - \w(w_i x_i)$.
	Consequently, $Y\setminus \{z\} \cup \{x_i\}$ is also a solution.
	Therefore, we can remove all false leaves and then we are done.
\end{proof}

In the proof of Theorem~\ref{thm:Net-PSC->MAPPD}, we can see that in the root $\rho$, we model an operation that ensures that at least $D$ of the children of $\rho$ are selected and further, these tree vertices ensure that at least one of the reticulations below them are selected.
It might appear that by adding more layers of reticulations and tree vertices to the construction of $\Net$, one could reduce from problems even more complex than \wpSC, and thereby show that \MAPPD has an an even higher position in the W-hierarchy.
%Therefore, it seemed likely that \MAPPD with respect to $k$ could take an even higher position in the \texttt{W}-hierarchy with reductions from \textsc{Weighted Monotonous $t$-Normalized Satisfiability}.
%
This, however, is unlikely, because of the following reduction to \wpSC.

\begin{theorem} \label{thm:Net-MAPPD->PSC}
	For every instance $\Instance=(\Net,k,D)$ of \MAPPD, we can compute an equivalent instance~$(\mathcal U,\mathcal F,\w,k',D')$ of \wpSC with $k'=k$, $D'=D$ and $\max_{\w'}=\max_{\w}$ in time polynomial in $|\Instance|$.
\end{theorem}
\begin{proof}
	\proofpara{Reduction}
	Let $\Instance = (\Net,k,D)$ be an instance of \MAPPD.
	We define an instance $\Instance' = (\mathcal{U},\mathcal{F},\w',k,D)$ of \wpSC as follows.
	Let $k$ and $D$ stay unchanged.
	For each edge $e$ of \Net, define an item $u_e$ with a weight of $\w'(u_e)=\w(e)$ and let $\mathcal{U}$ be the set of these $u_e$.
	For each taxon $x$, define a set $F_x$ which contains item $u_e$ if and only if $e$ is affected by $\{x\}$.
	Let $\mathcal{F}$ be the family of these sets.

	\proofpara{Correctness}
	Clearly, the reduction is computed in polynomial time.
	We show the equivalence of the two instances.
	Let $Y$ be a solution for the instance \Instance of \MAPPD.
	Without loss of generality, assume $Y=\{x_1,\dots,x_\ell\}$ with $\ell\le k$.
	We show that~$F_1,\dots,F_\ell$ is a solution for instance~$\Instance'$ of \wpSC.
	We know by definition that~$\ell\le k$.
	Let $E_Y$ be the edges affected by $Y$.
	Observe that $e$ is in $E_Y$ if and only if $u_e$ is in $F^+ := \bigcup_{i=1}^\ell F_i$.
	Then, $D \le \apPD(Y) = \sum_{e\in E_Y} \w(e) = \sum_{u_e\in F^+} \w'(u_e)$.
	Hence, $F_1,\dots,F_\ell$ is a solution for instance~$\Instance'$ of \wpSC.

	Now, without loss of generality, let $F_1,\dots,F_\ell$ be a solution for instance $\Instance'$ of \wpSC.
	Let~$u_{e_1},\dots,u_{e_p}$ be the items in the union of $F_1,\dots,F_\ell$.
	By the construction, the edges $e_1,\dots,e_p$ are affected by $Y=\{x_1,\dots,x_\ell\}$.
	Then,
	$$
	\apPD(Y)\ge \sum_{i=1}^p \w(e_i) = \sum_{i=1}^p \w'(u_{e_i}) \ge D.
	$$
	The size of $Y$ is at most $k$.
	Hence, $Y$ is a solution for $\Instance$ of \MAPPD.
\end{proof}

To the best of our knowledge, it is unknown if \wpSC is $\Wh 2$-complete, like \textsc{Set Cover}.
Nevertheless, we obtain the following connection between \wpSC and \MAPPD.

\begin{corollary}
	\label{cor:Net-k}
	\MAPPD is $\Wh i$-complete with respect to $k$ if and only if \wpSC is~$\Wh i$-complete with respect to $k$.
\end{corollary}

\section{Fixed-Parameter Tractability of MapPD}
\label{sec:Net-FPT}

\subsection{Preserved and Lost Diversity}
\label{sec:Net-FPTdiversity}
In this subsection, we show that \MAPPD is \FPT with respect to $D$, the threshold of phylogenetic diversity, and $\Dbar := \apPD(X)-D$, the acceptable loss of phylogenetic diversity.

Let \Instance be an instance of \MAPPD. If there is an edge $e$ with $\w(e)\ge D$ and $k\ge 1$,  then for each offspring $x$ of $e$ we have $\apPD(\{x\})\ge\w(e)\ge D$, and so $\{x\}$ is a solution for \Instance.
So, we may assume that $\max_\w < D$.
Therefore, each edge $e$ can be subdivided~$\w(e)-1$ times in $\Oh(D\cdot m)$ time such that $\w'(e)=1$ for each edge $e$ of the new  network $\Net'$.
Bläser showed that \wpSC can be solved in $\Oh^*(2^{\Oh(D)})$ time when $\w(u)=1$ for each item $u\in\mathcal{U}$~\cite{blaeser}.
Consequently, with Theorem~\ref{thm:Net-MAPPD->PSC} and the result from Bläser we conclude the following.
\begin{corollary}
	\label{cor:Net-D}
	\MAPPD can be solved in $\Oh^*(2^{\Oh(D)})$ time.
\end{corollary}

As \textsc{Set Cover} is a special case of \wpSC with $D = \sum_{u\in \mathcal{U}}\w(u)$, \wpSC is para-\NP-hard with respect to the dual $\sum_{u\in \mathcal{U}}\w(u) - D$.
Observe~$\Dbar > \kbar := |X| - k$.
By contrast, we show in the following that \MAPPD is \FPT with respect to \Dbar.
More precisely, we show the following.

\begin{theorem}
	\label{thm:Net-Dbar}
	\MAPPD can be solved in $\Oh(2^{\Dbar + \kbar + o(\Dbar)} \cdot n\log n)$~time.
\end{theorem}

To this end, we use the technique of color coding.
Recall that $\off(e)=\off(w)$ for each edge $e = vw$.
% and the strictly affected edges $T_Y$ for a set of taxa $Y\subseteq X$ is the set of edges~$e$ with $\off(e)\subseteq Y$.
%
We define the following auxiliary problem,
in which we assign a \emph{color}, red or green, to each taxon.
A set~$Y\subseteq X$ is \emph{color-fitting} if each taxon $x\in Y$ is red and for each vertex~$v \in V(\Net)$, at least one of the following conditions is satisfied:
\begin{itemize}
	\item $v$ has a green offspring,
	\item all offspring of $v$ are in $Y$, or
	\item all offspring of $v$ are in $X \setminus Y$.
\end{itemize}
We define the colored version of \MAPPD as follows.
\problemdef{\cMAPPDlong (\cMAPPD)}
{A phylogenetic $X$-network \Net, integers $k$ and $D$, and a coloring on the taxa $c: X\to \{\text{red},\text{green}\}$}
{Is there a subset $S\subseteq X$ of taxa such that $|S|\le k$, $\apPD(S)\ge D$, and~$X\setminus S$ is color-fitting}

\begin{lemma} \label{lem:Net-Dbar}
	\cMAPPD can be solved in~$\Oh(\Dbar \cdot m)$~time.
\end{lemma}
\begin{proof}
	\proofpara{Algorithm}
	Let $\Instance := (\Net:=(V,E,\w),k,D,c)$ be an instance of \cMAPPD.
	Delete all edges~$uv$ for which~$v$ has a green offspring, and then delete all isolated vertices.
	For any vertex~$u$ with an in-degree of~$0$ with children~$v_1, \dots, v_q$, replace~$u$ with~$q$ vertices~$u_1, \dots, u_q$ and add an edge~$u_iv_i$ with a weight of~$\w(uv_i)$, for each~$i \in [q]$.
	Each in-degree-$0$ vertex has one child, now.
	Let $G'$ be the resulting graph.
	An example of this transformation is depicted in Figure~\ref{fig:Net-transformation}.
	
	\begin{figure}[t]
		\centering
		\begin{tikzpicture}[scale=0.8,every node/.style={scale=0.7}]
			\draw[-{Stealth[length=6pt]}] (5,8) -> (7,7);
			\draw[-{Stealth[length=6pt]}] (5,8) -> (5,7.4);
			\draw[-{Stealth[length=6pt]}] (5,8) -> (2.5,7);
			
			\draw[-{Stealth[length=6pt]}] (2.5,7) -> (2,6);
			\draw[-{Stealth[length=6pt]}] (2.5,7) -> (3.45,6);
			\draw[-{Stealth[length=6pt]}] (2.5,7) -> node[above] {5} (4.25,6.5);
			
			\draw[-{Stealth[length=6pt]}] (5,7.4) -> node[left] {4} (4.25,6.5);
			\draw[-{Stealth[length=6pt]}] (5,7.4) -> (6,6.5);
			
			\draw[-{Stealth[length=6pt]}] (7,7) -> (6,6.5);
			\draw[-{Stealth[length=6pt]}] (7,7) -> node[left] {6} (7,6);
			
			\draw[-{Stealth[length=6pt]}] (6,6.5) -> (6,5.8);
			\draw[-{Stealth[length=6pt]}] (4.25,6.5) -> (4.25,5.3);

			\draw[-{Stealth[length=6pt]}] (3.45,6) -> node[left] {2} (4.25,5.3);
			\draw[-{Stealth[length=6pt]}] (3.45,6) -> node[left] {3} (3.45,4.8);
			\draw[-{Stealth[length=6pt]}] (3.45,6) -> (2.6,5.3);
			\draw[-{Stealth[length=6pt]}] (2,6) -> (2.6,5.3);
			\draw[-{Stealth[length=6pt]}] (2,6) -> (2,4.8);
			\draw[-{Stealth[length=6pt]}] (4.25,5.3) -> (4.25,4.5);
			\draw[-{Stealth[length=6pt]}] (2.6,5.3) -> (2.6,4.5);

			\draw[-{Stealth[length=6pt]}] (6,5.8) -> (5.3,5.3);
			\draw[-{Stealth[length=6pt]}] (6,5.8) -> (6,5);
			\draw[-{Stealth[length=6pt]}] (6,5.8) -> (6.5,5.3);
			\draw[-{Stealth[length=6pt]}] (7,6) -> (6.5,5.3);
			\draw[-{Stealth[length=6pt]}] (7,6) -> (7,5);
			
			\draw[-{Stealth[length=6pt]}] (6.5,5.3) -> (6.5,4.5);
			\draw[-{Stealth[length=6pt]}] (5.3,5.3) -> (5,4.5);
			\draw[-{Stealth[length=6pt]}] (5.3,5.3) -> (5.5,4.5);
			
			\node[draw,fill=red,inner sep=2pt,circle] at (2,4.7) {};
			\node[draw,fill=green,inner sep=2pt,circle] at (2.6,4.4) {};
			\node[draw,fill=red,inner sep=2pt,circle] at (3.45,4.7) {};
			\node[draw,fill=red,inner sep=2pt,circle] at (4.25,4.4) {};
			
			\node[draw,fill=red,inner sep=2pt,circle] at (5,4.4) {};
			\node[draw,fill=red,inner sep=2pt,circle] at (5.5,4.4) {};
			\node[draw,fill=green,inner sep=2pt,circle] at (6,4.9) {};
			\node[draw,fill=red,inner sep=2pt,circle] at (6.5,4.4) {};
			\node[draw,fill=red,inner sep=2pt,circle] at (7,4.9) {};
			
			\node at (6.5,8) {(1)};
			
			\path[draw] (7.5,8.2) -> (7.5,4.1);
		\end{tikzpicture}
		\begin{tikzpicture}[scale=0.8,every node/.style={scale=0.7}]
			\draw[-{Stealth[length=6pt]}] (2.5,7) -> node[above] {5} (4.25,6.5);
			
			\draw[-{Stealth[length=6pt]}] (5,7.4) -> node[left] {4} (4.25,6.5);
			
			\draw[-{Stealth[length=6pt]}] (7,7) -> node[left] {6} (7,6);
			
			\draw[-{Stealth[length=6pt]}] (4.25,6.5) -> (4.25,5.3);

			\draw[-{Stealth[length=6pt]}] (3.45,6) -> node[left] {2} (4.25,5.3);
			\draw[-{Stealth[length=6pt]}] (3.45,6) -> node[left] {3} (3.45,4.8);
			\draw[-{Stealth[length=6pt]}] (2,6) -> (2,4.8);
			\draw[-{Stealth[length=6pt]}] (4.25,5.3) -> (4.25,4.5);

			\draw[-{Stealth[length=6pt]}] (6,5.8) -> (5.3,5.3);
			\draw[-{Stealth[length=6pt]}] (6,5.8) -> (6.5,5.3);
			\draw[-{Stealth[length=6pt]}] (7,6) -> (6.5,5.3);
			\draw[-{Stealth[length=6pt]}] (7,6) -> (7,5);
			
			\draw[-{Stealth[length=6pt]}] (6.5,5.3) -> (6.5,4.5);
			\draw[-{Stealth[length=6pt]}] (5.3,5.3) -> (5,4.5);
			\draw[-{Stealth[length=6pt]}] (5.3,5.3) -> (5.5,4.5);
			
			\node[draw,fill=red,inner sep=2pt,circle] at (2,4.7) {};
			\node[draw,fill=red,inner sep=2pt,circle] at (3.45,4.7) {};
			\node[draw,fill=red,inner sep=2pt,circle] at (4.25,4.4) {};
			
			\node[draw,fill=red,inner sep=2pt,circle] at (5,4.4) {};
			\node[draw,fill=red,inner sep=2pt,circle] at (5.5,4.4) {};
			\node[draw,fill=red,inner sep=2pt,circle] at (6.5,4.4) {};
			\node[draw,fill=red,inner sep=2pt,circle] at (7,4.9) {};
			
			\node at (6.5,8) {(2)};
			
			\path[draw] (7.5,8.2) -> (7.5,4.1);
		\end{tikzpicture}
		\begin{tikzpicture}[scale=0.8,every node/.style={scale=0.7}]
			\draw[-{Stealth[length=6pt]}] (2,5.8) to (2,4.8);
			\draw[-{Stealth[length=6pt]}] (2.9,5.8) to node[left] {3} (3.45,4.8);
			\draw[-{Stealth[length=6pt]}] (3.5,6.3) to node[left] {2} (4.25,5.3);
			\draw[-{Stealth[length=6pt]}] (3.75,7.5) to node[left] {5} (4.25,6.5);
			\draw[-{Stealth[length=6pt]}] (4.75,7.5) to node[right] {4} (4.25,6.5);
			\draw[-{Stealth[length=6pt]}] (5.3,6.3) to (5.3,5.3);
			\draw[-{Stealth[length=6pt]}] (6.5,6.3) to (6.5,5.3);
			\draw[-{Stealth[length=6pt]}] (7,7) to node[left] {6} (7,6);
			
			\draw[-{Stealth[length=6pt]}] (4.25,6.5) -> (4.25,5.3);
			\draw[-{Stealth[length=6pt]}] (4.25,5.3) -> (4.25,4.5);
			
			\draw[-{Stealth[length=6pt]}] (7,6) -> (6.5,5.3);
			\draw[-{Stealth[length=6pt]}] (7,6) -> (7,5);
			
			\draw[-{Stealth[length=6pt]}] (6.5,5.3) -> (6.5,4.5);
			\draw[-{Stealth[length=6pt]}] (5.3,5.3) -> (5,4.5);
			\draw[-{Stealth[length=6pt]}] (5.3,5.3) -> (5.5,4.5);
			
			\node[draw,fill=red,inner sep=2pt,circle] at (2,4.7) {};
			\node[draw,fill=red,inner sep=2pt,circle] at (3.45,4.7) {};
			\node[draw,fill=red,inner sep=2pt,circle] at (4.25,4.4) {};
			
			\node[draw,fill=red,inner sep=2pt,circle] at (5,4.4) {};
			\node[draw,fill=red,inner sep=2pt,circle] at (5.5,4.4) {};
			\node[draw,fill=red,inner sep=2pt,circle] at (6.5,4.4) {};
			\node[draw,fill=red,inner sep=2pt,circle] at (7,4.9) {};
			
			\node at (4.5,6) {$C_1$};
			\node at (5.55,6) {$C_2$};
			
			\node at (6.5,8) {(3)};
		\end{tikzpicture}
		\caption{In this figure, an example for the transformation in the proof of Lemma~\ref{lem:Net-Dbar} is given.
			A hypothetical network \Net with a coloring is given in (1).
			In (2), the graphs $G$, and in (3) the graph $G'$ is depicted.
			Some edges are labeled (to increase readability).
			Unlabeled edges have weight~1.
			The connected components $C_1$ and $C_2$ have weight~$13$ and~$3$ and value~$1$ and~$2$, respectively.}
		\label{fig:Net-transformation}
	\end{figure}
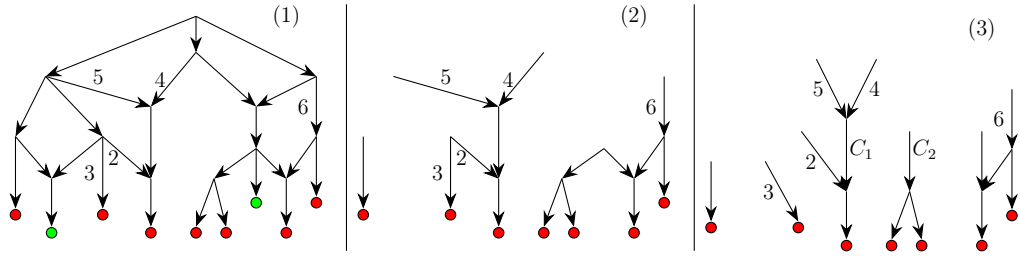%

	Compute the set of connected components of the underlying graph of $G'$.
	For connected components $C = (V_C,E_C)$, proceed as follows.
	Define an item $I_C$ with a \emph{weight} of~$\w(E_C)$ and a \emph{value} of~$|Y_C|$, where $Y_C$ is the set of taxa in $V_C$.

	Let~$M$ be the set of these items.
	Now, return \yes if there is a subset of items in~$M$ whose total weight is at most $\Dbar$ and whose total value is at least $\kbar=|X|-k$, and \no, otherwise.
	Here, \kbar and \Dbar are taken from the original network~$\Net$.	
	Observe that this can be determined by solving an instance of \KP with a set of items~$M$, budget~$\Dbar$, and target value $\kbar$.
	This can be done in~$\Oh(\Dbar \cdot |M|) \in \Oh(\Dbar \cdot |X|)$ time~\cite{weingartner,rehs}.

	\proofpara{Correctness}
	We show first that if $\Instance$ is a \yes-instance of \cMAPPD, then the algorithm return~\yes and secondly we show the converse.
	
	Assume that \Instance is a \yes-instance of \cMAPPD.
	Thus, there is a set~$S \subseteq X$ of size at most~$k$ with $\apPD(S) \ge D$ and $X\setminus S$ is color-fitting.
	We claim that~$X\setminus S$ is also color-fitting in $G'$.
	Indeed, suppose for a contradiction that this is not the case.
	Then, there exists some vertex $v$ in $G'$ that $v$ has an offspring in $X\setminus S$ and an offspring in $S$.
	Let~$v$ be a lowest vertex satisfying this condition.
	Then, $v$ does not have an in-degree of~$0$ in~$G'$, as in this case $v$ has a unique child with the same offspring.
	Furthermore, $v$ does not have any green offspring in~$\Net$, as the incoming edges of $v$ were not deleted in the construction of $G'$.
	Then, in fact, $v$ has the same offspring in $\Net$ as in $G'$, implying that $X\setminus S$ is not color-fitting.

	Since $X \setminus S$ is color-fitting in $G'$, and $G'$ has no green taxa, every vertex $v$ in $G$ satisfies $\off(v) \subseteq X\setminus S$ or $\off(v) \subseteq S$.
	Each edge $e$ in $\Net$ that is not affected by $S$ is strictly affected by $X\setminus S$.
	Because $S$ is a solution we conclude each taxon in $X\setminus S$ is red and so
	$e$ is still an edge in $G'$.
	Let $C_e$ be the connected component of the underlying undirected graph of $G'$ that contains $e$.
	Observe that every edge in $C_e$ is also strictly affected by $X \setminus S$---indeed, for any edge $uv$ in $C$, $u$ has all its offspring in $X\setminus S$ if and only if $v$ has all its offspring in $X\setminus S$.
	Thus, $C_e$ satisfies the conditions to be in~$M$ for each edge $e$ that is not affected by $S$.
	Let $C_1,\dots,C_t$ be the unique connected components that contain the edges that are not affected by $S$.
	We now conclude\lb that~$\w(C_1 \cup \dots \cup C_t) \le \Dbar$ and~$C_1 \cup \dots \cup C_t$ contain the leaves $X\setminus S$, which are at least $\kbar$.
	Hence, $I_{C_1},\dots,I_{C_t}$ is a solution for the \KP-instance and the algorithm returns \yes.
	
	For the converse, assume that the algorithm returns \yes and let $I_{C_1},\dots,I_{C_t}$ be a solution for the \KP-instance.
	Let $Y_i$ be the set of taxa of $C_i$ and define the set~$Y := \bigcup_{i=1}^t Y_i$.
	We prove that $S := X \setminus Y$ is a solution for the instance~\Instance of \cMAPPD. 
	By the construction, we conclude that~$Y_i$ are colored red and $Y_i$ and $Y_j$ are disjoint for any $i\neq j$.
	% 	Further in~$G$, the root~$\rho_G$ is the only vertex with green offspring.
	% 	Each vertex in $V(\Net) \setminus V(G)$ has green offspring.
	Each vertex $v \in \Net$ that has some but not all offspring in $Y_i$ has at least one green offspring.
	Consequently, $Y$ is color-fitting.
	Further, for any edge $e=uv$ that is strictly affected by $Y$ in $\Net$, we have that $v$  has no green offspring, and therefore $e$ is not deleted in the construction of~$G'$.
	Moreover, as $v$ has offspring in $Y_i$ for some $i \in [t]$, we conclude that~$e$ is in~$C_i$.
	Because $\sum_{i=1}^t \w(E_{C_i}) \le \Dbar$, we conclude that the phylogenetic diversity of~$S$ is~$\apPD(S) = \apPD(X) - \sum_{i=1}^t \w(E_{C_i}) \ge \apPD(X) - \Dbar = D$.
	Likewise,\lb because~$\sum_{i=1}^t |Y_i| \ge \kbar$, we conclude~$|S| = |X| - \sum_{i=1}^t |Y_i| \le |X| - \kbar = k$.

	\proofpara{Running Time}
	The graph $G'$ can be computed from $\Net$ in~$\Oh(m)$~time.
	The connected components of the underlying graph of $G'$ can be computed in~$\Oh(m)$~time as well.
	For each connected component~$C = (V_C,E_C)$ the item~$I_C$ of~$N$ is computed in~$\Oh(|E_C|)$~time.
	Consequently, we can compute~$N$ in~$\Oh(m)$~time.
	As the instance of \KP can be solved in~$\Oh(\Dbar \cdot |X|)$~time~\cite{weingartner,rehs}, we have an overall running time of $\Oh(m + \Dbar \cdot |X|) \in \Oh(\Dbar \cdot m)$.
\end{proof}

To show that \MAPPD is \FPT with respect to \Dbar, we present an reduction from \MAPPD to \cMAPPD using standard color coding techniques.
In particular, we show that there exists a family ${\cal F}$ of $2$-colorings $c:E\to \{\text{red},\text{green}\}$, with~$|{\cal F}|$ bounded by a function of $\Dbar$ times a polynomial in~$n$, such that $(\Net, k, D)$ is a \yes-instance of \MAPPD if and only if $(\Net,k,D,c)$ is a \yes-instance of \cMAPPD for some $c \in {\cal F}$.

Recall that an~$(n,k)$-universal set is a family ${\cal U}$ of subsets of $[n]$ such that for any~$S\subseteq [n]$ of size $k$, $\{A \cap S \mid A \in {\cal U}\}$ contains all $2^k$ subsets of $S$.
For any $n,k \geq 1$, one can construct an $(n,k)$-universal set of size~$2^k k^{\Oh(\log k)}\log n$ in time~$2^k k^{\Oh(\log k)}n\log n$~\cite{Naor1995SplittersAN}.

\begin{proof}[Proof of Theorem~\ref{thm:Net-Dbar}]
	\proofpara{Algorithm}
	Let $\Instance := (\Net:=(V,E,\w),k,D)$ be an instance of \MAPPD.
	Arbitrarily order the taxa $x_1,\dots, x_n$. 
	Construct an~$(n,\kbar+\Dbar)$-universal set $\mathcal{U}$.
	
	Now, for each $A \in \mathcal{U}$, construct a $2$-coloring $c_A: X\to \{\text{red},\text{green}\}$ where $x_i$ is colored green if and only if $i \in A$, and solve \cMAPPD on $(\Net, k, D, c_A)$.
	Return \yes, if~$(\Net, k, D, c_A)$ is a \yes-instance for some~$A \in \mathcal{U}$.
	Otherwise, return \no.
	
	\proofpara{Correctness}
	First observe that if $(\Net, k, D, c)$ is a \yes-instance of \cMAPPD for any coloring $c: X\to \{\text{red},\text{green}\}$, then $(\Net,k,D)$ is also a \yes-instance of \MAPPD.
	
	Now, suppose $(\Net,k,D)$ is a \yes-instance.
	Let $S\subseteq X$ be a subset of taxa of size at most~$k$ and with~$\apPD(S)\ge D$.
	If necessary, add taxa to~$S$ until~$|S| = k$.
	Consequently, $X\setminus S$ has a size of $\kbar$.
	Let $V_Y$ be the set of vertices~$u$ of \Net which have an offspring~$x_u$ in~$S$ and have a child~$v$ with~$\off(v) \subseteq X\setminus S$.
	Define~$Y := \{ x_u \mid u\in V_Y \}$.
	Observe that if we can define a coloring which colors the taxa in~$X\setminus S$ in red and the taxa in~$Y$ in green, then~$X\setminus S$ would be color-fitting.
	
	Define an operation $\Index : 2^X \to 2^{[n]}$ by~$\Index(X') := \{ i \mid x_i \in X' \}$ for sets~$X' \subseteq X$.
	
	For each~$u\in V_Y$ and each child~$v$ of~$u$ with~$\off(v) \subseteq X\setminus S$,
	the edge~$uv$ is strictly affected by~$X\setminus S$.
	We conclude $|Y| \le |V_Y| \le \sum_{u\in V_Y} \w(uv) \le \Dbar$.
	Therefore, $Z := Y \cup (X\setminus S) \subseteq X$ is a set of size at most~$\kbar + \Dbar$.
	If necessary, add taxa until~$Z$ has a size of~$\Dbar+\kbar$.
	Consequently, there is a set~$A\in \mathcal{U}$ with~$A \cap \Index(Z) = \Index(Y)$.
	So, $S$ is a solution for the instance~$(\Net, k, D, c_A)$ of \cMAPPD.
	
	\proofpara{Running Time}
	The construction of $\mathcal{U}$ takes $2^{\Dbar + \kbar + \Oh(\log^2 (\Dbar))} n\log n$ time, and for each of the $2^{\Dbar + \kbar + \Oh(\log^2 (\Dbar))}\log n$ sets in $\mathcal{U}$ we solve an instance of \cMAPPD.
	This can be done in $\Oh(\Dbar \cdot m)$~time by Lemma~\ref{lem:Net-Dbar}.
	
	The overall running time is~$\Oh(2^{\Dbar + \kbar + o(\Dbar)} \cdot n\log n)$.
\end{proof}

\subsection{Proximity to Trees}
\label{sec:Net-FPTtreelike}
\MAPPD can be solved in polynomial time with Faith's Greedy-Algorithm, if the given network is a tree \cite{FAITH1992,steel}.
Therefore, in this subsection, we examine \MAPPD with respect to two parameters that classify the network's proximity to a tree, the number of reticulations $\ret$ and the smaller parameter treewidth $\twN$.

\begin{theorem}
	\label{thm:Net-ret}
	\MAPPD can be solved in $\Oh(2^{\ret} \cdot k \cdot m)$ time.
\end{theorem}
Observe that by Corollary~\ref{cor:Net-X}, \MAPPD can not be solved in $O^*(2^{\epsilon\cdot \ret})$ time for any $\epsilon<1$, unless \SETH fails.
Therefore, the running time of this theorem is tight, to some extent.
\begin{proof}
	\proofpara{Algorithm}
	For a reticulation $v$ in a network \Net with child $u$, let $E^{(\uparrow vu)}$ be the set of edges of \Net that are between two vertices of $\anc(v) \cup \{u\}$.
	Recall that~$\off(e)\subseteq X$ is the set of offspring of $w$ for an edge $e = vw$ and
	the strictly affected edges~$T_Y$ for a set of taxa $Y\subseteq X$ is the set of edges $e$ with $\off(e)\subseteq Y$.
	Define two operations, called \take and \leave, that for an instance $\Instance = (\Net,k,D)$ and a reticulation~$v$ of~\Net return another instance of \MAPPD.
	Every subset of taxa $Y$ that \emph{does} contain an offspring of $v$ should be a solution for \Instance if and only if $Y$ is a solution for $\take(\Instance,v)$.
	Similarly, every subset of taxa $Y$ that \emph{does not} contain an offspring of $v$ should be a solution for \Instance if and only if $Y$ is a solution for $\leave(\Instance,v)$.

	We define $\leave(\Instance,v)$ to be the instance $\Instance'=(\Net',k,D)$ of \MAPPD, in which~$k$ and~$D$ are unchanged and $\Net'$ is the  network that results from deleting the edges~$T_{\off(v)}$ and the resulting isolated vertices from \Net.
	Recall that $\Dbar := \sum_{e\in E} \w(e)-D$.
	We define $\take(\Instance,v)$ to be the instance $\Instance'=(\Net',k,D')$ of \MAPPD with an unchanged~$k$ and~$D':=D+\Dbar$.
	Here, $\Net'$ is the network that results from \Net by deleting the edges~$E^{(\uparrow vu)}$, merging all the ancestors of $v$ to a single vertex $\rho$, adding an edge $\rho u$, and setting the weight of $\rho u$ to $\w(E^{(\uparrow vu)})+\Dbar$.
	For each vertex $w\ne u$ with~$t\ge 1$\lb parents~$u_1,\dots,u_t$ in $\anc(v)$, we add an edge $\rho w$ that has a weight of $\sum_{i=1}^t \w(u_i w)$.
	Observe that $\apPDsub{\Net'}(X) = \apPD(X) + \Dbar$.
	Figure~\ref{fig:Net-leave+take} depicts an example of the operations \take and \leave.
	\begin{figure}[t]
		\centering
		\begin{tikzpicture}[scale=0.8,every node/.style={scale=0.7}]
			\draw[-{Stealth[length=6pt]}] (5,8) -> (7,7);
			\draw[-{Stealth[length=6pt]}] (5,8) -> node[left] {3} (4.5,7);
			\draw[-{Stealth[length=6pt]}] (5,8) -> (3,7);
			\draw[-{Stealth[length=6pt]}] (5,8) -> node[right] {2} (6,6);
			
			\draw[-{Stealth[length=6pt]}] (3,7) -> node[left] {2} (2,6);
			\draw[-{Stealth[length=6pt]}] (3,7) -> node[left] {3} (3.75,6);
			\draw[-{Stealth[length=6pt]}] (4.5,7) -> node[right] {4} (4.5,6);
			\draw[-{Stealth[length=6pt]}] (4.5,7) -> node[above] {2} (6,6);
			\draw[-{Stealth[length=6pt]}] (4.5,7) -> (3.75,6);
			\draw[-{Stealth[length=6pt]}] (7,7) -> (6,6);
			\draw[-{Stealth[length=6pt]}] (7,7) -> (7,6);
			
			\draw[-{Stealth[length=6pt]}] (2,6) -> (2,5);
			\draw[-{Stealth[length=6pt]}] (2,6) -> (3,5);
			\draw[-{Stealth[length=6pt]}] (2,6) -> node[above] {2} (3.75,5.4);
			\draw[-{Stealth[length=6pt]}] (3.75,6) -> (3.75,5.4);
			\draw[-{Stealth[length=6pt]}] (4.5,6) -> (4.5,5);
			\draw[-{Stealth[length=6pt]}] (4.5,6) -> (6,5.4);
			\draw[-{Stealth[length=6pt]}] (6,6) -> (6,5.4);
			
			\draw[-{Stealth[length=6pt]}] (3.75,5.4) -> (3.75,4.8);
			\draw[-{Stealth[length=6pt]}] (6,5.4) -> (6,4.8);
			
			\node at (6.5,8) {(1)};
			\node at (4,6) {$v$};
			\node at (4,5.4) {$u$};
			\node at (2.5,8) {$k=3$};
			\node at (2.5,7.5) {$D=28$};
			
			\path[draw] (7.5,8) -> (7.5,4.8);
		\end{tikzpicture}
		\begin{tikzpicture}[scale=0.8,every node/.style={scale=0.7}]
			\draw[-{Stealth[length=6pt]}] (5,8) -> (7,7);
			\draw[-{Stealth[length=6pt]}] (5,8) -> node[left] {3} (4.5,7);
			\draw[-{Stealth[length=6pt]}] (5,8) -> (3,7);
			\draw[-{Stealth[length=6pt]}] (5,8) -> node[right] {2} (6,6);
			
			\draw[-{Stealth[length=6pt]}] (3,7) -> node[left] {2} (2,6);
			\draw[-{Stealth[length=6pt]}] (4.5,7) -> node[left] {4} (4.5,6);
			\draw[-{Stealth[length=6pt]}] (4.5,7) -> node[above] {2} (6,6);
			\draw[-{Stealth[length=6pt]}] (7,7) -> (6,6);
			\draw[-{Stealth[length=6pt]}] (7,7) -> (7,6);
			
			\draw[-{Stealth[length=6pt]}] (2,6) -> (2,5);
			\draw[-{Stealth[length=6pt]}] (2,6) -> (3,5);
			\draw[-{Stealth[length=6pt]}] (4.5,6) -> (4.5,5);
			\draw[-{Stealth[length=6pt]}] (4.5,6) -> (6,5.4);
			\draw[-{Stealth[length=6pt]}] (6,6) -> (6,5.4);
			
			\draw[-{Stealth[length=6pt]}] (6,5.4) -> (6,4.8);
			
			\node at (6.5,8) {(2)};
			\node at (2.5,8) {$k=3$};
			\node at (2.5,7.5) {$D=28$};
			
			\path[draw] (7.5,8) -> (7.5,4.8);
		\end{tikzpicture}
		\begin{tikzpicture}[scale=0.8,every node/.style={scale=0.7}]
			\draw[-{Stealth[length=6pt]}] (5,8) -> (7,7);
			\draw[-{Stealth[length=6pt]}] (5,8) -> node[right] {4} (6,6);
			\draw[-{Stealth[length=6pt]}] (5,8) -> node[right] {2} (2,6);
			\draw[-{Stealth[length=6pt]}] (5,8) -> node[right] {4} (4.5,6);
			\draw[-{Stealth[length=6pt]}] (5,8) -> node[left] {12} (3.75,5.4);
			
			\draw[-{Stealth[length=6pt]}] (7,7) -> (6,6);
			\draw[-{Stealth[length=6pt]}] (7,7) -> (7,6);
			
			\draw[-{Stealth[length=6pt]}] (2,6) -> (2,5);
			\draw[-{Stealth[length=6pt]}] (2,6) -> (3,5);
			\draw[-{Stealth[length=6pt]}] (2,6) -> node[above] {2} (3.75,5.4);
			\draw[-{Stealth[length=6pt]}] (4.5,6) -> (4.5,5);
			\draw[-{Stealth[length=6pt]}] (4.5,6) -> (6,5.4);
			\draw[-{Stealth[length=6pt]}] (6,6) -> (6,5.4);
			
			\draw[-{Stealth[length=6pt]}] (3.75,5.4) -> (3.75,4.8);
			\draw[-{Stealth[length=6pt]}] (6,5.4) -> (6,4.8);
			
			\node at (6.5,8) {(3)};
			\node at (5.8,6) {$w$};
			\node at (2.5,8) {$k=3$};
			\node at (2.5,7.5) {$D=31$};
		\end{tikzpicture}
		\caption{In this figure, an example for the usage of $\leave$ and $\take$ is given.
			A hypothetical instance \Instance is given in (1).
			Here, the value of $\Dbar$ is $3$.
			In (2), the instance~$\leave(\Instance,v)$, and in (3) the instance $\take(\Instance,v)$ is depicted.
			Unlabeled edges have a weight of~1.
			Observe in (3), the weight of the edge $\rho w$ is 4, as $w$ has two edges from ancestors of~$v$ in \Instance which have a weight of 2 each.
			The weight of $\rho u$ is $12$, as in \Instance the edges of~$E^{(\uparrow vu)}$ have a combined weight of 9.}
		\label{fig:Net-leave+take}
	\end{figure}
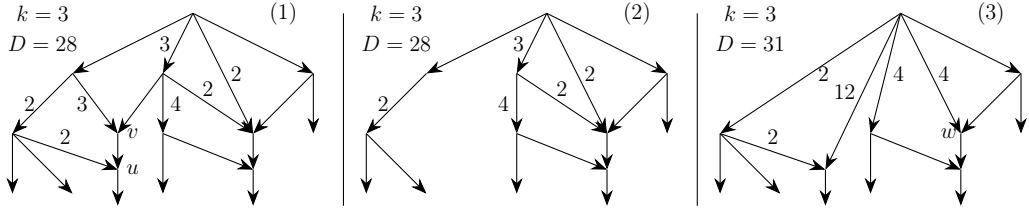%

	Now, we are at the position to define the branching algorithm. Let $\Instance=(\Net,k,D)$ be an instance of \MAPPD.
	If $\Net$ is a phylogenetic tree, solve the instance \Instance with Faith's Algorithm~\cite{steel,FAITH1992}.
	Otherwise, let $v$ be a reticulation of $\Net$.
	Then, return \yes if $\take(\Instance,v)$ or $\leave(\Instance,v)$ is a \yes-instance of \MAPPD and \no, otherwise.

	\proofpara{Correctness}
	The correctness of the base case is given by the correctness of Faith's Algorithm.
	We show that if \Net contains a reticulation $v$, then \Instance is a \yes-instance of \MAPPD if and only if $\take(\Instance,v)$ or $\leave(\Instance,v)$ is a \yes-instance of \MAPPD.

	%irst, let \Instance be a \yes-instance of \MAPPD and let $Y\subseteq X$ be a solution.
	%Thus, $|Y|\le k$ and $\apPD(Y)\ge D$.
	Consider any set of taxa $Y \subseteq X$.
	Firstly, we claim that if $Y \cap \off(e) = \emptyset$, then~$\apPDsub{\Net'}(Y) = \apPD(Y)$, where $\Net'$ is the network in $\leave(\Instance,v)$.
	%
	%If $Y$ and $\off(v)$ are disjoint, then $Y$ is a solution for $\leave(\Instance,v)=(\Net',k,D)$:
	%
	Indeed,~$\Net'$ contains all the vertices and edges of \Net that have an offspring outside of $\off(v)$. Therefore, $\apPDsub{\Net'}(Y) = \apPD(Y)$. % \ge D$ and $|Y|\le k$.
	%
	%
	%Otherwise, let the intersection $Y'$ of $Y$ and $\off(v)$ be non-empty. 
	Secondly, we claim that if $Y$ contains a vertex of~$\off(v)$, then $\apPDsub{\Net'}(Y) = \apPD(Y) + \Dbar$, where $\Net'$ is the network in~$\take(\Instance,v)$.
	%
	% We show that $Y$ is a solution for $\take(\Instance,v)$.
	%
	Recall that each edge $e = u_1 u_2$ with $u_1 \ne \rho$ of $E(\Net')$ is also an edge of $\Net$ and $\w'(e) = \w(e)$.
	Further, for each edge $e = \rho u_2$ with $u_2\ne u$ of~$E(\Net')$ there are edges $e_1 = u_{i_1} u_2, \dots, e_t = u_{i_t} u_2$ of $E(\Net)$ with $\w'(e) = \sum_{i=1}^t \w(e_i)$.
	Now, let~$Q = Q_1 \cup Q_2 \cup \{ \rho u \}$ be the edges of $\Net'$ that have at least one offspring in~$Y$, of which edges in $Q_1$ have both endpoints in $V(\Net')\setminus\{ \rho \}$, and $Q_2$ are outgoing edges of~$\rho$.
	Further, let $P = P_1 \cup P_2 \cup E^{(\uparrow vu)}$ be the edges of $\Net$ that have at least one offspring in $Y$, of which edges in $P_1$ have both endpoints in $V(\Net')$, and $P_2$ are edges with one endpoint in $\anc(v)\setminus \{v\}$ and one endpoint in $V(\Net')\setminus\{ \rho \}$.
	Observe that since any vertex in $V(\Net')$ has the same offspring in $\Net$ as in $\Net'$, we know $Q_1 = P_1$, and~$\w'_\Sigma(Q_1) = \w_\Sigma(P_1)$. 
	Further, $\w'_\Sigma(Q_2) = \w_\Sigma(P_2)$ as for each $u_2 \in V(\Net')\setminus\{ \rho \}$, the total weight of edges~$u_1 u_2$ with $u_1 \in \anc(v)\setminus \{v\}$ in $\Net$ is equal to the weight of the edge $\rho u_2$ in $\Net'$.
	It follows that
	\begin{eqnarray*}
		\apPDsub{\Net'}(Y) 
		 & = & \w'_\Sigma(Q_1) + \w'_\sigma(Q_2) + \w'( \rho u )\\
		 & = & \w_\Sigma(P_1) + \w_\Sigma(P_2) + \w_\Sigma(E^{(\uparrow vu)}) + \Dbar\\
		 & = & \apPD(Y) + \Dbar.
	\end{eqnarray*}
	
	By the above it follows that if $Y$ is a solution for \Instance, then either $Y$ is a solution for $\leave(\Instance,v)$ or $Y$ is a solution for $\take(\Instance,v)$.
	Conversely, if~$Y$ is a solution for~$\leave(\Instance,v)$ then $Y \cap \off(e) = \emptyset$ and thus $ \apPD(Y) = \apPDsub{\Net'}(Y) \geq D$, so~$Y$ is also a solution for \Instance.
	Finally, if $Y$ is a solution for $\take(\Instance,v)$ then $Y \cap \off(e) \neq \emptyset$, as otherwise
	\begin{eqnarray*}
		\apPDsub{\Net'}(Y) 
		& \le & \apPDsub{\Net'}(X) - \w'(\rho y)\\
		& = & D + 2\Dbar - (\w_\Sigma(E^{(\uparrow vu)}) + \Dbar) \le D + \Dbar - 1 < D'.
	\end{eqnarray*}
	Then, $\apPDsub{\Net'}(Y) = \apPD(Y) + \Dbar$.
	It follows that $\apPD(Y) \geq D' - \Dbar = D$ and $Y$ is also a solution for \Instance.

	%Secondly, let $\leave(\Instance,v)$ be a \yes-instance of \MAPPD.
	%
	%Like before, we can directly follow that the solution $Y$ of $\leave(\Instance,v)$ is also a solution for \Instance.
	%
	%
	%Finally, let $\take(\Instance,v)=(\Net',k,D'=D+\Dbar)$ be a \yes-instance of \MAPPD and $Y$ an according solution.
	%
	%Let $e=(r,u)$ of $E(\Net')$ be the edge from the root $r$ of $\Net'$ to $u$, the child of $v$ in \Net.
	%The weight of all edges in $\Net'$ is $\w'(E') = \w'(e) + \w'(E' \setminus \{e\}) = \w(E^{(\uparrow vu)}) + \Dbar + \w'(E' \setminus \{e\}) = \w(E(\Net)) + \Dbar = D + 2\Dbar$.
	%
	%Thus, $\apPDsub{\Net'}(X \setminus \off(e)) \le \apPDsub{\Net'}(X) - \w(e) = D + 2\Dbar - (\w(E^{(\uparrow vu)}) + \Dbar) \le D + \Dbar -1 < D'$.
	%Consequently, the solution $Y$ contains a taxon $x\in \off(e)$.
	%
	%Then, the phylogenetic diversity of $Y$ in \Net is $\apPD(Y) = \apPDsub{\Net'}(Y) + \w(E^{(\uparrow vu)}) - \w'(e) \ge \apPDsub{\Net'}(Y) - \Dbar \ge D + \Dbar - \Dbar = D$.
	%
	%Hence, $Y$ is a solution of the instance \Instance of \MAPPD.

	\proofpara{Running Time}
	Let \Instance be an instance of \MAPPD that contains a reticulation $v$.
	The number of reticulations in \Instance is greater than the number of reticulations in~$\take(\Instance,v)$ and $\leave(\Instance,v)$, because at least the reticulation $v$ is removed and no new reticulations are added.
	Therefore, the search tree contains $\Oh(2^{\ret})$ nodes.
	It can be checked in $\Oh(m)$ time, if $\Net$ contains a reticulation.
	Faith's Algorithm takes~$\Oh(k \cdot m)$ time \cite{steel}.
	
	The sets $\off(v)$ and $\anc(v)$ for a vertex $v$, and $T_Y$ for a set $Y$ can be computed in~$\Oh(m)$ time.
	Once $\anc(v)$ is computed, we can iterate over $E$ to find the edges that are outgoing from $\anc(v)$ and in $\Oh(m)$ time we can compute the value for an edge $\rho w$ in $\Net'$, which is also the time needed to compute $\w(\rho u)$ which needs \Dbar and the weight of $E^{(\uparrow vu)}$.
	Therefore, the instances $\take(\Instance,v)$ and $\leave(\Instance,v)$ can be computed in~$\Oh(m)$~time.
	
	Thus, a solution for \MAPPD can be computed in $\Oh(2^{\ret} \cdot k \cdot m)$ time.
\end{proof}

Bordewich et al. showed that \MAPPD can be solved in polynomial time on level-1-networks \cite{bordewichNetworks}.
We extend this result by showing that \MAPPD is fixed-parameter tractable with respect to treewidth.

\begin{theorem}
	\label{thm:Net-tw}
	\MAPPD can be solved in $\Oh(9^\twN \cdot \twN \cdot k^2 \cdot m)$ time.
\end{theorem}
We first provide a sketch of the main ideas, here.
% and here we want to give a scratch of the algorithm instead.
This algorithm has many similarities with Theorem~\ref{thm:PDD-tw}, however, here we have a tree-decomposition over a network and not a food-web.

We aim to find a set of edges $E'$ that have an overall weight of at least $D$ and that are incident with at most $k$ leaves.
Further, for each edge $e = uv \in E'$ we require that either~$v$ is a leaf or there is an edge $vw \in E'$.
In this dynamic program algorithm over a nice tree decomposition, we index feasible partial solutions by a~3-coloring of the vertices.
At a given node of the tree decomposition, a vertex~$v$ is colored:
\begin{itemize}
	\item red, if it is still mandatory that we select an outgoing edge of $v$ (because we have selected an incoming edge of $v$),
	\item green, if we can select incoming edges of~$v$ and do not need to select an outgoing edge of~$v$ (because $v$ is a leaf or we have already selected an outgoing edge of~$v$),
	\item black, if we have to not yet selected an edge incident with~$v$ (such that only the selection of an incoming edge of~$v$ makes the selection of an outgoing edge of $v$ necessary).
\end{itemize}
We introduce each leaf as a green vertex and the other vertices as black vertices.
In order to consider only feasible solutions, a vertex must be green or black when it is forgotten.
The most important step of the algorithm is in the introduction of an edge, where colors may be adjusted depending on whether or not the new edge is included in $E'$.

\begin{proof}
	We define a dynamic programming algorithm over a nice tree-decomposition~$T$ of the underlying undirected graph of an $X$-network $\Net=(V,E,\w)$ in which vertices are introduced without incident edges and all edges are induced exactly once.
	We note that at every join bag the graph~$\Net_t[Q_t]$ is edgeless, if we only introduce edges right before we forget one of the incident vertices.
	Here, $\Net_t$ is the graph with vertices~$V_t$ and edges~$E_t$ that are introduced at~$t$ or at decedants of~$t$ and are not forgotten.
	
	\proofpara{Definition of the Table}
	Let $\Instance=(\Net,k,D)$ be an instance of \MAPPD and let~$T$ be a nice tree decomposition of  $\Net$.
	We index solutions by a partition $R\cup G\cup B$ of $Q_t$, and a non-negative integer $s$.
	%
	%Vertices in $R,~G$ and $B$ are called \textit{red}, \textit{green} and \textit{black}, respectively.
	%\todosi{I suggest we use \textit{red}, \textit{green} and \textit{black} to talk about colors with respect to some $F$.
		%i.e.
		%then replace conditions F2-F4 with ``the vertices of $R/G/B$ are red/gree/black with respect to $F$''.
		%}
	%
	For a set of edges~$F \subseteq E_t$, we call a vertex $u \in V_t$ \emph{\textit{green} with respect to $F$} if $u$ is a leaf or has an outgoing edge in $F$. We call $u$ \emph{\textit{red} with respect to $F$} if $u$ is not a leaf and has an incoming but no outgoing edge in $F$. Finally, we call $u$  \emph{\textit{black} with respect to $F$} if $u$ is not a leaf and has no incident edges in $F$.
	For a node~$t\in T$, a partition $R\cup G\cup B$ of $Q_t$ and an integer $s$, we call a set of edges~$F$ over $V_t$ \textit{feasible for~$t$, $R$, $G$, $B$, and~$s$}, if all the following conditions hold:
	\begin{enumerate}
		\item[(F1)]\label{it:safeness} If $uv$ is an edge in $F$ and $v\notin Q_t$, then $v$ is a leaf of $\Net$ or $v$ has an outgoing edge in $F$.
		\item[(F2)]\label{it:red} The vertices $R\subseteq Q_t$ are red with respect to $F$.
		%are not leaves in \Net and have an incoming but no outgoing edge in $F$.
		\item[(F3)]\label{it:green} The vertices $G\subseteq Q_t$ are green with respect to $F$. 
		%leaves in \Net or have an outgoing edge in $F$.
		\item[(F4)]\label{it:black} The vertices $B\subseteq Q_t$ are black with respect to $F$.
		%not leaves in \Net and are not incident with an edge of $F$.
		\item[(F5)]\label{it:taxa} The number of leaves in $\Net$ with an incoming edge in $F$ is $s$.
	\end{enumerate}
	We define \Sstar to be the family of all sets~$F$ that are feasible for~$t$, $R$, $G$, $B$, and~$s$.
	%
	%Because $B=V\setminus (R\cup B)$, we will no longer mention $B$.\todos{I think we should}

	We define a dynamic programming algorithm over a nice tree decomposition~$T$.
	In a table entry $\DP[t,A,R,G,B,s]$, we store the greatest weight $\w_\Sigma(F)$ of a set~$F$ that is feasible for~$t$, $R$, $G$, $B$, and~$s$.
	If there is no feasible $F$, then we store a big negative value.
	Let~$r$ be the root of the nice tree-decomposition $T$.
	Then, $\DP[r,\emptyset,\emptyset,\emptyset,k]$ stores the greatest phylogenetic diversity that can be preserved with a budget of $k$.
	Hence, we can return \yes if $\DP[r,\emptyset,\emptyset,\emptyset,k]\ge D$ and \no, otherwise.

	Now, we have everything we need to define the dynamic programming algorithm.
	In the calculations that follow, any time a value $\DP[t,R,G,B,s]$ is called for which~$\DP[t,R,G,B,s]$ is not defined (in particular, if $s < 0$), we take $\DP[t,R,G,B,s]$ to be a large negative value. For our purposes a value of $-m\cdot \max_\w-1$ suffices, as even if every edge would be chosen, the entry at the root of the tree decomposition is still negative.

	\proofpara{Leaf Node}
	For a leaf~$t$ of~$T$ the bags~$Q_t$ and $V_t$ are empty. So if $s = 0$, we trivially store
	\begin{equation}
		\label{tw:leaf}
		\DP[t,\emptyset,\emptyset,\emptyset,s] = 0.
	\end{equation}
	Otherwise, we store $\DP[t,R,G,B,s] = -m\cdot \max_\w-1$.

	\proofpara{Introduce Vertex Node}
	Suppose now that~$t$ is an \textit{introduce vertex node}, i.e.~$t$ has a single child~$t'$,~$Q_t = Q_{t'} \cup \{v\}$, and $v$ is an isolated vertex in $\Net_t$.
	If either~$v\in B$, or~$v\in G$ and~$v$ is a leaf, we store
	\begin{equation}
		\label{tw:insertvertex}
		\DP[t,R,G,B,s] = \DP[t',R,G\setminus\{v\},B\setminus\{v\},s].
	\end{equation}
	Otherwise, we store $\DP[t,R,G,B,s] = -m\cdot \max_\w-1$.

	\proofpara{Introduce Edge Node}
	Suppose now that~$t$ is an \textit{introduce edge node}, i.e.~$t$ has a single child~$t'$,~$Q_t = Q_{t'}$, and $e = vw$ is introduced at $t'$.
	The algorithm must decide whether $e$ is affected by the solution, or not.
	If $v\notin G$ or $w\in B$, edge~$e$ can not be an affected edge and so we store $\DP[t,R,G,B,s] = \DP[t',R,G,B,s]$.
	Otherwise, we store
	\begin{equation}
		\label{tw:insertedge}
		\DP[t,R,G,B,s] = \max\{
		% e not selected
		\DP[t',R,G,B,s];
		% e is selected
		\max_{R',G',B'}\DP[t',R',G',B',s'] + \w(e)
		\},
	\end{equation}
	where the second maximum is over all possible partitions $R'\cup G'\cup B'$ of $Q_t$, such that $S\setminus\{v,w\} = S'\setminus\{v,w\}$ for all $S \in \{R,G,B\}$ (i.e. the two partitions agree on $Q_T\setminus\{v,w\}$), and such that if $w \in G'$ then $w \in G$; and $w \in R$, otherwise.
	Here,~$s' = s-1$ if $w$ is a leaf, and~$s' = s$ if not.

	\proofpara{Forget Node}
	Suppose now that~$t$ is \textit{forget node}, i.e.~$t$ has a single child~$t'$\lb and~$Q_t = Q_{t'} \setminus \{v\}$.
	We store
	\begin{equation}
		\label{tw:forget}
		\DP[t,R,G,B,s] =  \max \{
		% we never chose an incidnet edge
		\DP[t',R,G,B\cup\{v\},s];
		% we did
		\DP[t',R,G\cup\{v\},B,s]\}.
	\end{equation}

	\proofpara{Join Node}
	Suppose now that~$t\in T$ is an \textit{join node}, i.e.~$t$ in $T$ has two children~$t_1$ and~$t_2$ with~$Q_t = Q_{t_1} = Q_{t_2}$.
	We call two partitions $R_1\cup G_1\cup B_1$ and $R_2\cup G_2 \cup B_2$ of $Q_t$ \textit{qualified} for $R\cup G\cup B$ if $R = (R_1\cup R_2) \setminus (G_1 \cup G_2)$ and $G = G_1\cup G_2$ (and consequently $B = B_1 \cap B_2$). See~Figure~\ref{fig:Net-joinColorings}.
	We store
	\begin{equation}
		\label{tw:join}
		\DP[t,R,G,B,s] = \max_{(\pi_1, \pi_2) \in{\cal Q},s'} ~ \DP[t_1,R_1,G_1,B_1,s'] + \DP[t_2,R_2,G_2,B_2,s-s'],
	\end{equation}
	%maximum is taken over all partitions $R_1\cup G_1\cup B_1$ and $R_2\cup G_2 \cup B_2$  that are qualified for $R,G,B$ and all $s' \in [s]$.
	where $s' \in [s]_0$, and ${\cal Q}$ is the set of pairs of partitions $\pi_1 = R_1\cup G_1\cup B_1$\lb and $\pi_2 = R_2\cup G_2 \cup B_2$ that are qualified for $R,G,B$.
	
	\newcommand{\darkgreen}{green!50!black}
	\begin{figure}
		\begin{center}
			\begin{tabular}{c|c|c|c|}\cline{2-4}
				& $B_1$ & \color{red}{$R_1$} & \color{\darkgreen}{$G_1$}\\\hline
				\multicolumn{1}{|l|}{$B_2$}  & $B$ & \color{red}{$R$} & \color{\darkgreen}{$G$} \\\hline
				\multicolumn{1}{|l|}{\color{red}{$R_2$}}  & \color{red}{$R$} & \color{red}{$R$} & \color{\darkgreen}{$G$} \\\hline
				\multicolumn{1}{|l|}{\color{\darkgreen}{$G_2$}}  & \color{\darkgreen}{$G$} & \color{\darkgreen}{$G$} & \color{\darkgreen}{$G$} \\\hline
			\end{tabular} 
			
			\caption{This table shows the relationship between the three partitions $R_1\cup G_1\cup B_1$, $R_2\cup G_2 \cup B_2$ and $R\cup G\cup B$ in the case of a join node, when $R_1\cup G_1\cup B_1$ and $R_2\cup G_2 \cup B_2$ are qualified for $R\cup G\cup B$. The table shows which of the sets $R$, $G$, or~$B$ an element~$v \in Q_t$ will be in, depending on its membership in $R_1,G_1,B_1,R_2,G_2$, and $B_2$. For example if~$v \in R_1$ and $v \in B_2$, then $v \in R$.}
			\label{fig:Net-joinColorings}
		\end{center}
	\end{figure}

	\proofpara{Correctness}
	Let $t$ be a node of the nice tree-decomposition $T$, $s$ an integer,\lb and~$R\cup G\cup B$ a partition of $Q_t$.
	We show that the value of $\DP[t,R,G,B,s]$ is correct, for each type of node individually.

	If $t$ is a \textit{leaf node}, then $Q_t$ and $V_t$ are empty.
	Consequently, $R=G=B=\emptyset$ and any feasible set $F$ is also empty.
	Therefore, we also conclude with (F5) that~$s=0$\lb and~$\w_\Sigma(F)=0$ for any $F \in \Sstar$.
	Hence, we store the correct value in Recurrence~(\ref{tw:leaf}).

	Let $t$ be an \textit{introduce vertex node} with child $t'$ and $Q_t=Q_{t'}\cup\{v\}$.
	Then, the graph $\Net_t$ is the graph $\Net_{t'}$ with an additional isolated vertex $v$.
	Thus, $v$ is black with respect to any $F \subset E_t$, unless $v$ is a leaf in which case $v$ is green.
	Thus, there is no feasible set $F$ for $\Sstar$ unless either $v \in B$ or $v \in G$ and $v$ is a leaf.
	Assuming this is the case, any set $F \subseteq E_t$ is feasible for~$t$, $R$, $G$, $B$, and~$s$ if and only if it is feasible for $t'$, $R$, $G\setminus\{v\}$, $B\setminus\{v\}$, and~$s$.
	Hence, we store the right value.

	Let $t$ be an \textit{introduce edge node} with child $t'$ and $Q_t=Q_{t'}$ and $\Net_t-e=\Net_{t'}$. Define $e := vw$.
	Clearly, every set $F\subseteq E_t\setminus\{e\}$  is feasible for~$t$, $R$, $G$, $B$, and~$s$ if and only if $F$ is also feasible for~$t'$, $R$, $G$, $B$, and~$s$, because then $\Net_t[F]$ is isomorphic to $\Net_{t'}[F]$.
	%
	% Assume otherwise that $e=(v,w)\in A$.
	Now, consider a set~$F\subseteq E_t$ with $vw \in F$.
	Note that $v$ is green, and $w$ cannot be black, with respect to any such set $F$.
	Therefore, such an $F$ only exists if~$v \in G$ and $w \notin B$.
	Now, let $F'$ be~$F\setminus\{vw\}$. 
	Every vertices $u$, that is not $v$ or $w$, is red/green/black with respect to $F$ if and only if~$u$ is red/green/black with respect to $F'$.
	Observe that if $w$ is green with respect to $F$ if and only if it is green with respect to~$F'$ (as it has the same outgoing edges in $F$ and $F'$).
	If $w$ is red or black with respect to~$F'$, then it is red with respect to~$F$ (as it has an incoming edge and no outgoing edges).
	On the other hand, $v$ could be any color with respect to $F'$, but is necessarily green with resepct to $F$.
	Likewise, we can show the correctness of $s'$.
	Thus, the correct value is stored.

	Let $t$ be a \textit{forget node} with child $t'$ and $Q_t=Q_{t'}\setminus\{v\}$.
	We show that a set $F$ is feasible for~$t$, $R$, $G$, $B$, and~$s$ if and only if $F$ is also feasible for~$t'$, $R$, $G$, $B\cup\{v\}$, and~$s$ or for~$t'$, $R$, $G \cup \{v\}$, $B$, and~$s$.
	It then follows that $\DP[t,R,G,B,s]$ stores the correct value.
	To this end, let $F$ be a set of edges that is feasible for~$t$, $R$, $G$, $B$, and~$s$.
	Since~$v \notin Q_t$ and condition (F1) satisfies, $v$ either has an outgoing edge in $F$, is a leaf, or does not have any incoming edge in $F$.
	It follows that $v$ is green or black with respect to $F$.
	From this it is straightforward to confirm that $F$ is feasible\lb for~$t'$, $R$, $G$, $B\cup\{v\}$, and~$s$, or for~$t'$, $R$, $G \cup \{v\}$, $B$, and~$s$.
	Conversely, if $F$ is feasible for $t'$, $R$, $G$, $B\cup\{v\}$, and~$s$, or for $t'$, $R$, $G \cup \{v\}$, $B$, and~$s$, then~$F$ also satisfies condition (F1) for~$t'$, $R$, $G$, $B\cup\{v\}$, and~$s$ (in particular, the condition is satisfied for~$v$ as~$v$ is green or black with respect to $F$).
	It is straightforward to confirm that~$F$ the remaining conditions to be feasible with respect to $\Sstar$.

	Let $t$ be a \textit{join node} with children $t_1$ and $t_2$.
	We show that there is an $F\in \Sstar$ if and only if
	there are partitions $R_1 \cup G_1 \cup B_1$ and $R_2 \cup G_2 \cup B_2$, qualified for~$R$,~$G$, and~$B$
	and~$s'\le s$ such that
	there are $F_1\in \SstarPrime{t_1,R_1,G_1,B_1,s'}$ and $F_2\in \SstarPrime{t_2,R_2,G_2,B_2,s-s'}$
	with~$F_1\cup F_2 = F$.

	First, let $F$ be feasible for~$t$, $R$, $G$, $B$, and~$s$.
	Let $F_i$ be the subset of edges of~$F$ that occur in~$\Net_{t_i}$ for $i\in\{1,2\}$.
	Because every edge is introduced exactly once, the sets~$F_1$ and~$F_2$ are disjoint.
	We define $s'$ to be the number of edges of $F_1$ that are incident with a leaf.
	Further, we define $R_i$ to be the set of vertices in $Q_t$ that are not leaves in \Net and have an incoming but no outgoing edge in $F_i$ for $i\in \{1,2\}$.
	Similarly, we define $G_i$ to be the set of vertices that are leaves or have an outgoing edge in~$F_i$ for~$i\in \{1,2\}$.
	We show that $F_1\in \SstarPrime{t_1,R_1,G_1,B_1,s'}$ and $F_2\in \SstarPrime{t_2,R_2,G_2,B_2,s-s'}$.
	Let $e = uv$ be an edge in $F_i$ and let $v$ be an internal-vertex and $v\not\in Q_{t_i}$ for~$i\in \{1,2\}$.
	Because $Q_t=Q_{t_i}$, we conclude that $v\not\in Q_{t}$.
	Thus, with (F1) we conclude that~$v$ has an outgoing edge~$e^*$ in~$F$.
	Because $v$ is not in $Q_t$ and $v$ is in $V_{t_i}$, the vertex $v$ is not in~$V_{t_{3-i}}$.
	Thus, the edge $e^*$ is also in $F_i$ and so $F_1$ and $F_2$ satisfy (F1).
	By definition, $F_1$ and~$F_2$ satisfy conditions (F2) to (F4).
	Also by definition, $F_1$ satisfies~(F5).
	As $F_2$ contains the edges of $F$ that are not in~$F_1$, there are $s-s'$ edges in~$F_2$ that are incident with a leaf.
	Hence, $F_2$ satisfies (F5) and $F_1\in \SstarPrime{t_1,R_1,G_1,B_1,s'}$ and $F_2\in \SstarPrime{t_2,R_2,G_2,B_2,s-s'}$.
	It remains to show that~$R_1\cup G_1 \cup B_1$ and $R_2\cup G_2 \cup B_2$ are qualified for~$R$, $G$, and~$B$.
	If~$v$ is a leaf, we conclude that $v$ is in $G_1,G_2$ and $G$.
	Further, $v$ has an outgoing edge in $F$, if and only if $v$ has an outgoing edge in~$F_1$ or~$F_2$.
	We conclude that $G=G_1\cup G_2$.
	For each $v\in R$ there is an incoming edge~$e$ but no outgoing edge in $F$.
	Consequently, if~$v\in R$ then $v\not\in G_1\cup G_2$.
	Further, without loss of generality, $e\in F_1$.
	However, there is no outgoing edge of $v$ in $F_1$.
	Thus, $v\in R_1$ and so~$R\subseteq (R_1\cup R_2) \setminus (G_1 \cup G_2)$.
	If $v\in (R_1\cup R_2) \setminus (G_1 \cup G_2)$ then there is an incoming edge of $v$ in $F_1$ or $F_2$ and therefore in $F$, but no outgoing edges.
	Thus, $v$ is red and in~$R$.

	Secondly, assume that there are $R_1,G_1,R_2$, and $G_2$ qualified for~$R$ and $G$,\lb
	and~$s'\le s$ such that
	there are $F_1\in \SstarPrime{t_1,R_1,G_1,B_1,s'}$ and $F_2\in \SstarPrime{t_2,R_2,G_2,B_2,s-s'}$.
	We show that $F:=F_1\cup F_2$ is feasible for~$t$, $R$, $G$, $B$, and~$s$.
	Let $e = uv$ be an edge in $F$ with~$v$ is not a leaf and $v\not\in Q_t$. 
	Without loss of generality, $e\in F_1$.
	Because $Q_t=Q_{t_1}$, we conclude that $v\not\in Q_{t_1}$.
	Thus, $v$ has an outgoing edge in $F_1$ and therefore in~$F$, we know~$F$ satisfies condition~(F1).
	Let~$v$ be a vertex of~$Q_t$.
	If~$v\in R=(R_1\cup R_2) \setminus (G_1 \cup G_2)$, then there is an edge $e$ incoming at $v$ in $F_1$ or $F_2$, but~$F_1$ and~$F_2$ do not contain an outgoing edge of $v$.
	Consequently, $v$ has an incoming edge in~$F$ but no outgoing edges.
	Analogously, we can see that if $v\in G=G_1\cup G_2$ then $v$ has an outgoing edge in $F$ and if $v\in B=B_1\cap B_2$, then~$v$ is not incident with edges.
	Because that covers all options, $F$ satisfies conditions (F2) to (F4).
	Let~$\hat E_1$ and $\hat E_2$ be the edges of $F_1$ and $F_2$, respectively, that are incident with a leaf. Thus, $|\hat E_1|=s'$ and $|\hat E_2|=s-s'$.
	Because every edge is introduced exactly once, the sets~$F_1$ and $F_2$ and therefore $\hat E_1$ and $\hat E_2$ are disjoint.
	Consequently, $\hat E_1\cup \hat E_2$ contains the $s$ edges that are incoming at leaves.
	Thus, $F$ satisfies condition (F5).
	Hence, the value that is stored with Recurrence~(\ref{tw:join}) is correct.

	\proofpara{Running Time}
	The tree decomposition $T$ has $\Oh(m)$ nodes.

	For a leaf node, an introduce vertex nodes and a forget node, the value of each entry~$\DP[t,R,G,B,s]$ can clearly be computed in time linear in $\twN$.
	Checking the conditions in an introduce edge node requires checking the colors of $v$ and $w$.
	Further, the partition~$R'\cup G'\cup B'$ agrees on all vertices but $v$ and $w$.
	Therefore, we can also compute the values of entries in time linear in $\twN$ for an introduce edge node.
	Altogether, we can compute the value of all entries of nodes that are not join nodes in $\Oh(3^\twN \cdot \twN \cdot k \cdot m)$ time.

	For the computation in a join node $t$, we first store a big negative integer in each entry $\DP[t,R,G,B,s]$ for partitions $R\cup G\cup B$ of $Q_t$ and an integer $s\in [k]_0$.
	Then, iterate over partitions $R_1\cup G_1\cup B_1$ and~$R_2\cup G_2\cup B_2$ of $Q_t$.
	From the part-\lb itions~$R_1\cup G_1\cup B_1$ and~$R_2\cup G_2\cup B_2$, compute the implicitly defined part-\lb ition~$R\cup G\cup B$ of~$Q_t$ in~$\Oh(\twN)$~time.
	See Figure~\ref{fig:Net-joinColorings}.
	Iterate over $s\in [k]_0$ and $s'\in [s]_0$.
	If $\DP[t_1,R_1,G_1,B_1,s'] + \DP[t_2,R_2,G_2,B_2,s-s']$ exceeds the value of~$\DP[t,R,G,B,s]$, then replace it. Otherwise, let $\DP[t,R,G,B,s]$ stay unchanged. 
	After the iterations, we have computed the values of $\DP[t,R,G,B,s]$ for all partitions $R\cup G\cup B$ of $Q_t$ and integer $s\in [k]_0$.
	Therefore, $\Oh(9^{\twN} \cdot \twN \cdot k^2)$ time is required to compute all values of entries for each join node $t$.

	Hence, a solution for \MAPPD can be computed in $\Oh(9^\twN \cdot \twN \cdot k^2 \cdot m)$ time.
\end{proof}

\section{A Kernelization for Reticulation-Edges}
\label{sec:Net-kernel}
In Theorem~\ref{thm:Net-ret}, we presented a branching algorithm that proves that \MAPPD is \FPT when parameterized by the number of reticulations, $\vret$.
In this section, we show that \MAPPD admits a kernelization algorithm of polynomial size {with respect to~\eret.}
Recall that~\eret is the number of reticulation-edges which need to be removed such that~\Net is a tree.
Observe~$\eret \ge \vret$ and in binary networks they are equal.

{We first show how to bound the number of vertices and edges by a polynomial in~\eret, without giving any such bound on the weights of the edges. Afterwards, we will apply a result from \cite{etscheid,frank} to get an appropriate bound on the edge weights.}

\begin{theorem}
	\label{thm:Net-vkernel}
	Given an instance $\Instance = (\Net, k, D)$ of \MAPPD,
	an equivalent instance~$\Instance^* = (\Net^* = (V^*,E^*,{\w^*}), k^*, D^*)$ of \MAPPD
	with $|V^*|,|E^*| \in \Oh(\eret^2)$ and $k^* \in \Oh(\eret)$
	can be computed in~{$\Oh(m^2 \log^2 m \cdot \log \max_\w)$}~time.
\end{theorem}

Throughout this section, assume that $\Instance = (\Net = (V,E,\w), k, D)$ is an instance of \MAPPD with $\rho$ being the root of \Net.
Let $R \subseteq V$ be the set of reticulation vertices of \Net.
{We apply the following reduction rules exhaustively, and each rule is applied only if none of the previous rules apply.}
After any of the reduction rules let~$\Net' = (V,E',w')$ denote the new network.

\begin{rr}
	\label{rr:Net-trees}
	Let $v\in V$ be a vertex with children $x$ and~$y$ that are leaves.
	Assume~$\w(vx) \ge \w(vy)$.
	If $v\ne \rho$, then replace the edge~$vy$ with an edge~$\rho y$ of weight~$\w'(\rho y) = \w(vx)$.
\end{rr}
\begin{lemma}
	\label{lem:Net-trees}
	Reduction Rule~\ref{rr:Net-trees} is correct and can be applied
	in~$\Oh(n)$~time.
\end{lemma}
\begin{proof}
	\proofpara{Correctness}
	We first show that if \Instance is a \yes-instance of \MAPPD\lb then $\Instance' := (\Net',k,D)$ is a \yes-instance of \MAPPD.
	Afterward, we show the converse.
	If \Instance is a \yes-instance then let $S\subseteq X$ be a solution.
	If~$y\in S$ but~$x\not\in S$,\lb then define~$S' := {(S\cup \{x\})} \setminus \{y\}$.
	We conclude that~$S'$ is also a solution\lb because $\apPD(S') = \apPD(S) + \w(vx) - \w(vy) \ge \apPD(S)$.
	Therefore, we may assume without loss of generality that~$x \in S$ or~$y \notin S$.
	If~$x \notin S$ and~$y \notin S$, then observe that $\apPD(S) = \apPDsub{\Net'}(S)$ since all edges affected by~$S$ appear in both networks with the same weight. 
	Similarly, if $x \in S$ and $y \notin S$,\lb then $\apPD(S) = \apPDsub{\Net'}(S)$.
	Finally if $x \in S$ and $y \in S$, then any\lb edge $e \in E(\Net)\setminus \{vy\}$ is affected by~$S$ in~$\Net$ if and only if it is affected by~$S$ in~$\Net'$ (in particular, if $y$ is an offspring of $e$ in $\Net$ then $x$ is an offspripng of $e$ in $\Net'$).
	It follows that  $\apPD(S) = \apPDsub{\Net'}(S) - \w'(\rho y) + \w(vx) = \apPDsub{\Net'}(S)$.
	% 	Now observe $\apPD(S) = \apPDsub{\Net'}(S)$ if $x\in S$ or $y\not\in S$, proving that
	Thus, $S$ is a solution of $\Instance'$ in all cases.

	Conversely, observe that if an edge $e \in E(\Net')\setminus \{\rho y\}$ is affected by $S \subseteq X$ in~$\Net'$ then it is also affected by~$S$ in~$\Net$.
	Thus, $\apPDsub{\Net'}(S) \le \apPD(S)$ and consequently each solution for $\Instance'$ is also a solution for \Instance. Thus, if $\Instance'$ is a \yes-instance of \MAPPD, then so is \Instance.

	\proofpara{Running Time}
	In~$\Oh(n)$ time, we can determine whether there exists a vertex~$v$ which has two children in~$X$.
	If such a vertex~$v$ exists, we only need to compare~$\w(vx)$ and~$\w(vy)$ in~$\Oh(\log\max_{\w})$~time.
\end{proof}

Note that Reduction Rule~\ref{rr:Net-trees} might create degree-2-vertices; these are handled by the next reduction rule.
\begin{rr}
	\label{rr:Net-deg-2}
	Let $v\in V$ be a vertex of degree-2.
	Let $u$ be the parent and $w$ be the child of $v$.
	Remove~$v$, its incident edges and create an edge~$uw$ with a weight of~$\w'(uw) = \w(uv) + \w(vw)$.
\end{rr}
\begin{lemma}
	\label{lem:Net-deg-2}
	Reduction Rule~\ref{rr:Net-deg-2} is correct and can be applied
	in~$\Oh(n)$~time.
\end{lemma}
\begin{proof}
	\proofpara{Correctness}
	The correctness follows from the observation that for any~$S \subseteq X$ we have~$\apPD(S) = \apPDsub{\Net'}(S)$.
	
	\proofpara{Running Time}
	In~$\Oh(n)$ time, we can find out whether there exists a vertex with a degree of~2.
	If such a vertex $v$ exists, it takes constant~time to remove $v$ and its incident edges and create the new edge $uw$ with the appropriate weight.
\end{proof}

To evaluate the size of the network after applying reduction rules, we adapt the language of \emph{network generators}.
We refer the reader to~\cite{GambetteGenerators2009, GeneratorsDef2009} for an in-depth study.

\begin{definition}
	\label{def:cor-side}
	~
	\begin{propEnum}
		\item A vertex $v \in V\setminus R$ is a \emph{core-vertex} if $v$ has two children, $u_1$ and $u_2$,\lb where~$u_1 \ne u_2$, and $\desc(u_i) \cap R \ne \emptyset$ for each $i \in \{1, 2\}$.

		\item Let $Q$ be the set of core-vertices of \Net.
		
		\item A \emph{side-path} in $\Net$ is a path from $u$ to $w$ for two vertices $u,w \in R \cup Q$ with no internal vertices in $R\cup Q$.
		The internal vertices of a side-path are \emph{side-vertices}.

		\item Let $Z$ be the set of side-vertices of \Net.
	\end{propEnum}
\end{definition}

Note that after applying Reticulation Rules~\ref{rr:Net-trees} and~\ref{rr:Net-deg-2}, every non-leaf vertex is either a reticulation or has at least one reticulation as a descendant.
Therefore, every non-leaf vertex is in $Q$, in $R$, or in $Z$.
Every side-vertex has exactly one child which is a leaf.
The following result is similar to one in~\cite{GambetteGenerators2009}. For completeness we prove it here.

\begin{observation}
	\label{obs:Net-core-retic-bound}
	Every network~\Net has $\Oh(\eret)$ side-paths.
	Further, it is satisfied that~$|R|+|Q| \in \Oh(\eret)$.
\end{observation}
\begin{proof}
	Observe that $|R| = \vret \leq \eret$. 
	As every side-path in $\Net$ ends at a core vertex or reticulation, we have that the number of side-paths is at most 
	\begin{eqnarray}
		\sum_{v \in Q\cup R} \deg^-(v)
		&=& \sum_{r \in  R} (\deg^-(r) - 1) + |Q| + |R|\\
		&=& \eret +  |Q| + |R| \leq 2\cdot \eret + |Q|.
	\end{eqnarray}

	Conversely, each core-vertex has at least two side-paths leaving it, from\lb which it follows that the number of side-paths is at least $2\cdot |Q|$.
	We conclude\lb that~$2\cdot |Q| \leq  2\cdot \eret + |Q|$ and so $|Q| \leq 2\cdot \eret$.
	This implies that the number of side-paths is at most $4\cdot \eret$.
\end{proof}

\begin{rr}
	\label{rr:Net-a>a+b}
	Let~$v$ and~$w$ be side-vertices and let~$v$ be the parent of~$w$.
	Further, let~$x_v$ and~$x_w$ be leaves which are children of $v$ and~$w$, respectively.\lb
	If~$\w(vx_v) \le \w(wx_w) + \w(vw)$ then replace the edge~$vx_v$ with an edge~$\rho x_v$ with a weight of~$\w'(\rho x_v) = \w(vx_v)$.
\end{rr}
\begin{lemma}
	\label{lem:Net-a>a+b}
	Reduction Rule~\ref{rr:Net-a>a+b} is correct and can be applied
	in~$\Oh(n)$~time.
\end{lemma}
\begin{proof}
	
	\proofpara{Correctness}
	The proof follows similar lines as the proof of Lemma~\ref{lem:Net-trees}.\lb
	Let $S\subseteq X$ be a solution of $\Instance$.
	Assume that $x_v \in S$ but $S \cap \off(w) = \emptyset$.
	Then,\lb let  $S' := {(S\cup \{x_w\})} \setminus \{x_v\}$.
	Observe that $wx_w$ and $vw$ are affected by~$S'$\lb but not by~$S$, while the only edge affected by~$S$ and not by~$S'$ is~$vx_v$.\lb
	Thus, $\apPD(S') = \apPD(S) - \w(vx_v) + \w(wx_w) + \w(vw)$, which by the  condition of the reduction rule is at least  $\apPD(S)$.
	Therefore, also~$S'$ is a solution of~\Instance with $x_v\not\in S$.
	Thus, we may now assume that $S$ contains a taxon from~$\off(w)$ or~$x_v\not\in S$. 
	In either case, we can observe that $\apPD(S) = \apPDsub{\Net'}(S)$.

	Conversely, observe that any edge $e \in E(\Net')\setminus\{\rho x_v\}$ affected by $S$ in $\Net'$ is also affected by $S$ in $\Net$. Thus $\apPDsub{\Net'}(S) \le \apPD(S)$ for any $S \subseteq X$, and so each solution for $\Instance'$ is also a solution for \Instance.
	% 	Thus, if $\Instance'$ is a \yes-instance of \MAPPD, then so is \Instance.
	
	\proofpara{Running Time}
	In~$\Oh(n)$ time, we can find out whether there exists a vertex satisfying the conditions of~$v$ in Reduction Rule~\ref{rr:Net-a>a+b}, and if so edit the network accordingly.
\end{proof}

\begin{rr}
	\label{rr:Net-strings}
	% 	Apply Reduction Rules \ref{rr:Net-trees}, \ref{rr:Net-deg-2}, and \ref{rr:Net-a>a+b} exhaustively.
	
	Let~$u$ and~$w$ be core-vertices or reticulations with $u \neq \rho$ and
	for some~$\ell > 1$
	let~$v_1,\dots,v_\ell$ be side-vertices such that $uv_1, v_\ell w, v_iv_{i+1} \in E$ for each $i\in [\ell-1]$.
	Remove the edge~$v_1 v_2$ and add edges~$\rho v_2$ and $v_1 w$ with a weight of~$\w'(\rho v_2) = \w(v_1 v_2)$, and~$\w'(v_1 w) = 0$.
\end{rr}
An application of this reduction rule is depicted in Figure~\ref{fig:Net-strings}.
Here, we slightly bend our own definitions in which we required that $\w(e) > 0$ for each edge~$e$ and that each vertex either has an in-degree of~1 or an out-degree of~1.
% In an instance in which we can not apply any of the reduction rules any more,
We note that after applying the reduction rules exhaustively, we can adjust the instance to ensure these requirements are met.
First, we can replace each vertex~$v$ which has an in-degree and an out-degree of larger than~1 with~$v_{\text{in}}$ receiving all the incoming edges of~$v$ and~$v_{\text{out}}$ receiving $v$ in all the outgoing edges and adding an edge~$v_{\text{in}} v_{\text{out}}$ with a weight of~0.
Furthermore, we can multiply each weight and $D$ by {$m+1$}.
Then, we can set the edges with a weight of~0 to a weight of~1.

% By this, we would again have an instance as we defined it.
% \todoj{Think about how to rephrase}

\begin{figure}[t]
	\centering
	\begin{tikzpicture}[scale=0.7,every node/.style={scale=0.7}]
		\node (u) {$u$}
		child {node {} edge from parent[draw=none]}
		child {node {$v_1$} edge from parent[-{Stealth[length=6pt]}]
			child {node {$x_1$} edge from parent[-{Stealth[length=6pt]}]}
			child {node {$v_2$} edge from parent[-{Stealth[length=6pt]}]
				child {node {$x_2$} edge from parent[-{Stealth[length=6pt]}]}
				child {node {$v_\ell$} edge from parent[dotted,-{Stealth[length=.1pt]}]
					child {node {$x_\ell$} edge from parent[solid,-{Stealth[length=6pt]}]}
					child {node {$w$} edge from parent[solid,-{Stealth[length=6pt]}]};
				};
			};
		};
		
		\node (root) at (2,0) {$\rho$};
		\draw[bend right, dashed,-{Stealth[length=6pt]}] (root) to (u);
		
		\node at (3,0) {(1)};
		\draw (4,.25) to (4,-6.25);
	\end{tikzpicture}
	\begin{tikzpicture}[scale=0.7,every node/.style={scale=0.7}]
		\node (root) {$\rho$}
		child {node (v2) {$v_2$} edge from parent[-{Stealth[length=6pt]}]
			child {node {$v_3$} edge from parent[-{Stealth[length=6pt]}]
				child {node {$v_\ell$} edge from parent[dotted,-{Stealth[length=.1pt]}]
					child {node (w) {$w$} edge from parent[solid,-{Stealth[length=6pt]}]}
					child {node {$x_\ell$} edge from parent[solid,-{Stealth[length=6pt]}]}
				}
				child {node {$x_3$} edge from parent[-{Stealth[length=6pt]}]}
			}
			child {node {$x_2$} edge from parent[-{Stealth[length=6pt]}]}
		};
		\node[left of=v2, xshift=-35mm] (u) {$u$}
		child {node {} edge from parent[draw=none]}
		child {node (v1) {$v_1$} edge from parent[-{Stealth[length=6pt]}]
			child {node {$x_1$} edge from parent[-{Stealth[length=6pt]}]}
			child {node {} edge from parent[draw=none]}
		};
		\draw[-{Stealth[length=6pt]}] (v1) to (w);
		\draw[bend right, dashed,-{Stealth[length=6pt]}] (root) to (u);
		
		\node at (-.75,-.5) {$\w(u v_1)$};
		\node at (-3.25,-4.75) {$0$};
		
		\node at (1,0) {(2)};
	\end{tikzpicture}
	\caption{This figure in (1) depicts the path from a core-vertex~$v$ to another core vertex~$w$ and in~(2) an application of Reduction Rule~\ref{rr:Net-strings} to the path depicted in (1).}
	\label{fig:Net-strings}
\end{figure}
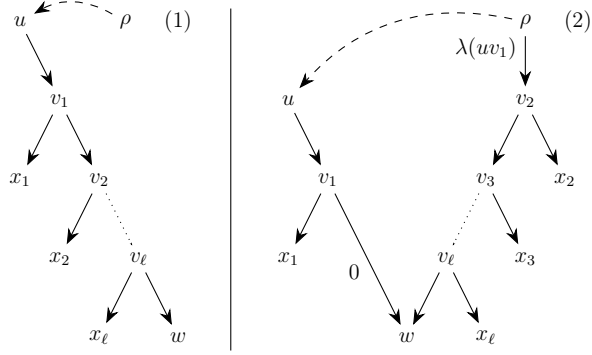%

\begin{lemma}
	\label{lem:Net-strings}
	Reduction Rule~\ref{rr:Net-strings} is correct and can be applied 
	in $\Oh(m)$ time.
\end{lemma}
\begin{proof}
	\proofpara{Correctness}
	Observe that if an edge $e \in E(\Net')\setminus\{\rho v_2, v_1w\}$ is affected by~$S \subseteq X$ in $\Net'$ then $e$ is also affected by $S$ in $\Net$.
	Furthermore, $\rho v_2$ is affected by~$S$ in~$\Net'$ if and only if $v_1v_2$ is affected by~$S$ in~$\Net$.
	Therefore, for every set~$S \subseteq X$ we have $\apPDsub{\Net'}(S) \le \apPD(S) + w(v_1w) = \apPD(S)$, and each solution for $\Instance'$ is also a solution for \Instance.
	
	We now show the converse.
	We first observe some facts.
	Let $x_i$ denote the leaf child of $v_i$ for each $i \in [\ell]$.
	We observe that due to Reduction Rule~\ref{rr:Net-a>a+b}, we may assume that $\w(v_ix_i) \ge \w(v_i v_{i+1}) + \w(v_{i+1} x_{i+1})$ for each~$i \in [\ell - 1]$.
	Consequently,
	\begin{eqnarray}
		\w(v_1x_1)
		&\geq& \w(v_2x_2) + \w(v_1 v_2)\\
		&\geq& \w(v_3x_3) + \w(v_2 v_3) + \w(v_2 v_1)\\
		&\geq& \w(v_i x_i) + \sum_{j = 1}^{i-1}\w(v_j v_{j+1})
	\end{eqnarray}
	for each $i \in [\ell - 1]$. 
	
	Let $S$ be a solution of \Instance.
	Observe that if $S$ contains $x_1$ or an offspring of~$w$, then~$\apPD(S) = \apPDsub{\Net'}(S)$.
	Similarly, if $S$ does not contain one of the taxa~$x_1, \dots, x_\ell$ then~$\apPD(S) = \apPDsub{\Net'}(S)$.
	So, assume now that $S$ does not contain~$x_1$ nor an offspring of~$w$ but~$x_i \in S$ for some~$i\in [\ell]$ with~$i>1$.
	Define~$S':= (S \cup \{x_1\}) \setminus \{x_i\}$.
	Then,
	\begin{eqnarray}
		\apPD(S') &\ge& \apPD(S) + \w(v_1x_1) - \w(v_i x_i) - \sum_{j = 1}^{i-1}\w(v_j v_{j+1}),
	\end{eqnarray}
	which by the observation above is at least $\apPD(S)$.
	As $x_1 \in S'$, we have that~$\apPDsub{\Net'}(S') = \apPD(S')$, and so $S'$ is  a solution for $\Instance'$.

	\proofpara{Running Time}
	To find the vertices~$u$ and~$w$ we iterate over the vertices in~$Q$ as vertex~$u$ and consider each outgoing edge as a path to~$u$.
	Therefore, all possible combinations of~$u$ and~$w$ are found in~$\Oh(m)$~time.
	Each operation can be executed in constant time.
\end{proof}

We now categorize the vertices in some sets to be able to easier refer to them.

\begin{definition}
	\label{def:top-strings}
	~
	\begin{propEnum}
		\item Define $A$ to be the intersection of $X$ with children of $\rho$.
		\item Define $B$ to be the set of vertices which can be reached from $R \cup (Q \setminus \{\rho\})$.
		\item Define $Y$ to be the side-vetices which are children of $\rho$ and define $Y_X$ to be the intersection of $X$ with the children of $Y$.
		\item Define $Z$ to be the side-vetices which are not in $B\cup Y$ and define $Z_X$ to be the intersection of $X$ with the children of $Z$.
		\item Define $a^* := \max_{x\in A} \w(\rho x)$ and $x^* := \arg\max_{x\in A} \w(\rho x)$.
		\item Define $c^* := \max_{y\in Y_X}\w(v_y y)$ 
		and $y^* := \arg\max_{y\in Y_X} \w(v_y y)$ where $v_y$ is the parent of $y$.
	\end{propEnum}
\end{definition}

In the following, for a taxon~$x\in X$, we use an operation~\emph{remove~$x$} when we delete~$x$ from~$X$ and~$V$ as well as the incoming edge of~$x$ from~$E$.
Further, we use an operation \emph{save~$x$} consisting of these steps:
\begin{itemize}
	\itemsep-.35em
	\item Reduce $k$ by $1$,
	\item reduce $D'$ by $\apPD(\{x\})$,
	\item delete~$x$ from~$X$,
	\item remove all edges in $E_{\{x\}}$, and
	\item identify all vertices with no incoming edge to a single root.
\end{itemize}

\begin{lemma}
	\label{lem:Net-save-remove}
	For a given instance \Instance of \MAPPD and a taxon~$x$,
	\begin{propEnum}
		\item \Instance has a solution~$S$ with~$x \in S$ if and only if~$S\setminus \{x\}$ is a solution for the instance after~$x$ is saved, and
		\item \Instance has a solution~$S$ with~$x \notin S$ if and only if~$S$ is a solution for the instance after~$x$ is removed.
	\end{propEnum}
\end{lemma}
\begin{proof}
	To see the first claim, let $\Net'$ be the network after saving $x$.
	Any edge in $\Net'$ affected by some $S \subseteq X\setminus\{x\}$ has a corresponding edge in $E(\Net)\setminus E_{\{x\}}$ with the same offspring, from which it follows that $\apPD(S\cup\{x\}) \geq \apPDsub{\Net'}(S) + D'$.
	Conversely, if $x \in S \subseteq X$ then any edge in $E_S \setminus E_{\{x\}}$ has a corresponding edge in $\Net'$ affected by $S \setminus \{x\}$, from which it followed that $\apPDsub{\Net'}(S\setminus \{x\}) \geq \apPD(S) - D'$.
	
	The second claim follows immediately from the definition of $\apPD(S)$.
\end{proof}

\begin{rr}
	\label{rr:Net-top-strings}
	If $k > |B| + |Y|$ and~$a^* > c^*$, then save~$x^*$.
	If $k > |B| + |Y|$ and~$a^* \le c^*$, then save $y^*$.
\end{rr}
\begin{lemma}
	\label{lem:Net-top-strings}
	Reduction Rule~\ref{rr:Net-top-strings} is correct and can be applied 
	in $\Oh(m)$ time.
\end{lemma}
\begin{proof}
	\proofpara{Correctness}
	Let $v_{1,1}, v_{2,1}, \dots v_{|Y|,1}$ denote the vertices of $Y$.
	Let~$v_{j,1} v_{j,2},\dots v_{j, \ell_j}$ be the side-vertices on the path from $v_{j,1}$ to a core vertex for each $j \in [|Y|]$.
	Let~$z_{j,i}$ be the leaf child of $v_{j,i}$ for each $j \in [|Y|]$, and~$i \in [\ell_j]$.
	This mapping is unique after~Reduction Rule~\ref{rr:Net-trees} has been applied exhaustively.
	Observe that $Z$ is the set~$\{v_{j,i} \mid 2 \leq j \leq |Y|, i \in [\ell_j]\}$.

	Similarly, we have $\w(v_{j,1} y_{j,1}) \geq \w(v_{j,i} y_{j,i}) + \sum_{h=1}^{i-1}\w(v_{j,h} v_{j,h+1})$ for each $j \in [|Y|]$, and each~$i \in [\ell_j]$.
	Consequently, we may assume that if $y_{j,i}$ is in a solution for some~$i > 1$, then so is~$y_{j,1}$.
	
	Furthermore, for any solution that contains $y_{j,i}$ and $y_{j,1}$ for $i>1$, we can assume the solution contains~$y^*$ as otherwise replacing $y_{j,i}$ with $y^*$ gives another solution, because~$w(v_{y^*} y^*) \geq\w(v_{j,i}y_{j,i})$ where~$v_{y^*}$ the parent of~$y^*$ and for each $j \in [|Y|]$, and each~$i \in [\ell_j]$.
	Thus, if a solution~$S$ contains any element of~$Z_X$ we can assume~$S$ also contains $y^*$.
	
	Now, suppose $k > |B| + |Y|$.
	Then, any solution~$S$ contains at least one element of~$A \cup Z_X$.
	This implies that $S$ contains at least one taxon of $A \cup \{y^*\}$ as~$S$ contains~$y^*$ if it contains any taxon in $Z_X$.
	If $a^* > c^*$, then $S$ contains $x^*$, as otherwise we could replace a taxon from~$(A\setminus \{x^*\}) \cup \{y^*\}$.
	So, in this case, $S$ contains $x^*$, and we can save $x^*$ by~Lemma~\ref{lem:Net-save-remove}.
	Otherwise, we may assume $S$ contains $y^*$, as otherwise we can replace an element from $A$ with $y^*$.

	\proofpara{Running Time}
	We can compute the size of~$B$ and~$Y$ and find $a^*$ and~$c^*$ in $\Oh(n)$ time.
	Saving $x^*$ or $y^*$ takes $\Oh(m)$ time.
\end{proof}

\begin{rr}
	\label{rr:Net-A}
	Let $x_1,\dots,x_{|A|}$ be the taxa in $A$ such that $\w(\rho x_i) \ge \w(\rho x_{i+1})$ for each $i\in [|A|-1]$.
	If $k > |A|$, then remove $x_{k+1},\dots,x_{|A|}$ and their incident edges from~$\Net$ if $|A| > k$.
\end{rr}
\begin{lemma}
	\label{lem:Net-A}
	Reduction Rule~\ref{rr:Net-A} is correct and can be applied 
	in $\Oh(n\log n)$~time.
\end{lemma}
\begin{proof}
	\proofpara{Correctness}
	Clearly, any solution for the instance~$\Instance'$ is also a solution for~\Instance.
	Therefore, let~$S$ be a solution for~$\Instance$.
	If $S \cap \{x_{k+1},\dots,x_{|A|}\} = \emptyset$, then $S$ is also a solution for~$\Instance'$.
	Assume that $x_i \in S$ for some~$i\in \{k+1, \dots, |A|\}$.
	As~$|S| \le k$ we conclude that there is a taxon~$x_j$ for $j\in [k]$ with $x_j \not\in S$.
	Because $\w(\rho x_j) \ge \w(\rho x_i)$, we conclude that $(S \cup \{x_j\}) \setminus \{x_i\}$ is also a solution for \Instance.
	Thus, we may assume that~$S \cap \{x_{k+1},\dots,x_{|A|}\} = \emptyset$.

	\proofpara{Running Time}
	We can sort $A$ in $\Oh(n\log n)$~time.
\end{proof}

\begin{rr}
	\label{rr:Net-pathlength}
	% 	Apply Reduction Rule~\ref{rr:Net-strings} exhaustively.
	
	Let $z_0,\dots,z_{k+1}\in Y\cup Z$ be vertices with $z_{i}$ being the parent of $z_{i+1}$ for $i\in [k]_0$.
	Add an edge~$z_{k-1} z_{k+1}$ of weight $\w(z_{k-1} z_{k}) + \w(z_{k} z_{k+1})$.
	Remove the vertices which are reachable by~$z_k$ but not by $z_{k+1}$ and the incident edges from~$\Net$.
\end{rr}
\begin{lemma}
	\label{lem:Net-pathlength}
	Reduction Rule~\ref{rr:Net-pathlength} is correct and can be applied in $\Oh(n)$~time.
\end{lemma}
\begin{proof}
	\proofpara{Correctness}
	Let~$x$ be a taxon which is reachable by~$z_k$ but not by $z_{k+1}$.
	Because we applied Reduction Rules~\ref{rr:Net-trees} and~\ref{rr:Net-a>a+b} exhaustively, each vertex $z_i$ contains at most one child~$x_i$ in~$X$ and~$\w(z_i x_i) > \w(z_{i+1} x_{i+1}) + \w(z_i z_{i+1})$.
	Consequently, we can assume that if~$x_i$ is in a solution then so are~$x_0, \dots, x_{i-1}$.
	Therefore, $x_{k}$ can not be in a solution, and an application of~Lemma~\ref{lem:Net-save-remove} is correct.
	
	\proofpara{Running Time}
	We can find an appropriate vertex $z_k$ in $\Oh(n)$ time and edit the network in constant time.
\end{proof}

Finally, we have everything to proof this section's main theorem.
\begin{proof}[Proof of Theorem~\ref{thm:Net-vkernel}]
	For a given instance~\Instance apply all of the reduction rules exhaustively to receive instance~$\Instance^* = (\Net^* = (V^*,E^*), k^*, D^*)$ of \MAPPD.
	We denote with~$R$, $Q$, and so the respective set in the original instance and with $R^*$, $Q^*$, and so the respective set in~$\Instance^*$.
	
	The correctness follows from the correctness of the reticulation rules.
	
	\proofpara{Running time}
	Observe that Reduction Rules~\ref{rr:Net-top-strings},~\ref{rr:Net-A}, and~\ref{rr:Net-pathlength} reduce $|X|$ by $1$, while Reduction Rules~\ref{rr:Net-trees} and~\ref{rr:Net-a>a+b} reduce $|X\setminus A|$ by $1$.
	As none of the reduction rules increase $|X|$ or $|X\setminus A|$, these rules are applied at most $|X| + |X\setminus A| \leq 2n$ times in total.
	For Reduction Rule~\ref{rr:Net-deg-2}, we note that only Reduction Rules~\ref{rr:Net-trees},~\ref{rr:Net-a>a+b},~\ref{rr:Net-top-strings},~\ref{rr:Net-A}, and~\ref{rr:Net-pathlength} can create a degree-2 vertex.
	Thus, every application of Reduction Rule~\ref{rr:Net-deg-2} occurs after one of these other rules, and thus Reduction Rule~\ref{rr:Net-deg-2} is applied at most~$2 \cdot |X|$ times.
	
	For Reduction Rule~\ref{rr:Net-strings}, we observe that each application of this rule reduces the number of side-paths starting at a non-root vertex of length of at least~2.
	None of the reduction rules increase this measure, and so the number of application of Reduction Rule~\ref{rr:Net-strings} is bounded by the number of side-paths in the original network, which is $\Oh(\eret)$ by Observation~\ref{obs:Net-core-retic-bound}.
	
	So, the total number of applications of all reduction rules is $\Oh(|X| + \eret)$.\lb
	Observe that a single application of any reduction rule is done in~$\Oh(m+ n\log n)$~time by Lemmas~\ref{lem:Net-trees},~\ref{lem:Net-deg-2},~\ref{lem:Net-a>a+b},~\ref{lem:Net-strings},~\ref{lem:Net-top-strings},~\ref{lem:Net-A}, and~\ref{lem:Net-pathlength}.
	We conclude that applying all rules exhaustively takes $\Oh((|X| + \eret) \cdot (m+ n\log n)\cdot \log m \cdot  \log\max_\w)$ time{, which can be summarized as~$\Oh(m^2 \log^2 m \cdot \log \max_\w)$}.
	
	\proofpara{Size}
	Each application of Reduction Rule~\ref{rr:Net-strings} increases the number of reticulation-edges by one.
	Therefore, observe that with \eret we refer to the parameter in the original network.
	
	Observe that none of the reduction rules (except for Reduction Rule~\ref{rr:Net-strings})
	change the number of reticulations or core vertices in the network.
	Reduction Rule \ref{rr:Net-strings} may turn some core vertices into reticulations but otherwise does not create new core vertices or reticulations.
	Therefore, we have $|R^*| + |Q^*| = |R| + |Q| \in \Oh(\eret)$ by Observation~\ref{obs:Net-core-retic-bound}.
	
	As Reduction Rule~\ref{rr:Net-strings} is exhaustively applied, we have that each side-path in~$\Net^*$ not leaving the root has at most one internal vertex, and this vertex has one leaf child. 
	There are $\Oh(\eret)$ side-paths in $\Net^*$ by Observation~\ref{obs:Net-core-retic-bound} and the fact that none of the reduction rules increase the number of side-paths not leaving the root.
	Thus, we have that the total number of vertices reachable from $R^* \cup (Q\setminus \rho)$ is at most $|R^*| + |Q^*| + \Oh(\eret) = \Oh(\eret)$. That is, $|B^*| \in \Oh(\eret)$.
	
	The size of~$Y^*$ is the number of paths from~$\rho$ to a core-vertex or a reticulation.
	There are at most~$|Q| + \eret$ such paths in the original instance.
	Each application of Reduction Rule~\ref{rr:Net-strings} adds one such path.
	We saw that~$|Q| \in \Oh(\eret)$ and Reduction Rule~\ref{rr:Net-strings} can be applied~$\Oh(\eret)$ times.
	We conclude~$|Y^*| \in \Oh(\eret)$.
	
	After Reduction Rule~\ref{rr:Net-top-strings} has been applied exhaustively, we may conclude\lb that~$k^* \le |B^*| + |Y^*| \in \Oh(\eret)$.
	After Reduction Rule~\ref{rr:Net-A} has been applied exhaustively, there are at most~$k^*$ vertices in~$A^*$
	and 
	after Reduction Rule~\ref{rr:Net-pathlength} has been applied exhaustively, each path in~$Z^*$ has length of most~$k^*-1$ such that~$|Y^*| + |Z^*| \le |Y^*| \cdot k^* \in \Oh(\eret^2)$.
	We conclude~$|V^*| \in \Oh(\eret^2)$.
	
	We have~$\eretwithoutN_{\Net^*} \le \eret + |Q^*| + |R^*| \in \Oh(\eret)$
	because~$\eret = |E| - |V|$ in any network.
	Hence, we conclude
	$|E^*| = \eretwithoutN_{\Net^*} + |V^*| \in \Oh(\eret^2)$.
\end{proof}

From Theorem~\ref{thm:Net-vkernel}, we have that in polynomial time we can reduce any instance~\Instance of \MAPPD to an equivalent instance $\Instance^* = (\Net^* = (V^*,E^*,{\w^*}), k^*, D^*)$ in which~$|V^*|$, $|E^*|$, and~$k^*$ are all bounded by a polynomial in $\eret$.
This does not guarantee a polynomial kernel, as the encoding size of $D^*$ or $\max_{\w^*}$
could be much larger than~$|V^*|$ or~$|E^*|$.
Fortunately, we can apply a result of \cite{etscheid,frank} to bound these values, as follows.

Let~$e_1,\dots,e_{m^*}$ be an order of the edges after applying all reduction rules.
We define~$w_i := \w^*( e_i )$ for each~$i\in [m^*]$ and~$W := D^*$.
In polynomial time, we can compute positive numbers~$w_1', \dots, w_{m^*}'$, and~$W'$ such that the total encoding length is~$\Oh( (m^*)^4 ) \in \Oh( \eret^8 )$ with~$\sum_{i \in S} w_i \ge W$ if and only if~$\sum_{i \in S} w_i' \ge W'$ for every~$S \subseteq [m^*]$ by~\cite[Corollary~2]{etscheid}.

We directly conclude the following.

\begin{theorem}
	\label{thm:Net-kernel}
	\MAPPD admits a polynomial size kernelization algorithm for the number of reticulation-edges~\eret.
\end{theorem}

\section{Hardness of Max-Net-PD}
\label{sec:Net-reduction-MaxNPD}
In this section, we examine another measure of phylogenetic diversity on phylogenetic networks than \MAPPD, namely \MaxNPD.
Recall that in \MaxNPD, each reticulation-edge is given an \iprop.
Similar to the computation of diversity in \GNAP, also in \MaxNPD we consider \emph{expected} phylogenetic diversity.

In this section, we present a polynomial-time reduction from \ucNAP to \MaxNPD, in which the level of the phylogenetic network is~1.
Recall that in \ucNAP, we are given a phylogenetic tree, a \sprop for every taxon, and two integers~$k$ and~$D$,
and it is asked whether there is a set~$S$ of size at most~$k$ taxa such that has an expected phylogenetic diversity of at least~$D$.

By Theorem~\ref{thm:GNAP-ucNAP-hardness}, \ucNAP is \NP-hard even if the phylogenetic tree has a height of~2.
We conclude the following.

\begin{theorem}
	\label{thm:Net-level-1}
	\MaxNPD is \NP-hard even if the input network has a level of~1 and the distance between the root and each leaf is~4.
\end{theorem}

\begin{proof}
	By Theorem~\ref{thm:GNAP-ucNAP-hardness},
	\ucNAP is \NP-hard on trees with a height of~2.\lb
	Let \Tree be an $X$-tree with a height of~2 for some set of taxa $X$ and let an\lb instance~$\Instance = (\Tree, \w, q, k, D)$ of \ucNAP be given.
	
	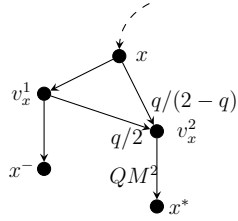
\begin{figure}[t]
		\centering
		\begin{tikzpicture}[scale=1,every node/.style={scale=0.7}]
			\tikzstyle{txt}=[circle,fill=white,draw=white,inner sep=0pt]
			\tikzstyle{nde}=[circle,fill=black,draw=black,inner sep=2.5pt]
			\node[txt] at (5.2,2.9) {$QM^2$};
			%\foreach \v/\n/\x/\y in {l/x/5/4.5, }
			\node[txt] at (5.1,3.4) {$q/2$};
			\node[txt] at (6,3.9) {$q/(2-q)$};
			
			\node[txt] at (5.3,4.5) {$x$};
			\node[txt] at (3.7,4) {$v_x^{1}$};
			\node[txt] at (5.9,3.5) {$v_x^{2}$};
			\node[txt] at (3.7,3) {$x^-$};
			\node[txt] at (5.8,2.5) {$x^*$};
			
			\node[nde] (l) at (5,4.5) {};
			\node[nde] (vl1) at (4,4) {};
			\node[nde] (vl2) at (5.5,3.5) {};
			\node[nde] (lm) at (4,3) {};
			\node[nde] (ls) at (5.5,2.5) {};
			
			\draw[dashed,-stealth] (5.4,5.2) to[bend right] (l);
			\draw[-stealth] (l) to (vl1);
			\draw[-stealth] (l) to (vl2);
			\draw[-stealth] (vl1) to (vl2);
			\draw[-stealth] (vl1) to (lm);
			\draw[-stealth] (vl2) to (ls);
		\end{tikzpicture}
		\caption{Illustration of the leaf-gadget. Omitted edge-weights are~1 and $q(x)$ is abbreviated to $q$.}
		\label{fig:Net-leaf-gadget}
	\end{figure}%
	We define a leaf-gadget which is illustrated in Figure~\ref{fig:Net-leaf-gadget}.
	Let $x\in X$ be a leaf of~\Tree with \sprop~$q(x)$.
	%We assume that $p(x) > 0$ as otherwise we can remove $x$ from the tree.
	Add four vertices $v_x^{1}$, $v_x^{2}$, $x^*$, and~$x^-$, and edges~$x v_x^{1}$, $x v_x^{2}$, $v_x^{1}v_x^{2}$, $v_x^{1}x^-$, and~$v_x^{2}x^*$.
	The only reticulation in this gadget is $v_x^{2}$ with incoming edges $x v_x^{2}$ and $v_x^{1}v_x^{2}$.
	We set \iprops of these reticulation edges~$p(x v_x^{2}) := q(x) / (2 - q(x))$ and~$p(v_x^{1}v_x^{2}) := q(x) / 2$ which are both in~$\mathbb{R}_{[0,1]}$ because~$q(x) \in \mathbb{R}_{[0,1]}$.
	
	Let \Net be the network which results from replacing each leaf of \Tree with the corresponding leaf-gadget.
	The leaves of \Net are $X' := \{ x^*, x^- \mid x \in X \}$.
	Let $d$ denote the largest denominator of a \sprop $q(x)$ for some~$x \in X$, so that every~$q(x)$ is expressible as~$c'/d'$ for some pair of integers~$c$, and~$d$ such that~$d' \leq d$. 
	%\todo{MJ: not sure what notation is best here}
	Let~$M$ and~$Q$ be large integers, such that $M$ is bigger than $\PDsub{\Tree}(X) \ge |X| \ge k$, 
	and $Q\cdot D$ and $Q\cdot d^{-k}$ are both bigger than~3.
	% and $Q\cdot D$ and $Q\cdot 2^{-(\max_x \text{bin-enc}(q(x))) \cdot k}$ are both bigger than~3.

	Observe that the number of bits necessary to write $M$ and $Q$ is polynomial in the size of~\Instance.
	We set the weight of edges $e \in E(\Tree)$ in \Net to $\w'(e) = kQ \cdot \w(e)$.
	For each taxon~$x\in X$, set~$\w'(v_x^{2}x^*) := Q\cdot M^2$ and $\w'(e) := 1$ for each~$e \in \{x v_x^{1}, x v_x^{2}, v_x^{1}v_x^{2}, v_x^{1}x^-\}$.

	Finally, let $\Instance' := (\Net, \w', p, k, D' := kQ(M^2 + D))$ be an instance of \MaxNPD.
	Each leaf-gadget is a level-1-network.
	As the leaf-gadgets are connected by a tree, \Net is a level-1-network.
	Recall that the height of the tree \Tree is 2, and as such the distance between the root and each leaf in in \Net  is~4.
	
	We first show that $\gam{Z}(e) = q(x)$ for the edge~$e$ incoming at~$x$, if $Z$ contains~$x^*$ but not~$x^-$.
	Indeed, because $x^-\not\in Z$, we conclude~$\gam{Z}(x v_x^{2}) = p(x v_x^{2}) = q(x) / (2 - q(x))$ and~$\gam{Z}(x v_x^{1}) = \gam{Z}(v_x^{1}v_x^{2}) = p(v_x^{1}v_x^{2}) = q(x) / 2$.
	Subsequently, 
	\begin{eqnarray}
		\label{eqn:gamma}
		\gam{Z}(e)
		&=& 1 - \left(1 - \gam{Z}(x v_x^{1})\right)\left(1 - \gam{Z}(x v_x^{2})\right)\\
		\nonumber
		&=& 1 - \frac{2 - q(x)}{2} \cdot \frac{2 - 2q(x)}{2 - q(x)}\\
		\nonumber
		&=& 1 - \frac{2 - q(x)}{2 - q(x)} \cdot \frac{2 - 2q(x)}{2}\\
		\nonumber
		&=& 1 - \frac{2 - 2q(x)}{2}
		= q(x).
	\end{eqnarray}
	
	We now show that if \Instance is a \yes-instance of \ucNAP, then $\Instance'$ is a \yes-instance of \MaxNPD.
	Afterward, we show the converse. 
	
	Suppose that \Instance is a \yes-instance of \ucNAP and that $S\subseteq X$ is a solution of \Instance, that is~$|S|\le k$ and~$\PDsub{\Tree}(S) \ge D$.
	Let $S' := \{ x^* \mid x \in S \}$ be a subset of~$X'$.
	Clearly, $|S'| = |S| \le k$.
	Because~\Tree does not contain reticulation edges and~$\gam{Z}(e) = q(x)$ with~$e$ being the edge incoming at $x$, we conclude that
	\begin{equation}
		\NetPD(S')
		\ge kQ\cdot \PDsub{\Tree}(S) + k\cdot \w'(v_x^{2}x^*)
		\ge kQ\cdot (D + M^2)
		= D'.
	\end{equation}
	Hence, $S'$ is a solution of $\Instance'$.
	
	Conversely,
	let $S'$ be a solution of $\Instance'$.
	We define $S^- = S \cap \{ x^- \mid x\in X \}$ and~$S^* = S \cap \{ x^* \mid x\in X \}$.
	Towards a contradiction, assume that~$S^-$ is non-empty.
	Then however, using $3 < Q\cdot D$,
	\begin{align*}
		\NetPD(S')
		& \le  \sum_{x^- \in S^-}(\w'(v_x^{1}x^-) + \w'(x v_x^{1})) +  |S^*| \cdot (QM^2 + 3)
		+ \sum_{e\in E(\Tree)} \w'(e)\\
		& \le 2 \cdot |S^-| + |S^*| \cdot (QM^2 + 3) +  kQM \\
		& \le  2 + (k-1)(QM^2 + 3) + kQM \\
		&	< kQM^2 - QM^2 + kQM + 3k\\
		&   < k(QM^2 + 3)
		< k(QM^2 + QD)
		= D'
	\end{align*}
	contradicts that $S'$ is a solution.
	%Recall that $3 < QD$ by the definition of $Q$.\todo{MW: this comes too late}
	Therefore, we conclude that $S' \subseteq \{ x^* \mid x\in L \}$ and we assume that the size of~$S'$ is~$k$.
	Define $S := \{ x \mid x^*\in S' \}$.
	Subsequently, with Equation~\eqref{eqn:gamma} we conclude 
	\begin{align*}
		kQ(M^2 + D) 
		= D' 
		& \le \NetPD(S') \\
		& = k\cdot QM^2 + \sum_{x \in S}\bigg(\underbrace{\frac{q(x)}2 + \frac{q(x)}2 + \frac{q(x)}{2-q(x)}}_{\leq 3}\bigg) + kQ\cdot \PDsub{\Tree}(S).
	\end{align*}
	It follows that $\PDsub{\Tree}(S) \ge \frac{1}{kQ} \cdot (kQ(M^2 + D) - kQM^2 - 3k) = D - 3/Q$.
	
	It remains to show that $\PDsub{\Tree}(S)$ cannot take any values between~$D - 3/Q$ and~$D$, and therefore~$\PDsub{\Tree}(S) \ge D - 3/Q$ implies~$\PDsub{\Tree}(S) \ge D$.
	To this end, let $c_x$, and~$d_x$ be the unique positive integers such that $q(x) = c_x/d_x$ for each taxon $x \in X$.
	Then, $q(x)$ and $(1 - q(x))$ are multiples of $1/d_x$, by construction.
	It follows that for any edge~$e$ of the tree $\Tree$ that $\gamma'_S(e) = (1 - \prod_{x \in \off(e) \cap S}(1- q(x)))$ is a multiple of $1/(\prod_{x \in S}d_x)$.
	As all edge-weights are integers, we also have that $\PDsub{\Tree}(S)$ is a multiple of $1/(\prod_{x \in S}d_x)$. It follows that either $\PDsub{\Tree}(S) \geq D$ or $D - \PDsub{\Tree}(S) \geq 1/(\prod_{x \in S}d_x)$.
	Because~$d_x \leq d$ for any $x$, this difference is at least $d^{-k} > 3/Q$.
	It follows that if $\PDsub{\Tree}(S) \ge D - 3/Q$ then in fact $\PDsub{\Tree}(S) \ge D$.
	
	% We observe $3/Q < 2^{- \max_x \text{bin-enc}(q(x)) \cdot {k}}$ by the definition of $Q$.
	% The probability that at least one taxon out of a set of taxa with size at most~$k$ survives can be written with $\max_x \text{bin-enc}(q(x)) \cdot k$ bits and is therefore too big to represent a difference of $3/Q$ in diversity, as the weights on edges are integers.
	% \todo[inline]{MW: this is way too hand-wavy and I did not understand it}
	% Therefore, $\PDsub{\Tree}(Y)\ge A$ for each set $Y$ of leaves with $|Y|\le k$ and $\PDsub{\Tree}(Y)\ge A - 3/Q$ for any $A$.
	We conclude $\PDsub{\Tree}(S)\ge D$.
	And therefore with $|S| = |S'| = k$ we follow that $S$ is a solution of~\Instance.
	Hence, \Instance is a \yes-instance of \ucNAP.
\end{proof}

% Mark's Conjecture:
% Conj: Any solution for PD instance either
% 1. corresponds to a valid solution for unit-cost-NAP, OR
% 2. has score < D' - delta.
% Argument: any suboptimal solution to ucNAP has
% score < D - epsilon
% so any* suboptimal solution for networkPD has
% score < Q(D-epsilon)+ 2k +kM
% whereas an opt solution has
% score >= QD + 0 + kM

% Choose Q such that Q*epsilon >= 2k + delta
% Set D' = kM + QD

% A set of *-leaves corresponding to a suboptimal solution for ucNAP gets score
% < kM + Q(D-epsilon) + 2k
% = kM + QD - Q*epsilon + 2k
% <= kM + QD -2k - delta + 2k
% = kM + QD - delta

% *"any suboptimal solution for networkPD" := any solution corresponding to a solution for ucNAP i.e. one where we choose k *-leaves.

% \sum_{chosen i} (p_i + p_i/(2-p_i)) < 2k

% Question: how small is epsilon? how big does Q need to be?

\section{Discussion}
\label{sec:Net-discussion}
In this chapter, we considered two problems, \MAPPDlong and~\MaxNPD, in maximizing phylogenetic diversity in networks, and analyzed them within the framework of parameterized complexity.

We showed that \MAPPD is \Wh 2-hard parameterized by $k$, the size of the solution.
We further were able to show an equivalence between \MAPPD parameterized by $k$ and \wpSClong parameterized by the size of the solution.
Thus, establishing the exact complexity class of \wpSClong, would also establish the exact complexity class of \MAPPD.
On the positive side, we showed that \MAPPD is \FPT when parameterized with the number of reticulations~$\vret$ and with respect to the treewidth~$\twN$ of the network.
We further showed that \MAPPD admits a kernelization algorithm with respect to the number of reticulation-edges~$\eret$.
Finally, we have shown that \MaxNPD remains \NP-hard on phylogenetic networks with a level of~1.

Because $\vret$ is smaller than~$\eret$, it is natural to ask if the kernelization result can also be shown for $\vret$.
We also ask whether \MAPPD parameterized by $k$ is~\Wh 2-complete.

For \MaxNPD we also raise some questions.
We have proven that \MAPPD is \FPT when parameterized by the threshold of diversity~$D$ and the acceptable loss of diversity~$\Dbar$.
Also, we showed that \MAPPD admits a kernelization algorithm of polynomial size with respect to \eret.
It is natural to ask, if these positive results transfer to \MaxNPD.
Further, based on the presented hardness for \MaxNPD we ask for some improvements.
First of all, can \MaxNPD be solved in pseudo-polynomial time on level-1-networks?
Secondly, is \MaxNPD polynomial time solvable on level-1-networks if we require the network to be ultrametric, i.e. when all root-leaf paths have the same length?
Finally, does \MaxNPD{} remain \Wh{1}-hard when the parameter is the number of species~$k$ to save plus the level of the network?

\chapter{Conclusion}
\label{ctr:conclusion}
We studied several problems concerning the maximization of preserved phylogenetic diversity within the constraint of limited resources, which allow for the selection of only a few taxa.
The considered problem definitions model biological processes more realistically compared to the basic problem of maximization phylogenetic diversity,~\MPD.
For these problem definitions, within the framework of parameterized algorithms, we presented several lower and upper running time bounds for algorithms.

Before presenting a broader perspective on this work, we will briefly review the different problem definitions and the results obtained.
Finally, we will offer an outlook on potential future research directions.

\section{Summary of Problems and Results}
The first problem we considered is \GNAPLong, in Chapter~\ref{ctr:GNAP}.
With \GNAP one is able to model different costs for saving taxa and also an uncertainty as to whether an intervention actually saves a taxon.

We showed that \GNAP is \XP with respect to the number of unique costs plus the number of unique \sprops.
We further proved that \GNAP is \Wh{1}-hard with respect to the number of taxa.
We also showed that \ucNAP, the special case of \GNAP in which all projects with a positive \sprop have a cost of~1, is \NP-hard even on phylogenetic trees with a height of~2.
It remains open if~\GNAP, or even \ucNAP, is strongly \NP-hard or if \GNAP can be solved in pseudo-polynomial running time.

In Chapter~\ref{ctr:TimePD} we considered \tPDs and \tPDws.
These problems are motivated by the concern that some taxa need to be treated earlier than others because not all taxa go extinct at the same point in time.
As it is necessary for a solution of an instance of either of the two problems to indicate how to schedule the available time in order to save all of these taxa,
a connection to the field of scheduling is given.

Both problems, \tPDs and \tPDws, are \FPT with respect to the target diversity $D$.
Further, \tPDs is also \FPT when parameterized by the acceptable loss of phylogenetic diversity $\Dbar$.
In contrast, this result does not hold for \tPDws, unless \PneqNP.
It remains open whether \tPDs is solvable in pseudo-polynomial time.

With \MPD it is not possible to check whether taxa in a selected set have crucial dependencies to taxa which are not selected.
In \PDDlong, which we studied in Chapter~\ref{ctr:FoodWebs}, one not only searches for a phylogenetically diverse set of taxa but also for one where in the given food-web each taxon either is a source or has an edge from another saved taxon.
Therefore, with \PDD it is possible to compute a small set of taxa that has a phylogenetic diversity above a certain threshold and each selected is self-sufficient in the ecological system or feeds on the saved set of taxa.

We presented algorithms with a color-coding approach to show that \PDD is \FPT when parameterized with the solution size plus the height of the phylogenetic tree.
\sPDD---the special case of \PDD in which the phylogenetic tree is a star---is \Wh{1}-hard when parameterized by the acceptable loss of phylogenetic diversity.
We also considered the structure of the food-web.
Here, we proved that \PDD remains \NP-hard if every connected component of the food-web is a star or a clique.
However, \PDD is \FPT with respect to the distance to co-cluster of the food-web.
We showed further that \sPDD is \FPT with respect to the treewidth of the food-web.
It is open whether \PDD is \FPT with respect to the size of the solution,
or if \PDD can be solved in polynomial time if every connected component in the food-web contains at most two vertices.

In some constellations, the evolutionary history of taxa may be explained better with a phylogenetic network than with a phylogenetic tree.
Models for generalizing phylogenetic diversity on explicit phylogenetic networks have been defined rather recently~\cite{WickeFischer2018,bordewichNetworks}.
In Chapter~\ref{ctr:Networks}, we considered two corresponding problems,~\MAPPDlong and~\MaxNPD.

We presented reductions to showcase an equivalence between \MAPPD and a special case of \SC, both parameterized by the size of the solution and thereby establishing a \Wh{2}-hardness of \MAPPD with respect to~$k$.
We further showed that \MAPPD is \FPT when parameterized with the number of reticulations or with the treewidth of the network.
Finally, \MaxNPD remains \NP-hard on phylogenetic networks of level~1.
These results have the practical implications that exact solutions can be computed in a reasonable time if the phylogenetic network is relatively treelike.
It remains open whether \MAPPD admits a kernelization of polynomial size for $\vret$.
Further, we wonder if \MaxNPD is \FPT when parameterized by~$k$ plus the level of the phylogenetic network.

\section{A Broader View on Our Results}
Let us now provide a broader view of the results of this work.
Firstly, we want to observe that the regarded problems can be distinguished between two categories.
While \BNAP, \tPDs, and \tPDws are primarily extending \MPD in a direction to better model interventions of men,
\PDD, \MAPPD, and \MaxNPD solely model biological processes.

The algorithms that we presented are predominantly dynamic programming algorithms or use the technique of color-coding.
It is rather not surprising to have a lot of dynamic programming approaches when dealing with a tree-structure.
In the color-coding algorithms, we mostly used perfect hash-families and sometimes universal sets.
While these techniques provide the desired \FPT-results, representative sets~\cite[Chapter~12.3]{cygan,monien1985} actually yield faster running times for other problems.
We, therefore, ask whether it is possible to apply representative sets to at least improve some of the running times.

We observe that certain hardness results, both~\NP-hardness and \Wh{$i$}-hardness, were mostly given even on phylogenetic trees or networks that only have a constant height.
On the contrary, it seemed that if a certain algorithm works on trees of height three or four, then the approach also generalizes to arbitrary heights.
As phylogenetic trees tend to have significant height in practice, this is a rather positive observation because some algorithmic ideas are not strongly impeded by tall trees.

The \ucNAP problem presented itself as an interesting problem for reductions.
This is for example utilized in Theorem~\ref{thm:Net-level-1}.
We have been able to prove that \ucNAP is \NP-hard in Theorem~\ref{thm:GNAP-ucNAP-hardness}.
However, we did not succeed in presenting a pseudo-polynomial running-time algorithm for \ucNAP, nor in showing that \ucNAP is strongly \NP-hard.
If \ucNAP was strongly \NP-hard, then we could also exclude pseudo-polynomial running-time algorithms for \GNAP and \MaxNPD restricted to level-1-networks, unless~\PeqNP.

We now have a look into some specific parameters considered in this work.
The parameterization by the acceptable loss of phylogenetic diversity, \Dbar, is to the best of our knowledge introduced by us.
Considering the fact that we are probably interested in preserving the majority of available phylogenetic diversity, we may assume \Dbar to be relatively small in real-world instances.
Considering \Dbar as a parameter proved to be fruitful, as \tPDs and \MAPPD are \FPT with respect to~\Dbar.
Both these algorithms use the technique of color-coding.
Theorem~\ref{thm:timePD-DBar}, in which we prove that \tPDs is \FPT when parameterized with \Dbar, is arguably the most technically advanced algorithm in this thesis.

Most of the considered problems are \Wh{1}-hard with respect to the size of the solution.
This results from the fact that these problems are generalizations of \KP or \SC, which are known to be \Wh{1}-hard for the solution size.
Further, most of the problems can be solved with trivial $\Oh(2^{|X|})$ running-time algorithms.
One can wonder if any improvements can be made to these algorithms.
At least for \MAPPD, we know that the algorithm running in~$\Oh(2^{|X|})$ time is tight under \SETH, by~Corollary~\ref{cor:Net-X}.

\section{Research Ideas Based on This Work}
In the remainder of the discussion, we want to delve a bit into general ideas for further research.
Within this thesis, we presented many algorithmic ideas and discussed their correctness and running time upper bounds.
Given the theoretical nature of this work, it would be valuable to evaluate the practical efficiency of these algorithms.
Therefore, one would need to implement some of the algorithmic ideas to evaluate their practical efficiency.
For instance, the branching algorithm for \MAPPD (Theorem~\ref{thm:Net-ret}) and the algorithms for the parameterization with~$D$ (Theorems~\ref{thm:timePD-D} and~\ref{thm:PDD-D}) or with~$\Dbar$ (Theorems~\ref{thm:timePD-DBar} and~\ref{thm:Net-Dbar}) would be recommendable.
Implemented algorithms could do both, give further insight into the complexity of the problems, and also act as a bridge between theory and practical decision-making in conservation projects.

In some real-world preservation actions, decision-makers choose to preserve natural reservoirs instead of specific taxa, for example be due to political pressure.
In \textsc{Optimizing Diversity with Coverage} (\textsc{ODC}), one may select regions for protection in which all taxa are preserved~\cite{moulton}.
\textsc{ODC} is \NP-hard by a reduction from \SC~\cite{moulton} and a generalization has been studied from an approximation point-of-view~\cite{bordewich2008,BudgetedNatureReserveSelection}.
It would be interesting to see whether \textsc{ODC} is \FPT for reasonable parameters.

Naturally, one could ask to combine the aspects of the individual problems.
Given the presented individual hardness results for the problems, this attempt seems rather not so fruitful---and extremely technical.
Nevertheless, the consideration of ecological constraints, as done in \PDD, is indeed important to consider in conservation interventions.
It could present an interesting study to combine these viability constraints with phylogenetic networks.
Do the resulting problems become much harder than~\PDD?

In Chapter~\ref{ctr:Networks}, we considered two existent definitions of phylogenetic diversity on phylogenetic networks.
We have not discussed how well these problems model biological processes.
One of the most important tasks in this field now would be to assess which variants of phylogenetic diversity on networks are biologically most accurate.
This could of course depend on the type of species considered and in particular on the type of reticulate evolutionary events.
Even if the maximization problem cannot be solved efficiently, having a good measure of phylogenetic diversity can still be of great practical use by measuring how diverse a given set of species is.

Finally, we observe that we have considered problems in maximizing phylogenetic diversity.
When it comes to the preservation of taxa we are arguably very interested in the different features of a selected set of taxa and therefore rather in feature or functional diversity.
Computing functional diversity, however, is even impossible in cases~\cite{MPC+18} and so we use phylogenetic diversity as somewhat fitting albeit imperfect proxy of functional diversity~\cite{winter2013}.
We wonder if in special cases parameterized algorithms could be a tool to handle the intractability.
Further, it has been reported that maximizing phylogenetic diversity is only marginally better than selecting a random set of species when it comes to maximizing the functional diversity of the surviving species~\cite{MPC+18}.
The situation could be different, however, when some constraints of this thesis, especially the biological constraints considered in the last two chapters, are incorporated.
Here, investigating the following two questions seems fruitful:
First, do randomly selected viable species sets have a higher functional diversity than randomly selected species?
Second, do viable sets with maximal phylogenetic diversity have a higher functional diversity than randomly selected viable sets?

\newpage\null\thispagestyle{empty}\newpage

\bibliographystyle{alpha}

\newcommand{\etalchar}[1]{$^{#1}$}

\end{document}